\documentclass[runningheads]{llncs}
\usepackage{amsmath}%
\usepackage{amssymb}%
\usepackage{stmaryrd}%
\usepackage{array}%
\usepackage{epic}%
\usepackage[fixamsmath]{mathtools}
\usepackage{cancel}
\usepackage{url}
\usepackage{bigstrut}
\usepackage{wrapfig}

\DeclareMathAlphabet\mathbb{U}{msb}{m}{n}
\newtheorem{fact}{Fact}
\DeclareMathAlphabet{\mathrm}    {OT1}{cmr}{m}{n}
\DeclareMathAlphabet{\mathrmbf}  {OT1}{cmr}{bx}{n}
\DeclareMathAlphabet{\mathrmit}  {OT1}{cmr}{m}{it}
\DeclareMathAlphabet{\mathrmbfit}{OT1}{cmr}{bx}{it}
\DeclareMathAlphabet{\mathsf}    {OT1}{cmss}{m}{n}
\DeclareMathAlphabet{\mathsfbf}  {OT1}{cmss}{bx}{n}
\DeclareMathAlphabet{\mathsfit}  {OT1}{cmss}{m}{sl}
\DeclareMathAlphabet{\mathtt}    {OT1}{cmtt}{m}{n}
\DeclareMathAlphabet{\mathttbf}  {OT1}{cmtt}{bx}{n}
\DeclareMathAlphabet{\mathttit}  {OT1}{cmtt}{m}{it}
\DeclareMathAlphabet{\mathpzc}   {OT1}{pzc}{m}{it}
\newcommand{\keywords}[1]{\par\addvspace\baselineskip\noindent\enspace\ignorespaces{\bfseries Keywords:\,}#1}
\newcommand{\comment}[1]{}
\begin{document}

\pagestyle{headings}
\title{The {\ttfamily FOLE} Table} 
\titlerunning{The {\ttfamily FOLE} Table}  
\author{Robert E. Kent}
\institute{Ontologos}
\maketitle

\begin{abstract}
This paper discusses the representation of ontologies
in the first-order logical environment {\ttfamily FOLE}
(Kent~\cite{kent:iccs2013}).
An ontology defines the primitives with which to model the knowledge resources for a community of discourse
(Gruber~\cite{gruber:eds2009}).
These primitives, consisting of classes, relationships and properties,
are represented by the entity-relationship-attribute {\ttfamily ERA} data model
(Chen~\cite{chen:76}). 
An ontology uses formal axioms to constrain the interpretation of these primitives. 
In short, an ontology specifies a logical theory.
A series of three papers
provide a rigorous mathematical representation
for the {\ttfamily ERA} data model in particular,
and ontologies in general,
within the first-order logical environment {\ttfamily FOLE}.
The first two papers, 
which provide a \emph{foundation} and \emph{superstructure} for {\ttfamily FOLE},
represent the formalism and semantics of (many-sorted) first-order logic in a classification form
corresponding to ideas discussed in the Information Flow Framework 
(IFF~\cite{iff}).
The third paper (Kent~\cite{kent:fole:era:interp}) will define an \emph{interpretation} 
of {\ttfamily FOLE} in terms of the transformational passage,
first described in Kent~\cite{kent:iccs2013},
from the classification form of first-order logic
to an equivalent interpretation form,
thereby defining the formalism and semantics of first-order logical/relational database systems. 
Two papers will provide a precise mathematical basis for \texttt{FOLE} interpretation:
the current paper develops the notion of a \texttt{FOLE} relational table
following the relational model (Codd~\cite{codd:90}), and
a follow-up paper will develop the notion of a \texttt{FOLE} relational database.
Both of these papers expand on material found in the paper (Kent~\cite{kent:db:sem}).
%
Although 
the classification form (and \texttt{FOLE} itself)
follow the entity-relationship-attribute data model of Chen,
the interpretation form incorporates the relational data model of Codd.
In general,
the {\ttfamily FOLE} representation uses a conceptual structures approach,
that is completely compatible 
with formal concept analysis (Ganter and Wille~\cite{ganter:wille:99}) 
and information flow (Barwise and Seligman~\cite{barwise:seligman:97}).
\keywords{signature, type domain, signed domain, table.}
\end{abstract}

\tableofcontents


\section{Introduction}\label{sec:intro}

\subsection{Philosophy.}\label{sub:sec:phil}

The relational model is an approach to information management 
using the semantics and formalism of first-order predicate logic. 
%
\footnote{``The relational model for database management : version 2'' by E.F. Codd \cite{codd:90}.} 
%
The first-order logical environment \texttt{FOLE} is a framework for defining the semantics and formalism of 
logic and databases in an integrated and coherent fashion. 
Hence,
the relational model for information management can be framed in terms of 
the first-order logical environment \texttt{FOLE}.


\subsection{Background}\label{sub:sec:background}

The author's ``Systems Consequence'' paper (Kent~\cite{kent:iccs2009}) 
is a very general theory and methodology for specification and inter-operation of systems of information resources. 
The generality comes from the fact that it is independent of the logical/semantic system (institution) being used. 
This is a wide-ranging theory, 
based upon ideas from 
information flow (Barwise and Seligman~\cite{barwise:seligman:97}), 
formal concept analysis (Wille and Ganter et al~\cite{ganter:wille:99}), 
the theory of institutions (Goguen et al~\cite{goguen:burstall:92}), and 
the lattice of theories notion (Sowa~\cite{sowa:kr}), 
for the integration of both formal and semantic systems independent of logical environment. 
In order to better understand the motivations of that paper and to be able more readily to apply its concepts, 
in the future it will be important to study system consequence in various particular logical/semantic systems. 
This paper aims to do just that for the logical/semantic system of relational databases. 
The paper, 
which was inspired by and which extends a recent set of papers on the theory of relational database systems 
(Spivak~\cite{spivak:sd},\cite{spivak:fdm}), 
is linked with work on the Information Flow Framework (IFF~\cite{iff}) connected with the ontology standards effort (SUO), 
since relational databases naturally embed into first order logic. 
We offer both an intuitive and a technical discussion. 
Corresponding to the notions of primary and foreign keys, 
relational database semantics takes two forms: 
a distinguished form where entities are distinguished from relations, and 
a unified form where relations and entities coincide. 
The distinguished form corresponds to the theory presented in the paper (Spivak~\cite{spivak:sd}). 
We extend Spivak's treatment of tables 
from the static case of a single entity classification (type specification) 
to the dynamic case of classifications varying along infomorphisms. 
Our treatment of relational databases as diagrams of tables differs from 
Spivak's sheaf theory of databases. 
The unified form, 
a special case of the distinguished form, 
corresponds to the theory presented in the paper (Spivak~\cite{spivak:fdm})). 
The unified form has a graphical presentation, 
which corresponds to the sketch theory of databases (Johnson and Rosebrugh~\cite{johnson:rosebrugh:07}) 
and the resource description framework (RDF). 
This paper, 
which is the first step to connect relational databases with system consequence, 
is concerned with the semantics of relational databases. 
A later paper will discuss various formalisms of relational databases, 
such as relational algebra and first order logic.

\subsection{Architecture.}\label{sub:sec:arch}

\begin{wrapfigure}{r}{3cm}
\begin{center}
\setlength{\unitlength}{0.36pt}
\begin{picture}(120,100)(3,30)
\put(60.5,80.5){\makebox(0,0){\tiny{$\equiv$}}}
\qbezier(100,87)(60,97)(20,87)
\put(20,87){\vector(-4,-1){0}}
\qbezier(20,75)(60,65)(100,75)
\put(100,75){\vector(4,1){0}}
\put(0.3,80){\makebox(0,0){\huge{$\circ$}}}
\put(120.3,80){\makebox(0,0){\huge{$\circ$}}}
\put(2,10){\makebox(0,0){\huge{$\circ$}}}
\put(122,10){\makebox(0,0){\huge{$\bullet$}}}
%
\put(0,68){\line(0,-1){40}}
\put(120,68){\line(0,-1){40}}
\put(12,-25){\scriptsize{\textsf{architecture}}}
\end{picture}
\end{center}
\caption{\texttt{FOLE}}
\label{hierarchy}
\end{wrapfigure}
The \texttt{FOLE} architecture,
as briefly pictured in Fig.\ref{hierarchy} 
and more completely in Fig.\,1
in the preface of \cite{kent:fole:era:equiv},
consists of four nodes
divided into two branches.
%
The classification form of \texttt{FOLE}
(left hand side of Fig.\,\ref{hierarchy})
consists of 
``The {\ttfamily FOLE} Foundation''
at the bottom
and
``The {\ttfamily FOLE} Superstructure''
at the top.
%
The interpretation form of \texttt{FOLE} 
(right hand side of Fig.\ref{hierarchy})
consists of 
``The {\ttfamily FOLE} Table''
at the bottom
and
``The {\ttfamily FOLE} Database''
at the top.
%
%
The equivalence 
between
the classification form
and
the interpretation form 
is define in the paper 
``{\ttfamily FOLE} Equivalence''
\cite{kent:fole:era:equiv}.
%
%
The current paper is concerned with the \texttt{FOLE} table concept.
%

\subsection{Overview}
\label{sub:sec:overvu}

Section~\ref{sec:tbl:basics} provides material on the basic structures underpinning the \texttt{FOLE} table concept: 
signatures, type domains and signed domains.
Section~\ref{sec:tbl:hier} describes our representation for the table concept 
by defining the multi-path fibered context of tables 
\footnote{The original discussion of {\ttfamily FOLE} (Kent~\cite{kent:iccs2013}) took place within the knowledge representation community,
where the term \emph{category} is defined to be a division within a system of classification or a mode of existence.
Hence following (Kent~\cite{kent:iccs2013}), we use 
``mathematical context'' (Goguen \cite{goguen:cm91}) for the mathematical term ``category'',
``passage'' for the term ``functor'', and
``bridge'' for the term ``natural transformation''.}
(illustrated in Tbl.~\ref{fig:tbl:gro:constr} of \S~\ref{sub:sec:fbr:cxt:tbl})
\begin{center}
{{\begin{tabular}{@{\hspace{-10pt}}c@{\hspace{40pt}}c}
{{\scriptsize{\begin{tabular}{r@{\hspace{5pt}}l}
\textit{signature} & $\mathcal{S}$ \\
\textit{type domain} & $\mathcal{A}$ \\ 
\textit{signed domain} & ${\langle{\mathcal{S},\mathcal{A}}\rangle}$
\end{tabular}}}}
&
{{\begin{tabular}{c}
{\setlength{\unitlength}{0.55pt}\begin{picture}(120,120)(0,-10)
\put(60,131){\makebox(0,0){\scriptsize{$\mathrmbf{Tbl}$}}}
\put(-10,60){\makebox(0,0)[r]{\scriptsize{$\mathrmbf{Tbl}(\mathcal{S})$}}}
\put(130,60){\makebox(0,0)[l]{\scriptsize{$\mathrmbf{Tbl}(\mathcal{A})$}}}
\put(60,-12){\makebox(0,0){\scriptsize{$\mathrmbf{Tbl}(\mathcal{S},\mathcal{A})$}}}
\put(60,10){\vector(0,1){100}}
\put(5,65){\vector(1,1){50}}
\put(115,65){\vector(-1,1){50}}
\dottedline[$\cdot$]{4}(52,8)(8,52)\put(5,55){\vector(-1,1){0}}
\dottedline[$\cdot$]{4}(68,8)(112,52)\put(115,55){\vector(1,1){0}}
\put(60,120){\circle*{5}}
\put(60,0){\circle*{5}}
\put(0,60){\circle*{5}}
\put(120,60){\circle*{5}}
\put(-125,5){\makebox(0,0)[l]{\scriptsize{\textbf{Tbl.~\ref{fig:tbl:gro:constr}} (simplified)}}}
\end{picture}}
\end{tabular}}}
\end{tabular}}}
\end{center}
%
---
one fibered path goes directly via signed domains (\S~\ref{sub:sec:tbl:sign:dom}),
while two other paths go indirectly via signatures (\S~\ref{sub:sec:tbl:sign})
and type domains (\S~\ref{sub:sec:tbl:typ:dom}).
Section~\ref{sec:tbl:lim:colim} 
uses properties of comma contexts and the Grothendieck construction
to prove that the various (sub)contexts of \texttt{FOLE} tables are complete (joins exist) and cocomplete (unions exist).
Table~\ref{tbl:figs:tbls} lists the figures and tables in this paper.

%
\begin{table}
\begin{center}
{{\scriptsize{\setlength{\extrarowheight}{1.6pt}
{\begin{tabular}{l@{\hspace{10pt}}l}
\\
{\fbox{
\begin{tabular}[t]
{|
l@{\hspace{8pt}}
l@{\hspace{2pt}:\hspace{5pt}}
l|}
\hline

\S\ref{sub:sec:arch}
& Fig.~\ref{hierarchy} 
& \texttt{FOLE} Architecture
\\\hline\hline

\S\ref{sub:sec:overview}
& Fig.~\ref{fig:fole:tbl:basics} 
& \texttt{FOLE} Table Basics
\\\cline{1-2}

\S\ref{sub:sec:fole:comps:sign}
& Fig.~\ref{fig:sign:mor} 
& List Morphism
\\\cline{1-2}

\S\ref{sub:sub:sec:tup:bridge:sign}
& Fig.~\ref{fig:tup:bridge:sign} 
& Tuple Bridge: Signature
\\\cline{1-1}

\S\ref{sub:sub:sec:tup:bridge:typ:dom}
& Fig.~\ref{fig:tup:bridge:typ:dom} 
& Tuple Bridge: Type Domain
\\\cline{1-1}

\S\ref{sub:sub:sec:tup:fn:fact}
& Fig.~\ref{fig:tup:fn:fact} 
& Tuple Function Factorization
\\\hline\hline

\S\ref{sub:sec:tbl}
& Fig.~\ref{fig:fole:tbl} 
& \texttt{FOLE} Table
\\

& Fig.~\ref{fig:tbl:mor} 
& \texttt{FOLE} Table Morphism
\\

& Fig.~\ref{fig:tbl:cxt} 
& \texttt{FOLE} Table Mathematical Context
\\\cline{1-2}

\S\ref{sub:sec:tbl:sign:dom}
& Fig.~\ref{fig:tbl:mor:large} 
& Table Morphism: Signed Domain
\\\cline{1-2}

\S\ref{sub:sub:sec:tbl:sign:low}
& Fig.~\ref{fig:tbl:mor:medium:fixed:S} 
& $\mathcal{S}$-Table Morphism
\\\cline{1-1}

\S\ref{sub:sub:sec:tbl:sign:up}
& Fig.~\ref{fig:tbl:fbr:fact:sign} 
& Table Fiber Passage Factorization
\\
& Fig.~\ref{fig:tbl:mor:large:sign} 
& Table Morphism: Signature
\\\cline{1-2}

\S\ref{sub:sub:sec:tbl:typ:dom:low}
& Fig.~\ref{fig:tbl:mor:medium} 
& $\mathcal{A}$-Table Morphism
\\\cline{1-1}
\S\ref{sub:sub:sec:tbl:typ:dom:up}
& Fig.~\ref{fig:tbl:fbr:fact:typ:dom} 
& Table Fiber Adjunction Factorization
\\

& Fig.~\ref{fig:tbl:mor:large:typ:dom} 
& Table Morphism: Type Domain
\\

& Fig.~\ref{fig:ind:fbr:tbl:cls} 
& Indexed Adjunction of Tables
\\\cline{1-2}

\S\ref{sub:sec:fbr:cxt:tbl}
& Fig.~\ref{fig:tbl:gro:constr} 
& The Fibered Hierarchy of \texttt{FOLE} Tables
\\\hline\hline

\S\ref{sub:sec:eg}
& Fig.~\ref{binary:join} 
& Binary Join
\\\hline\hline

\S\ref{sub:sec:present:FOLE}
& Fig.~\ref{fig:tbl:papers} 
& \texttt{FOLE} Papers: Sequence \& Dependency
\\\hline

\end{tabular}}}
&
\begin{tabular}[t]{l}
{\fbox{
\begin{tabular}[t]
{|
l@{\hspace{8pt}}
l@{\hspace{2pt}
:\hspace{5pt}}
l|}
\hline

\S\ref{sub:sec:overvu}
& Tbl.~\ref{tbl:figs:tbls} 
& Figures \& Tables
\\\hline\hline

\S\ref{sub:sec:fole:comps:sign}
& Tbl.~\ref{tbl:list-sub:refl} 
& Sort List/Subset Reflection
\\\cline{1-2}

\S\ref{sub:sub:sec:tup:fn:fact}
& Tbl.~\ref{tbl:tup:fns} 
& Tuple Functions
\\\hline\hline

\S\ref{sub:sec:tbl:sign:dom}
& Tbl.~\ref{tbl:tbl-rel:refl:sign:dom:mor} 
& Reflection: Signed Domain
\\\cline{1-2}

\S\ref{sub:sub:sec:tbl:sign:low}
& Tbl.~\ref{tbl:tbl-rel:refl:sign:mor} 
& Reflection: Signature
\\\cline{1-2}

\S\ref{sub:sub:sec:tbl:typ:dom:low}
& Tbl.~\ref{tbl:tbl-rel:refl:sml:shrt} 
& Reflection: Type Domain
\\\cline{1-2}

\S\ref{sub:sec:fbr:cxt:tbl}
& Tbl.~\ref{tbl:grothen:construct} 
& Grothendieck Constructions
\\\hline\hline

\S\ref{sub:sub:sec:props}
& Tbl.~\ref{tbl:co:complete:cxts} 
& Complete/Cocomplete Contexts
\\\hline\hline

\S\ref{sub:sec:paper:review}
& Tbl.~\ref{tbl:thms:props:lems:cors} 
& Lemmas, Propositions \& Theorems
\\\hline

\end{tabular}}}
\\\\
\end{tabular}
\\\\
\end{tabular}}}}}
\end{center}
\caption{Figures and Tables}
\label{tbl:figs:tbls}
\end{table}
%


\newpage
\section{Table Basics}\label{sec:tbl:basics}

\subsection{Overview}\label{sub:sec:overview}

A table in the relational model is represented as an array, organized into rows and columns. 
The rows are called the tuples (records) of the table, whereas
the columns are called the attributes of the table.
The rows are indexed by keys.
Both rows 
and columns 
are unordered;
instead of indexing headers and tuples as $n$-tuples,
the \texttt{FOLE} approach uses attribute names for tuples
(as advocated by Codd~\cite{codd:90}).
In the relational model, 
all components can be resolved into sets and functions.
%
\footnote{The relational data model is based upon the context $\mathrmbf{Set}$ of sets and functions.}
%
%
\begin{center}
{{\begin{tabular}{c}
\setlength{\unitlength}{0.5pt}
\begin{picture}(400,230)(0,-40)
\thicklines
\put(0,150){\line(1,0){74}}
\put(0,0){\line(1,0){74}}
\put(80,150){\line(1,0){320}}
\put(80,0){\line(1,0){320}}
\put(0,0){\line(0,1){150}}
\put(74,0){\line(0,1){150}}
\put(80,0){\line(0,1){150}}
\put(400,0){\line(0,1){150}}
\thinlines
\put(80,100){\line(1,0){320}}
\put(80,50){\line(1,0){320}}
\put(0,100){\line(1,0){74}}
\put(0,50){\line(1,0){74}}
\put(160,0){\line(0,1){150}}
\put(240,0){\line(0,1){150}}
\put(320,0){\line(0,1){150}}
\dottedline{5}(162.5,150)(162.5,0)
\dottedline{5}(167.5,150)(167.5,0)
\dottedline{5}(172.5,150)(172.5,0)
\dottedline{5}(177.5,150)(177.5,0)
\dottedline{5}(182.5,150)(182.5,0)
\dottedline{5}(187.5,150)(187.5,0)
\dottedline{5}(192.5,150)(192.5,0)
\dottedline{5}(197.5,150)(197.5,0)
\dottedline{5}(202.5,150)(202.5,0)
\dottedline{5}(207.5,150)(207.5,0)
\dottedline{5}(212.5,150)(212.5,0)
\dottedline{5}(217.5,150)(217.5,0)
\dottedline{5}(222.5,150)(222.5,0)
\dottedline{5}(227.5,150)(227.5,0)
\dottedline{5}(232.5,150)(232.5,0)
\dottedline{5}(237.5,150)(237.5,0)
\dottedline{5}(80,95)(400,95)
\dottedline{5}(80,90)(400,90)
\dottedline{5}(80,85)(400,85)
\dottedline{5}(80,80)(400,80)
\dottedline{5}(80,75)(400,75)
\dottedline{5}(80,70)(400,70)
\dottedline{5}(80,65)(400,65)
\dottedline{5}(80,60)(400,60)
\dottedline{5}(80,55)(400,55)
\put(200,177){\makebox(0,0){\footnotesize{$
\overset{\rule[-2pt]{0pt}{5pt}\textstyle{
{\mathsf{attribute}}}}
{\overbrace{\rule{40pt}{0pt}}}$}}}
\put(0,75){\makebox(0,0)[r]{\footnotesize{$
{\mathsf{key}}\left\{\rule{0pt}{15pt}\right.$}}}
\put(69,74){\makebox(0,0)[l]{\normalsize{$\mapsto$}}}
\put(402,75){\makebox(0,0)[l]{\footnotesize{$\left.\rule{0pt}{15pt}\right\}{\;\mathsf{tuple}}$}}}
\put(200,-37){\makebox(0,0){\footnotesize{$\underset{\rule[3pt]{0pt}{5pt}\textstyle{\underset{\rule[1pt]{0pt}{5pt}
\underset{\underset{}{}}{}
}{\mathsf{table}}}}{\rule[5pt]{0pt}{5pt}\underbrace{\rule{199pt}{0pt}}}$}}}
\end{picture}
\end{tabular}}}
\end{center}

\begin{figure}
\begin{center}
{{\begin{minipage}{200pt}
{{\begin{tabular}{@{\hspace{0pt}}c@{\hspace{0pt}}c}
\\&\\
{{\begin{tabular}{c}
\setlength{\unitlength}{0.5pt}
\begin{picture}(200,120)(-10,-10)
\put(0,80){\makebox(0,0){\footnotesize{$I$}}}
\put(80,80){\makebox(0,0){\footnotesize{$X$}}}
\put(80,0){\makebox(0,0){\footnotesize{$Y$}}}
\put(90,40){\makebox(0,0)[l]{\scriptsize{$\models_{\mathcal{A}}$}}}
\put(40,90){\makebox(0,0){\scriptsize{$s$}}}
\put(15,80){\vector(1,0){50}}
\put(80,65){\line(0,-1){50}}
\put(40,120){\makebox(0,0){\footnotesize{
$\overset{\textstyle{\stackrel{\text{\emph{signature}}}{\mathcal{S}\rule{0pt}{1pt}}}}{\overbrace{\rule{40pt}{0pt}}}$}}}
\put(120,40){\makebox(0,0)[l]{\footnotesize{$\left.\rule{0pt}{22pt}\right\}$}}}
\put(130,30){\makebox(0,0)[l]{\footnotesize{
$\underset{\underset{{\scriptstyle\text{\emph{domain}}}}{\scriptstyle{\text{\emph{type}}}}}{\textstyle{\mathcal{A}}\rule[-2pt]{0pt}{1pt}}$}}}
\put(-25,45){\makebox(0,0)[r]{\footnotesize{$\left\{\rule{0pt}{30pt}\right.$}}}
\put(-45,35){\makebox(0,0)[r]{\footnotesize{$\underset{\underset{{\scriptstyle\text{\emph{domain}}}}
{\scriptstyle{\text{\emph{signed}}}}}{\textstyle{\mathcal{D}}}
\!\!\footnote{A \texttt{FOLE} table 
$\mathcal{T}={\langle{K,t,\mathcal{D}}\rangle}$
(defined in \S\ref{sub:sec:tbl}) 
consists of 
a signed domain $\mathcal{D}$,
with a key set $K$ and tuple map 
$K\xrightarrow{t}\mathrmbfit{tup}_{\mathcal{A}}(I,s)=
\mathrmbfit{ext}_{\mathrmbf{List}(\mathcal{A})}(I,s)$.
}
$}}}
\end{picture}
\end{tabular}}}
&
{{\begin{tabular}{c}
\setlength{\unitlength}{0.5pt}
\begin{picture}(200,120)(-10,-10)
\put(10,80){\makebox(0,0){\footnotesize{$1$}}}
\put(87,80){\makebox(0,0)[l]{\footnotesize{$\mathrmbf{List}(X)$}}}
\put(87,0){\makebox(0,0)[l]{\footnotesize{$\mathrmbf{List}(Y)$}}}
\put(128,40){\makebox(0,0)[l]{\scriptsize{$\models_{\mathrmbf{List}(\mathcal{A})}$}}}
\put(50,90){\makebox(0,0){\scriptsize{${\langle{I,s}\rangle}$}}}
\put(20,80){\vector(1,0){60}}
\put(120,65){\line(0,-1){50}}
\put(60,40){\makebox(0,0){\footnotesize{$\mathcal{D}$}}}
\end{picture}
\end{tabular}}}
\\&\\
\end{tabular}}}
\end{minipage}}}
\end{center}
\caption{\texttt{FOLE} Table Basics}
\label{fig:fole:tbl:basics}
\end{figure}
%

\subsection{Signatures}\label{sub:sec:fole:comps:sign}

A signature, which represents the header of a relational table, provides typing for the tuples permitted in the table.
%
\footnote{The use of lists for signatures (and tuples) 
follows Codd's recommendation to use attribute names to index the tuples of a relation instead of a numerical ordering.}

\paragraph{Fibers.}

Let $X$ be a sort set.
The fiber mathematical context of $X$-signatures is the comma context
\[\mbox{\footnotesize{$
\mathrmbf{List}(X) 
= \mathrmbf{List}(X) 
= {\bigl(\mathrmbf{Set}{\,\downarrow\,}X\bigr)}
$}\normalsize}\]
associated with the sort set (constant passage)
$\mathrmbf{1}\xrightarrow{\;X\;}\mathrmbf{Set}$.
It has the index and trivial projection passages
$\mathrmbf{Set}\xleftarrow{\;\mathrmbfit{ind}_{X}\;}\mathrmbf{List}(X)\xrightarrow{\;\Delta\;}\mathrmbf{1}$
and the defining bridge
$\mathrmbfit{ind}_{X}\xRightarrow{\,\sigma_{X}\;}\Delta{\;\circ\;}X$. 
An $\mathrmbf{List}(X)$-object $\mathcal{S} = {\langle{I,s}\rangle}$, 
called an $X$-signature (header),
consists of an indexing set (arity) $I$ and a map $I\xrightarrow{\,s\,}X$ from $I$ to the set of sorts $X$.
An $\mathrmbf{List}(X)$-morphism 
$\mathcal{S}' = {\langle{I',s'}\rangle} \xrightarrow{h} {\langle{I,s}\rangle} = \mathcal{S}$
is an arity function $I'\xrightarrow{h}I$
that preserves signatures by
satisfying the naturality condition $h{\,\cdot\,}s = s'$.
%

\paragraph{Fibered Context.}

The fibered context of signatures is the comma context
%
\[\mbox{\footnotesize{$
\mathrmbf{List} = \bigl(\mathrmbf{Set}{\,\downarrow\,}\mathrmbf{Set}\bigr)
$.}\normalsize}\]
It has index and set projection passages
$\mathrmbf{Set}\xleftarrow{\mathrmbfit{arity}}\mathrmbf{List}\xrightarrow{\mathrmbfit{sort}} \mathrmbf{Set}$
and the defining bridge
$\mathrmbfit{arity}\xRightarrow{\,\sigma\;} \mathrmbfit{sort}$. 
A $\mathrmbf{List}$-object (signature, sort list) ${\langle{I,s,X}\rangle}$ 
consists of a sort set $X$ and an $X$-signature ${\langle{I,s}\rangle}$.
A $\mathrmbf{List}$-morphism (signature morphism) 
${\langle{I_{2},s_{2},X_{2}}\rangle}\xrightarrow{{\langle{h,f}\rangle}}{\langle{I_{1},s_{1},X_{1}}\rangle}$
consists of 
a sort function $X_{2}\xrightarrow{f}X_{1}$ and an arity function $I_{2}\xrightarrow{h}I_{1}$
satisfying the naturality condition $h{\,\cdot\,}s_{1} = s_{2}{\,\cdot\,}f$.
\begin{figure}
\begin{center}
{{\begin{tabular}{c}
\setlength{\unitlength}{0.5pt}
\begin{picture}(180,80)(0,10)
\put(0,80){\makebox(0,0){\footnotesize{$I_{2}$}}}
\put(180,80){\makebox(0,0){\footnotesize{$I_{1}$}}}
\put(0,0){\makebox(0,0){\footnotesize{$X_{2}$}}}
\put(180,0){\makebox(0,0){\footnotesize{$X_{1}$}}}
\put(45,45){\makebox(0,0){\footnotesize{$\hat{I}_{2}$}}}
\put(-6,40){\makebox(0,0)[r]{\scriptsize{$s_{2}$}}}
\put(22,16){\makebox(0,0)[l]{\scriptsize{$\hat{s}_{2}$}}}
\put(188,40){\makebox(0,0)[l]{\scriptsize{$s_{1}$}}}
\put(90,94){\makebox(0,0){\scriptsize{$h$}}}
\put(82,66){\makebox(0,0){\scriptsize{$\hat{f}$}}}
\put(25,68){\makebox(0,0)[l]{\scriptsize{$\hat{h}$}}}
\put(90,14){\makebox(0,0){\scriptsize{$f$}}}
\put(20,80){\vector(1,0){140}}
\put(20,0){\vector(1,0){140}}
\put(0,70){\vector(0,-1){60}}
\put(180,70){\vector(0,-1){60}}
\put(57,47){\vector(4,1){100}}
\qbezier(10,10)(20,20)(30,30)\put(10,10){\vector(-1,-1){0}}
\qbezier(10,70)(20,60)(30,50)\put(30,50){\vector(1,-1){0}}
%
\qbezier(160,20)(165,20)(170,20)
\qbezier(160,20)(160,15)(160,10)
\end{picture}
\end{tabular}}}
\end{center}
\caption{List Morphism}
\label{fig:sign:mor}
\end{figure}
%
This condition gives two alternate and adjoint definitions.
In terms of fibers,
a signature morphism
consists of
a sort function $X_{2}\xrightarrow{f}X_{1}$ 
\underline{and} 
either a morphism 
${\langle{I_{2},s_{2}}\rangle}\xrightarrow{\;\hat{h}\;}{f}^{\ast}(I,s)$
in the fiber context $\mathrmbf{List}(X_{2})$ 
or a morphism 
${\scriptstyle\sum}_{f}(I_{2},s_{2})\xrightarrow{\;h\;}{\langle{I,s}\rangle}$
in the fiber context $\mathrmbf{List}(X_{1})$.
\begin{equation}
{{\begin{picture}(120,10)(0,-4)
\put(60,0){\makebox(0,0){\footnotesize{$
\underset{\textstyle{\text{in}\;\mathrmbf{List}(X_{1})}}
{{\scriptstyle\sum}_{f}(I_{2},s_{2})\xrightarrow{\;h\;}{\langle{I_{1},s_{1}}\rangle}}
{\;\;\;\;\;\;\;\;\rightleftarrows\;\;\;\;\;\;\;\;}
\underset{\textstyle{\text{in}\;\mathrmbf{List}(X_{2})}}
{{\langle{I_{2},s_{2}}\rangle}\xrightarrow{\;\hat{h}\;}f^{\ast}(I_{1},s_{1})}
$}}}
\end{picture}}}
\end{equation}
The 
$X_{1}$-signature morphism 
${\scriptstyle\sum}_{f}(I_{2},s_{2})\xrightarrow{\;h\;}{\langle{I_{1},s_{1}}\rangle}$
is the composition (Fig.~\ref{fig:sign:mor}) 
of the fiber morphism
${\scriptstyle\sum}_{f}\Bigl({\langle{I_{2},s_{2}}\rangle}\xrightarrow{\;\hat{h}\;}f^{\ast}(I_{1},s_{1})\Bigr)$
with the ${\langle{I_{1},s_{1}}\rangle}^{\,\text{th}}$ counit component 
${\scriptstyle\sum}_{f}\Bigl(f^{\ast}(I_{1},s_{1})\Bigr)
\xrightarrow{\hat{f}}
{\langle{I_{1},s_{1}}\rangle}$
for the fiber adjunction
$\mathrmbf{List}(X_{2})
{\;\xrightarrow[{\langle{{\scriptscriptstyle\sum}_{f}{\;\dashv\;}f^{\ast}}\rangle}]{\mathrmbfit{list}(f)}\;}
\mathrmbf{List}(X_{1})$.
This fiber adjunction 
(top part of Tbl.~\ref{tbl:list-sub:refl})
is a component of
the sort indexed adjunction of signatures
$\mathrmbf{Set}\xrightarrow{\;\mathrmbfit{list}\;}\mathrmbf{Adj}$.


%
\begin{table}
\begin{center}
\begin{tabular}{c@{\hspace{45pt}}c}
\begin{tabular}[b]{c}
\setlength{\unitlength}{0.65pt}
\begin{picture}(120,80)(0,-15)
\put(0,80){\makebox(0,0){\footnotesize{$\underset{=\,(\mathrmbf{Set}{\downarrow}X_{2})}{\mathrmbf{List}(X_{2})}$}}}
\put(120,80){\makebox(0,0){\footnotesize{$\underset{=\,(\mathrmbf{Set}{\downarrow}X_{1})}{\mathrmbf{List}(X_{1})}$}}}
\put(0,0){\makebox(0,0){\footnotesize{${\wp}X_{2}$}}}
\put(120,0){\makebox(0,0){\footnotesize{${\wp}X_{1}$}}}
\put(60,100){\makebox(0,0){\scriptsize{${\scriptstyle\sum}_{f}$}}}
\put(62,82){\makebox(0,0){\scriptsize{$f^{\ast}$}}}
\put(60,58){\makebox(0,0){\scriptsize{${\scriptstyle\prod}_{f}$}}}
\put(60,20){\makebox(0,0){\scriptsize{$\exists_{f}$}}}
\put(64,2){\makebox(0,0){\scriptsize{$f^{{\scriptscriptstyle-}1}$}}}
\put(60,-20){\makebox(0,0){\scriptsize{$\forall_{f}$}}}
\put(-6,40){\makebox(0,0)[r]{\scriptsize{$\mathrmbfit{im}_{X_{2}}$}}}
\put(10,40){\makebox(0,0)[l]{\scriptsize{$\mathrmbfit{inc}_{X_{2}}$}}}
\put(114,40){\makebox(0,0)[r]{\scriptsize{$\mathrmbfit{im}_{X_{1}}$}}}
\put(130,40){\makebox(0,0)[l]{\scriptsize{$\mathrmbfit{inc}_{X_{1}}$}}}
\put(35,92){\vector(1,0){50}}
\put(35,80){\vector(-1,0){0}}\qbezier(35,80)(45,80)(51,80)\qbezier(67,80)(75,80)(85,80)
\put(35,68){\vector(1,0){50}}
\put(25,12){\vector(1,0){70}}
\put(25,0){\vector(-1,0){0}}\qbezier(25,0)(40,0)(51,0)\qbezier(67,0)(80,0)(95,0)
\put(25,-12){\vector(1,0){70}}
\put(-6,65){\vector(0,-1){50}}
\put(6,15){\vector(0,1){50}}
\put(114,65){\vector(0,-1){50}}
\put(126,15){\vector(0,1){50}}
\end{picture}
\end{tabular}
&
{\scriptsize\setlength{\extrarowheight}{2pt}$\begin{array}[b]{|l|}
\multicolumn{1}{l}{\rule[-5pt]{0pt}{10pt}X_{2}\xrightarrow{\;f\;}X_{1}}
\\\hline
{\scriptstyle\sum}_{f} \dashv f^{\ast} \dashv {\scriptstyle\prod}_{f}																							\\
\exists_{f} \dashv f^{{\scriptscriptstyle-}1} \dashv \forall_{f}																									\\
\mathrmbfit{im}_{X_{2}} \dashv \mathrmbfit{inc}_{X_{2}},\;\mathrmbfit{im}_{X_{1}} \dashv \mathrmbfit{inc}_{X_{1}}	\\ \hline
\mathrmbfit{inc}_{X_{1}} \circ f^{\ast} = f^{{\scriptscriptstyle-}1} \circ \mathrmbfit{inc}_{X_{2}}								\\
\mathrmbfit{inc}_{X_{2}} \circ {\scriptstyle\prod}_{f} = \forall_{f} \circ \mathrmbfit{inc}_{X_{1}}								\\
{\scriptstyle\sum}_{f} \circ \mathrmbfit{im}_{X_{1}} \cong \mathrmbfit{im}_{X_{2}} \circ \exists_{f}							\\
f^{\ast} \circ \mathrmbfit{im}_{X_{2}} \cong \mathrmbfit{im}_{X_{1}} \circ f^{{\scriptscriptstyle-}1}							\\ \hline
\end{array}$}
\\
\end{tabular}
\end{center}
\caption{Sort List/Subset Reflection}
\label{tbl:list-sub:refl}
\end{table}
\begin{theorem}\label{thm:fib:cxt:sign:set}
The fibered context of signatures
$\mathrmbf{List}\xrightarrow{\mathrmbfit{sort}}\mathrmbf{Set}$
is the Grothendieck construction of 
the sort indexed adjunction of signatures
$\mathrmbf{Set}\xrightarrow{\;\mathrmbfit{list}\;}\mathrmbf{Adj}$.
\footnote{\label{grothendieck}A fibration (fibered context over $\mathrmbf{B}$) (nLab \cite{nlab:grothendieck}) is a passage 
$\mathrmbf{E}\xrightarrow{\;\mathrmbfit{P}\;}\mathrmbf{B}$ 
such that the fibers 
$\mathrmbf{E}_{B} = \mathrmbfit{P}^{{\scriptscriptstyle-}1}(B)$ 
depend (contravariantly) pseudofunctorially on $B{\,\in\,}\mathrmbf{B}$. 
Dually, in an opfibration the dependence is covariant.
There is an equivalence of 2-contexts
\[\mbox{\footnotesize{$\int : [\mathrmbf{B}^{\mathrm{op}},\mathrmbf{Cxt}]\stackrel{\cong}{\;\longleftrightarrow\;}\mathrmbf{Fib}(\mathrmbf{B})$}}\]
between the 2-context $\mathrmbf{Fib}(\mathrmbf{B})$ of fibrations over $\mathrmbf{B}$ 
and the 2-context $[\mathrmbf{B}^{\mathrm{op}},\mathrmbf{Cxt}]$ of contravariant pseudo-passages from $\mathrmbf{B}$ to $\mathrmbf{Cxt}$, 
also called $\mathrmbf{B}$-indexed contexts.
The construction 
$\int 
: [\mathrmbf{B}^{\mathrm{op}},\mathrmbf{Cxt}]{\;\rightarrow\;}\mathrmbf{Fib}(\mathrmbf{B})
: \mathrmbf{F}\mapsto\int\mathrmbf{F}$ 
of a fibered context from an indexed context is called the \emph{Grothendieck construction}.
We say that fibered context $\int\mathrmbf{F}$ is the oplax sum of indexed context $\mathrmbf{F}$.
\S~\ref{sub:sec:ind:fbr} has a more detailed discussion of fibered contexts.}
\end{theorem}
%

\subsection{Type Domains}\label{sub:sec:fole:comps:typ:dom}

A type domain, which constrains the body of a relational table, is an indexed collection of data types from which a table's tuples are chosen.

\paragraph{Fiber.}

Let $X$ be a sort set.
The fiber mathematical context of $X$-sorted type domains
%
\footnote{In the {\ttfamily ERA} data model
(Kent \cite{kent:fole:era:found}),
attributes are represented by a typed domain
consisting of a collection of data types.
In {\ttfamily FOLE},
a typed domain is represented by an attribute classification 
$\mathcal{A} = {\langle{X,Y,\models_{\mathcal{A}}}\rangle}$
consisting of a set of attribute types (sorts) $X$,
a set of attribute instances (data values) $Y$ and
an attribute classification relation $\models_{\mathcal{A}}{\,\subseteq\,}X{\times}Y$. 
For each sort (attribute type) $x \in X$,
the data domain of that type is the $\mathcal{A}$-extent
$\mathcal{A}_{x}=\mathrmbfit{ext}_{\mathcal{A}}(x) = \{ y \in Y \mid y \models_{\mathcal{A}} x \}$.
The passage
$X \xrightarrow{\mathrmbfit{ext}_{\mathcal{A}}} {\wp}Y$
maps a sort $x{\,\in\,}X$ to its data domain ($\mathcal{A}$-extent) $\mathcal{A}_{x}{\;\subseteq\;}Y$.
The attribute list classification 
$\mathrmbf{List}(\mathcal{A}) = {\langle{\mathrmbf{List}(X),\mathrmbf{List}(Y),\models_{\mathrmbf{List}(\mathcal{A})}}\rangle}$
has $X$-signatures as types and
$Y$-tuples as instances,
with classification by common arity and universal $\mathcal{A}$-classification:
a $Y$-tuple ${\langle{J,t}\rangle}$ 
is classified by 
an $X$-signature ${\langle{I,s}\rangle}$ 
when
$J = I$ and
$t_{k} \models_{\mathcal{A}} s_{k}$
for all $k \in J = I$.}
%
is the context $\mathrmbf{Cls}(X)$ described as follows.
An $X$-sorted type domain $\mathcal{A} = {\langle{X,Y,\models_{\mathcal{A}}}\rangle}$
consists of a data value set $Y$ and a classification relation $\models_{\mathcal{A}}{\;\subseteq\;}X{\times}Y$;
hence,
a data-type collection $\{ \mathcal{A}_{x}{\;\subseteq\;}Y \mid x \in X \}$,
with each sort $x \in X$ indexing the data-type $\mathrmbfit{ext}_{\mathcal{A}}(x) = A_{x}$.
An $X$-sorted type domain morphism is an infomorphism
$\mathcal{A}_{2}\xrightleftharpoons{{\langle{\mathrmit{1}_{X},g}\rangle}}\mathcal{A}_{1}$
satisfying
$g(y_{1})\models_{\mathcal{A}_{2}}x$ \underline{iff} $y_{1}\models_{\mathcal{A}_{1}}x$
for each sort $x\in{X}$ and data value $y_{1}\in{Y_{1}}$;
hence,
consisting of a data value function $Y_{2}\xleftarrow{g}Y_{1}$
satisfying
$g(y_{1})\in\mathcal{A}_{2,x}$
for each sort $x \in X$ and each value $y_{1}\in\mathcal{A}_{1,x}$;
thus, 
defining the restrictions
$\{ \mathcal{A}_{2,x}\xleftarrow{\;g_{x}}\mathcal{A}_{1,x} \mid x \in X \}$.
%
\footnote{More generally,
let $\mathcal{A}_{2}\xrightleftharpoons{{\langle{f,g}\rangle}}\mathcal{A}_{1}$ be any infomorphism.
The condition
$g(y_{1})\models_{\mathcal{A}_{2}}x_{2}\;\text{\underline{iff}}\;y_{1}\models_{\mathcal{A}_{2}}f(x_{2})$
is equivalent to the abstraction
$g^{-1}(\mathrmbfit{ext}_{\mathcal{A}_{2}}(x_{2})) = \mathrmbfit{ext}_{\mathcal{A}_{1}}(f(x_{2}))$.
Hence,
there is a function
$\mathrmbfit{ext}_{\mathcal{A}_{2}}(x_{2})\xleftarrow{\;g_{x_{2}}\;}\mathrmbfit{ext}_{\mathcal{A}_{1}}(f(x_{2}))$
that is a restriction of the instance function
$Y_{2}\xleftarrow{\;g\;}Y_{1}$.}
%

\paragraph{Fibered Context.}

The fibered context of type domains $\mathrmbf{Cls}\xrightarrow{\mathrmbfit{sort}}\mathrmbf{Set}$ is described as follows.
A type domain $\mathcal{A}={\langle{X,Y,\models_{\mathcal{A}}}\rangle}$ is a classification;
and hence consists of 
a sort set $\mathrmbfit{sort}(\mathcal{A}) = X$ and
an $X$-sorted type domain $\mathcal{A} = {\langle{X,Y,\models_{\mathcal{A}}}\rangle}$.
A type domain morphism $\mathcal{A}_{2}\xrightleftharpoons{{\langle{f,g}\rangle}}\mathcal{A}_{1}$ is an infomophism
consisting of
a sort function $X_{2}\xrightarrow[\mathrmbfit{sort}(f,g)]{f}X_{1}$
and
a data value function $Y_{2}\xleftarrow{g}Y_{1}$
that satisfy the infomorphism condition
$g(y_{1}){\;\models_{\mathcal{A}_{2}}\;}x_{2}$
\underline{iff}
$y_{1}{\;\models_{\mathcal{A}_{1}}\;}f(x_{2})$
for any source sort $x_{2}{\,\in\,}X_{2}$ and target data value $y_{1}{\,\in\,}Y_{1}$.
%
This condition gives an alternate definition.
In terms of fibers,
a type domain morphism
consists of
a sort function $X_{2}\xrightarrow{f}X_{1}$ 
\underline{and} 
a morphism 
$\mathcal{A}_{2}\xrightleftharpoons{{\langle{\mathrmit{1}_{X},g}\rangle}}{f}^{-1}(\mathcal{A}_{1})$
in the fiber context $\mathrmbf{Cls}(X_{2})$.
%
\begin{center}
{{\begin{tabular}{c}
\setlength{\unitlength}{0.45pt}
\begin{picture}(200,110)(0,0)
\put(0,100){\makebox(0,0){\footnotesize{$X_{2}$}}}
\put(188,100){\makebox(0,0){\footnotesize{$X_{1}$}}}
\put(0,0){\makebox(0,0){\footnotesize{$Y_{2}$}}}
\put(188,0){\makebox(0,0){\footnotesize{$Y_{1}$}}}
\put(8,50){\makebox(0,0)[l]{\scriptsize{$\models_{\mathcal{A}_{2}}$}}}
\put(196,50){\makebox(0,0)[l]{\scriptsize{$\models_{\mathcal{A}_{1}}$}}}
\put(95,55){\makebox(0,0)[l]{\scriptsize{$\models_{f^{-1}(\mathcal{A}_{1})}$}}}
\put(90,114){\makebox(0,0){\scriptsize{$f$}}}
\put(90,14){\makebox(0,0){\scriptsize{$g$}}}
%
\put(23,100){\vector(1,0){140}}
\put(163,0){\vector(-1,0){140}}
\put(0,85){\line(0,-1){70}}
\put(186,85){\line(0,-1){70}}

\put(15,85){\line(2,-1){150}}

%
%
\end{picture}
\end{tabular}}}
\end{center}
For any sort function $X_{2}\xrightarrow{f}X_{1}$, 
there is a type domain fiber passage
\footnote{For any sort function $X_{2}\xrightarrow{f}X_{1}$,
there is an inverse image fiber passage
$\mathrmbf{Cls}(X_{2})\xleftarrow{{f}^{-1}}\mathrmbf{Cls}(X_{1})
:{f}^{-1}(\mathcal{A}_{1})\mapsfrom\mathcal{A}_{1}$,
where
$y_{1}{\;\models_{{f}^{-1}(\mathcal{A}_{1})}\;}x_{2}$
\underline{iff}
$y_{1}{\;\models_{\mathcal{A}_{1}}\;}f(x_{2})$
for any source sort $x_{2}{\,\in\,}X_{2}$ and target data value $y_{1}{\,\in\,}Y_{1}$;
or in terms of data types,
${f}^{-1}(\mathcal{A}_{1}) = \bigl\{ {f}^{-1}(\mathcal{A}_{1})_{x_{2}} \mid x_{2}{\,\in\,}X_{2} \bigr\}
= \bigl\{ {\mathcal{A}_{1}}_{f(x_{2})} \mid x_{2}{\,\in\,}X_{2} \bigr\}$.}
\newline
$\mathrmbf{Cls}(X_{2})\xleftarrow{\mathrmbfit{cls}(f)}\mathrmbf{Cls}(X_{1}):\mathcal{A}_{1}\mapsto{f}^{-1}(\mathcal{A}_{1})$.
This fiber passage is a component of
the sort indexed context of type domains
$\mathrmbf{Set}^{\mathrm{op}}\xrightarrow{\;\mathrmbfit{cls}\;}\mathrmbf{Cxt}
:X\mapsto{\mathrmbf{Cls}(X)}$.
\begin{theorem}\label{thm:fib:cxt:cls:set}
The fibered context of type domains
(a fibration)
$\mathrmbf{Cls}\xrightarrow{\mathrmbfit{sort}}\mathrmbf{Set}$
is the Grothendieck construction of 
the sort indexed context of type domains
$\mathrmbf{Set}^{\mathrm{op}}\xrightarrow{\;\mathrmbfit{cls}\;}\mathrmbf{Cxt}$.$^{\ref{grothendieck}}$
\end{theorem}
%

\subsection{Signed Domains}\label{sub:sec:fole:comps:sign:dom}

A signed domain represents both the header and the body of a relational table.

\paragraph{Signed Domains.}

Signed domains are a fundamental component used 
in the definition of database tables and 
in the database interpretation of \texttt{FOLE}.
Signed domains are used to denote the valid tuples for a database header (signature).

A signed (headed/typed) domain
$\mathcal{D}={\langle{I,s,\mathcal{A}}\rangle}$
consists of
a type domain $\mathcal{A}={\langle{X,Y,\models_{\mathcal{A}}}\rangle}$ 
with sort set $X$
and a signature (database header) ${\langle{I,s,X}\rangle}$.
\footnote{Signed domains were called semidesignations in ``Database Semantics'' \cite{kent:db:sem}. Indeed,
a signed domain ${\langle{I,s,\mathcal{A}}\rangle}$
is a list designation ${\langle{{\langle{I,s}\rangle},0}\rangle} : \textbf{10} \rightrightarrows \mathrmbf{List}(\mathcal{A})$
from the trivial entity classification $\textbf{10}={\langle{\textbf{1},\emptyset,\models_{\textbf{10}}}\rangle}$
with element signature map $\textbf{1}\xrightarrow{{\langle{I,s}\rangle}}\mathrmbf{List}(X)$
and empty tuple map $\emptyset\xrightarrow{0} \mathrmbf{List}(Y)$.}
A signed domain morphism
$\mathcal{D}_{2}={\langle{I_{2},s_{2},\mathcal{A}_{2}}\rangle}
\xrightarrow{{\langle{h,f,g}\rangle}}
{\langle{I_{1},s_{1},\mathcal{A}_{1}}\rangle}=\mathcal{D}_{1}$
consists of 
a signature morphism ${\langle{I_{2},s_{2},X_{2}}\rangle}\xrightarrow{{\langle{h,f}\rangle}}{\langle{I_{1},s_{1},X_{1}}\rangle}$ and 
a type domain morphism $\mathcal{A}_{2}\xrightleftharpoons{{\langle{f,g}\rangle}}\mathcal{A}_{1}$
with a common sort function $X_{2}\xrightarrow{f}X_{1}$.
Hence,
the mathematical context of signed domains $\mathrmbf{Dom}$ is the comma context
\[\mbox{\footnotesize{$
\mathrmbf{Set}\xleftarrow{\mathrmbfit{arity}}
\mathrmbf{Dom} = {\bigl(\mathrmbf{Set}{\,\downarrow\,}\mathrmbfit{sort}\bigr)}
\xrightarrow{\mathrmbfit{data}}\mathrmbf{Cls}
$}\normalsize}\]
associated with the sort passage
$\mathrmbf{Cls}\xrightarrow{\mathrmbfit{sort}}\mathrmbf{Set}$.
There is a sign mediating passage
$\mathrmbf{Dom}\xrightarrow{\mathrmbfit{sign}}\mathrmbf{List}
:{\langle{I,s,\mathcal{A}}\rangle}\mapsto{\langle{I,s,X}\rangle}$.

From a different point-of-view,
a signed domain
$\mathcal{D} = {\langle{\mathcal{S},\mathcal{A}}\rangle}$
consists of 
a signature $\mathcal{S}={\langle{I,s,X}\rangle}$ and 
a type domain $\mathcal{A}={\langle{X,Y,\models_{\mathcal{A}}}\rangle}$ 
with common sort set $X$,
and
a signed domain morphism
${\langle{\mathcal{S}_{2},\mathcal{A}_{2}}\rangle}\xrightarrow{{\langle{h,f,g}\rangle}}{\langle{\mathcal{S}_{1},\mathcal{A}_{1}}\rangle}$
consists of a signature morphism $\mathcal{S}_{2}\xrightarrow{{\langle{h,f}\rangle}}\mathcal{S}_{1}$
and a type domain morphism $\mathcal{A}_{2}\xrightleftharpoons{{\langle{f,g}\rangle}}\mathcal{A}_{1}$
with common sort function $X_{2}\xrightarrow{f}X_{1}$.
Hence,
$\mathrmbf{Dom}$ can also be defined as the fibered product
\[\mbox{\footnotesize{$
\mathrmbf{List}\xleftarrow{\mathrmbfit{sign}}
\mathrmbf{Dom}=\mathrmbf{List}{\times_{\mathrmbf{Set}}}\mathrmbf{Cls}
\xrightarrow{\mathrmbfit{data}}\mathrmbf{Cls}
$,}\normalsize}\]
for the opspan of passages
$\mathrmbf{List}\xrightarrow{\mathrmbfit{sort}}\mathrmbf{Set}\xleftarrow{\mathrmbfit{sort}}\mathrmbf{Cls}$.
%
\begin{center}
{{\begin{tabular}{c}
\setlength{\unitlength}{0.48pt}
\begin{picture}(320,100)(0,-5)
\put(0,80){\makebox(0,0){\footnotesize{$\mathrmbf{Set}$}}}
\put(163,80){\makebox(0,0){\footnotesize{$\mathrmbf{Dom}=\bigl(\mathrmbf{Set}{\,\downarrow\,}\mathrmbfit{sort}\bigr)$}}}
\put(320,80){\makebox(0,0){\footnotesize{$\mathrmbf{Cls}$}}}
\put(0,0){\makebox(0,0){\footnotesize{$\mathrmbf{Set}$}}}
\put(163,0){\makebox(0,0){\footnotesize{$\mathrmbf{List}=\bigl(\mathrmbf{Set}{\,\downarrow\,}\mathrmbf{Set}\bigr)$}}}
\put(320,0){\makebox(0,0){\footnotesize{$\mathrmbf{Set}$}}}
\put(45,92){\makebox(0,0){\scriptsize{$\mathrmbfit{arity}$}}}
\put(275,92){\makebox(0,0){\scriptsize{$\mathrmbfit{data}$}}}
\put(55,-12){\makebox(0,0){\scriptsize{$\mathrmbfit{arity}$}}}
\put(265,-12){\makebox(0,0){\scriptsize{$\mathrmbfit{sort}$}}}
\put(-8,40){\makebox(0,0)[r]{\scriptsize{$\mathrmbfit{1}$}}}
\put(152,40){\makebox(0,0)[r]{\scriptsize{$\mathrmbfit{sign}$}}}
\put(328,40){\makebox(0,0)[l]{\scriptsize{$\mathrmbfit{sort}$}}}
\put(68,80){\vector(-1,0){43}}
\put(252,80){\vector(1,0){43}}
\put(78,0){\vector(-1,0){55}}
\put(242,0){\vector(1,0){55}}
\put(0,65){\vector(0,-1){50}}
\put(160,65){\vector(0,-1){50}}
\put(320,65){\vector(0,-1){50}}
\qbezier(280,30)(290,30)(300,30)\qbezier(280,30)(280,20)(280,10)
\end{picture}
\end{tabular}}}
\end{center}
\begin{definition}\label{def:sign:dom:tup}
There is a tuple passage
$\mathrmbfit{tup} : \mathrmbf{Dom}^{\mathrm{op}} \rightarrow \mathrmbf{Set}$.
\end{definition}
\begin{proof}
%
The tuple passage
$\mathrmbf{Dom}^{\mathrm{op}}\xrightarrow{\mathrmbfit{tup}}\mathrmbf{Set}$
maps a signed domain
${\langle{I,s,\mathcal{A}}\rangle}$
to its set of tuples
$\mathrmbfit{tup}(I,s,\mathcal{A})$,
\footnote{This important concept
can intuitively be regarded as the set of legal tuples under the database header
$\mathcal{S}={\langle{I,s,X}\rangle}$.
It is define to be the extent in the list type domain $\mathrmbf{List}(\mathcal{A})$:
$\mathrmbfit{tup}(I,s,\mathcal{A})
=\mathrmbfit{ext}_{\mathrmbf{List}(\mathcal{A})}(I,s) 
= \{ {\langle{J,t}\rangle} \in \mathrmbf{List}(Y) 
\mid {\langle{J,t}\rangle} \models_{\mathrmbf{List}(\mathcal{A})} {\langle{I,s}\rangle} \}$.
Various notations are used for this concept depending upon circumstance:
$\mathrmbfit{tup}(\mathcal{S},\mathcal{A})
= \mathrmbfit{tup}(I,s,\mathcal{A})\;\text{in \S~\ref{sub:sec:fole:comps:sign:dom},\,\ref{sub:sec:tbl:sign:dom}};
= \mathrmbfit{tup}_{\mathcal{A}}(\mathcal{S})
= \mathrmbfit{tup}_{\mathcal{A}}(I,s)\;\text{in \S~\ref{sub:sub:sec:tup:bridge:typ:dom},\,\ref{sub:sec:tbl:typ:dom}};
= \mathrmbfit{tup}_{\mathcal{S}}(\mathcal{A})
= \mathrmbfit{tup}_{\mathcal{S}}(Y,\models_{\mathcal{A}})\;\text{in \S~\ref{sub:sub:sec:tup:bridge:sign},\,\ref{sub:sec:tbl:sign}}$.}
and maps a signed domain morphism
${\langle{I_{2},s_{2},\mathcal{A}_{2}}\rangle}\xrightarrow{{\langle{h,f,g}\rangle}}{\langle{I_{1},s_{1},\mathcal{A}_{1}}\rangle}$
to its tuple function
$\mathrmbfit{tup}(I_{2},s_{2},\mathcal{A}_{2})
\xleftarrow[(h{\cdot}{(\mbox{-})})\cdot({(\mbox{-})}{\cdot}g)]{\mathrmbfit{tup}(h,f,g)}
\mathrmbfit{tup}(I_{1},s_{1},\mathcal{A}_{1})$;
\newline
or visually,
$({\cdots\,}g(t_{h(i_{2})}){\,\cdots}{\,\mid\,}i_{2}{\,\in\,}I_{2})
\mapsfrom
({\cdots\,}t_{i_{1}}{\,\cdots}{\,\mid\,}i_{1}{\,\in\,}I_{1})$.
\mbox{}\hfill\rule{5pt}{5pt}
\end{proof}
\subsection{Inclusion/Tuple Bridges}\label{sub:sec:tup:bridge}

\subsubsection{Signatures.}\label{sub:sub:sec:tup:bridge:sign}


Let
$\mathcal{S} = {\langle{I,s,X}\rangle}$
be a signature.
There is an inclusion passage
$\mathrmbf{Cls}(X)\xrightarrow{\mathrmbfit{inc}_{\mathcal{S}}}\mathrmbf{Dom}$
that maps an $X$-sorted type domain
$\mathcal{A} = {\langle{X,Y,\models_{\mathcal{A}}}\rangle}$
to the signed domain
${\langle{I,s,\mathcal{A}}\rangle}$
and maps an 
$X$-sorted type domain morphism 
$\mathcal{A}_{2} = {\langle{X,Y_{2},\models_{\mathcal{A}_{2}}}\rangle}
\xrightleftharpoons{{\langle{1_{X},g}\rangle}}
{\langle{X,Y_{1},\models_{\mathcal{A}_{1}}}\rangle} = \mathcal{A}_{1}$
to the 
signed domain morphism 
${\langle{I,s,\mathcal{A}_{2}}\rangle}\xrightarrow{{\langle{1_{I},1_{X},g}\rangle}}{\langle{I,s,\mathcal{A}_{1}}\rangle}$. 
Composition of the inclusion passage with the signed domain tuple passage 
(Def.~\ref{def:sign:dom:tup})
gives a signature tuple passage 
\[\mbox{\footnotesize{$
\mathrmbf{Cls}(X)^{\mathrm{op}}
\xrightarrow[\mathrmbfit{inc}_{\mathcal{S}}^{\mathrm{op}}{\,\circ\,}\mathrmbfit{tup}]{\mathrmbfit{tup}_{\mathcal{S}}}
\mathrmbf{Set}
$.}\normalsize}\]
which maps an $X$-sorted type domain $\mathcal{A} = {\langle{X,Y,\models_{\mathcal{A}}}\rangle}$ 
to the tuple set
$\mathrmbfit{tup}(\mathcal{S},\mathcal{A}) = \mathrmbfit{tup}(I,s,\mathcal{A}) = \mathrmbfit{tup}_{\mathcal{S}}(Y,\models_{\mathcal{A}})$
and maps an
$X$-sorted type domain morphism (infomorphism)
$\mathcal{A}_{2} = {\langle{X,Y_{2},\models_{\mathcal{A}_{2}}}\rangle}
\xrightleftharpoons{{\langle{1_{X},g}\rangle}}
{\langle{X,Y_{1},\models_{\mathcal{A}_{1}}}\rangle} = \mathcal{A}_{1}$
to the tuple function associated with $g$:
$\mathrmbfit{tup}_{\mathcal{S}}(Y_{2},\models_{\mathcal{A}_{2}})
\xleftarrow[(\text{-}){\,\cdot\,}g]{\mathrmbfit{tup}_{\mathcal{S}}(g)}
\mathrmbfit{tup}_{\mathcal{S}}(Y_{1},\models_{\mathcal{A}_{1}})$;
or visually,
\newline
$({\cdots\,}g(t_{i}){\,\cdots}{\,\mid\,}i{\,\in\,}I)
\mapsfrom
({\cdots\,}t_{i}{\,\cdots}{\,\mid\,}i{\,\in\,}I)$.
%
\footnote{The tuple passage
$\mathrmbfit{tup}_{\mathcal{S}} : \mathrmbf{Cls}(X)^{\mathrm{op}} \rightarrow \mathrmbf{Set}$
maps an $X$-sorted type domain $\mathcal{A} = {\langle{X,Y,\models_{\mathcal{A}}}\rangle}$
to the tuple set 
$\mathrmbfit{tup}_{\mathcal{S}}(\mathcal{A}) 
= \mathrmbfit{tup}_{\mathcal{A}}(I,s)
= \prod_{i\in{I}}\mathcal{A}_{s(i)}$
and maps an $X$-sorted type domain morphism
$\mathcal{A}
\xrightleftharpoons{{\langle{\mathrmit{1}_{X},g}\rangle}} 
\widetilde{\mathcal{A}}$
to the tuple function
\newline
\mbox{}\hfill
$\mathrmbfit{tup}_{\mathcal{A}}(I,s) = \prod_{i\in{I}}\mathcal{A}_{s(i)}
\xleftarrow[\prod_{i\in{I}}g_{s(i)}]{\mathrmbfit{tup}_{\mathcal{S}}(g)={(\mbox{-})}{\cdot}{g}}
\prod_{i\in{I}}\widetilde{\mathcal{A}}_{s(i)} = \mathrmbfit{tup}_{\widetilde{\mathcal{A}}}(I,s)$,
\hfill\mbox{}\newline
a restriction of the tuple function
$\mathrmbf{List}(Y)\xleftarrow[{\scriptstyle\sum}_{g}]{\;\mathrmbfit{list}(g)}\mathrmbf{List}(\widetilde{Y})$.}

Let
$\mathcal{S}_{2}\xrightarrow{{\langle{h,f}\rangle}}\mathcal{S}_{2}$
be a signature morphism.
There is an inclusion bridge
$f^{-1}{\circ\;}\mathrmbfit{inc}_{\mathcal{S}_{2}}
\xRightarrow{\;\iota_{{\langle{h,f}\rangle}}\;} 
\mathrmbfit{inc}_{\mathcal{A}_{1}}$
(illustrated below right).
\begin{center}
{{\begin{tabular}{c@{\hspace{55pt}}c}
{{\begin{tabular}{c}
\setlength{\unitlength}{0.40pt}
\begin{picture}(200,190)(0,0)
\put(0,180){\makebox(0,0){\footnotesize{$I_{2}$}}}
\put(180,180){\makebox(0,0){\footnotesize{$I_{1}$}}}
\put(0,100){\makebox(0,0){\footnotesize{$X_{2}$}}}
\put(180,100){\makebox(0,0){\footnotesize{$X_{1}$}}}
\put(0,0){\makebox(0,0){\footnotesize{$Y$}}}
\put(180,0){\makebox(0,0){\footnotesize{$Y$}}}
\put(8,140){\makebox(0,0)[l]{\scriptsize{$s_{2}$}}}
\put(188,140){\makebox(0,0)[l]{\scriptsize{$s_{1}$}}}
\put(8,50){\makebox(0,0)[l]{\scriptsize{$\models_{\mathcal{A}_{2}}$}}}
\put(188,50){\makebox(0,0)[l]{\scriptsize{$\models_{\mathcal{A}_{1}}$}}}
\put(90,194){\makebox(0,0){\scriptsize{$h$}}}
\put(90,114){\makebox(0,0){\scriptsize{$f$}}}
\put(90,14){\makebox(0,0){\scriptsize{$=$}}}
\put(-40,55){\makebox(0,0)[r]{\footnotesize{$\left\{\rule{0pt}{23pt}\right.$}}}
\put(-65,55){\makebox(0,0)[r]{\scriptsize{$f^{-1}(\mathcal{A}_{1})$}}}
\put(20,180){\vector(1,0){140}}
\put(20,100){\vector(1,0){140}}
\put(160,0){\line(-1,0){140}}
\put(0,170){\vector(0,-1){60}}
\put(0,85){\line(0,-1){70}}
\put(180,170){\vector(0,-1){60}}
\put(180,85){\line(0,-1){70}}
\end{picture}
\end{tabular}}}
&
{{\begin{tabular}{c}
\setlength{\unitlength}{0.6pt}
\begin{picture}(120,90)(0,10)
\put(0,80){\makebox(0,0){\footnotesize{$\mathrmbf{Cls}(X_{2})$}}}
\put(124,80){\makebox(0,0){\footnotesize{$\mathrmbf{Cls}(X_{1})$}}}
\put(60,5){\makebox(0,0){\footnotesize{$\mathrmbf{Dom}$}}}
\put(62,92){\makebox(0,0){\scriptsize{$f^{-1}$}}}
\put(24,38){\makebox(0,0)[r]{\scriptsize{$\mathrmbfit{inc}_{\mathcal{S}_{2}}$}}}
\put(97,38){\makebox(0,0)[l]{\scriptsize{$\mathrmbfit{inc}_{\mathcal{S}_{1}}$}}}
\put(60,54){\makebox(0,0){\shortstack{\scriptsize{$\iota_{{\langle{h,f}\rangle}}$}\\\large{$\Longrightarrow$}}}}
\put(80,80){\vector(-1,0){40}}
\put(9,68){\vector(3,-4){40}}
\put(111,68){\vector(-3,-4){40}}
\end{picture}
\end{tabular}}}
\end{tabular}}}
\end{center}
For any target type domain $\mathcal{A}_{1} = {\langle{X,Y_{1},\models_{\mathcal{A}_{1}}}\rangle} \in \mathrmbf{Cls}(X_{1})$,
the signed domain morphism
\newline\mbox{}\hfill
$\mathrmbfit{inc}_{\mathcal{S}_{2}}(f^{-1}(\mathcal{A}_{1}))
\xrightarrow[{\langle{h,f,1_{Y}}\rangle}]{\iota_{{\langle{h,f}\rangle}}(\mathcal{A}_{1})} 
\mathrmbfit{inc}_{\mathcal{S}_{1}}(\mathcal{A}_{1})$
\hfill\mbox{}\newline
is illustrated above left.
This is natural in type domain.
Hence, there is an inclusion passage
$\mathrmbf{List}\xrightarrow{\mathrmbfit{inc}}\bigl(\mathrmbf{Cxt}{\,\Downarrow\,}\mathrmbf{Dom}\bigr)^{\mathrm{op}}$.
%
Composition of the inclusion bridge
$f^{-1}{\;\circ\;}\mathrmbfit{inc}_{\mathcal{S}_{2}}
\xRightarrow{\iota_{{\langle{h,f}\rangle}}}
\mathrmbfit{inc}_{\mathcal{S}_{1}}$
with the signeated above leftd domain tuple passage 
(Def.~\ref{def:sign:dom:tup})
gives a signature tuple bridge 
(Fig.~\ref{fig:tup:bridge:sign})
\begin{equation}\label{sign:tup:bridge}
\text{\footnotesize{$
{f^{-1}}^{\mathrm{op}}{\!\circ\!\!\!\!}
\underset{\overbrace{\mathrmbfit{inc}_{\mathcal{S}_{2}}^{\mathrm{op}}{\;\circ\;}\mathrmbfit{tup}}}
{\!\!\!\!\mathrmbfit{tup}_{\mathcal{S}_{2}}}
\xLeftarrow[\iota_{{\langle{h,f}\rangle}}^{\mathrm{op}}{\!\circ\;}\mathrmbfit{tup}]{\tau_{{\langle{h,f}\rangle}}}
\underset{\overbrace{\mathrmbfit{inc}_{\mathcal{S}_{1}}^{\mathrm{op}}{\;\circ\;}\mathrmbfit{tup}}}
{\mathrmbfit{tup}_{\mathcal{S}_{1}}}
\doteq\;\;
\bigl(f^{-1}{\circ\;}\mathrmbfit{inc}_{\mathcal{S}_{2}}
\xRightarrow{\iota_{{\langle{h,f}\rangle}}}
\mathrmbfit{inc}_{\mathcal{S}_{1}}\bigr)^{\mathrm{op}}
{\;\circ\;}\mathrmbfit{tup}.
$}}
\end{equation}
%
For any target type domain $\mathcal{A}_{1} = {\langle{X,Y_{1},\models_{\mathcal{A}_{1}}}\rangle} \in \mathrmbf{Cls}(X_{1})$,
the $\mathcal{A}_{1}^{\text{th}}$-component of the signature tuple bridge is
the tuple function
$\tau_{{\langle{h,f}\rangle}}(\mathcal{A}_{1}) = h{\,\cdot\,}{(\mbox{-})} : 
\mathrmbfit{tup}_{\mathcal{S}_{1}}(\mathcal{A}_{1})
= \mathrmbfit{tup}_{{\langle{I_{1},s_{1},\mathcal{A}_{1}}\rangle}} 
\rightarrow 
\mathrmbfit{tup}_{{\langle{I_{2},s_{2},f^{-1}(\mathcal{A}_{1})}\rangle}} =
\mathrmbfit{tup}_{\mathcal{S}_{2}}(f^{-1}(\mathcal{A}_{1}))$.
This is natural in signature.
Hence, there is a tuple passage
$\mathrmbf{List}\xrightarrow{\mathrmbfit{tup}}\bigl(\mathrmbf{Cxt}{\,\Uparrow\,}\mathrmbf{Set}\bigr)^{\mathrm{op}}$.
%
%
\begin{figure}
\begin{center}
{{\begin{tabular}{c}
\setlength{\unitlength}{0.6pt}
\begin{picture}(120,80)(0,20)
\put(0,80){\makebox(0,0){\footnotesize{$\mathrmbf{Cls}(X_{2})^{\mathrm{op}}$}}}
\put(124,80){\makebox(0,0){\footnotesize{$\mathrmbf{Cls}(X_{1})^{\mathrm{op}}$}}}
\put(60,5){\makebox(0,0){\footnotesize{$\mathrmbf{Set}$}}}
\put(60,92){\makebox(0,0){\scriptsize{$(f^{-1})^{\mathrm{op}}$}}}
\put(24,38){\makebox(0,0)[r]{\scriptsize{$\mathrmbfit{tup}_{\mathcal{S}_{2}}$}}}
\put(97,38){\makebox(0,0)[l]{\scriptsize{$\mathrmbfit{tup}_{\mathcal{S}_{1}}$}}}
\put(60,54){\makebox(0,0){\shortstack{\scriptsize{$\tau_{{\langle{h,f}\rangle}}$}\\\large{$\Longleftarrow$}}}}
\put(80,80){\vector(-1,0){40}}
\put(9,68){\vector(3,-4){40}}
\put(111,68){\vector(-3,-4){40}}
\end{picture}
\end{tabular}}}
\end{center}
\caption{Tuple Bridge: Signature}
\label{fig:tup:bridge:sign}
\end{figure}
%

\subsubsection{Type Domains.}\label{sub:sub:sec:tup:bridge:typ:dom}

Let
$\mathcal{A} = {\langle{X,Y,\models_{\mathcal{A}}}\rangle}$
be a type domain.
There is an inclusion passage
$\mathrmbf{List}(X)\xrightarrow{\mathrmbfit{inc}_{\mathcal{A}}}\mathrmbf{Dom}$
that maps an $X$-signature
${\langle{I,s}\rangle}$
to the signed domain
${\langle{I,s,\mathcal{A}}\rangle}$
and maps an 
$X$-signature morphism 
${\langle{I_{2},s_{2}}\rangle}\xrightarrow{h}{\langle{I_{1},s_{1}}\rangle}$
to the 
signed domain morphism 
${\langle{I_{2},s_{2},\mathcal{A}}\rangle}\xrightarrow{{\langle{h,1_{X},1_{Y}}\rangle}}{\langle{I_{1},s_{1},\mathcal{A}}\rangle}$. 
Composition of the inclusion passage with the signed domain tuple passage 
(Def.~\ref{def:sign:dom:tup})
gives a type domain tuple passage 
\[\mbox{\footnotesize{$
\mathrmbf{List}(X)^{\mathrm{op}}
\xrightarrow[\mathrmbfit{inc}_{\mathcal{A}}^{\mathrm{op}}{\,\circ\,}\mathrmbfit{tup}]
{\mathrmbfit{tup}_{\mathcal{A}}=\mathrmbfit{ext}_{\mathrmbf{List}(\mathcal{A})}}
\mathrmbf{Set}
$,}\normalsize}\]
which maps an $X$-signature
${\langle{I,s}\rangle}$
to the tuple set 
(its $\mathrmbf{List}(\mathcal{A})$-extent)
$\mathrmbfit{tup}_{{\langle{I,s,\mathcal{A}}\rangle}}=\mathrmbfit{tup}_{\mathcal{A}}(I,s)$
and maps an 
$X$-signature morphism 
${\langle{I_{2},s_{2}}\rangle}\xrightarrow{h}{\langle{I_{1},s_{1}}\rangle}$
to the tuple function associated with $h$:
$\mathrmbfit{tup}_{\mathcal{A}}(I_{2},s_{2})
\xleftarrow[h{\,\cdot\,}{(\mbox{-})}]{\mathrmbfit{tup}_{\mathcal{A}}(h)} 
\mathrmbfit{tup}_{\mathcal{A}}(I_{1},s_{1})$;
or visually,
\newline
$({\cdots\,}t_{h(i_{2})}{\,\cdots}{\,\mid\,}i_{2}{\,\in\,}I_{2})
\mapsfrom
({\cdots\,}t_{i_{1}}{\,\cdots}{\,\mid\,}i_{1}{\,\in\,}I_{1})$.
%
\newline
\begin{description}
\item[levo:] 
Let
$\mathcal{A}_{2}\xrightleftharpoons{{\langle{f,g}\rangle}}\mathcal{A}_{1}$
be a type domain morphism.
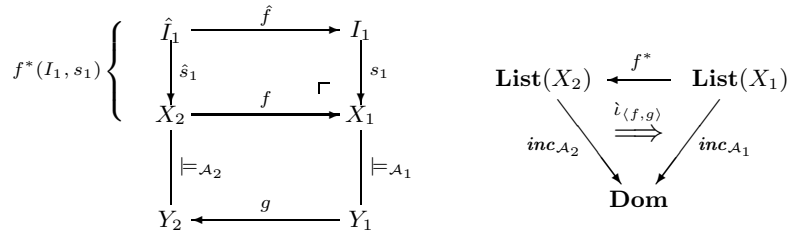
\begin{center}
{{\begin{tabular}{c@{\hspace{55pt}}c}
{{\begin{tabular}{c}
\setlength{\unitlength}{0.40pt}
\begin{picture}(200,190)(0,0)
\put(0,180){\makebox(0,0){\footnotesize{$\hat{I}_{1}$}}}
\put(180,180){\makebox(0,0){\footnotesize{$I_{1}$}}}
\put(0,100){\makebox(0,0){\footnotesize{$X_{2}$}}}
\put(180,100){\makebox(0,0){\footnotesize{$X_{1}$}}}
\put(0,0){\makebox(0,0){\footnotesize{$Y_{2}$}}}
\put(180,0){\makebox(0,0){\footnotesize{$Y_{1}$}}}
\put(8,140){\makebox(0,0)[l]{\scriptsize{$\hat{s}_{1}$}}}
\put(188,140){\makebox(0,0)[l]{\scriptsize{$s_{1}$}}}
\put(8,50){\makebox(0,0)[l]{\scriptsize{$\models_{\mathcal{A}_{2}}$}}}
\put(188,50){\makebox(0,0)[l]{\scriptsize{$\models_{\mathcal{A}_{1}}$}}}
\put(90,194){\makebox(0,0){\scriptsize{$\hat{f}$}}}
\put(90,114){\makebox(0,0){\scriptsize{$f$}}}
\put(90,14){\makebox(0,0){\scriptsize{$g$}}}
\put(-40,142){\makebox(0,0)[r]{\footnotesize{$\left\{\rule{0pt}{22pt}\right.$}}}
\put(-65,142){\makebox(0,0)[r]{\scriptsize{$f^{\ast}(I_{1},s_{1})$}}}
%
\put(20,180){\vector(1,0){140}}
\put(20,100){\vector(1,0){140}}
\put(160,0){\vector(-1,0){140}}
\put(0,170){\vector(0,-1){60}}
\put(0,85){\line(0,-1){70}}
\put(180,170){\vector(0,-1){60}}
\put(180,85){\line(0,-1){70}}
%
\qbezier(140,130)(140,125)(140,120)
\qbezier(140,130)(145,130)(150,130)
\end{picture}
\end{tabular}}}
&
{{\begin{tabular}{c}
\setlength{\unitlength}{0.6pt}
\begin{picture}(120,90)(0,10)
\put(0,80){\makebox(0,0){\footnotesize{$\mathrmbf{List}(X_{2})$}}}
\put(124,80){\makebox(0,0){\footnotesize{$\mathrmbf{List}(X_{1})$}}}
\put(60,5){\makebox(0,0){\footnotesize{$\mathrmbf{Dom}$}}}
\put(62,92){\makebox(0,0){\scriptsize{$f^{\ast}$}}}
\put(24,38){\makebox(0,0)[r]{\scriptsize{$\mathrmbfit{inc}_{\mathcal{A}_{2}}$}}}
\put(97,38){\makebox(0,0)[l]{\scriptsize{$\mathrmbfit{inc}_{\mathcal{A}_{1}}$}}}
\put(60,54){\makebox(0,0){\shortstack{\scriptsize{$\grave{\iota}_{{\langle{f,g}\rangle}}$}\\\large{$\Longrightarrow$}}}}
\put(80,80){\vector(-1,0){40}}
\put(9,68){\vector(3,-4){40}}
\put(111,68){\vector(-3,-4){40}}
\end{picture}
\end{tabular}}}
\end{tabular}}}
\end{center}
There is an inclusion bridge
$f^{\ast}{\circ\;}\mathrmbfit{inc}_{\mathcal{A}_{2}}
\xRightarrow{\;\grave{\iota}_{{\langle{f,g}\rangle}}\;} 
\mathrmbfit{inc}_{\mathcal{A}_{1}}$
(illustrated above right).
For any target signature ${\langle{I_{1},s_{1}}\rangle} \in \mathrmbf{List}(X_{1})$,
the signed domain morphism
\newline\mbox{}\hfill
$\mathrmbfit{inc}_{\mathcal{A}_{2}}(f^{\ast}(I_{1},s_{1}))
\xrightarrow[{\langle{\hat{f},f,g}\rangle}]{\grave{\iota}_{{\langle{f,g}\rangle}}(I_{1},s_{1})} 
\mathrmbfit{inc}_{\mathcal{A}_{1}}(I_{1},s_{1})$
\hfill\mbox{}\newline
is define by pullback (illustrated above left).
This is natural in signature.
%
\footnote{For any $X_{1}$-signature morphism ${\langle{I_{1}',s_{1}'}\rangle} \xrightarrow{h} {\langle{I_{1},s_{1}}\rangle}$
with inverse image $X_{2}$-signature morphism 
${\langle{\hat{I}_{1}',\hat{s}_{1}'}\rangle}\xrightarrow{f^{\ast}(h)}{\langle{\hat{I}_{1},\hat{s}_{1}}\rangle}$
we have the commutative diagram
$\grave{\iota}_{{\langle{f,g}\rangle}}(I'_{1},s'_{1}){\,\cdot\,}\mathrmbfit{inc}_{\mathcal{A}_{1}}(h)
= \mathrmbfit{inc}_{\mathcal{A}_{2}}(f^{\ast}(h)){\,\cdot\,}\grave{\iota}_{{\langle{f,g}\rangle}}(I_{1},s_{1})$.}
%
Hence, there is an inclusion passage
$\mathrmbf{Cls}\xrightarrow{\grave{\mathrmbfit{inc}}}\bigl(\mathrmbf{Cxt}{\,\Downarrow\,}\mathrmbf{Dom}\bigr)^{\mathrm{op}}$.
\footnote{For any context $\mathrmbf{C}$,
the ``super-comma'' context 
$\bigl(\mathrmbf{Cxt}{\,\Downarrow\,}\mathrmbf{C}\bigr)$
is defined \cite{maclane:71} as follows:
\newline
(1) an object is a $\mathrmbf{C}$-diagram ${\langle{\mathrmbf{I},\mathrmbfit{D}}\rangle}$
with indexing context $\mathrmbf{I}$ and passage $\mathrmbf{I}\xrightarrow{\mathrmbfit{D}}\mathrmbf{C}$;
\newline
(2) a morphism is a $\mathrmbf{C}$-diagram morphism 
${\langle{\mathrmbf{I}_{2},\mathrmbfit{D}_{2}}\rangle}
\xrightarrow{\langle{\mathrmbfit{F},\alpha}\rangle}
{\langle{\mathrmbf{I}_{1},\mathrmbfit{D}_{1}}\rangle}$
with indexing passage $\mathrmbf{I}_{2}\xrightarrow{\mathrmbfit{F}}\mathrmbf{I}_{1}$
and bridge $\mathrmbfit{D}_{2}\xLeftarrow{\;\,\alpha}\mathrmbfit{F}{\;\circ\;}\mathrmbfit{D}_{1}$.}
%
Composition of the inclusion bridge
$f^{\ast}{\;\circ\;}\mathrmbfit{inc}_{\mathcal{A}_{2}}
\xRightarrow{\grave{\iota}_{{\langle{f,g}\rangle}}}
\mathrmbfit{inc}_{\mathcal{A}_{1}}$
with the signed domain tuple passage 
(Def.\ref{def:sign:dom:tup})
gives a type domain tuple bridge 
(left-side Fig.~\ref{fig:tup:bridge:sign})
\begin{equation}\label{typ:dom:tup:bridge}
\text{\footnotesize{$
{f^{\ast}}^{\mathrm{op}}{\!\circ\!\!\!\!}
\underset{\overbrace{\mathrmbfit{inc}_{\mathcal{A}_{2}}^{\mathrm{op}}{\;\circ\;}\mathrmbfit{tup}}}
{\!\!\!\!\mathrmbfit{tup}_{\mathcal{A}_{2}}}
\xLeftarrow[\grave{\iota}_{{\langle{f,g}\rangle}}^{\mathrm{op}}{\!\circ\;}\mathrmbfit{tup}]{\acute{\tau}_{{\langle{f,g}\rangle}}}
\underset{\overbrace{\mathrmbfit{inc}_{\mathcal{A}_{1}}^{\mathrm{op}}{\;\circ\;}\mathrmbfit{tup}}}
{\mathrmbfit{tup}_{\mathcal{A}_{1}}}
\!\!\!\!\doteq\;\;
\bigl(f^{\ast}{\circ\;}\mathrmbfit{inc}_{\mathcal{A}_{2}}
\xRightarrow{\grave{\iota}_{{\langle{f,g}\rangle}}}
\mathrmbfit{inc}_{\mathcal{A}_{1}}\bigr)^{\mathrm{op}}
{\;\circ\;}\mathrmbfit{tup}.
$}}
\end{equation}
For any target signature ${\langle{I_{1},s_{1}}\rangle} \in \mathrmbf{List}(X_{1})$,
the tuple function
$\acute{\tau}_{{\langle{f,g}\rangle}}(I_{1},s_{1}) = \hat{f} \cdot {(\mbox{-})} \cdot g : 
\mathrmbfit{tup}_{\mathcal{A}_{1}}(I_{1},s_{1}) \rightarrow \mathrmbfit{tup}_{\mathcal{A}_{2}}(f^{\ast}(I_{1},s_{1}))$
is define by pullback (illustrated above left).
This is natural in signature.
Hence, there is a tuple passage
$\mathrmbf{Cls}\xrightarrow{\acute{\mathrmbfit{tup}}}\bigl(\mathrmbf{Cxt}{\,\Uparrow\,}\mathrmbf{Set}\bigr)^{\mathrm{op}}$.
\newline
\item[dextro:] 
Let
$\mathcal{A}_{2}\xrightleftharpoons{{\langle{f,g}\rangle}}\mathcal{A}_{1}$
be a type domain morphism.
%
\begin{center}
\begin{tabular}{c@{\hspace{60pt}}c}
\\
{{\begin{tabular}{c}
\setlength{\unitlength}{0.40pt}
\begin{picture}(200,190)(0,0)
\put(0,180){\makebox(0,0){\footnotesize{$I_{2}$}}}
\put(180,180){\makebox(0,0){\footnotesize{$I_{2}$}}}
\put(0,100){\makebox(0,0){\footnotesize{$X_{2}$}}}
\put(180,100){\makebox(0,0){\footnotesize{$X_{1}$}}}
\put(0,0){\makebox(0,0){\footnotesize{$Y_{2}$}}}
\put(180,0){\makebox(0,0){\footnotesize{$Y_{1}$}}}
\put(8,140){\makebox(0,0)[l]{\scriptsize{$s_{2}$}}}
\put(188,140){\makebox(0,0)[l]{\scriptsize{$
{s_{2}{\,\cdot\,}f}$}}}
\put(8,50){\makebox(0,0)[l]{\scriptsize{$\models_{\mathcal{A}_{2}}$}}}
\put(188,50){\makebox(0,0)[l]{\scriptsize{$\models_{\mathcal{A}_{1}}$}}}
\put(90,194){\makebox(0,0){\scriptsize{$\mathrmit{1}$}}}
\put(90,114){\makebox(0,0){\scriptsize{$f$}}}
\put(90,14){\makebox(0,0){\scriptsize{$g$}}}
%
\put(20,180){\vector(1,0){140}}
\put(20,100){\vector(1,0){140}}
\put(160,0){\vector(-1,0){140}}
\put(0,170){\vector(0,-1){60}}
\put(0,85){\line(0,-1){70}}
\put(180,170){\vector(0,-1){60}}
\put(180,85){\line(0,-1){70}}
%
\end{picture}
\end{tabular}}}
&
{{\begin{tabular}{c}
\setlength{\unitlength}{0.6pt}
\begin{picture}(120,90)(0,10)
\put(0,80){\makebox(0,0){\footnotesize{$\mathrmbf{List}(X_{2})$}}}
\put(124,80){\makebox(0,0){\footnotesize{$\mathrmbf{List}(X_{1})$}}}
\put(60,5){\makebox(0,0){\footnotesize{$\mathrmbf{Dom}$}}}
\put(60,92){\makebox(0,0){\scriptsize{${\scriptstyle\sum}_{f}$}}}
\put(24,38){\makebox(0,0)[r]{\scriptsize{$\mathrmbfit{inc}_{\mathcal{A}_{2}}$}}}
\put(97,38){\makebox(0,0)[l]{\scriptsize{$\mathrmbfit{inc}_{\mathcal{A}_{1}}$}}}
\put(60,54){\makebox(0,0){\shortstack{\scriptsize{$\acute{\iota}_{{\langle{f,g}\rangle}}$}\\\large{$\Longrightarrow$}}}}
\put(40,80){\vector(1,0){40}}
\put(9,68){\vector(3,-4){40}}
\put(111,68){\vector(-3,-4){40}}
\end{picture}
\end{tabular}}}
\\\\
\end{tabular}
\end{center}
%
There is an inclusion bridge
$\mathrmbfit{inc}_{\mathcal{A}_{2}}
\xRightarrow{\acute{\iota}_{{\langle{f,g}\rangle}}} 
{\scriptstyle\sum}_{f}{\;\circ\;}\mathrmbfit{inc}_{\mathcal{A}_{1}}$ 
(illustrated above right).
For any source signature ${\langle{I_{2},s_{2}}\rangle} \in \mathrmbf{List}(X_{2})$,
the signed domain morphism
\newline\mbox{}\hfill
$\mathrmbfit{inc}_{\mathcal{A}_{2}}(I_{2},s_{2})
\xrightarrow[{\langle{1_{I_{2}},f,g}\rangle}]{\acute{\iota}_{{\langle{f,g}\rangle}}(I_{2},s_{2})} 
\mathrmbfit{inc}_{\mathcal{A}_{1}}({\scriptstyle\sum}_{f}(I_{2},s_{2}))$
\hfill\mbox{}\newline
is define by composition (illustrated above left).
This is natural in signature.
Hence, there is an inclusion passage
$\mathrmbf{Cls}\xrightarrow{\acute{\mathrmbfit{inc}}}\bigl(\mathrmbf{Cxt}{\,\Uparrow\,}\mathrmbf{Dom}\bigr)$.
%
Composition of the inclusion bridge
$\mathrmbfit{inc}_{\mathcal{A}_{2}}
\xRightarrow{\acute{\iota}_{{\langle{f,g}\rangle}}} 
{\scriptstyle\sum}_{f}{\;\circ\;}\mathrmbfit{inc}_{\mathcal{A}_{1}}$ 
with the signed domain tuple passage 
(Def.\ref{def:sign:dom:tup})
gives a type domain tuple bridge 
(right-side Fig.~\ref{fig:tup:bridge:typ:dom})
\[\mbox{\footnotesize{$
\underset{\overbrace{\mathrmbfit{inc}_{\mathcal{A}_{2}}^{\mathrm{op}}{\;\circ\;}\mathrmbfit{tup}}}
{\!\!\!\!\mathrmbfit{tup}_{\mathcal{A}_{2}}}
\xLeftarrow[\acute{\iota}_{{\langle{f,g}\rangle}}^{\,\mathrm{op}}{\!\circ\;}\mathrmbfit{tup}]{\grave{\tau}_{{\langle{f,g}\rangle}}\;}
{\scriptstyle\sum}_{f}^{\mathrm{op}}{\;\circ\!\!\!\!\!}
\underset{\overbrace{\mathrmbfit{inc}_{\mathcal{A}_{1}}^{\mathrm{op}}{\;\circ\;}\mathrmbfit{tup}}}
{\mathrmbfit{tup}_{\mathcal{A}_{1}}}
\doteq\;\;
\bigl(
\mathrmbfit{inc}_{\mathcal{A}_{2}}
\xRightarrow{\acute{\iota}_{{\langle{f,g}\rangle}}} 
{\scriptstyle\sum}_{f}{\,\circ\;}\mathrmbfit{inc}_{\mathcal{A}_{1}}
\bigr)^{\mathrm{op}}
{\!\circ\;\,}\mathrmbfit{tup}.
$}}\]
For any source signature ${\langle{I_{2},s_{2}}\rangle} \in \mathrmbf{List}(X_{2})$,
the tuple function
$\grave{\tau}_{{\langle{f,g}\rangle}}(I_{2},s_{2}) = {(\mbox{-})} \cdot g : 
\mathrmbfit{tup}_{\mathcal{A}_{1}}({\scriptstyle\sum}_{f}(I_{2},s_{2})) \rightarrow \mathrmbfit{tup}_{\mathcal{A}_{2}}(I_{2},s_{2})$
is define by composition (illustrated above left).
This is natural in signature.
Hence, there is a tuple passage
$\mathrmbf{Cls}\xrightarrow{\grave{\mathrmbfit{tup}}}\bigl(\mathrmbf{Cxt}{\,\Downarrow\,}\mathrmbf{Set}\bigr)$.

%
\end{description}

\newpage

\begin{lemma}\label{lem:nat:iso}
There are natural isomorphisms
%
\footnote{For adjunction 
$\mathrmbf{A}_{2} \xrightarrow{{\langle{\mathrmbfit{F},\mathrmbfit{G},\eta,\varepsilon}\rangle}} \mathrmbf{A}_{1}$
with 
left adjoint $\mathrmbf{A}_{2}\xrightarrow{\mathrmbfit{F}}\mathrmbf{A}_{1}$,
right adjoint $\mathrmbf{A}_{2}\xleftarrow{\mathrmbfit{G}}\mathrmbf{A}_{1}$,
unit $\mathrmbfit{1}_{\mathrmbf{A}_{2}}\xRightarrow{\;\eta\;}\mathrmbfit{F}{\;\circ\;}\mathrmbfit{G}$ and 
counit $\mathrmbfit{G}{\;\circ\;}\mathrmbfit{F}\xRightarrow{\;\varepsilon\;}\mathrmbfit{1}_{\mathrmbf{A}_{1}}$,
there is an natural isomorphism
\[\mbox{\footnotesize$
\mathrmbfit{F}{\,\circ\,}\mathrmbfit{A}_{1}\xLeftarrow{\;\,\acute{\alpha}}\mathrmbfit{A}_{2}
{\;\;\cong\;\;}
\mathrmbfit{A}_{1}\xLeftarrow{\;\,\grave{\alpha}}\mathrmbfit{G}{\,\circ\,}\mathrmbfit{A}_{2}
$\normalsize}\]
with 
$\acute{\alpha} = (\eta{\,\circ\,}\mathrmbfit{A}_{2}){\;\bullet\;}(\mathrmbfit{F}{\,\circ\,}\grave{\alpha})$
and  
$\grave{\alpha} = (\mathrmbfit{G}{\,\circ\,}\acute{\alpha}){\;\bullet\;}(\varepsilon{\,\circ\,}\mathrmbfit{A}_{1})$.}
%
\[\mbox{\footnotesize$
{\scriptstyle\sum}_{f}{\;\circ\;}\mathrmbfit{inc}_{\mathcal{A}_{1}} 
\xLeftarrow{\;\acute{\iota}_{{\langle{f,g}\rangle}}} 
\mathrmbfit{inc}_{\mathcal{A}_{2}}
{\;\;\cong\;\;}
\mathrmbfit{inc}_{\mathcal{A}_{1}} 
\xLeftarrow{\;\grave{\iota}_{{\langle{f,g}\rangle}}} 
f^{\ast}{\;\circ\;}\mathrmbfit{inc}_{\mathcal{A}_{2}}
$\normalsize}\]
with
$\acute{\iota}_{{\langle{f,g}\rangle}} = 
\bigl(\eta_{f}{\;\circ\;}\mathrmbfit{inc}_{\mathcal{A}_{2}}\bigr)
{\;\bullet\;} 
\bigl({\scriptstyle\sum}_{f}{\;\circ\;}\grave{\iota}_{{\langle{f,g}\rangle}}\bigr)$
and
$\grave{\iota}_{{\langle{f,g}\rangle}} = 
\bigl(f^{\ast}{\;\circ\;}\acute{\iota}_{{\langle{f,g}\rangle}}\bigr)
{\;\bullet\;}
\bigl(\varepsilon_{f}{\;\circ\;}\mathrmbfit{inc}_{\mathcal{A}_{1}}\bigr)$;
and
\[\mbox{\footnotesize$
{f^{\ast}}^{\mathrm{op}}{\;\circ\;}\mathrmbfit{tup}_{\mathcal{A}_{2}}
\xLeftarrow[\;\grave{\iota}_{{\langle{f,g}\rangle}}^{\mathrm{op}}{\!\circ\;}\mathrmbfit{tup}]{\;\acute{\tau}_{{\langle{f,g}\rangle}}}
\mathrmbfit{tup}_{\mathcal{A}_{1}} 
{\;\;\cong\;\;}
\mathrmbfit{tup}_{\mathcal{A}_{2}}
\xLeftarrow[\;\acute{\iota}_{{\langle{f,g}\rangle}}^{\,\mathrm{op}}{\!\circ\;}\mathrmbfit{tup}]{\;\grave{\tau}_{{\langle{f,g}\rangle}}}
{\scriptstyle\sum}_{f}^{\mathrm{op}}{\;\circ\;}\mathrmbfit{tup}_{\mathcal{A}_{1}} 
$\normalsize}\]
with
$\acute{\tau}_{{\langle{f,g}\rangle}} = 
(\varepsilon_{f}^{\mathrm{op}}{\;\circ\;}\mathrmbfit{tup}_{\mathcal{A}_{1}})
{\;\bullet\;}
({f^{\ast}}^{\mathrm{op}}{\;\circ\;}\grave{\tau}_{{\langle{f,g}\rangle}})$
and
$\grave{\tau}_{{\langle{f,g}\rangle}} = 
({\scriptstyle\sum}_{f}^{\mathrm{op}}{\;\circ\;}\acute{\tau}_{{\langle{f,g}\rangle}})
{\;\bullet\;}
(\eta_{f}^{\mathrm{op}}{\;\circ\;}\mathrmbfit{tup}_{\mathcal{A}_{2}})$.
\end{lemma}
\begin{figure}
\begin{center}
\begin{tabular}{c@{\hspace{65pt}}c}
\textbf{levo}
&
\textbf{dextro}
\\&\\
{{\begin{tabular}{c}
\setlength{\unitlength}{0.6pt}
\begin{picture}(120,90)(0,10)
\put(0,80){\makebox(0,0){\footnotesize{$\mathrmbf{List}(X_{2})$}}}
\put(124,80){\makebox(0,0){\footnotesize{$\mathrmbf{List}(X_{1})$}}}
\put(60,5){\makebox(0,0){\footnotesize{$\mathrmbf{Dom}$}}}
\put(62,92){\makebox(0,0){\scriptsize{$f^{\ast}$}}}
\put(24,38){\makebox(0,0)[r]{\scriptsize{$\mathrmbfit{inc}_{\mathcal{A}_{2}}$}}}
\put(97,38){\makebox(0,0)[l]{\scriptsize{$\mathrmbfit{inc}_{\mathcal{A}_{1}}$}}}
\put(60,54){\makebox(0,0){\shortstack{\scriptsize{$\grave{\iota}_{{\langle{f,g}\rangle}}$}\\\large{$\Longrightarrow$}}}}
\put(80,80){\vector(-1,0){40}}
\put(9,68){\vector(3,-4){40}}
\put(111,68){\vector(-3,-4){40}}
\end{picture}
\end{tabular}}}
&
{{\begin{tabular}{c}
\setlength{\unitlength}{0.6pt}
\begin{picture}(120,90)(0,10)
\put(0,80){\makebox(0,0){\footnotesize{$\mathrmbf{List}(X_{2})$}}}
\put(124,80){\makebox(0,0){\footnotesize{$\mathrmbf{List}(X_{1})$}}}
\put(60,5){\makebox(0,0){\footnotesize{$\mathrmbf{Dom}$}}}
\put(60,92){\makebox(0,0){\scriptsize{${\scriptstyle\sum}_{f}$}}}
\put(24,38){\makebox(0,0)[r]{\scriptsize{$\mathrmbfit{inc}_{\mathcal{A}_{2}}$}}}
\put(97,38){\makebox(0,0)[l]{\scriptsize{$\mathrmbfit{inc}_{\mathcal{A}_{1}}$}}}
\put(60,54){\makebox(0,0){\shortstack{\scriptsize{$\acute{\iota}_{{\langle{f,g}\rangle}}$}\\\large{$\Longrightarrow$}}}}
\put(40,80){\vector(1,0){40}}
\put(9,68){\vector(3,-4){40}}
\put(111,68){\vector(-3,-4){40}}
\end{picture}
\end{tabular}}}
\\&\\&\\
{{\begin{tabular}{c}
\setlength{\unitlength}{0.6pt}
\begin{picture}(120,90)(0,10)
\put(0,80){\makebox(0,0){\footnotesize{$\mathrmbf{List}(X_{2})^{\mathrm{op}}$}}}
\put(124,80){\makebox(0,0){\footnotesize{$\mathrmbf{List}(X_{1})^{\mathrm{op}}$}}}
\put(60,5){\makebox(0,0){\footnotesize{$\mathrmbf{Set}$}}}
\put(60,92){\makebox(0,0){\scriptsize{$(f^{\ast})^{\mathrm{op}}$}}}
\put(24,38){\makebox(0,0)[r]{\scriptsize{$\mathrmbfit{tup}_{\mathcal{A}_{2}}$}}}
\put(97,38){\makebox(0,0)[l]{\scriptsize{$\mathrmbfit{tup}_{\mathcal{A}_{1}}$}}}
\put(60,54){\makebox(0,0){\shortstack{\scriptsize{$\acute{\tau}_{{\langle{f,g}\rangle}}$}\\\large{$\Longleftarrow$}}}}
\put(80,80){\vector(-1,0){40}}
\put(9,68){\vector(3,-4){40}}
\put(111,68){\vector(-3,-4){40}}
\end{picture}
\end{tabular}}}
&
{{\begin{tabular}{c}
\setlength{\unitlength}{0.6pt}
\begin{picture}(120,90)(0,10)
\put(0,80){\makebox(0,0){\footnotesize{$\mathrmbf{List}(X_{2})^{\mathrm{op}}$}}}
\put(124,80){\makebox(0,0){\footnotesize{$\mathrmbf{List}(X_{1})^{\mathrm{op}}$}}}
\put(60,5){\makebox(0,0){\footnotesize{$\mathrmbf{Set}$}}}
\put(60,92){\makebox(0,0){\scriptsize{$({\scriptstyle\sum}_{f})^{\mathrm{op}}$}}}
\put(24,38){\makebox(0,0)[r]{\scriptsize{$\mathrmbfit{tup}_{\mathcal{A}_{2}}$}}}
\put(97,38){\makebox(0,0)[l]{\scriptsize{$\mathrmbfit{tup}_{\mathcal{A}_{1}}$}}}
\put(60,54){\makebox(0,0){\shortstack{\scriptsize{$\grave{\tau}_{{\langle{f,g}\rangle}}$}\\\large{$\Longleftarrow$}}}}
\put(40,80){\vector(1,0){40}}
\put(9,68){\vector(3,-4){40}}
\put(111,68){\vector(-3,-4){40}}
\end{picture}
\end{tabular}}}
\end{tabular}
\end{center}
\caption{Tuple Bridge: Type Domain}
\label{fig:tup:bridge:typ:dom}
\end{figure}
%

%
\begin{center}
{{\footnotesize\setlength{\extrarowheight}{4pt}$\begin{array}{|@{\hspace{5pt}}l@{\hspace{15pt}}l@{\hspace{5pt}}|}
\multicolumn{1}{l}{\text{\bfseries levo}} & \multicolumn{1}{l}{\text{\bfseries dextro}}
\\ \hline
\acute{\tau}_{{\langle{f,g}\rangle}}:{f^{\ast}}^{\mathrm{op}}{\;\circ\;}\mathrmbfit{tup}_{\mathcal{A}_{2}}\Leftarrow\mathrmbfit{tup}_{\mathcal{A}_{1}} 
&
\grave{\tau}_{{\langle{f,g}\rangle}}:\mathrmbfit{tup}_{\mathcal{A}_{2}}\Leftarrow{\scriptstyle\sum}_{f}^{\mathrm{op}}{\;\circ\;}\mathrmbfit{tup}_{\mathcal{A}_{1}} 
\\
\acute{\tau}_{{\langle{f,g}\rangle}}=(\varepsilon_{f}^{\mathrm{op}}{\;\circ\;}\mathrmbfit{tup}_{\mathcal{A}_{1}})\bullet({f^{\ast}}^{\mathrm{op}}{\;\circ\;}\grave{\tau}_{{\langle{f,g}\rangle}})
&
\grave{\tau}_{{\langle{f,g}\rangle}}=({\scriptstyle\sum}_{f}^{\mathrm{op}}{\;\circ\;}\acute{\tau}_{{\langle{f,g}\rangle}})\bullet(\eta_{f}^{\mathrm{op}}{\;\circ\;}\mathrmbfit{tup}_{\mathcal{A}_{2}})
\rule[-7pt]{0pt}{10pt}
\\\hline
\multicolumn{1}{l}{}
\end{array}$}}
\end{center}
%

\newpage
\begin{proposition}\label{prop:typ:dom:inc:tup}
There are inclusion/tuple passages
from the context of type domains
to the lax comma context of adjointly connected presheaves:
\begin{center}
{\footnotesize\setlength{\extrarowheight}{2pt}$\begin{array}{rcl}
\mathrmbf{Cls}
&\xrightarrow{\mathrmbfit{inc}}
&\bigl(\mathrmbf{Adj}{\,\Uparrow\,}\mathrmbf{Dom}\bigr)
\\
\mathrmbf{Cls}^{\mathrm{op}}\!\!
&\xrightarrow{\mathrmbfit{tup}}
&\bigl(\mathrmbf{Adj}{\,\Uparrow\,}\mathrmbf{Set}\bigr)
\end{array}$}
\end{center}
\end{proposition}
%

\newpage
\subsubsection{Tuple Function Factorization}\label{sub:sub:sec:tup:fn:fact}

In \S~\ref{sub:sec:tup:bridge}
we composed with the signed domain tuple passage
(Def.~\ref{def:sign:dom:tup})
to define the tuple passage and bridge for both signatures and type domains.
Here,
we factor components of the signed domain tuple passage in terms of components of these defined notions.
\begin{lemma}\label{lem:tup:fn:fact}
For any signed domain morphism
${\langle{I_{2},s_{2},\mathcal{A}_{2}}\rangle}\xrightarrow{{\langle{h,f,g}\rangle}}{\langle{I_{1},s_{1},\mathcal{A}_{1}}\rangle}$,
the tuple function 
$\mathrmbfit{tup}_{\mathcal{A}_{2}}(I_{2},s_{2})\xleftarrow{\mathrmbfit{tup}(h,f,g)}\mathrmbfit{tup}_{\mathcal{A}_{1}}(I_{1},s_{1})$
has two factorizations:
%
\begin{itemize}
\item 
(Fig.~\ref{fig:tup:fn:fact} left side)
in terms of the signature tuple bridge of \S\ref{sub:sub:sec:tup:bridge:sign} (Fig.~\ref{fig:tup:bridge:sign})
(used in the table fiber passage along a signature morphism).
\newline
\item 
(Fig.~\ref{fig:tup:fn:fact} right side)
in terms of the type domain tuple bridges of \S\ref{sub:sub:sec:tup:bridge:typ:dom} (Fig.~\ref{fig:tup:bridge:typ:dom})
(used in the table fiber adjoint passages along a type domain morphism).
\end{itemize}
\begin{figure}
\begin{center}
{{\begin{tabular}{@{\hspace{-20pt}}c@{\hspace{50pt}}c}
{{\begin{tabular}{c}
\setlength{\unitlength}{0.5pt}
\begin{picture}(320,210)(-40,-20)
\put(0,160){\makebox(0,0){\footnotesize{$
\underset{{=\;\mathrmbfit{ext}_{\mathrmbf{List}(\mathcal{A}_{2})}(I_{2},s_{2})}}
{\mathrmbfit{tup}_{\mathcal{S}_{2}}(\mathcal{A}_{2})}$}}}
%
\put(-5,0){\makebox(0,0){\footnotesize{$
\underset{{=\;\mathrmbfit{ext}_{\mathrmbf{List}(f^{-1}(\mathcal{A}_{1}))}(I_{2},s_{2})}}
{\mathrmbfit{tup}_{\mathcal{S}_{2}}(f^{-1}(\mathcal{A}_{1}))}$}}}
\put(240,0){\makebox(0,0){\footnotesize{$
\underset{{=\;\mathrmbfit{ext}_{\mathrmbf{List}(\mathcal{A}_{1})}(I_{1},s_{1})}}
{\mathrmbfit{tup}_{\mathcal{S}_{1}}(\mathcal{A}_{1})}$}}}
%
\put(130,10){\makebox(0,0){\scriptsize{$\tau_{{\langle{h,f}\rangle}}(\mathcal{A}_{1})$}}}
\put(130,-10){\makebox(0,0){\scriptsize{$=\;h{\,\cdot\,}{(\mbox{-})}$}}}
\put(-10,90){\makebox(0,0)[r]{\scriptsize{$\mathrmbfit{tup}_{\mathcal{S}_{2}}(g)$}}}
\put(-10,70){\makebox(0,0)[r]{\scriptsize{$=\;{(\mbox{-})}{\,\cdot\,}g$}}}
\put(60,80){\makebox(0,0){\scriptsize{$\mathrmbfit{tup}(h,f,g)$}}}
\put(20,60){\makebox(0,0)[l]{\scriptsize{$=\;(h{\cdot}{(\mbox{-})})\cdot({(\mbox{-})}{\cdot}g)$}}}
%
\put(180,0){\vector(-1,0){100}}
\put(0,35){\vector(0,1){90}}
\put(200,25){\vector(-3,2){160}}
\end{picture}
\end{tabular}}}
&
{{\begin{tabular}{c}
\setlength{\unitlength}{0.5pt}
\begin{picture}(320,210)(-40,-20)
\put(0,160){\makebox(0,0){\footnotesize{$
\underset{{=\;\mathrmbfit{ext}_{\mathrmbf{List}(\mathcal{A}_{2})}(I_{2},s_{2})}}
{\mathrmbfit{tup}_{\mathcal{A}_{2}}(I_{2},s_{2})}$}}}
\put(240,160){\makebox(0,0){\footnotesize{$
\underset{{=\;\mathrmbfit{ext}_{\mathrmbf{List}(\mathcal{A}_{1})}(I_{1},s_{1})}}
{\mathrmbfit{tup}_{\mathcal{A}_{1}}({\scriptstyle\sum}_{f}(I_{2},s_{2})})$}}}
\put(0,0){\makebox(0,0){\footnotesize{$
\underset{{=\;\mathrmbfit{ext}_{\mathrmbf{List}(\mathcal{A}_{2})}({f}^{\ast}(I_{1},s_{1}))}}
{\mathrmbfit{tup}_{\mathcal{A}_{2}}({f}^{\ast}(I_{1},s_{1}))}$}}}
\put(240,0){\makebox(0,0){\footnotesize{$
\underset{{=\;\mathrmbfit{ext}_{\mathrmbf{List}(\mathcal{A}_{1})}(I_{1},s_{1})}}
{\mathrmbfit{tup}_{\mathcal{A}_{1}}(I_{1},s_{1})}$}}}
\put(120,170){\makebox(0,0){\scriptsize{$\grave{\tau}_{{\langle{f,g}\rangle}}(I_{2},s_{2})$}}}
\put(120,150){\makebox(0,0){\scriptsize{$=\;{(\mbox{-})}{\,\cdot\,}g$}}}
\put(130,10){\makebox(0,0){\scriptsize{$\acute{\tau}_{{\langle{f,g}\rangle}}(I_{1},s_{1})$}}}
\put(130,-10){\makebox(0,0){\scriptsize{$=\;\hat{f}{\,\cdot\,}{(\mbox{-})}{\,\cdot\,}g$}}}
\put(-10,90){\makebox(0,0)[r]{\scriptsize{$\mathrmbfit{tup}_{\mathcal{A}_{2}}(\hat{h})$}}}
\put(-10,70){\makebox(0,0)[r]{\scriptsize{$=\;\hat{h}{\,\cdot\,}{(\mbox{-})}$}}}
\put(250,90){\makebox(0,0)[l]{\scriptsize{$\mathrmbfit{tup}_{\mathcal{A}_{1}}(h)$}}}
\put(250,70){\makebox(0,0)[l]{\scriptsize{$=\;h{\,\cdot\,}{(\mbox{-})}$}}}
\put(60,80){\makebox(0,0){\scriptsize{$\mathrmbfit{tup}(h,f,g)$}}}
\put(20,60){\makebox(0,0)[l]{\scriptsize{$=\;(h{\cdot}{(\mbox{-})})\cdot({(\mbox{-})}{\cdot}g)$}}}
\put(170,160){\vector(-1,0){100}}
\put(180,0){\vector(-1,0){100}}
\put(0,35){\vector(0,1){90}}
\put(240,35){\vector(0,1){90}}
\put(200,25){\vector(-3,2){160}}
\end{picture}
\end{tabular}}}
\\ & \\
\textit{signature}
& 
\textit{type domain} 
\end{tabular}}}
\end{center}
\caption{Tuple Function Factorization}
\label{fig:tup:fn:fact}
\end{figure}
\end{lemma}
\begin{proof}
We prove the type domain case (Fig.~\ref{fig:tup:fn:fact} left side).
\mbox{ }
\vspace{-10pt}
\begin{center}
{\footnotesize{
\begin{tabular}[t]{p{220pt}c}
For any target signature ${\langle{I_{1},s_{1}}\rangle} \in (\mathrmbf{Set}{\downarrow}X_{1})$,
if ${\langle{\hat{I}_{2},\hat{s}_{2}}\rangle} = f^{\ast}(I_{1},s_{1})$ is its substitution signature (defined by pullback),
the tuple function
$\acute{\tau}_{{\langle{f,g}\rangle}}(I_{1},s_{1})  : 
\mathrmbfit{tup}_{\mathcal{A}_{1}}(I_{1},s_{1}) \rightarrow \mathrmbfit{tup}_{\mathcal{A}_{2}}(f^{\ast}(I_{1},s_{1}))$
(Fig.~\ref{fig:tup:bridge:typ:dom})
maps a target tuple 
$t_{1} \in \mathrmbfit{tup}(\mathcal{S}_{1}) = \mathrmbfit{tup}_{\mathcal{A}_{1}}(I_{1},s_{1})$
to the intermediate tuple 
$\hat{t}_{1} = \hat{f} \cdot t_{1} \cdot g \in \mathrmbfit{tup}_{\mathcal{A}_{2}}(\hat{I}_{1},\hat{s}_{1})$,
where
$\hat{f}=\varepsilon^{f}_{{\langle{I_{1},s_{1}}\rangle}}$
is the 
${\langle{I_{1},s_{1}}\rangle}^{\text{th}}$-component of 
the counit $\varepsilon^{f}$
&
\setlength{\unitlength}{0.40pt}
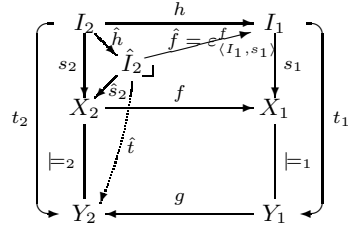
\begin{picture}(120,0)(10,170)
%
\put(0,180){\makebox(0,0){\footnotesize{$I_{2}$}}}
\put(180,180){\makebox(0,0){\footnotesize{$I_{1}$}}}
\put(0,100){\makebox(0,0){\footnotesize{$X_{2}$}}}
\put(180,100){\makebox(0,0){\footnotesize{$X_{1}$}}}
\put(0,0){\makebox(0,0){\footnotesize{$Y_{2}$}}}
\put(180,0){\makebox(0,0){\footnotesize{$Y_{1}$}}}
\put(45,145){\makebox(0,0){\footnotesize{$\hat{I}_{2}$}}}
\qbezier[40](45,125)(40,70)(15,10)\put(15,10){\vector(-1,-3){0}}
\put(40,65){\makebox(0,0)[l]{\scriptsize{$\hat{t}$}}}
\put(-6,140){\makebox(0,0)[r]{\scriptsize{$s_{2}$}}}
\put(22,116){\makebox(0,0)[l]{\scriptsize{$\hat{s}_{2}$}}}
\put(188,140){\makebox(0,0)[l]{\scriptsize{$s_{1}$}}}
\put(-8,50){\makebox(0,0)[r]{\scriptsize{$\models_{2}$}}}
\put(188,50){\makebox(0,0)[l]{\scriptsize{$\models_{1}$}}}
\put(90,194){\makebox(0,0){\scriptsize{$h$}}}
\put(78,162){\makebox(0,0)[l]{\scriptsize{$\hat{f}=\varepsilon^{f}_{{\langle{I_{1},s_{1}}\rangle}}$}}}
\put(25,168){\makebox(0,0)[l]{\scriptsize{$\hat{h}$}}}
\put(90,114){\makebox(0,0){\scriptsize{$f$}}}
\put(90,14){\makebox(0,0){\scriptsize{$g$}}}
\put(-52,90){\makebox(0,0)[r]{\scriptsize{$t_{2}$}}}
\put(235,90){\makebox(0,0)[l]{\scriptsize{$t_{1}$}}}
\put(20,180){\vector(1,0){140}}
\put(20,100){\vector(1,0){140}}
\put(160,0){\vector(-1,0){140}}
\put(0,170){\vector(0,-1){60}}
\put(0,85){\line(0,-1){70}}
\put(180,170){\vector(0,-1){60}}
\put(180,85){\line(0,-1){70}}
\put(210,90){\oval(30,180)[r]}\put(204,0){\vector(-1,0){0}}
\put(-30,90){\oval(30,180)[l]}\put(-24,0){\vector(1,0){0}}
\put(57,147){\vector(4,1){100}}
\qbezier(10,110)(20,120)(30,130)\put(10,110){\vector(-1,-1){0}}
\qbezier(10,170)(20,160)(30,150)\put(30,150){\vector(1,-1){0}}
\qbezier(55,130)(60,130)(65,130)
\qbezier(65,140)(65,135)(65,130)
\end{picture}
\\
\multicolumn{2}{p{350pt}}{
of the signature fiber adjunction
$\mathrmbf{List}(X_{2})\xrightarrow{{\langle{{\scriptscriptstyle\sum}_{f}{\;\dashv\;}f^{\ast}}\rangle}}\mathrmbf{List}(X_{1})$.
Signature preservation, 
$s_{2} \cdot f = h \cdot s_{1}$
means that
$X_{2} \xleftarrow{s_{2}} I_{2} \xrightarrow{h} I_{1}$
is a span of the opspan 
$X_{2} \xrightarrow{f} X_{1} \xleftarrow{s_{1}} I_{1}$.
Let $I_{2} \xrightarrow{\hat{h}} \hat{I}_{2}$ be the mediating function,
so that
${\langle{\hat{h},1_{X_{2}}}\rangle} : {\langle{I_{2},s_{2}}\rangle} \rightarrow {\langle{\hat{I}_{2},\hat{s}_{2}}\rangle}$ is a signature morphism
and $\hat{h} \cdot \hat{f} = h$.
Then the tuple function
$\mathrmbfit{tup}_{\mathcal{A}_{2}}(\hat{h}) : 
\mathrmbfit{tup}_{\mathcal{A}_{2}}(f^{\ast}(I_{1},s_{1})) \rightarrow \mathrmbfit{tup}_{\mathcal{A}_{2}}(I_{2},s_{2})$
maps the intermediate tuple
$\hat{t}_{1} \in \mathrmbfit{tup}_{\mathcal{A}_{2}}(\hat{I}_{1},\hat{s}_{1})$
to the source tuple
$t_{2} = \hat{h} \cdot \hat{t} \in \mathrmbfit{tup}_{\mathcal{A}_{2}}(I_{2},s_{2})$.
Since pullbacks compose,
this is functorial.
\mbox{}\hfill\rule{5pt}{5pt}
%
%
}
\end{tabular}}}
\end{center}
\end{proof}
\begin{table}
\begin{center}
\begin{tabular}{c}
{\scriptsize\setlength{\extrarowheight}{2pt}$\begin{array}[b]{c}
\text{signed domain morphism} 
\\
{\langle{I_{2},X_{2},\mathcal{A}_{2}}\rangle}\xrightarrow{{\langle{h,f,g}\rangle}}{\langle{I_{1},X_{1},\mathcal{A}_{1}}\rangle}
\\
\mathrmbfit{tup}_{\mathcal{A}_{2}}(I_{2},s_{2})\xleftarrow[h{\,\cdot\,}{(\mbox{-})}{\,\cdot\,}g]{\mathrmbfit{tup}(h,f,g)}\mathrmbfit{tup}_{\mathcal{A}_{1}}(I_{1},s_{1})
\rule[-10pt]{0pt}{10pt}
\end{array}$}
\\\hline
{\scriptsize\setlength{\extrarowheight}{2pt}$\begin{array}[b]{c}
\text{signature morphism} 
\\
\mathcal{S}_{2}\xrightarrow{{\langle{h,f}\rangle}}\mathcal{S}_{2}
\\
\text{--- tuple bridge ---} 
\\
\tau_{{\langle{h,f}\rangle}}:
(f^{-1})^{\mathrm{op}}{\!\circ\;}\mathrmbfit{tup}_{\mathcal{S}_{2}}\Leftarrow\mathrmbfit{tup}_{\mathcal{S}_{1}}
\rule[-5pt]{0pt}{10pt}
\\
\mathrmbfit{tup}_{\mathcal{S}_{2}}(f^{-1}(\mathcal{A}_{1}))
\xleftarrow[h{\,\cdot\,}{(\mbox{-})}]{\tau_{{\langle{h,f}\rangle}}(\mathcal{A}_{1})} 
\mathrmbfit{tup}_{\mathcal{S}_{1}}(\mathcal{A}_{1})
\rule[-10pt]{0pt}{10pt}
\end{array}$}
\\\hline
{\scriptsize\setlength{\extrarowheight}{2pt}$\begin{array}[b]{c}
\text{type domain morphism} 
\\
\mathcal{A}_{2}\xrightleftharpoons{{\langle{f,g}\rangle}}\mathcal{A}_{1}
\\
\text{--- levo tuple bridge ---} 
\\
\acute{\tau}_{{\langle{f,g}\rangle}}:(f^{\ast})^{\mathrm{op}}{\;\circ\;}\mathrmbfit{tup}_{\mathcal{A}_{2}}\Leftarrow\mathrmbfit{tup}_{\mathcal{A}_{1}} 
\rule[-5pt]{0pt}{10pt}
\\
\mathrmbfit{tup}_{\mathcal{A}_{2}}(f^{\ast}(I_{1},s_{1}))
\xleftarrow[\hat{f}{\,\cdot\,}{(\mbox{-})}{\,\cdot\,}g]{\acute{\tau}_{{\langle{f,g}\rangle}}(I_{1},s_{1})} 
\mathrmbfit{tup}_{\mathcal{A}_{1}}(I_{1},s_{1})
\rule[-10pt]{0pt}{10pt}
\end{array}$}
\\
{\scriptsize\setlength{\extrarowheight}{2pt}$\begin{array}[b]{c}
\text{--- dextro tuple bridge ---} 
\\
\grave{\tau}_{{\langle{f,g}\rangle}}:\mathrmbfit{tup}_{\mathcal{A}_{2}}\Leftarrow{\scriptstyle\sum}_{f}^{\mathrm{op}}{\;\circ\;}\mathrmbfit{tup}_{\mathcal{A}_{1}} 
\rule[-5pt]{0pt}{10pt}
\\
\mathrmbfit{tup}_{\mathcal{A}_{2}}(I_{2},s_{2})
\xleftarrow[{(\mbox{-})}{\,\cdot\,}g]{\grave{\tau}_{{\langle{f,g}\rangle}}(I_{2},s_{2})} 
\mathrmbfit{tup}_{\mathcal{A}_{1}}({\scriptstyle\sum}_{f}(I_{2},s_{2}))
\end{array}$}
\\
\end{tabular}
\end{center}
\caption{Tuple Functions}
\label{tbl:tup:fns}
\end{table}
%

\newpage
\section{Table Hierarchy}\label{sec:tbl:hier}

The relational table is the basic concept in the relational model for databases.

\subsection{\texttt{FOLE} Tables}\label{sub:sec:tbl}

\paragraph{Tables.}

A table (database relation) 
$\mathcal{T} = {\langle{\mathcal{D},K,t}\rangle}$
consists of 
a signed domain $\mathcal{D}$,
a set $K$ of (primary) keys and
a tuple function $K\xrightarrow{t}\mathrmbfit{tup}(\mathcal{D})$
mapping keys to $\mathcal{D}$-tuples.
%
\footnote{
\texttt{FOLE} tables correspond to improper relations (Codd~\cite{codd:90}),
since they strictly violate the property the ``all rows are distinct from one another in content''.
Proper relations correspond to \texttt{FOLE} relations (\S~\ref{append:sub:sec:rel:tbl}).
One method for converting to the proper relations of Codd,
and thus getting an injective tuple function, 
is to incorporate keys into their corresponding tuple by defining a key datatype.
This was done in Kent~\cite{kent:fole:era:supstruc}.}
%
Equivalently,
it is an object in the comma context 
$\bigl(\mathrmbf{Set}{\,\downarrow\,}\mathrmbfit{tup}\bigr)$
defined by the tuple passage
$\mathrmbf{Dom}^{\mathrm{op}}\xrightarrow{\mathrmbfit{tup}}\mathrmbf{Set}$
(Def.~\ref{def:sign:dom:tup} of \S~\ref{sub:sec:fole:comps:sign:dom}).
A precise description of the 
\texttt{FOLE} Table is given in Fig.~\ref{fig:fole:tbl}.
\begin{figure}
\begin{center}
{{\begin{tabular}{c@{\hspace{25pt}}c}
{\setlength{\extrarowheight}{2pt}{\scriptsize{$\begin{array}[c]{c@{\hspace{20pt}}c}
\mathcal{T}={\langle{K,t,\mathcal{D}}\rangle}
\\
\mathcal{D}={\langle{\mathcal{S},\mathcal{A}}\rangle}
\\
{K}\xrightarrow{t}\mathrmbfit{tup}_{\mathcal{A}}(I,s)
\\
t(k) = {\langle{I,t_{k}}\rangle},\;I\xrightarrow{t_{k}}Y
\\
t_{k,i}\in\mathcal{A}_{s_{i}}
\end{array}$}}}
&
{{\begin{tabular}[c]{c}
\setlength{\unitlength}{0.5pt}
\begin{picture}(180,135)(-60,-3)
\put(-10,90){\makebox(0,0){\scriptsize{$\mathcal{T}$}}}
\put(60,120){\makebox(0,0){\shortstack{\scriptsize{$\mathcal{S}={\langle{I,s,X}\rangle}$}\\$\overbrace{\rule{50pt}{0pt}}$}}}
\put(-55,40){\makebox(0,0){\scriptsize{${K}\left\{\rule{0pt}{20pt}\right.$}}}
\put(25,90){\makebox(0,0){\scriptsize{$\cdots$}}}
\put(60,90){\makebox(0,0){\scriptsize{$i\!:\!s_{i}$}}}
\put(97,90){\makebox(0,0){\scriptsize{$\cdots$}}}
\put(-10,40){\makebox(0,0){\scriptsize{$k$}}}
\put(25,40){\makebox(0,0){\scriptsize{$\cdots$}}}
\put(60,40){\makebox(0,0){\scriptsize{$t_{k,i}$}}}
\put(97,40){\makebox(0,0){\scriptsize{$\cdots$}}}
\put(-20,78){\line(1,0){140}}
\put(2,0){\line(0,1){100}}
\end{picture}
\end{tabular}}}
\end{tabular}}}
\end{center}
\caption{\texttt{FOLE} Table}
\label{fig:fole:tbl}
\end{figure}
\newline
Hence,
a table $\mathcal{T} = {\langle{\mathcal{S},\mathcal{A},K,t}\rangle}$
consists of 
a signature $\mathcal{S} = {\langle{I,s,X}\rangle}$ and
type domain $\mathcal{A} = {\langle{X,Y,\models_{\mathcal{A}}}\rangle}$
that share a common sort set $X$,
a key set $K$, and
a tuple function $K\xrightarrow{t}\mathrmbfit{tup}(\mathcal{D})=\mathrmbfit{tup}_{\mathcal{A}}(I,s)$.
Three alternate expressions are as follows.
\begin{description}
\item[signed domain:] 
Given a signed domain $\mathcal{D} = {\langle{\mathcal{S},\mathcal{A}}\rangle}$,
a table 
$\mathcal{T} = {\langle{K,t}\rangle}$
consists of 
a set $K$ of keys and
a tuple function $K\xrightarrow{t}\mathrmbfit{tup}(\mathcal{D})$.
Hence,
a table is an object in the fiber context $\mathrmbf{Tbl}(\mathcal{D})$.
(See \S\ref{sub:sec:tbl:sign:dom}.)
\item[signature:] 
Given a signature $\mathcal{S} = {\langle{I,s,X}\rangle}$,
a table 
$\mathcal{T} = {\langle{\mathcal{A},K,t}\rangle}$
consists of 
an $X$-sorted type domain $\mathcal{A} = {\langle{X,Y,\models_{\mathcal{A}}}\rangle}$, 
a set $K$ of keys, and
a tuple function $K\xrightarrow{t}\mathrmbfit{tup}_{\mathcal{A}}(I,s)$.
Hence,
a table is an object in the fiber context $\mathrmbf{Tbl}(\mathcal{S})$.
(See \S\ref{sub:sub:sec:tbl:sign:low}.)
\item[type domain:] 
Given a type domain $\mathcal{A} = {\langle{X,Y,\models_{\mathcal{A}}}\rangle}$,
a table 
$\mathcal{T} = {\langle{\mathcal{S},K,t}\rangle}$
consists of 
an $X$-sorted signature $\mathcal{S} = {\langle{I,s,X}\rangle}$,
a set $K$ of keys, and
a tuple function $K\xrightarrow{t}\mathrmbfit{tup}_{\mathcal{A}}(I,s)$.
Hence,
a table is an object in the fiber context $\mathrmbf{Tbl}(\mathcal{A})$.
(See \S\ref{sub:sub:sec:tbl:typ:dom:low}.)
\end{description}
%

%
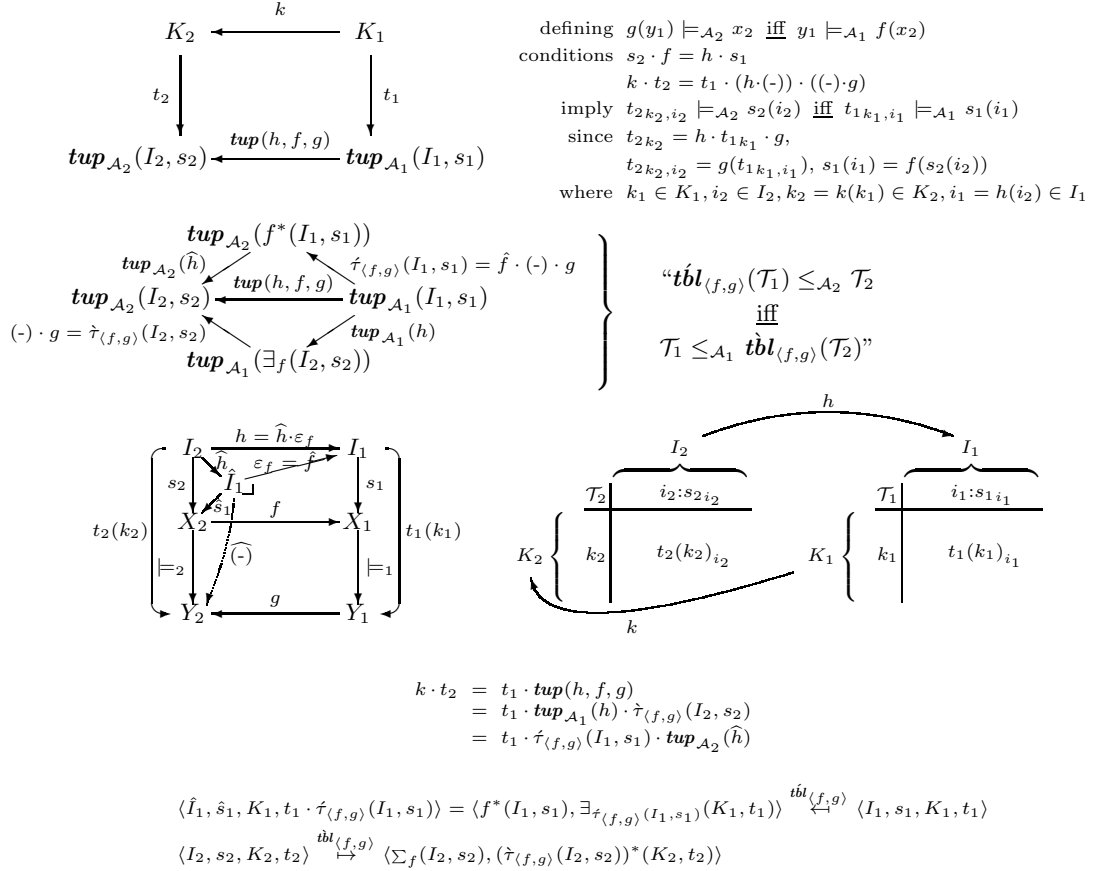
\begin{figure}
\begin{center}
\begin{tabular}{@{\hspace{25pt}}c@{\hspace{45pt}}l}
\begin{tabular}{c}
\setlength{\unitlength}{0.6pt}
\begin{picture}(120,80)(0,-5)
\put(0,80){\makebox(0,0){\footnotesize{$K_{2}$}}}
\put(120,80){\makebox(0,0){\footnotesize{$K_{1}$}}}
\put(-27,0){\makebox(0,0){\footnotesize{$\mathrmbfit{tup}_{\mathcal{A}_{2}}(I_{2},s_{2})$}}}
\put(148,0){\makebox(0,0){\footnotesize{$\mathrmbfit{tup}_{\mathcal{A}_{1}}(I_{1},s_{1})$}}}
\put(63,95){\makebox(0,0){\scriptsize{$k$}}}
\put(63,12){\makebox(0,0){\scriptsize{$\mathrmbfit{tup}(h,f,g)$}}}
\put(-6,40){\makebox(0,0)[r]{\scriptsize{$t_{2}$}}}
\put(128,40){\makebox(0,0)[l]{\scriptsize{$t_{1}$}}}
\put(0,65){\vector(0,-1){50}}
\put(120,65){\vector(0,-1){50}}
\put(100,80){\vector(-1,0){80}}
\put(100,0){\vector(-1,0){80}}
\end{picture}
\end{tabular}
&
{\scriptsize
\setlength{\extrarowheight}{2pt}
\begin{tabular}{@{\hspace{-5pt}}r@{\hspace{5pt}}l}
defining &
$g(y_{1}) \models_{\mathcal{A}_{2}} x_{2} \;\;\text{\underline{iff}}\;\; y_{1} \models_{\mathcal{A}_{1}} f(x_{2})$ \\ 
conditions & 
$s_{2}{\,\cdot\,}f = h \cdot s_{1}$ \\ 
& $k \cdot t_{2} = t_{1} \cdot (h{\cdot}{(\mbox{-})})\cdot({(\mbox{-})}{\cdot}g)$ \\ 
imply &
${t_{2}}_{k_{2},i_{2}} \models_{\mathcal{A}_{2}} s_{2}(i_{2})
\;\;\text{\underline{iff}}\;\;
{t_{1}}_{k_{1},i_{1}} \models_{\mathcal{A}_{1}} s_{1}(i_{1})$ \\ 
since
& ${t_{2}}_{k_{2}} = h{\,\cdot\,}{t_{1}}_{k_{1}}{\,\cdot\,}g$,
\\ &
  ${t_{2}}_{k_{2},i_{2}} = g({t_{1}}_{k_{1},i_{1}})$,
  $s_{1}(i_{1}) = f(s_{2}(i_{2}))$ \\
where 
& $k_{1} \in K_{1}, i_{2} \in I_{2}, k_{2} = k(k_{1}) \in K_{2}, i_{1} = h(i_{2}) \in I_{1}$
\end{tabular}}
\\ & \\
\begin{tabular}{c}
\setlength{\unitlength}{0.6pt}
\begin{picture}(160,80)(-22,-45)
\put(-27,0){\makebox(0,0){\footnotesize{$\mathrmbfit{tup}_{\mathcal{A}_{2}}(I_{2},s_{2})$}}}
\put(148,0){\makebox(0,0){\footnotesize{$\mathrmbfit{tup}_{\mathcal{A}_{1}}(I_{1},s_{1})$}}}
\put(60,40){\makebox(0,0){\footnotesize{$\mathrmbfit{tup}_{\mathcal{A}_{2}}(f^{\ast}(I_{1},s_{1}))$}}}
\put(60,-40){\makebox(0,0){\footnotesize{$\mathrmbfit{tup}_{\mathcal{A}_{1}}(
{\scriptstyle\sum}_{f}
(I_{2},s_{2}))$}}}
\put(63,10){\makebox(0,0){\scriptsize{$\mathrmbfit{tup}(h,f,g)$}}}
\put(105,-22){\makebox(0,0)[l]{\scriptsize{$\mathrmbfit{tup}_{\mathcal{A}_{1}}(h)$}}} 
\put(15,-22){\makebox(0,0)[r]{\scriptsize{${(\mbox{-})} \cdot g = \grave{\tau}_{{\langle{f,g}\rangle}}(I_{2},s_{2})$}}}
\put(105,22){\makebox(0,0)[l]{\scriptsize{$\acute{\tau}_{{\langle{f,g}\rangle}}(I_{1},s_{1}) = \hat{f} \cdot {(\mbox{-})} \cdot g$}}} 
\put(15,22){\makebox(0,0)[r]{\scriptsize{$\mathrmbfit{tup}_{\mathcal{A}_{2}}(\widehat{h})$}}}
\put(100,0){\vector(-1,0){80}}
\put(42,-30){\vector(-3,2){30}}
\put(42,30){\vector(-3,-2){30}}
\put(108,-10){\vector(-3,-2){30}}
\put(108,10){\vector(-3,2){30}}
\end{picture}
\end{tabular}
& 
{\footnotesize{\begin{tabular}{@{\hspace{25pt}}l@{\hspace{15pt}}c}
$\left.\rule[0pt]{0pt}{32pt}\right\}$
& 
\setlength{\extrarowheight}{2pt}\begin{tabular}{c}
``$\acute{\mathrmbfit{tbl}}_{{\langle{f,g}\rangle}}(\mathcal{T}_{1}) \leq_{\mathcal{A}_{2}} \mathcal{T}_{2}$
\\ \underline{iff} \\
$\mathcal{T}_{1} \leq_{\mathcal{A}_{1}} \grave{\mathrmbfit{tbl}}_{{\langle{f,g}\rangle}}(\mathcal{T}_{2})$''
\end{tabular}
\end{tabular}}\normalsize}
\\ & \\ & \\
\begin{tabular}{c}
\setlength{\unitlength}{0.347pt}
\begin{picture}(180,180)(0,0)
%
\put(0,180){\makebox(0,0){\footnotesize{$I_{2}$}}}
\put(180,180){\makebox(0,0){\footnotesize{$I_{1}$}}}
\put(0,100){\makebox(0,0){\footnotesize{$X_{2}$}}}
\put(180,100){\makebox(0,0){\footnotesize{$X_{1}$}}}
\put(0,0){\makebox(0,0){\footnotesize{$Y_{2}$}}}
\put(180,0){\makebox(0,0){\footnotesize{$Y_{1}$}}}
\put(45,145){\makebox(0,0){\footnotesize{$\hat{I}_{1}$}}}
\qbezier[40](45,125)(40,70)(15,10)\put(15,10){\vector(-1,-3){0}}
\put(40,65){\makebox(0,0)[l]{\scriptsize{${\widehat{(\mbox{-})}}$}}}
\put(-6,140){\makebox(0,0)[r]{\scriptsize{$s_{2}$}}}
\put(22,116){\makebox(0,0)[l]{\scriptsize{$\hat{s}_{1}$}}}
\put(188,140){\makebox(0,0)[l]{\scriptsize{$s_{1}$}}}
\put(-8,50){\makebox(0,0)[r]{\scriptsize{$\models_{2}$}}}
\put(188,50){\makebox(0,0)[l]{\scriptsize{$\models_{1}$}}}
\put(90,194){\makebox(0,0){\scriptsize{$h = \widehat{h}{\cdot}\varepsilon_{f}$}}}
\put(100,165){\makebox(0,0){\scriptsize{$\varepsilon_{f}=\hat{f}$}}}
\put(25,168){\makebox(0,0)[l]{\scriptsize{$\widehat{h}$}}}
\put(90,114){\makebox(0,0){\scriptsize{$f$}}}
\put(90,14){\makebox(0,0){\scriptsize{$g$}}}
\put(-48,90){\makebox(0,0)[r]{\scriptsize{${t_{2}(k_{2})}$}}}
\put(232,90){\makebox(0,0)[l]{\scriptsize{${t_{1}(k_{1})}$}}}
\put(20,180){\vector(1,0){140}}
\put(20,100){\vector(1,0){140}}
\put(160,0){\vector(-1,0){140}}
\put(0,170){\vector(0,-1){60}}
\put(0,90){\vector(0,-1){80}}
\put(180,170){\vector(0,-1){60}}
\put(180,90){\vector(0,-1){80}}
\put(210,90){\oval(30,180)[r]}\put(204,0){\vector(-1,0){0}}
\put(-30,90){\oval(30,180)[l]}\put(-24,0){\vector(1,0){0}}
\put(57,147){\vector(4,1){100}}
\qbezier(10,110)(20,120)(30,130)\put(10,110){\vector(-1,-1){0}}
\qbezier(10,170)(20,160)(30,150)\put(30,150){\vector(1,-1){0}}
\qbezier(55,130)(60,130)(65,130)
\qbezier(65,140)(65,135)(65,130)
\end{picture}
\end{tabular}
&
\begin{tabular}{c}
\setlength{\unitlength}{0.45pt}
\begin{picture}(125,120)(-60,0)
\put(0,0){\begin{picture}(0,0)(0,0)
\put(-10,90){\makebox(0,0){\scriptsize{$\mathcal{T}_{2}$}}}
\put(60,120){\makebox(0,0){\shortstack{\scriptsize{$I_{2}$}\\$\overbrace{\rule{50pt}{0pt}}$}}}
\put(-55,40){\makebox(0,0){\scriptsize{$K_{2}\left\{\rule{0pt}{20pt}\right.$}}}
\put(43,90){\makebox(0,0)[l]{\scriptsize{$i_{2}{:}{s_{2}}_{i_{2}}$}}}
\put(-10,40){\makebox(0,0){\scriptsize{$k_{2}$}}}
\put(72,40){\makebox(0,0){\scriptsize{${t_{2}(k_{2})}_{i_{2}}$}}}
\put(-20,78){\line(1,0){140}}
\put(2,0){\line(0,1){100}}
\put(185,170){\makebox(0,0){\scriptsize{$h$}}}
\qbezier(80,140)(185,180)(290,140)\put(290,140){\vector(4,-1){0}}
\end{picture}}
\put(245,0){\begin{picture}(0,0)(0,0)
\put(-10,90){\makebox(0,0){\scriptsize{$\mathcal{T}_{1}$}}}
\put(60,120){\makebox(0,0){\shortstack{\scriptsize{$I_{1}$}\\$\overbrace{\rule{50pt}{0pt}}$}}}
\put(-55,40){\makebox(0,0){\scriptsize{$K_{1}\left\{\rule{0pt}{20pt}\right.$}}}
\put(43,90){\makebox(0,0)[l]{\scriptsize{$i_{1}{:}{s_{1}}_{i_{1}}$}}}
\put(-10,40){\makebox(0,0){\scriptsize{$k_{1}$}}}
\put(72,40){\makebox(0,0){\scriptsize{${t_{1}(k_{1})}_{i_{1}}$}}}
\put(-20,78){\line(1,0){140}}
\put(2,0){\line(0,1){100}}
\put(-225,-20){\makebox(0,0){\scriptsize{$k$}}}
\qbezier(-90,25)(-300,-40)(-310,20)\put(-310,20){\vector(-1,3){0}}
\end{picture}}
\end{picture}
\end{tabular}
\\ & \\ & \\
\multicolumn{2}{c}{{\scriptsize{$\begin{array}{l@{\hspace{5pt}=\hspace{5pt}}l}
k{\,\cdot\,}t_{2} 
& t_{1}{\,\cdot\,}\mathrmbfit{tup}(h,f,g)
\\
& t_{1}{\,\cdot\,}\mathrmbfit{tup}_{\mathcal{A}_{1}}(h){\,\cdot\,}\grave{\tau}_{{\langle{f,g}\rangle}}(I_{2},s_{2})
\\
& t_{1}{\,\cdot\,}\acute{\tau}_{{\langle{f,g}\rangle}}(I_{1},s_{1}){\,\cdot\,}\mathrmbfit{tup}_{\mathcal{A}_{2}}(\widehat{h})
\end{array}$}}}
\\ & \\
%
\multicolumn{2}{c}{{\scriptsize{$\begin{array}{l}
{\langle{\hat{I}_{1},\hat{s}_{1},K_{1},t_{1}{\,\cdot\,}\acute{\tau}_{{\langle{f,g}\rangle}}(I_{1},s_{1})}\rangle}
=
{\langle{f^{\ast}(I_{1},s_{1}),
{\scriptstyle\sum}
_{\acute{\tau}_{{\langle{f,g}\rangle}}(I_{1},s_{1})}(K_{1},t_{1})}\rangle}
\stackrel{\acute{\mathrmbfit{tbl}}_{{\langle{f,g}\rangle}}}{\mapsfrom}
{\langle{I_{1},s_{1},K_{1},t_{1}}\rangle}
\\
{\langle{I_{2},s_{2},K_{2},t_{2}}\rangle}
\stackrel{\grave{\mathrmbfit{tbl}}_{{\langle{f,g}\rangle}}}{\mapsto}
{\langle{{\scriptstyle\sum}_{f}(I_{2},s_{2}),(\grave{\tau}_{{\langle{f,g}\rangle}}(I_{2},s_{2}))^{\ast}(K_{2},t_{2})}\rangle}
\end{array}$}}}
\\ & \\
\multicolumn{2}{c}{{\scriptsize\begin{tabular}{p{268pt}}
This four-part figure illustrates the defining conditions on table morphisms.
It has been annotated to help guide the understanding.
The condition is symbolically stated in terms of set functions in the line of text just above.
The top left diagram illustrates the condition,
and the bottom left diagram expands on this.
The top right diagram text is more detailed in terms of 
a source row (tuple) $k_{1} \in K_{1}$ and a target column (attribute) $i_{2} \in I_{2}$.
Here we see appearance of the infomorphism condition 
\[
g({t_{1}(k_{1})}_{i_{1}}\!) \models_{\mathcal{A}_{2}} {s_{2}}_{i_{2}} 
\;\;\text{\underline{iff}}\;\;
{t_{1}(k_{1})}_{i_{1}} \models_{\mathcal{A}_{1}} f({s_{2}}_{i_{2}}).
\]
Finally,
the bottom right figure illustrates the meaning of the morphism's defining condition 
with respect to source/target tables
$\mathcal{T}_{1}$ and $\mathcal{T}_{2}$.
\end{tabular}}}
\end{tabular}
\end{center}
\caption{\texttt{FOLE} Table Morphism}
\label{fig:tbl:mor}
\end{figure}
%

\paragraph{Table Morphisms.}

A table morphism (morphism of database relations)
\[\mbox{\footnotesize{$
\mathcal{T}_{2} = {\langle{{\langle{I_{2},s_{2},\mathcal{A}_{2}}\rangle},K_{2},t_{2}}\rangle}
\xleftarrow{{\langle{{\langle{h,f,g}\rangle},k}\rangle}} 
{\langle{{\langle{I_{1},s_{1},\mathcal{A}_{1}}\rangle},K_{1},t_{1}}\rangle} = \mathcal{T}_{1}
$}\normalsize}\]
consists of 
a signed domain morphism
${\langle{I_{2},s_{2},\mathcal{A}_{2}}\rangle}\xrightarrow{{\langle{h,f,g}\rangle}}{\langle{I_{1},s_{1},\mathcal{A}_{1}}\rangle}$
and
a key function $K_{2}\xleftarrow{k}K_{1}$,
which satisfy the naturality condition
$k{\;\cdot\;}t_{2} = t_{1}{\;\cdot\;}\mathrmbfit{tup}(h,f,g)$.
\footnote{Since
the table tuple function embodies the entity/domain integrity constraints, 
this condition on morphisms asserts the preservation of data integrity.}
%
Hence,
a table morphism
$\mathcal{T}_{2} = {\langle{\mathcal{S}_{2},\mathcal{A}_{2},K_{2},t_{2}}\rangle}
\xleftarrow{{\langle{h,f,g,k}\rangle}} 
{\langle{\mathcal{S}_{1},\mathcal{A}_{1},K_{1},t_{1}}\rangle} = \mathcal{T}_{1}$
consists of 
a signature morphism
$\mathcal{A}_{2}={\langle{I_{2},s_{2},X_{2}}\rangle}\xrightarrow{{\langle{h,f}\rangle}}{\langle{I_{1},s_{1},X_{1}}\rangle}=\mathcal{S}_{1}$ and
a type domain morphism 
$\mathcal{A}_{2} = {\langle{X_{2},Y_{2},\models_{\mathcal{A}_{2}}}\rangle} 
\xrightleftharpoons{{\langle{f,g}\rangle}} 
{\langle{X_{1},Y_{1},\models_{\mathcal{A}_{1}}}\rangle} = \mathcal{A}_{1}$
with common sort function $X_{2}\xrightarrow{f}X_{1}$,
and a key function $K_{2}\xleftarrow{k}K_{1}$,
which satisfy the naturality condition above.

Table morphisms are illustrated in Fig.~\ref{fig:tbl:mor}.
Here we see that table morphisms have the pleasing property that
corresponding entries in the source and target tables
satisfy the infomorphism condition from the theory of information flow
(Barwise and Seligman~\cite{barwise:seligman:97}).
Composition of morphisms is defined component-wise.
Let 
\[\mbox{\footnotesize{$
\mathrmbf{Set}\xleftarrow{\;\mathrmbfit{key}\;}
\mathrmbf{Tbl} = {\bigl(\mathrmbf{Set}{\,\downarrow\,}\mathrmbfit{tup}\bigr)}
\xrightarrow{\;\mathrmbfit{dom}\;}\mathrmbf{Dom}^{\mathrm{op}}
$}\normalsize}\]
denote the comma context of tables 
(Fig.~\ref{fig:tbl:cxt})
with the key/signed-domain projection passages. 
There is a defining tuple bridge 
$\mathrmbfit{key}\xRightarrow{\;\tau\,\;}\mathrmbfit{dom}{\;\circ\;}\mathrmbfit{tup}$,
whose $\mathcal{T}^{\text{th}}$ component is the tuple function 
$K\xrightarrow{t}\mathrmbfit{tup}(\mathcal{D})$.
%
Composition yields
signature/type-domain projection passages
$\mathrmbf{List}^{\mathrm{op}}\xleftarrow{\mathrmbfit{sign}}\mathrmbf{Tbl}\xrightarrow{\mathrmbfit{data}}\mathrmbf{Cls}^{\mathrm{op}}$.
\begin{figure}
\begin{center}
{{\begin{tabular}{c@{\hspace{30pt}}c}
{{\begin{tabular}{c}
\setlength{\unitlength}{0.5pt}
\begin{picture}(100,120)(-20,0)
\put(0,118){\makebox(0,0){\footnotesize{$\mathrmbf{Tbl}$}}}
\put(60,60){\makebox(0,0){\footnotesize{$\mathrmbf{Dom}^{\mathrm{op}}$}}}
\put(0,0){\makebox(0,0){\footnotesize{$\mathrmbf{Set}$}}}
\put(-45,60){\makebox(0,0)[r]{\scriptsize{$\mathrmbfit{key}$}}}
\put(36,95){\makebox(0,0)[l]{\scriptsize{$\mathrmbfit{dom}$}}}
\put(36,24){\makebox(0,0)[l]{\scriptsize{$\mathrmbfit{tup}$}}}
\put(0,65){\makebox(0,0){{$\xRightarrow{\;\;\tau\;\;}$}}}
\put(15,105){\vector(1,-1){30}}
\put(45,45){\vector(-1,-1){30}}
\qbezier(-12,105)(-60,60)(-12,15)\put(-12,15){\vector(1,-1){0}}
\end{picture}
\end{tabular}}}
&
{{\begin{tabular}{c}
\setlength{\unitlength}{0.55pt}
\begin{picture}(320,160)(0,-35)
\put(20,80){\makebox(0,0){\footnotesize{$\mathrmbf{Set}$}}}
\put(20,0){\makebox(0,0){\footnotesize{$\mathrmbf{1}$}}}
\put(140,80){\makebox(0,0){\footnotesize{$\mathrmbf{Tbl}$}}}
\put(285,80){\makebox(0,0){\footnotesize{$\mathrmbf{Cls}^{\mathrm{op}}$}}}
\put(145,0){\makebox(0,0){\footnotesize{$\mathrmbf{List}^{\mathrm{op}}$}}}
\put(285,0){\makebox(0,0){\footnotesize{$\mathrmbf{Set}^{\mathrm{op}}$}}}
\put(215,40){\makebox(0,0){\footnotesize{$\mathrmbf{Dom}^{\mathrm{op}}$}}}
\put(80,92){\makebox(0,0){\scriptsize{$\mathrmbfit{key}$}}}
\put(182,66){\makebox(0,0)[l]{\scriptsize{$\mathrmbfit{dom}$}}}
\put(210,92){\makebox(0,0){\scriptsize{$\mathrmbfit{data}$}}}
\put(134,43){\makebox(0,0)[r]{\scriptsize{$\mathrmbfit{sign}$}}}
\put(210,-12){\makebox(0,0){\scriptsize{$\mathrmbfit{sort}^{\mathrm{op}}$}}}
\put(288,40){\makebox(0,0)[l]{\scriptsize{$\mathrmbfit{sort}^{\mathrm{op}}$}}}
\put(115,80){\vector(-1,0){70}}
\put(165,80){\vector(1,0){90}}
\put(115,0){\vector(-1,0){80}}
\put(165,0){\vector(1,0){90}}
\put(20,65){\vector(0,-1){50}}
\put(140,65){\vector(0,-1){50}}
\put(280,65){\vector(0,-1){50}}
\put(155,70){\vector(2,-1){36}}
\put(190,30){\vector(-2,-1){36}}
\put(230,50){\vector(2,1){36}}
\qbezier(250,20)(255,20)(260,20)
\qbezier(250,20)(250,15)(250,10)
\end{picture}
\end{tabular}}}
\end{tabular}}}
\end{center}
\caption{\texttt{FOLE} Table Mathematical Context}
\label{fig:tbl:cxt}
\end{figure}
We can have three indexing contexts for tables (above diagram):
signatures $\mathrmbf{List}$, type domains $\mathrmbf{Cls}$ and signed domains $\mathrmbf{Dom}$.
Each has their uses:
signature indexing follows the true formal-semantics distinction,
type domain indexing proves that the context of tables is complete (\S~\ref{sub:sub:sec:props:groth})
(and the fibers help explain database fibers), and
signed domain indexing proves that the context of tables is cocomplete (\S~\ref{sub:sub:sec:props:comma}).
Corresponding to this indexing
(as illustrated in Fig.~\ref{fig:tbl:gro:constr}),
there are two chains of fiber contexts:
fibers indexed by a signed domain 
$\mathcal{D} = {\langle{I,s,\mathcal{A}}\rangle}$
are smallest, and contained in either
fibers indexed by a type domain $\mathcal{A} = {\langle{X,Y,\models_{\mathcal{A}}}\rangle}$ or
fibers indexed by a signature $\mathcal{S} = {\langle{I,s,X}\rangle}$.
%


\paragraph{Restatement:}
We now sharpen the definition for the the context of tables.
This will be useful for defining and working with relational databases.
The fibered context of tables is the comma mathematical context 
\[\mbox{\footnotesize{$
\mathrmbf{Tbl} 
= {\bigl(\mathrmbf{Set}{\,\downarrow\,}\mathrmbfit{tup}\bigr)}
$}\normalsize}\]
for the opspan of passages
$\mathrmbf{Set}\xrightarrow{\mathrmbfit{1}}\mathrmbf{Set}\xleftarrow{\mathrmbfit{tup}}\mathrmbf{Dom}^{\mathrm{op}}$.
It has the key and signed domain projection passages
$\mathrmbf{Set}\xleftarrow{\mathrmbfit{key}}\mathrmbf{Tbl}\xrightarrow{\mathrmbfit{dom}}\mathrmbf{Dom}^{\mathrm{op}}$
and the defining tuple bridge
$\mathrmbfit{key}\xRightarrow{\;\tau\;}\mathrmbfit{dom}{\;\circ\;}\mathrmbfit{tup}$
such that
for any span of passages
$\mathrmbf{Set}\xleftarrow{\mathrmbfit{K}}\mathrmbf{R}\xrightarrow{\mathrmbfit{Q}}\mathrmbf{Dom}^{\mathrm{op}}$,
there is a bijection
$\mathrmbfit{T}\mapsto\widehat{\tau}=\mathrmbfit{T}{\;\circ\;}\tau$
between 
\begin{itemize}
\item 
passages $\mathrmbf{R}\xrightarrow{\mathrmbfit{T}}\mathrmbf{Tbl}$
satisfying $\mathrmbfit{T}{\;\circ\;}\mathrmbfit{key} = \mathrmbfit{K}$ 
and $\mathrmbfit{T}{\;\circ\;}\mathrmbfit{dom} = \mathrmbfit{Q}$, and
\item 
bridges $\mathrmbfit{K}\xRightarrow{\;\widehat{\tau}\;}\mathrmbfit{Q}{\;\circ\;}\mathrmbfit{tup}$.
\end{itemize}
By composition,
there are also signature and classification projection passages
$\mathrmbf{List}^{\mathrm{op}}\!\xleftarrow{\mathrmbfit{sign}}\mathrmbf{Tbl}\xrightarrow{\mathrmbfit{data}}\mathrmbf{Cls}^{\mathrm{op}}$.
%

\subsection{Signed Domain Indexing}\label{sub:sec:tbl:sign:dom}

In this section we show that the context of tables 
is a fibered context over signed domains.
We first define the table fiber for fixed signed domain.
We next move between table fibers along signed domain morphisms.
Finally,
we invoke the Grothendieck construction indexed by signed domains.

\paragraph{Fiber Contexts (small-size).}

Let ${\langle{I,s,\mathcal{A}}\rangle}$ be a fixed signed domain.
The fiber mathematical context of ${\langle{I,s,\mathcal{A}}\rangle}$-tables is the comma context
\[\mbox{\footnotesize{$
\mathrmbf{Tbl}(I,s,\mathcal{A}) 
= \mathrmbf{Tbl}_{\mathcal{A}}(I,s)
= {\bigl(\mathrmbf{Set}{\,\downarrow\,}\mathrmbfit{tup}_{\mathcal{A}}(I,s)\bigr)}
$}\normalsize}\]
associated with the tuple set (constant passage)
$\mathrmbf{1}\xrightarrow[\mathrmbfit{tup}(I,s,\mathcal{A})]{\;\mathrmbfit{tup}_{\mathcal{A}}(I,s))\;}\mathrmbf{Set}$.
It has the key and trivial projection passages
$\mathrmbf{Set}\xleftarrow{\;\mathrmbfit{key}_{\mathcal{A}}(I,s)\;}\mathrmbf{Tbl}_{\mathcal{A}}(I,s)\xrightarrow{\;\Delta\;}\mathrmbf{1}$
and the defining bridge
$\mathrmbfit{key}_{\mathcal{A}}(I,s)\xRightarrow{\,\tau_{\mathcal{A}}(I,s)\;}\Delta{\;\circ\;}\mathrmbfit{tup}_{\mathcal{A}}(I,s)$. 
A $\mathrmbf{Tbl}_{\mathcal{A}}(I,s)$-object $\mathcal{T} = {\langle{K,t}\rangle}$, 
called an ${\langle{I,s,\mathcal{A}}\rangle}$-table,
consists of a set $K$ of (primary) keys
and a tuple function
$K\xrightarrow{\,t\;}\mathrmbfit{tup}_{\mathcal{A}}(I,s)$
mapping each key to its descriptor $\mathcal{A}$-tuple of type (signature) ${\langle{I,s}\rangle}$.
A $\mathrmbf{Tbl}_{\mathcal{A}}(I,s)$-morphism 
$\mathcal{T}' = {\langle{K',t'}\rangle}\xleftarrow{k}{\langle{K,t}\rangle} = \mathcal{T}$
is a key function $K'\xleftarrow{\,k\,}K$
that preserves descriptors by
satisfying the naturality condition $k{\;\cdot\;}t' = t$.
%

\paragraph{Fibered Context (large-size).}

The fibered context of tables 
$\mathrmbf{Tbl}^{\mathrm{op}}\xrightarrow{\;\mathrmbfit{dom}\;}\mathrmbf{Dom}$
is defined as follows.
We use the same definitions as in \S~\ref{sub:sec:tbl}.
%
A $\mathrmbf{Tbl}$-object 
$\mathcal{T} = {\langle{I,s,\mathcal{A},K,t}\rangle}$, 
called an table,
consists of 
a signed domain $\mathrmbfit{dom}(\mathcal{T}) = {\langle{I,s,\mathcal{A}}\rangle}$ and 
an ${\langle{I,s,\mathcal{A}}\rangle}$-table ${\langle{K,t}\rangle}$.
A $\mathrmbf{Tbl}$-morphism 
$\mathcal{T}_{2} = {\langle{I_{2},s_{2},\mathcal{A}_{2},K_{2},t_{2}}\rangle} 
\xleftarrow{\,{\langle{h,f,g,k}\rangle}\;} 
{\langle{I_{1},s_{1},\mathcal{A}_{1},K_{1},t_{1}}\rangle} = \mathcal{T}_{1}$
called an table morphism,
consists of  Morphism
a signed domain morphism
${\langle{I_{2},X_{2},\mathcal{A}_{2}}\rangle}\xrightarrow[\mathrmbfit{dom}(h,f,g,k)]{{\langle{h,f,g}\rangle}}{\langle{I_{1},X_{1},\mathcal{A}_{1}}\rangle}$
and 
a key function
$K_{2}\xleftarrow{\,k\;}K_{1}$
satisfying the naturality condition 
$k{\;\cdot\;}t_{2} = t_{1}{\;\cdot\;}\mathrmbfit{tup}(h,f,g)$.
This condition gives two alternate and adjoint definitions.
%
\begin{figure}
\begin{center}
{{\begin{tabular}{@{\hspace{30pt}}r@{\hspace{30pt}}c@{\hspace{30pt}}l}
{{\begin{tabular}{c}
{\setlength{\unitlength}{0.45pt}\begin{picture}(240,180)(0,-55)
\put(0,80){\makebox(0,0){\footnotesize{$K_{2}$}}}
\put(120,80){\makebox(0,0){\footnotesize{$K_{1}$}}}
\put(240,80){\makebox(0,0){\footnotesize{$K_{1}$}}}
\put(60,-25){\makebox(0,0){\footnotesize{$\mathrmbfit{tup}(I_{2},s_{2},\mathcal{A}_{2})$}}}
\put(240,-25){\makebox(0,0){\footnotesize{$\mathrmbfit{tup}(I_{1},s_{1},\mathcal{A}_{1})$}}}
\put(30,-56){\makebox(0,0){\footnotesize{$\underset{\textstyle{\mathcal{T}_{2}}}{\underbrace{\rule{30pt}{0pt}}}$}}}
\put(105,-58){\makebox(0,0){\footnotesize{$\underset{\textstyle{{\scriptstyle\sum}_{{\langle{h,f,g}\rangle}}(\mathcal{T}_{1})}}{\underbrace{\rule{30pt}{0pt}}}$}}}
\put(240,-56){\makebox(0,0){\footnotesize{$\underset{\textstyle{\mathcal{T}_{1}}}{\underbrace{\rule{40pt}{0pt}}}$}}}
\put(5,26){\makebox(0,0)[l]{\scriptsize{$t_{2}$}}}
\put(102,26){\makebox(0,0)[l]{\scriptsize{$t_{1}{\,\cdot\,}\mathrmbfit{tup}(h,f,g)$}}}
\put(248,26){\makebox(0,0)[l]{\scriptsize{$t_{1}$}}}
\put(160,-10){\makebox(0,0){\scriptsize{$\mathrmbfit{tup}(h,f,g)$}}}
\put(120,125){\makebox(0,0){\scriptsize{$k$}}}
\put(60,92){\makebox(0,0){\scriptsize{$k$}}}
\put(100,80){\vector(-1,0){80}}
\put(4,59){\vector(2,-3){46}}
\put(108,59){\vector(-2,-3){46}}
\put(190,80){\makebox(0,0){\normalsize{$=$}}}
\put(240,60){\vector(0,-1){68}}
\put(165,-25){\vector(-1,0){35}}
\put(10,100){\oval(20,20)[tl]}
\put(10,110){\line(1,0){220}}
\put(230,100){\oval(20,20)[tr]}
\put(0,94){\vector(0,-1){0}}
\put(60,-98){\makebox(0,0){\footnotesize{$\underset{\textstyle{\mathrmbf{Tbl}(I_{2},s_{2},\mathcal{A}_{2})}}{\underbrace{\rule{60pt}{0pt}}}$}}}
\end{picture}}
\end{tabular}}}
${\;\;\;\;\;\;\;\;\;\;\cong\;\;\;\;\;\;\;\;\;\;}$
{{\begin{tabular}{c}
{\setlength{\unitlength}{0.45pt}\begin{picture}(240,180)(-120,-55)
\put(-120,80){\makebox(0,0){\footnotesize{$K_{2}$}}}
\put(0,80){\makebox(0,0){\footnotesize{$K_{1}$}}}
\put(120,80){\makebox(0,0){\footnotesize{$K_{1}$}}}
\put(60,-25){\makebox(0,0){\footnotesize{$\mathrmbfit{tup}(I_{1},s_{1},\mathcal{A}_{1})$}}}
\put(-120,-25){\makebox(0,0){\footnotesize{$\mathrmbfit{tup}(I_{2},s_{2},\mathcal{A}_{2})$}}}
\put(-120,-56){\makebox(0,0){\footnotesize{$\underset{\textstyle{\mathcal{T}_{2}}}{\underbrace{\rule{40pt}{0pt}}}$}}}
\put(30,-60){\makebox(0,0){\footnotesize{$\underset{\textstyle{{\langle{h,f,g}\rangle}^{\ast}(\mathcal{T}_{2})}}{\underbrace{\rule{30pt}{0pt}}}$}}}
\put(100,-56){\makebox(0,0){\footnotesize{$\underset{\textstyle{\mathcal{T}_{1}}}{\underbrace{\rule{30pt}{0pt}}}$}}}
\put(-60,92){\makebox(0,0){\scriptsize{$e$}}}
\put(15,26){\makebox(0,0)[r]{\scriptsize{$\hat{t}$}}}
\put(105,26){\makebox(0,0)[l]{\scriptsize{$t_{1}$}}}
\put(-125,26){\makebox(0,0)[r]{\scriptsize{$t_{2}$}}}
\put(-25,-10){\makebox(0,0){\scriptsize{$\mathrmbfit{tup}(h,f,g)$}}}
\put(60,92){\makebox(0,0){\scriptsize{$k'$}}}
\put(0,125){\makebox(0,0){\scriptsize{$k$}}}
\put(100,80){\vector(-1,0){80}}
\put(-25,80){\vector(-1,0){70}}
\put(4,59){\vector(2,-3){46}}
\put(116,59){\vector(-2,-3){46}}
\put(-120,60){\vector(0,-1){68}}
\put(-15,-25){\vector(-1,0){35}}
\put(-110,100){\oval(20,20)[tl]}
\put(-110,110){\line(1,0){220}}
\put(110,100){\oval(20,20)[tr]}
\put(-120,94){\vector(0,-1){0}}
%
\qbezier(-24,45)(-16,45)(-8,45)
\qbezier(-24,45)(-24,53)(-24,62)
\put(60,-98){\makebox(0,0){\footnotesize{$\underset{\textstyle{\mathrmbf{Tbl}(I_{1},s_{1},\mathcal{A}_{1})}}{\underbrace{\rule{60pt}{0pt}}}$}}}
\end{picture}}
\end{tabular}}}
\\&&
\\&&
\end{tabular}}}
\end{center}
\caption{Table Morphism: Signed Domain}
\label{fig:tbl:mor:large}
\end{figure}
%
In terms of fibers,
an table morphism
consists of
a signed domain morphism
${\langle{I_{2},s_{2},\mathcal{A}_{2}}\rangle}\xrightarrow{{\langle{h,f,g}\rangle}}{\langle{I_{1},s_{1},\mathcal{A}_{1}}\rangle}$
\underline{and} 
either a morphism 
$\mathcal{T}_{2}\xleftarrow{\;k\;}{\scriptstyle\sum}_{{\langle{h,f,g}\rangle}}(\mathcal{T}_{1})$
in the fiber context $\mathrmbf{Tbl}(I_{2},s_{2},\mathcal{A}_{2})$ 
or a morphism 
${\langle{h,f,g}\rangle}^{\ast}(\mathcal{T}_{2})\xleftarrow{\,k'\,}\mathcal{T}_{1}$
in the fiber context $\mathrmbf{Tbl}(I_{1},s_{1},\mathcal{A}_{1})$.
\begin{equation}
{{\begin{picture}(120,10)(0,-4)
\put(60,0){\makebox(0,0){\footnotesize{$
\underset{\textstyle{\text{in}\;\mathrmbf{Tbl}(I_{2},s_{2},\mathcal{A}_{2})}}
{\mathcal{T}_{2}\xleftarrow{\;k\;}{\scriptstyle\sum}_{{\langle{h,f,g}\rangle}}(\mathcal{T}_{1})}
{\;\;\;\;\;\;\;\;\rightleftarrows\;\;\;\;\;\;\;\;}
\underset{\textstyle{\text{in}\;\mathrmbf{Tbl}(I_{1},s_{1},\mathcal{A}_{1})}}
{{\langle{h,f,g}\rangle}^{\ast}(\mathcal{T}_{2})\xleftarrow{\,k'\,}\mathcal{T}_{1}}
$}}}
\end{picture}}}
\end{equation}
The 
${\langle{I_{2},s_{2},\mathcal{A}_{2}}\rangle}$-table morphism 
$\mathcal{T}_{2}\xleftarrow{\;k\;}{\scriptstyle\sum}_{{\langle{h,f,g}\rangle}}(\mathcal{T}_{1})$
is the composition (RHS of Fig.~\ref{fig:tbl:mor:large}) 
of the fiber morphism
${\scriptstyle\sum}_{{\langle{h,f,g}\rangle}}\Bigl({\langle{h,f,g}\rangle}^{\ast}(\mathcal{T}_{2})\xleftarrow{k'}\mathcal{T}_{1}\Bigr)$
with the $\mathcal{T}_{2}^{\,\text{th}}$ counit component 
$\mathcal{T}_{2}\xleftarrow{e}{\scriptstyle\sum}_{{\langle{h,f,g}\rangle}}\Bigl({\langle{h,f,g}\rangle}^{\ast}(\mathcal{T}_{2})\Bigr)$
for the fiber adjunction
\[\mbox{\footnotesize{
$
\mathrmbf{Tbl}(I_{2},s_{2},\mathcal{A}_{2})
{\;\xleftrightharpoons[\;\;\;\;{\langle{h,f,g}\rangle}^{\ast}]{\;{\scriptscriptstyle\sum}_{{\langle{h,f,g}\rangle}}\;\;\;}\;}
\mathrmbf{Tbl}(I_{1},s_{1},\mathcal{A}_{1})$.
}\normalsize}\]
This fiber adjunction 
(top part of Tbl.~\ref{tbl:tbl-rel:refl:sign:dom:mor})
is a component of
the signed domain indexed adjunction of tables
$\mathrmbf{Dom}^{\mathrm{op}}\!\xrightarrow{\;\mathrmbfit{tbl}\;}\mathrmbf{Adj}$.
\begin{table}
\begin{center}
\begin{tabular}{c}
{\scriptsize\setlength{\extrarowheight}{2pt}$\begin{array}[b]{c}
\\
{\langle{I_{2},X_{2},\mathcal{A}_{2}}\rangle}\xrightarrow{{\langle{h,f,g}\rangle}}{\langle{I_{1},X_{1},\mathcal{A}_{1}}\rangle}
\\
\mathrmbfit{tup}_{\mathcal{A}_{2}}(I_{2},s_{2})
\xleftarrow{\mathrmbfit{tup}(h,f,g)}
\mathrmbfit{tup}_{\mathcal{A}_{1}}(I_{1},s_{1})
\rule[-10pt]{0pt}{10pt}
\end{array}$}
\\\\
{{\begin{tabular}[b]{c}
\setlength{\unitlength}{0.55pt}
\begin{picture}(320,220)(-40,-30)
\put(0,160){\makebox(0,0){\footnotesize{$
\underset{=\,\bigl(\mathrmbf{Set}{\,\downarrow\,}\mathrmbfit{tup}_{\mathcal{A}_{2}}(I_{2},s_{2})\bigr)}
{\mathrmbf{Tbl}_{\mathcal{A}_{2}}(I_{2},s_{2})}$}}}
\put(0,0){\makebox(0,0){\footnotesize{$
\underset{=\,{\wp}\mathrmbfit{tup}_{\mathcal{A}_{2}}(I_{2},s_{2})}
{\mathrmbf{Rel}_{\mathcal{A}_{2}}(I_{2},s_{2})}$}}}
\put(240,160){\makebox(0,0){\footnotesize{$
\underset{=\,\bigl(\mathrmbf{Set}{\,\downarrow\,}\mathrmbfit{tup}_{\mathcal{A}_{1}}(I_{1},s_{1})\bigr)}
{\mathrmbf{Tbl}_{\mathcal{A}_{1}}(I_{1},s_{1})}$}}}
\put(240,0){\makebox(0,0){\footnotesize{$
\underset{=\,{\wp}\mathrmbfit{tup}_{\mathcal{A}_{1}}(I_{1},s_{1})}
{\mathrmbf{Rel}_{\mathcal{A}_{1}}(I_{1},s_{1})}$}}}
%
\put(120,190){\makebox(0,0){\scriptsize{${\scriptstyle\sum}_{{\langle{h,f,g}\rangle}}$}}}
\put(122,162){\makebox(0,0){\scriptsize{${\langle{h,f,g}\rangle}^{\ast}$}}}
\put(120,130){\makebox(0,0){\scriptsize{${\scriptstyle\prod}_{{\langle{h,f,g}\rangle}}$}}}
\put(120,30){\makebox(0,0){\scriptsize{$\exists_{{\langle{h,f,g}\rangle}}$}}}
\put(122,2){\makebox(0,0){\scriptsize{${\langle{h,f,g}\rangle}^{{\scriptscriptstyle-}1}$}}}
\put(120,-30){\makebox(0,0){\scriptsize{$\forall_{{\langle{h,f,g}\rangle}}$}}}
\put(-15,80){\makebox(0,0)[r]{\scriptsize{$\mathrmbfit{im}_{\mathcal{A}_{2}}(I_{2},s_{2})$}}}
\put(16,80){\makebox(0,0)[l]{\scriptsize{$\mathrmbfit{inc}_{\mathcal{A}_{2}}(I_{2},s_{2})$}}}
\put(225,80){\makebox(0,0)[r]{\scriptsize{$\mathrmbfit{im}_{\mathcal{A}_{1}}(I_{1},s_{1})$}}}
\put(256,80){\makebox(0,0)[l]{\scriptsize{$\mathrmbfit{inc}_{\mathcal{A}_{1}}(I_{1},s_{1})$}}}
\put(170,180){\vector(-1,0){100}}
\qbezier(70,160)(85,160)(100,160)\qbezier(140,160)(155,160)(170,160)\put(170,160){\vector(1,0){0}}
\put(170,140){\vector(-1,0){100}}
\put(170,20){\vector(-1,0){100}}
\qbezier(70,0)(85,0)(100,0)\qbezier(140,0)(155,0)(170,0)\put(170,0){\vector(1,0){0}}
\put(170,-20){\vector(-1,0){100}}
\put(-10,125){\vector(0,-1){90}}
\put(10,35){\vector(0,1){90}}
\put(230,125){\vector(0,-1){90}}
\put(250,35){\vector(0,1){90}}
\end{picture}
\end{tabular}}}
\\\\
{\scriptsize\setlength{\extrarowheight}{2.5pt}$\begin{array}[b]{|l|}
\hline
{\scriptstyle\sum}_{{\langle{h,f,g}\rangle}} \dashv {\langle{h,f,g}\rangle}^{\ast} \dashv {\scriptstyle\prod}_{{\langle{h,f,g}\rangle}} 
\\
\exists_{{\langle{h,f,g}\rangle}} \dashv {\langle{h,f,g}\rangle}^{{\scriptscriptstyle-}1} \dashv \forall_{{\langle{h,f,g}\rangle}} 
\\
\mathrmbfit{im}_{\mathcal{A}_{2}}(I_{2},s_{2}) \dashv \mathrmbfit{inc}_{\mathcal{A}_{2}}(I_{2},s_{2}),\; 
\mathrmbfit{im}_{\mathcal{A}_{1}}(I_{1},s_{1}) \dashv \mathrmbfit{inc}_{\mathcal{A}_{1}}(I_{1},s_{1}) 
\\ \hline
\mathrmbfit{inc}_{\mathcal{A}_{2}}(I_{2},s_{2}){\;\circ\;}{\langle{h,f,g}\rangle}^{\ast}	
 = {\langle{h,f,g}\rangle}^{{\scriptscriptstyle-}1}{\;\circ\;}\mathrmbfit{inc}_{\mathcal{A}_{1}}(I_{1},s_{1}) 
\\
\mathrmbfit{inc}_{\mathcal{A}_{1}}(I_{1},s_{1}){\;\circ\;}{\scriptstyle\prod}_{{\langle{h,f,g}\rangle}} 
 = \forall_{{\langle{h,f,g}\rangle}}{\;\circ\;}\mathrmbfit{inc}_{\mathcal{A}_{2}}(I_{2},s_{2}) 
\\ \hline
{\scriptstyle\sum}_{{\langle{h,f,g}\rangle}}{\;\circ\;}\mathrmbfit{im}_{\mathcal{A}_{2}}(I_{2},s_{2})
 \cong \mathrmbfit{im}_{\mathcal{A}_{1}}(I_{1},s_{1}){\;\circ\;}\exists_{{\langle{h,f,g}\rangle}}
\\
{\langle{h,f,g}\rangle}^{\ast}{\;\circ\;}\mathrmbfit{im}_{\mathcal{A}_{1}}(I_{1},s_{1})
 \cong \mathrmbfit{im}_{\mathcal{A}_{2}}(I_{2},s_{2}){\;\circ\;}{\langle{h,f,g}\rangle}^{{\scriptscriptstyle-}1}
\\ \hline
\end{array}$}
\\ \\
{\textit{small fibers -- long distance}}
\end{tabular}
\end{center}
%
\caption{Reflection: Signed Domain}
\label{tbl:tbl-rel:refl:sign:dom:mor}
\end{table}
\begin{theorem}\label{thm:fib:cxt:tbl:dom}
The fibered context of tables
$\mathrmbf{Tbl}\xrightarrow{\;\mathrmbfit{dom}\;}\mathrmbf{Dom}^{\mathrm{op}}$
is the Grothendieck construction $\int_{\mathrmbf{Dom}}$ of 
the signed domain indexed adjunction
$\mathrmbf{Dom}^{\mathrm{op}}\!\xrightarrow{\;\mathrmbfit{tbl}\;}\mathrmbf{Adj}$.$^{\ref{grothendieck}}$
\end{theorem}
%

\subsection{Signature Indexing}\label{sub:sec:tbl:sign}

In this section we show that the context of tables 
is a fibered context over signatures.
We first define the table fiber for fixed signature.
We next move between table fibers along signature morphisms.
Finally,
we invoke the Grothendieck construction indexed by signatures.


\subsubsection{Lower Aspect.}\label{sub:sub:sec:tbl:sign:low}

Let $\mathcal{S} = {\langle{I,s,X}\rangle}$ be a fixed signature.
For database tables,
the signature (header) $\mathcal{S}$
consists of a fixed sort set $X$ and a fixed $X$-signature ${\langle{I,s}\rangle}$.
Here, we show that the context of $\mathcal{S}$-tables $\mathrmbf{Tbl}(\mathcal{S})$ is fibered over $X$-sorted type domains
$\mathrmbf{Tbl}(\mathcal{S})\xrightarrow{\mathrmbfit{data}_{\mathcal{S}}}\mathrmbf{Cls}(X)^{\mathrm{op}}$.
We use the Grothendieck construction $\int_{\mathrmbf{Cls}(X)}$
on the indexed adjunction
$\mathrmbf{Cls}(X)\xrightarrow{\;\mathrmbfit{tbl}_{\mathcal{S}}}\mathrmbf{Adj}
:\mathcal{A}\mapsto{\mathrmbf{Tbl}_{\mathcal{S}}(\mathcal{A})}$.

\paragraph{Fiber(ed) Contexts (medium-size).}

The fiber(ed) mathematical context of $\mathcal{S}$-tables is the comma context
\[\mbox{\footnotesize{$
\mathrmbf{Set} 
\xleftarrow{\mathrmbfit{key}_{\mathcal{S}}}
\mathrmbf{Tbl}(\mathcal{S}) = {\bigl(\mathrmbf{Set}{\,\downarrow\,}\mathrmbfit{tup}_{\mathcal{S}}\bigr)}
\xrightarrow{\mathrmbfit{data}_{\mathcal{S}}} 
{\mathrmbf{Cls}(X)}^{\mathrm{op}}.
$}}\]
associated with the signature tuple passage
$\mathrmbf{Cls}(X)^{\mathrm{op}}\xrightarrow{\;\mathrmbfit{tup}_{\mathcal{S}}\;}\mathrmbf{Set}$
defined in \S~\ref{sub:sub:sec:tup:bridge:sign}.
A $\mathrmbf{Tbl}(\mathcal{S})$-object $\mathcal{T} = {\langle{\mathcal{A},K,t}\rangle}$,
called an $\mathcal{S}$-table,
consists of 
an $X$-sorted type domain $\mathcal{A} = {\langle{X,Y,\models_{\mathcal{A}}}\rangle}$
with data-type collection
$\{ \mathcal{A}_{x} \mid x \in X \}$,
a set $K$ of (primary) keys
and a tuple function
$K\xrightarrow{\,t\;}\mathrmbfit{tup}_{\mathcal{A}}(I,s)
= \prod_{i\in{I}}\mathcal{A}_{s(i)}$
mapping each key to its descriptor $\mathcal{A}$-tuple of type (signature) ${\langle{I,s}\rangle}$.
A $\mathrmbf{Tbl}(\mathcal{S})$-morphism
$\mathcal{T} = {\langle{\mathcal{A},K,t}\rangle} 
\xleftarrow{\;{\langle{g,k}\rangle}\;}
{\langle{\widetilde{\mathcal{A}},\widetilde{K},\widetilde{t}}\rangle} = \widetilde{\mathcal{T}}$
consists of
an $X$-sorted type domain morphism
$\mathcal{A}
\xrightleftharpoons{{\langle{\mathrmit{1}_{X},g}\rangle}} 
\widetilde{\mathcal{A}}$
and
a key function $K\xleftarrow{k}\widetilde{K}$,
which satisfies the condition
$k{\,\cdot\,}t = \widetilde{t}{\,\cdot\,}\mathrmbfit{tup}_{\mathcal{S}}(g)$.
\begin{figure}
\begin{center}
{{\begin{tabular}{@{\hspace{30pt}}r@{\hspace{30pt}}c@{\hspace{30pt}}l}
{{\begin{tabular}{c}
{\setlength{\unitlength}{0.45pt}\begin{picture}(240,180)(0,-55)
\put(0,80){\makebox(0,0){\footnotesize{$K$}}}
\put(120,80){\makebox(0,0){\footnotesize{$\widetilde{K}$}}}
\put(240,80){\makebox(0,0){\footnotesize{$\widetilde{K}$}}}
\put(60,-25){\makebox(0,0){\footnotesize{$\mathrmbfit{tup}_{\mathcal{S}}(\mathcal{A})$}}}
\put(240,-25){\makebox(0,0){\footnotesize{$\mathrmbfit{tup}_{\mathcal{S}}(\widetilde{\mathcal{A}})$}}}
\put(30,-56){\makebox(0,0){\footnotesize{$\underset{\textstyle{\mathcal{T}}}{\underbrace{\rule{30pt}{0pt}}}$}}}
\put(105,-58){\makebox(0,0){\footnotesize{$\underset{\textstyle{{\scriptstyle\sum}_{{g}}(\widetilde{\mathcal{T}})}}{\underbrace{\rule{30pt}{0pt}}}$}}}
\put(240,-56){\makebox(0,0){\footnotesize{$\underset{\textstyle{\widetilde{\mathcal{T}}}}{\underbrace{\rule{40pt}{0pt}}}$}}}
\put(5,26){\makebox(0,0)[l]{\scriptsize{$t$}}}
\put(102,26){\makebox(0,0)[l]{\scriptsize{$\widetilde{t}{\,\cdot\,}\mathrmbfit{tup}_{\mathcal{S}}(g)$}}}
\put(248,26){\makebox(0,0)[l]{\scriptsize{$\widetilde{t}$}}}
\put(160,-10){\makebox(0,0){\scriptsize{$\mathrmbfit{tup}_{\mathcal{S}}(g)$}}}
\put(120,125){\makebox(0,0){\scriptsize{$k$}}}
\put(60,92){\makebox(0,0){\scriptsize{$k$}}}
\put(100,80){\vector(-1,0){80}}
\put(4,59){\vector(2,-3){46}}
\put(108,59){\vector(-2,-3){46}}
\put(190,80){\makebox(0,0){\normalsize{$=$}}}
\put(240,60){\vector(0,-1){68}}
\put(185,-25){\vector(-1,0){65}}
\put(10,100){\oval(20,20)[tl]}
\put(10,110){\line(1,0){220}}
\put(230,100){\oval(20,20)[tr]}
\put(0,94){\vector(0,-1){0}}
\put(60,-98){\makebox(0,0){\footnotesize{$\underset{\textstyle{\mathrmbf{Tbl}_{\mathcal{S}}(\mathcal{A})}}{\underbrace{\rule{60pt}{0pt}}}$}}}
\end{picture}}
\end{tabular}}}
${\;\;\;\;\;\;\;\;\;\;\cong\;\;\;\;\;\;\;\;\;\;}$
{{\begin{tabular}{c}
{\setlength{\unitlength}{0.45pt}\begin{picture}(240,180)(-120,-55)
\put(-120,80){\makebox(0,0){\footnotesize{$K$}}}
\put(0,80){\makebox(0,0){\footnotesize{$\widehat{K}$}}}
\put(120,80){\makebox(0,0){\footnotesize{$\widetilde{K}$}}}
\put(60,-25){\makebox(0,0){\footnotesize{$\mathrmbfit{tup}_{\mathcal{S}}(\widetilde{\mathcal{A}})$}}}
\put(-120,-25){\makebox(0,0){\footnotesize{$\mathrmbfit{tup}_{\mathcal{S}}(\mathcal{A})$}}}
\put(-120,-56){\makebox(0,0){\footnotesize{$\underset{\textstyle{\mathcal{T}}}{\underbrace{\rule{40pt}{0pt}}}$}}}
\put(30,-60){\makebox(0,0){\footnotesize{$\underset{\textstyle{{g}^{\ast}(\mathcal{T})}}{\underbrace{\rule{30pt}{0pt}}}$}}}
\put(100,-56){\makebox(0,0){\footnotesize{$\underset{\textstyle{\widetilde{\mathcal{T}}}}{\underbrace{\rule{30pt}{0pt}}}$}}}
\put(-60,92){\makebox(0,0){\scriptsize{$e$}}}
\put(15,26){\makebox(0,0)[r]{\scriptsize{$\hat{t}$}}}
\put(105,26){\makebox(0,0)[l]{\scriptsize{$\widetilde{t}$}}}
\put(-125,26){\makebox(0,0)[r]{\scriptsize{$t$}}}
\put(-25,-10){\makebox(0,0){\scriptsize{$\mathrmbfit{tup}_{\mathcal{S}}(g)$}}}
\put(60,92){\makebox(0,0){\scriptsize{$\tilde{k}$}}}
\put(0,125){\makebox(0,0){\scriptsize{$k$}}}
\put(100,80){\vector(-1,0){80}}
\put(-25,80){\vector(-1,0){70}}
\put(4,59){\vector(2,-3){46}}
\put(116,59){\vector(-2,-3){46}}
\put(-120,60){\vector(0,-1){68}}
\put(5,-25){\vector(-1,0){65}}
\put(-110,100){\oval(20,20)[tl]}
\put(-110,110){\line(1,0){220}}
\put(110,100){\oval(20,20)[tr]}
\put(-120,94){\vector(0,-1){0}}
\qbezier(-100,15)(-92,15)(-84,15)
\qbezier(-84,-1)(-84,7)(-84,15)
\put(60,-98){\makebox(0,0){\footnotesize{$\underset{\textstyle{\mathrmbf{Tbl}_{\mathcal{S}}(\widetilde{\mathcal{A}})}}{\underbrace{\rule{60pt}{0pt}}}$}}}
\end{picture}}
\end{tabular}}}
\\&&
\end{tabular}}}
\end{center}
\caption{$\mathcal{S}$-Table Morphism}
\label{fig:tbl:mor:medium:fixed:S}
\end{figure}
%
%
In terms of fibers,
an $\mathcal{S}$-table morphism
consists of
an $X$-sorted type domain morphism
$\mathcal{A}\xrightleftharpoons{{\langle{\mathrmit{1}_{X},g}\rangle}} \widetilde{\mathcal{A}}$
\underline{and} 
either a morphism 
$\mathcal{T}\xleftarrow{\;k\;}{\scriptstyle\sum}_{g}(\widetilde{\mathcal{T}})$
in the fiber context $\mathrmbf{Tbl}_{\mathcal{S}}(\mathcal{A})$ 
or a morphism 
${g}^{\ast}(\mathcal{T})\xleftarrow{\;\tilde{k}}\widetilde{\mathcal{T}}$
in the fiber context $\mathrmbf{Tbl}_{\mathcal{S}}(\widetilde{\mathcal{A}})$.
The 
${\langle{I,s,\mathcal{A}}\rangle}$-table morphism 
$\mathcal{T}\xleftarrow{\;k\;}{\scriptstyle\sum}_{g}(\widetilde{\mathcal{T}})$
is the composition (RHS of Fig.~\ref{fig:tbl:mor:medium:fixed:S}) 
of the fiber morphism
${\scriptstyle\sum}_{g}\Bigl({g}^{\ast}(\mathcal{T})\xleftarrow{\;\tilde{k}}\widetilde{\mathcal{T}}\Bigr)$
with the $\mathcal{T}^{\,\text{th}}$ counit component 
$\mathcal{T}\xleftarrow{e}{\scriptstyle\sum}_{g}\Bigl({g}^{\ast}(\mathcal{T})\Bigr)$
for fiber adjunction
${\bigl\langle{{\scriptstyle\sum}_{g}{\;\dashv\;}{g}^{\ast}}\bigr\rangle}
:\mathrmbf{Tbl}_{\mathcal{S}}(\mathcal{A}){\;\leftrightarrows\;}\mathrmbf{Tbl}_{\mathcal{S}}(\widetilde{\mathcal{A}})$.
%
\footnote{Here,
the span 
$K'\xleftarrow{\,e\,}{\widehat{K}}\xrightarrow{\,\hat{t}\,}\mathrmbfit{tup}_{\mathcal{S}}(\widetilde{\mathcal{A}})$
is the pullback in the context $\mathrmbf{Set}$ of the opspan
$K\xrightarrow{\,t\,}\mathrmbfit{tup}_{\mathcal{S}}(\mathcal{A})\xleftarrow{\,\mathrmbfit{tup}_{\mathcal{S}}(g)\,}\mathrmbfit{tup}_{\mathcal{S}}(\widetilde{\mathcal{A}})$ and
$\widehat{K}\xleftarrow{\,\tilde{k}\,}\widetilde{K}$
is the mediating morphism for the span
$K\xleftarrow{\,k\,}\widetilde{K}\xrightarrow{\,\tilde{t}\,}\mathrmbfit{tup}_{\mathcal{S}}(\widetilde{\mathcal{A}})$.}
%
\begin{equation}
{{\begin{picture}(120,10)(0,-4)
\put(60,0){\makebox(0,0){\footnotesize{$
\underset{\textstyle{\text{in}\;\mathrmbf{Tbl}_{\mathcal{S}}(\mathcal{A})}}
{\mathcal{T}\xleftarrow{\;k\;}{\scriptstyle\sum}_{g}(\widetilde{\mathcal{T}})}
{\;\;\;\;\;\;\;\;\rightleftarrows\;\;\;\;\;\;\;\;}
\underset{\textstyle{\text{in}\;\mathrmbf{Tbl}_{\mathcal{S}}(\widetilde{\mathcal{A}})}}
{{g}^{\ast}(\mathcal{T})\xleftarrow{\;\tilde{k}\,}\widetilde{\mathcal{T}}}
$}}}
\end{picture}}}
\end{equation}
This fiber adjunction 
(top part of Tbl.~\ref{tbl:tbl-rel:refl:sign:mor})
is a component of
the signed domain indexed adjunction of tables
$\mathrmbf{Cls}(X)^{\mathrm{op}}\xrightarrow{\;\mathrmbfit{tbl}_{\mathcal{S}}}\mathrmbf{Adj}
:\mathcal{A}\mapsto{\mathrmbf{Tbl}_{\mathcal{S}}(\mathcal{A})}$.
\begin{table}
\begin{center}
\begin{tabular}{c}
{\scriptsize\setlength{\extrarowheight}{2pt}$\begin{array}[b]{c}
\mathcal{A}\xrightleftharpoons{{\langle{\mathrmit{1}_{X},g}\rangle}}\widetilde{\mathcal{A}}
\\
\underset{=\,\mathrmbfit{tup}_{\mathcal{A}}(I,s)}
{\mathrmbfit{tup}_{\mathcal{S}}(\mathcal{A})}
\xleftarrow[{(\mbox{-})}{\cdot}{g}]{\mathrmbfit{tup}_{\mathcal{S}}(g)}
\underset{=\,\mathrmbfit{tup}_{\widetilde{\mathcal{A}}}(I,s)}
{\mathrmbfit{tup}_{\mathcal{S}}(\widetilde{\mathcal{A}})}
\rule[-10pt]{0pt}{10pt}
\end{array}$}
\\\\
{{\begin{tabular}[b]{c}
\setlength{\unitlength}{0.55pt}
\begin{picture}(320,220)(-40,-30)
\put(0,160){\makebox(0,0){\footnotesize{$
\underset{=\,\bigl(\mathrmbf{Set}{\,\downarrow\,}\mathrmbfit{tup}_{\mathcal{A}}(I,s)\bigr)}
{\mathrmbf{Tbl}_{\mathcal{S}}(\mathcal{A})}$}}}
\put(0,0){\makebox(0,0){\footnotesize{$
\underset{=\,{\wp}\mathrmbfit{tup}_{\mathcal{A}}(I,s)}
{\mathrmbf{Rel}_{\mathcal{S}}(\mathcal{A})}$}}}
\put(240,160){\makebox(0,0){\footnotesize{$
\underset{=\,\bigl(\mathrmbf{Set}{\,\downarrow\,}\mathrmbfit{tup}_{\widetilde{\mathcal{A}}}(I,s)\bigr)}
{\mathrmbf{Tbl}_{\mathcal{S}}(\widetilde{\mathcal{A}})}$}}}
\put(240,0){\makebox(0,0){\footnotesize{$
\underset{=\,{\wp}\mathrmbfit{tup}_{\widetilde{\mathcal{A}}}(I,s)}
{\mathrmbf{Rel}_{\mathcal{S}}(\widetilde{\mathcal{A}})}$}}}
%
\put(120,190){\makebox(0,0){\scriptsize{${\scriptstyle\sum}_{g}$}}}
\put(122,162){\makebox(0,0){\scriptsize{${g}^{\ast}$}}}
\put(120,130){\makebox(0,0){\scriptsize{${\scriptstyle\prod}_{g}$}}}
\put(120,30){\makebox(0,0){\scriptsize{$\exists_{g}$}}}
\put(122,2){\makebox(0,0){\scriptsize{${g}^{{\scriptscriptstyle-}1}$}}}
\put(120,-30){\makebox(0,0){\scriptsize{$\forall_{g}$}}}
\put(-15,80){\makebox(0,0)[r]{\scriptsize{$\mathrmbfit{im}_{\mathcal{A}}(I,s)$}}}
\put(16,80){\makebox(0,0)[l]{\scriptsize{$\mathrmbfit{inc}_{\mathcal{A}}(I,s)$}}}
\put(225,80){\makebox(0,0)[r]{\scriptsize{$\mathrmbfit{im}_{\widetilde{\mathcal{A}}}(I,s)$}}}
\put(256,80){\makebox(0,0)[l]{\scriptsize{$\mathrmbfit{inc}_{\widetilde{\mathcal{A}}}(I,s)$}}}
\put(170,180){\vector(-1,0){100}}
\qbezier(70,160)(85,160)(100,160)\qbezier(140,160)(155,160)(170,160)\put(170,160){\vector(1,0){0}}
\put(170,140){\vector(-1,0){100}}
\put(170,20){\vector(-1,0){100}}
\qbezier(70,0)(85,0)(100,0)\qbezier(140,0)(155,0)(170,0)\put(170,0){\vector(1,0){0}}
\put(170,-20){\vector(-1,0){100}}
\put(-10,125){\vector(0,-1){90}}
\put(10,35){\vector(0,1){90}}
\put(230,125){\vector(0,-1){90}}
\put(250,35){\vector(0,1){90}}
\end{picture}
\end{tabular}}}
\\\\
{\scriptsize\setlength{\extrarowheight}{2.5pt}$\begin{array}[b]{|l|}
\hline
{\scriptstyle\sum}_{g} \dashv {g}^{\ast} \dashv {\scriptstyle\prod}_{g},\; 
\exists_{g} \dashv {g}^{{\scriptscriptstyle-}1} \dashv \forall_{g} 
\\
\mathrmbfit{im}_{\mathcal{A}}(I,s) \dashv \mathrmbfit{inc}_{\mathcal{A}}(I,s),\; 
\mathrmbfit{im}_{\widetilde{\mathcal{A}}}(I,s) \dashv \mathrmbfit{inc}_{\widetilde{\mathcal{A}}}(I,s) 
\\ \hline
\mathrmbfit{inc}_{\mathcal{A}}(I,s){\;\circ\;}{g}^{\ast}	
 = {g}^{{\scriptscriptstyle-}1}{\;\circ\;}\mathrmbfit{inc}_{\widetilde{\mathcal{A}}}(I,s) 
\\
\mathrmbfit{inc}_{\widetilde{\mathcal{A}}}(I,s){\;\circ\;}{\scriptstyle\prod}_{g} 
 = \forall_{g}{\;\circ\;}\mathrmbfit{inc}_{\mathcal{A}}(I,s) 
\\ \hline
{\scriptstyle\sum}_{g}{\;\circ\;}\mathrmbfit{im}_{\mathcal{A}}(I,s)
 \cong \mathrmbfit{im}_{\widetilde{\mathcal{A}}}(I_{1},s_{1}){\;\circ\;}\exists_{g}
\\
{g}^{\ast}{\;\circ\;}\mathrmbfit{im}_{\widetilde{\mathcal{A}}}(I,s)
 \cong \mathrmbfit{im}_{\mathcal{A}}(I,s){\;\circ\;}{g}^{{\scriptscriptstyle-}1}
\\ \hline
\end{array}$}
\\ \\
{\textit{small fibers -- short distance}}
\end{tabular}
\end{center}
\caption{Reflection: Signature}
\label{tbl:tbl-rel:refl:sign:mor}
\end{table}
\begin{theorem}\label{thm:fib:cxt:tbl:S:cls:X}
The fibered context of $\mathcal{S}$-tables
$\mathrmbf{Tbl}(\mathcal{S})\xrightarrow{\mathrmbfit{data}_{\mathcal{S}}}\mathrmbf{Cls}(X)^{\mathrm{op}}$
is the Grothendieck construction $\int_{\mathrmbf{Cls}(X)}$ of 
the type domain indexed adjunction
$\mathrmbf{Cls}(X)^{\mathrm{op}}\xrightarrow{\;\mathrmbfit{tbl}_{\mathcal{S}}}\mathrmbf{Adj}$.$^{\ref{grothendieck}}$
\end{theorem}
%

\subsubsection{Upper Aspect.}\label{sub:sub:sec:tbl:sign:up}

Here, we show that the context of tables $\mathrmbf{Tbl}$ is fibered over signatures
via the projection passage $\mathrmbf{Tbl}\xrightarrow{\mathrmbfit{sign}}\mathrmbf{List}^{\mathrm{op}}$.
We use the Grothendieck construction $\int_{\mathrmbf{List}}$
on the indexed context
$\mathrmbf{List}^{\mathrm{op}}\!\xrightarrow{\;\mathrmbfit{tbl}\;}\mathrmbf{Cxt}
:\mathcal{S}\mapsto\mathrmbf{Tbl}(\mathcal{S})$.
We use the same definitions as in \S~\ref{sub:sec:tbl}.
A $\mathrmbf{Tbl}$-object $\mathcal{T} = {\langle{\mathcal{S},\mathcal{A},K,t}\rangle}$,
called an table (database relation),
consists of 
a signature $\mathcal{S} = {\langle{I,s,X}\rangle} \in \mathrmbf{List}$ and
an $\mathcal{S}$-table ${\langle{\mathcal{A},K,t}\rangle}\in\mathrmbf{Tbl}(\mathcal{S})$.
A $\mathrmbf{Tbl}$-morphism
$\mathcal{T}_{2} = {\langle{I_{2},s_{2},\mathcal{A}_{2},K_{2},t_{2}}\rangle} 
\xleftarrow{{\langle{h,f,g,k}\rangle}} 
{\langle{I_{1},s_{1},\mathcal{A}_{1},K_{1},t_{1}}\rangle} = \mathcal{T}_{1}$,
called a table morphism
(see Fig.~\ref{fig:tbl:mor}),
consists of 
a signed domain morphism
${\langle{I_{2},s_{2},\mathcal{A}_{2}}\rangle}\xrightarrow{{\langle{h,f,g}\rangle}}{\langle{I_{1},s_{1},\mathcal{A}_{1}}\rangle}$
%
\footnote{A signed domain morphism factors into
a signature morphism
$\mathcal{S}_{2} = {\langle{I_{2},s_{2},X_{2}}\rangle} 
\xrightarrow{{\langle{h,f}\rangle}} 
{\langle{I_{1},s_{1},X_{1}}\rangle} = \mathcal{S}_{1}$ and
a type domain morphism 
$\mathcal{A}_{2} = {\langle{X_{2},Y_{2},\models_{\mathcal{A}_{2}}}\rangle} 
\xrightleftharpoons{{\langle{f,g}\rangle}} 
{\langle{X_{1},Y_{1},\models_{\mathcal{A}_{1}}}\rangle} = \mathcal{A}_{1}$
with common sort function $X_{2} \xrightarrow{f} X_{1}$.}
%
and a key function $K_{2}\xleftarrow{\;k\;}K_{1}$,
which satisfy the condition 
(using Lem.~\ref{lem:tup:fn:fact} in \S\ref{sub:sub:sec:tup:fn:fact}):
{\footnotesize{$k{\,\cdot\,}t_{2} 
= t_{1}{\,\cdot\,}\mathrmbfit{tup}(h,f,g)
= (t_{1}{\,\cdot\,}\tau_{{\langle{h,f}\rangle}}(\mathcal{A}_{1})){\,\cdot\,}\mathrmbfit{tup}_{\mathcal{S}_{2}}(g)$}}.
This gives an alternate, but equivalent, definition in terms of fibers.
%
%
\begin{lemma}\label{lem:tbl:fbr:fact:sign}
For any signed domain morphism
${\langle{I_{2},s_{2},\mathcal{A}_{2}}\rangle}\xrightarrow{{\langle{h,f,g}\rangle}}{\langle{I_{1},s_{1},\mathcal{A}_{1}}\rangle}$,
the tuple resolution 
{\footnotesize{$\mathrmbfit{tup}(h,f,g) = \tau_{{\langle{h,f}\rangle}}(\mathcal{A}_{1}){\;\cdot\;}\mathrmbfit{tup}_{\mathcal{S}_{2}}(g)$}}
(Lem.~\ref{lem:tup:fn:fact} in \S~\ref{sub:sub:sec:tup:fn:fact})
resolves the table fiber passage
\[\mbox{\footnotesize{
$
\mathrmbf{Tbl}(I_{2},s_{2},\mathcal{A}_{2})
{\;\xleftarrow{\;{\scriptscriptstyle\sum}_{{\langle{h,f,g}\rangle}}\;\;\;}\;}
\mathrmbf{Tbl}(I_{1},s_{1},\mathcal{A}_{1})$.
}\normalsize}\]
into the table fiber passage factorization in Fig.~\ref{fig:tbl:fbr:fact:sign}.
\end{lemma}

\begin{figure}
\begin{center}
{{\begin{tabular}{c}
\setlength{\unitlength}{0.54pt}
\begin{picture}(240,210)(0,-50)
\put(0,160){\makebox(0,0){\footnotesize{$
\underset{=\,(\mathrmbf{Set}{\,\downarrow\,}\mathrmbfit{tup}_{\mathcal{A}_{2}}(\mathcal{S}_{2}))}
{\mathrmbf{Tbl}_{\mathcal{S}_{2}}(\mathcal{A}_{2})}$}}}
\put(0,0){\makebox(0,0){\footnotesize{$
\underset{=\,(\mathrmbf{Set}{\,\downarrow\,}\mathrmbfit{tup}_{{f}^{-1}(\mathcal{A}_{1})}(\mathcal{S}_{2}))}
{\mathrmbf{Tbl}_{\mathcal{S}_{2}}({f}^{-1}(\mathcal{A}_{1}))}$}}}
\put(240,0){\makebox(0,0){\footnotesize{$
\underset{=\,(\mathrmbf{Set}{\,\downarrow\,}\mathrmbfit{tup}_{\mathcal{A}_{1}}(\mathcal{S}_{1}))}
{\mathrmbf{Tbl}_{\mathcal{S}_{1}}(\mathcal{A}_{1})}$}}}
\put(130,-8){\makebox(0,0){\scriptsize{${\scriptstyle\sum}_{\tau_{{\langle{h,f}\rangle}}(\mathcal{A}_{1})}$}}}
\put(-6,80){\makebox(0,0)[r]{\scriptsize{${\scriptstyle\sum}_{\mathrmbfit{tup}_{\mathcal{S}_{2}}(g)}$}}}
\put(130,90){\makebox(0,0)[l]{\scriptsize{${\scriptstyle\sum}_{{\langle{h,f,g}\rangle}}$}}}
\put(180,8){\vector(-1,0){100}}
\put(0,35){\vector(0,1){90}}
\put(194,28){\vector(-3,2){160}}
\put(360,100){\makebox(0,0)[l]{\footnotesize{$\underset{\textit{{sort function}}}
{X_{2}\xrightarrow{f}X_{1}}$}}}
\put(340,60){\makebox(0,0)[l]{\scriptsize{$\underset{X_{2}-\textit{{type domain morphism}}}
{\mathcal{A}_{2}\xrightarrow{g}{f}^{-1}(\mathcal{A}_{1})}$}}}
\put(120,-60){\makebox(0,0){\footnotesize{$\underset{\textit{{signature morphism}}}
{\mathcal{S}_{2}\xrightarrow{{\langle{h,f}\rangle}}\mathcal{S}_{1}}$}}}
\end{picture}
\end{tabular}}}
\end{center}
\caption{Table Fiber Passage Factorization}
\label{fig:tbl:fbr:fact:sign}
\end{figure}
\newpage
\mbox{}\newline
For any signature $\mathcal{S}={\langle{I,s,X}\rangle}$,
the fibered context of $\mathcal{S}$-tables $\mathrmbf{Tbl}(\mathcal{S})$
separates into the partition
$\mathrmbf{Tbl}(\mathcal{S})
= \coprod_{\!\!\!\!\!\!\underset{\in\mathrmbf{Cls}(X)}{\mathcal{A}}\!\!\!\!\!\!}\mathrmbf{Tbl}_{\mathcal{A}}(I,s)$.
For any signature morphism
${\mathcal{S}_{2}\xrightarrow{{\langle{h,f}\rangle}}\mathcal{S}_{1}}$,
we can sum the partitions of fibered passages as follows:
\begin{center}
{{\begin{tabular}{c}
\setlength{\unitlength}{0.54pt}
\begin{picture}(300,120)(0,-10)
\put(10,40){\makebox(0,0){\footnotesize{$
\underset{\coprod_{\!\!\!\!\!\!\underset{\in\mathrmbf{Cls}(X_{2})}{\mathcal{A}_{2}}\!\!\!\!\!\!\!\!}\mathrmbf{Tbl}_{\mathcal{S}_{2}}(\mathcal{A}_{2})}
{\underbrace{\mathrmbf{Tbl}(\mathcal{S}_{2})}}$}}}
\put(295,40){\makebox(0,0){\footnotesize{$
\underset{\coprod_{\!\!\!\!\!\!\underset{\in\mathrmbf{Cls}(X_{1})}{\mathcal{A}_{1}}\!\!\!\!\!\!\!\!}\mathrmbf{Tbl}_{\mathcal{S}_{1}}(\mathcal{A}_{1})}
{\underbrace{\mathrmbf{Tbl}(\mathcal{S}_{1})}}$}}}
\put(150,40){\makebox(0,0){\scriptsize{$
\xleftarrow
[\coprod_{\!\!\!\!\!\!\underset{\in\mathrmbf{Cls}(X_{1})}{\mathcal{A}_{1}}\!\!\!\!\!\!}{\scriptscriptstyle\sum}_{\tau_{{\langle{h,f}\rangle}}(\mathcal{A}_{1})}]
{\grave{\mathrmbfit{tbl}}_{{\langle{h,f}\rangle}}}$}}}
\end{picture}
\end{tabular}}}
\end{center}
\mbox{}\newline
The factorization in Fig.~\ref{fig:tbl:fbr:fact:sign},
suggests the following definition of table fiber passage,
where
the fiber passage
$\mathrmbf{Tbl}(\mathcal{S}_{2})\xleftarrow{\grave{\mathrmbfit{tbl}}_{{\langle{h,f}\rangle}}}\mathrmbf{Tbl}(\mathcal{S}_{1})$
is define in terms of the component tuple functions
{\footnotesize{$
\mathrmbfit{tup}_{\mathcal{S}_{2}}(f^{-1}(\mathcal{A}_{1}))
\xleftarrow[h{\,\cdot\,}{(\mbox{-})}]{\tau_{{\langle{h,f}\rangle}}(\mathcal{A}_{1})}
\mathrmbfit{tup}_{\mathcal{S}_{1}}(\mathcal{A}_{1})$}}
and the inverse image function
$\mathrmbf{Cls}(X_{2})\xleftarrow{\;f^{-1}\!}\mathrmbf{Cls}(X_{1})$.
\begin{definition}(table fiber passage)\label{def:fbr:pass:sign}
\begin{description}
\item[$\grave{\mathrmbfit{tbl}}_{{\langle{h,f}\rangle}}$:] 
An $\mathcal{S}_{1}$-table is mapped to an $\mathcal{S}_{1}$-table as follows:
\[\mbox{\footnotesize{$
{\langle{f^{-1}(\mathcal{A}_{1}),{\scriptstyle\sum}_{\tau_{{\langle{h,f}\rangle}}(\mathcal{A}_{1})}(K_{1},t_{1})}\rangle}
\stackrel{{\grave{\mathrmbfit{tbl}}_{{\langle{h,f}\rangle}}}}{\longmapsfrom}
{\langle{\mathcal{A}_{1},(K_{1},t_{1})}\rangle}
$}}\]
where 
${\langle{I_{2},s_{2},f^{-1}(\mathcal{A}_{1})}\rangle}$-tuple
${\scriptstyle\sum}_{\tau_{{\langle{h,f}\rangle}}(\mathcal{A}_{1})}(K_{1},t_{1})
= {\langle{K_{1},t_{1}{\cdot\,}\tau_{{\langle{h,f}\rangle}}(\mathcal{A}_{1})}\rangle}
 \in \mathrmbfit{tup}_{\mathcal{S}_{2}}(f^{-1}(\mathcal{A}_{1}))$
is the existential (direct) image of 
${\langle{I_{1},s_{1},\mathcal{A}_{1}}\rangle}$-tuple 
${\langle{K_{1},t_{1}}\rangle} \in \mathrmbfit{tup}_{\mathcal{S}_{1}}(\mathcal{A}_{1})$
along $\tau_{{\langle{h,f}\rangle}}(\mathcal{A}_{1})$.
A morphism of $\mathcal{S}_{1}$-tables
${\langle{\mathcal{A}_{1},K_{1},t_{1}}\rangle} 
\xleftarrow{\;{\langle{g,k}\rangle}\;}
{\langle{\widetilde{\mathcal{A}}_{1},\widetilde{K}_{1},\tilde{t}_{1}}\rangle}$
\newline
is mapped to the morphism of $\mathcal{S}_{2}$-tables
\newline\mbox{}\hfill
${\langle{f^{-1}(\mathcal{A}_{1}),{\scriptstyle\sum}_{\tau_{{\langle{h,f}\rangle}}(\mathcal{A}_{1})}(K_{1},t_{1})}\rangle}
\xleftarrow{\;{\langle{f^{-1}(g),k}\rangle}\;}
{\langle{f^{-1}(\widetilde{\mathcal{A}}_{1}),
{\scriptstyle\sum}_{\tau_{{\langle{h,f}\rangle}}(\widetilde{\mathcal{A}}_{1})}(\widetilde{K}_{1},\tilde{t}_{1})}\rangle}$.
\hfill\mbox{}
%
\end{description}
\end{definition}
%


%
\begin{figure}
\begin{center}
{{\begin{tabular}{c}
{\setlength{\unitlength}{0.45pt}\begin{picture}(430,200)(0,-95)
\put(0,80){\makebox(0,0){\footnotesize{$K_{2}$}}}
\put(120,80){\makebox(0,0){\footnotesize{$K_{1}$}}}
\put(240,80){\makebox(0,0){\footnotesize{$K_{1}$}}}
\put(430,80){\makebox(0,0){\footnotesize{$K_{1}$}}}
\put(60,-25){\makebox(0,0){\footnotesize{$\mathrmbfit{tup}_{\mathcal{S}_{2}}(\mathcal{A}_{2})$}}}
\put(240,-25){\makebox(0,0){\footnotesize{$\mathrmbfit{tup}_{\mathcal{S}_{2}}(
f^{-1}(\mathcal{A}_{1})
)$}}}
\put(430,-25){\makebox(0,0){\footnotesize{$\mathrmbfit{tup}_{\mathcal{S}_{1}}(\mathcal{A}_{1})$}}}
\put(30,-56){\makebox(0,0){\footnotesize{$\underset{\textstyle{\mathcal{T}_{2}}}{\underbrace{\rule{30pt}{0pt}}}$}}}
\put(240,-56){\makebox(0,0){\footnotesize{$\underset{\textstyle{
\grave{\mathrmbfit{tbl}}_{{\langle{h,f}\rangle}}(\mathcal{T}_{1})
}}{\underbrace{\rule{40pt}{0pt}}}$}}}
\put(430,-56){\makebox(0,0){\footnotesize{$\underset{\textstyle{\mathcal{T}_{1}}}{\underbrace{\rule{40pt}{0pt}}}$}}}
\put(15,26){\makebox(0,0)[r]{\scriptsize{$t_{2}$}}}
\put(102,26){\makebox(0,0)[l]{\scriptsize{$\widetilde{t}{\,\cdot\,}\mathrmbfit{tup}_{\mathcal{S}_{2}}(g)$}}}
\put(248,26){\makebox(0,0)[l]{\scriptsize{$\widetilde{t}=t_{1}{\,\cdot\,}\tau_{{\langle{h,f}\rangle}}(\mathcal{A}_{1})$}}}
\put(440,26){\makebox(0,0)[l]{\scriptsize{$t_{1}$}}}
\put(145,-10){\makebox(0,0){\scriptsize{$\mathrmbfit{tup}_{\mathcal{S}_{2}}(g)$}}}
\put(145,-35){\makebox(0,0){\scriptsize{${(\mbox{-})}{\,\cdot\,}g$}}}
\put(350,-10){\makebox(0,0){\scriptsize{$\tau_{{\langle{h,f}\rangle}}(\mathcal{A}_{1})$}}}
\put(350,-35){\makebox(0,0){\scriptsize{$h{\,\cdot\,}{(\mbox{-})}$}}}
\put(120,125){\makebox(0,0){\scriptsize{$k$}}}
\put(60,92){\makebox(0,0){\scriptsize{$k$}}}
\put(100,80){\vector(-1,0){80}}
\put(4,59){\vector(2,-3){46}}
\put(108,59){\vector(-2,-3){46}}
\put(180,80){\makebox(0,0){\normalsize{$=$}}}
\put(340,80){\makebox(0,0){\normalsize{$=$}}}
\put(240,60){\vector(0,-1){68}}
\put(160,-25){\vector(-1,0){45}}
\put(375,-25){\vector(-1,0){55}}
\put(430,60){\vector(0,-1){68}}
\put(10,100){\oval(20,20)[tl]}
\put(10,110){\line(1,0){220}}
\put(0,94){\vector(0,-1){0}}
\put(230,110){\line(1,0){190}}
\put(420,100){\oval(20,20)[tr]}
%
\put(150,-105){\makebox(0,0){\footnotesize{$\underset{\textstyle{\mathrmbf{Tbl}(\mathcal{S}_{2})}}{\underbrace{\rule{120pt}{0pt}}}$}}}
\put(435,-105){\makebox(0,0){\footnotesize{$\underset{\textstyle{\mathrmbf{Tbl}(\mathcal{S}_{1})}}{\underbrace{\rule{50pt}{0pt}}}$}}}
\end{picture}}
\end{tabular}}}
\end{center}
\caption{Table Morphism: Signature}
\label{fig:tbl:mor:large:sign}
\end{figure}
%
A table morphism (Fig.~\ref{fig:tbl:mor:large:sign})
consists of
signature morphism $\mathcal{S}_{2}\xrightarrow{{\langle{h,f}\rangle}}\mathcal{S}_{1}$
\underline{and} 
a morphism 
$\mathcal{T}_{2}\xleftarrow{{\langle{g,k}\rangle}}\grave{\mathrmbfit{tbl}}_{{\langle{h,f}\rangle}}(\mathcal{T}_{1})$
in the fiber context $\mathrmbf{Tbl}(S_{2})$,
where the fiber passage
$\mathrmbf{Tbl}(\mathcal{S}_{2})\xleftarrow{\grave{\mathrmbfit{tbl}}_{{\langle{h,f}\rangle}}}\mathrmbf{Tbl}(\mathcal{S}_{1})$
along the
signature morphism $\mathcal{S}_{2}\xrightarrow{{\langle{h,f}\rangle}}\mathcal{S}_{1}$
is defined in Def.~\ref{def:fbr:pass:sign}
of \S\,\ref{sub:sub:sec:tbl:sign:up}.
This fiber passage is a component of the signature indexed context of tables 
$\mathrmbf{List}^{\mathrm{op}}\!\xrightarrow{\;\mathrmbfit{tbl}\;}\mathrmbf{Cxt}$.
\begin{theorem}\label{thm:fib:cxt:tbl:sign}
The fibered context of tables (an opfibration)
$\mathrmbf{Tbl}\xrightarrow{\;\mathrmbfit{sign}\;}\mathrmbf{List}^{\mathrm{op}}$
is the Grothendieck construction 
$\int_{\mathrmbf{List}}$ 
--- visualized in the upper-left quadrant of Fig.~\ref{fig:tbl:gro:constr} ---
of the signature indexed context of tables 
$\mathrmbf{List}^{\mathrm{op}}\!\xrightarrow{\;\grave{\mathrmbfit{tbl}}\;}\mathrmbf{Cxt}$.
$^{\ref{grothendieck}}$
\end{theorem}
%

\subsection{Type Domain Indexing}\label{sub:sec:tbl:typ:dom}

In this section we show that the context of tables 
is a fibered context over type domains.
We first define the table fiber for fixed type domain.
We next move between table fibers along type domain morphisms.
Finally,
we invoke the Grothendieck construction indexed by type domains.

\subsubsection{Lower Aspect.}\label{sub:sub:sec:tbl:typ:dom:low}

Let $\mathcal{A}={\langle{X,Y,\models_{\mathcal{A}}}\rangle}$ be a fixed type domain.
For database tables,
the type domain $\mathcal{A}$
consists of a fixed sort set $X$ and a fixed $X$-indexed collection of data types $\{
\mathcal{A}_{x} = \mathrmbfit{ext}_{\mathcal{A}}(x) \mid x{\,\in\,}X \}$.
Here, we show that the context of $\mathcal{A}$-tables $\mathrmbf{Tbl}(\mathcal{A})$ is fibered over $X$-sorted signatures
$\mathrmbf{Tbl}(\mathcal{A})\xrightarrow{\mathrmbfit{sign}_{\mathcal{A}}}\mathrmbf{List}(X)^{\mathrm{op}}$.
We use the Grothendieck construction $\int_{\mathrmbf{List}(X)}$
on the indexed adjunction
$\mathrmbf{List}(X)^{\mathrm{op}}\xrightarrow{\;\mathrmbfit{tbl}_{\mathcal{A}}}\mathrmbf{Adj}
:{\langle{I,s}\rangle}\mapsto{\mathrmbf{Tbl}_{\mathcal{A}}(I,s)}$.

\paragraph{Fibered Context (medium-size).}

The fibered context of $\mathcal{A}$-tables 
%
\footnote{The context of $\mathcal{A}$-tables $\mathrmbf{Tbl}(\mathcal{A})$
corresponds to the context of tables $\mathrmbf{Tables}^{\pi}$
in (Spivak \cite{spivak:sd})
for a (fixed) datatype specification $U\xrightarrow{\pi}\mathrmbf{DT}$
with universe $U$ and set of datatypes $\mathrmbf{DT}$,
since a data-type specification is a special case of a type domain.
However, in \cite{spivak:sd} 
there is no connection between contexts of tables with different data-type specifications,
analogous to the fiber adjunction 
(Prop.~\ref{prop:typ:dom:tbl:fbr:adj} of \S~\ref{sub:sub:sec:tbl:typ:dom:up})
$\mathrmbf{Tbl}(\mathcal{A}_{2})
\xleftarrow{{\langle{\acute{\mathrmbfit{tbl}}_{{\langle{f,g}\rangle}}\!\dashv\;\grave{\mathrmbfit{tbl}}_{{\langle{f,g}\rangle}}}\rangle}}
\mathrmbf{Tbl}(\mathcal{A}_{1})$
for type domain morphism $\mathcal{A}_{2}\xrightleftharpoons{{\langle{f,g}\rangle}}\mathcal{A}_{1}$.}
%
is the comma mathematical context
\[\mbox{\footnotesize{$
\mathrmbf{Tbl}(\mathcal{A}) 
= {\bigl(\mathrmbf{Set}{\,\downarrow\,}\mathrmbfit{tup}_{\mathcal{A}}\bigr)}
$}\normalsize}\]
associated with the type domain tuple passage
$\mathrmbf{List}(X)^{\mathrm{op}}\xrightarrow{\;\mathrmbfit{tup}_{\mathcal{A}}\;}\mathrmbf{Set}$
defined in \S~\ref{sub:sub:sec:tup:bridge:typ:dom}.
It has key and signature projection passages
$\mathrmbf{Set}
\xleftarrow{\mathrmbfit{key}_{\mathcal{A}}}
\mathrmbf{Tbl}(\mathcal{A})
\xrightarrow{\mathrmbfit{sign}_{\mathcal{A}}}
\mathrmbf{List}(X)^{\mathrm{op}}$
and defining bridge
$\mathrmbfit{key}_{\mathcal{A}}\xRightarrow{\,\tau_{\mathcal{A}}\;} 
\mathrmbfit{sign}_{\mathcal{A}}{\;\circ\;}\mathrmbfit{tup}_{\mathcal{A}}$. 
A $\mathrmbf{Tbl}(\mathcal{A})$-object $\mathcal{T} = {\langle{I,s,K,t}\rangle}$, 
called an $\mathcal{A}$-table,
consists of 
an indexing $X$-sorted signature $\mathrmbfit{sign}_{\mathcal{A}}(\mathcal{T}) = {\langle{I,s}\rangle}$, 
a key set $\mathrmbfit{key}_{\mathcal{A}}(\mathcal{T}) = K$ and 
a tuple function $K\xrightarrow[t]{\tau_{\mathcal{A}}(\mathcal{T})}\mathrmbfit{tup}_{\mathcal{A}}(I,s)$.
A $\mathrmbf{Tbl}(\mathcal{A})$-morphism 
$\mathcal{T}'={\langle{I',s',K',t'}\rangle}\xleftarrow{\langle{h,k}\rangle}{\langle{I,s,K,t}\rangle}=\mathcal{T}$
consists of 
an indexing $X$-sorted signature morphism 
${\langle{I',s'}\rangle}\xrightarrow[\,h\,]{\mathrmbfit{sign}_{\mathcal{A}}(h,k)}{\langle{I,s}\rangle}$ and 
a key function $K'\xleftarrow{\,k\,}K$
satisfying the naturality condition $k{\;\cdot\;}t'=t{\;\cdot\;}\mathrmbfit{tup}_{\mathcal{A}}(h)$.
%
The naturality condition gives two alternate and adjoint definitions.
\begin{figure}
\begin{center}
{{\begin{tabular}{@{\hspace{30pt}}r@{\hspace{30pt}}c@{\hspace{30pt}}l}
{{\begin{tabular}{c}
{\setlength{\unitlength}{0.45pt}\begin{picture}(240,180)(0,-55)
\put(0,80){\makebox(0,0){\footnotesize{$K'$}}}
\put(120,80){\makebox(0,0){\footnotesize{$K$}}}
\put(240,80){\makebox(0,0){\footnotesize{$K$}}}
\put(60,-25){\makebox(0,0){\footnotesize{$\mathrmbfit{tup}_{\mathcal{A}}(I',s')$}}}
\put(240,-25){\makebox(0,0){\footnotesize{$\mathrmbfit{tup}_{\mathcal{A}}(I,s)$}}}
\put(30,-56){\makebox(0,0){\footnotesize{$\underset{\textstyle{\mathcal{T}'}}{\underbrace{\rule{30pt}{0pt}}}$}}}
\put(105,-58){\makebox(0,0){\footnotesize{$\underset{\textstyle{{\scriptstyle\sum}_{{h}}(\mathcal{T})}}{\underbrace{\rule{30pt}{0pt}}}$}}}
\put(240,-56){\makebox(0,0){\footnotesize{$\underset{\textstyle{\mathcal{T}}}{\underbrace{\rule{40pt}{0pt}}}$}}}
\put(5,26){\makebox(0,0)[l]{\scriptsize{$t'$}}}
\put(102,26){\makebox(0,0)[l]{\scriptsize{$t{\,\cdot\,}\mathrmbfit{tup}_{\mathcal{A}}(h)$}}}
\put(248,26){\makebox(0,0)[l]{\scriptsize{$t$}}}
\put(160,-10){\makebox(0,0){\scriptsize{$\mathrmbfit{tup}_{\mathcal{A}}(h)$}}}
\put(120,125){\makebox(0,0){\scriptsize{$k$}}}
\put(60,92){\makebox(0,0){\scriptsize{$k$}}}
\put(100,80){\vector(-1,0){80}}
\put(4,59){\vector(2,-3){46}}
\put(108,59){\vector(-2,-3){46}}
\put(190,80){\makebox(0,0){\normalsize{$=$}}}
\put(240,60){\vector(0,-1){68}}
\put(185,-25){\vector(-1,0){65}}
\put(10,100){\oval(20,20)[tl]}
\put(10,110){\line(1,0){220}}
\put(230,100){\oval(20,20)[tr]}
\put(0,94){\vector(0,-1){0}}
\put(60,-98){\makebox(0,0){\footnotesize{$\underset{\textstyle{\mathrmbf{Tbl}_{\mathcal{A}}(I',s')}}{\underbrace{\rule{60pt}{0pt}}}$}}}
\end{picture}}
\end{tabular}}}
${\;\;\;\;\;\;\;\;\;\;\cong\;\;\;\;\;\;\;\;\;\;}$
{{\begin{tabular}{c}
{\setlength{\unitlength}{0.45pt}\begin{picture}(240,180)(-120,-55)
\put(-120,80){\makebox(0,0){\footnotesize{$K'$}}}
\put(0,80){\makebox(0,0){\footnotesize{$\widehat{K}$}}}
\put(120,80){\makebox(0,0){\footnotesize{$K$}}}
\put(60,-25){\makebox(0,0){\footnotesize{$\mathrmbfit{tup}_{\mathcal{A}}(I,s)$}}}
\put(-120,-25){\makebox(0,0){\footnotesize{$\mathrmbfit{tup}_{\mathcal{A}}(I',s')$}}}
\put(-120,-56){\makebox(0,0){\footnotesize{$\underset{\textstyle{\mathcal{T}'}}{\underbrace{\rule{40pt}{0pt}}}$}}}
\put(30,-60){\makebox(0,0){\footnotesize{$\underset{\textstyle{{h}^{\ast}(\mathcal{T}')}}{\underbrace{\rule{30pt}{0pt}}}$}}}
\put(100,-56){\makebox(0,0){\footnotesize{$\underset{\textstyle{\mathcal{T}}}{\underbrace{\rule{30pt}{0pt}}}$}}}
\put(-60,92){\makebox(0,0){\scriptsize{$e$}}}
\put(15,26){\makebox(0,0)[r]{\scriptsize{$\hat{t}$}}}
\put(105,26){\makebox(0,0)[l]{\scriptsize{$t$}}}
\put(-125,26){\makebox(0,0)[r]{\scriptsize{$t'$}}}
\put(-25,-10){\makebox(0,0){\scriptsize{$\mathrmbfit{tup}_{\mathcal{A}}(h)$}}}
\put(60,92){\makebox(0,0){\scriptsize{$k'$}}}
\put(0,125){\makebox(0,0){\scriptsize{$k$}}}
\put(100,80){\vector(-1,0){80}}
\put(-25,80){\vector(-1,0){70}}
\put(4,59){\vector(2,-3){46}}
\put(116,59){\vector(-2,-3){46}}
\put(-120,60){\vector(0,-1){68}}
\put(5,-25){\vector(-1,0){65}}
\put(-110,100){\oval(20,20)[tl]}
\put(-110,110){\line(1,0){220}}
\put(110,100){\oval(20,20)[tr]}
\put(-120,94){\vector(0,-1){0}}
\qbezier(-100,15)(-92,15)(-84,15)
\qbezier(-84,-1)(-84,7)(-84,15)
\put(60,-98){\makebox(0,0){\footnotesize{$\underset{\textstyle{\mathrmbf{Tbl}_{\mathcal{A}}(I,s)}}{\underbrace{\rule{60pt}{0pt}}}$}}}
\end{picture}}
\end{tabular}}}
\\&&
\end{tabular}}}
\end{center}
\caption{$\mathcal{A}$-Table Morphism}
\label{fig:tbl:mor:medium}
\end{figure}
In terms of fibers,
an $\mathcal{A}$-table morphism
(see Fig.~\ref{fig:tbl:mor:medium}) 
consists of
a $X$-signature morphism ${\langle{I',s'}\rangle}\xrightarrow{\,h\,}{\langle{I,s}\rangle}$ 
\underline{and} 
either a morphism 
$\mathcal{T}'\xleftarrow{\;k\;}{\scriptstyle\sum}_{h}(\mathcal{T})$
in the fiber context $\mathrmbf{Tbl}_{\mathcal{A}}(I',s')$ 
or a morphism 
${h}^{\ast}(\mathcal{T}')\xleftarrow{\;k'}\mathcal{T}$
in the fiber context $\mathrmbf{Tbl}_{\mathcal{A}}(I,s)$.
The ${\langle{I',s',\mathcal{A}}\rangle}$-table morphism 
$\mathcal{T}'\xleftarrow{\;k\;}{\scriptstyle\sum}_{h}(\mathcal{T})$
is the composition (RHS of Fig.~\ref{fig:tbl:mor:medium}) 
of the fiber morphism
${\scriptstyle\sum}_{h}\Bigl({h}^{\ast}(\mathcal{T}')\xleftarrow{k'}\mathcal{T}\Bigr)$
with the $\mathcal{T}'^{\,\text{th}}$ counit component 
$\mathcal{T}'\xleftarrow{e}{\scriptstyle\sum}_{h}\Bigl({h}^{\ast}(\mathcal{T}')\Bigr)$
for the fiber adjunction
$\mathrmbf{Tbl}_{\mathcal{A}}(I',s')
{\;\xleftarrow{{\bigl\langle{{\scriptscriptstyle\sum}_{h}{\;\dashv\;}{h}^{\ast}}\bigr\rangle}}\;}
\mathrmbf{Tbl}_{\mathcal{A}}(I,s)$.
\begin{equation}\label{def:tbl:cxt}
{{\begin{picture}(120,10)(0,-4)
\put(60,0){\makebox(0,0){\footnotesize{$
\underset{\textstyle{\text{in}\;\mathrmbf{Tbl}_{\mathcal{A}}(I',s')}}
{\mathcal{T}'\xleftarrow{\;k\;}{\scriptstyle\sum}_{h}(\mathcal{T})}
{\;\;\;\;\;\;\;\;\rightleftarrows\;\;\;\;\;\;\;\;}
\underset{\textstyle{\text{in}\;\mathrmbf{Tbl}_{\mathcal{A}}(I,s)}}
{{h}^{\ast}(\mathcal{T}')\xleftarrow{\;k'\,}\mathcal{T}}
$}}}
\end{picture}}}
\end{equation}
This fiber adjunction 
(top part of Tbl.~\ref{tbl:tbl-rel:refl:sml:shrt})
is a component of
the $X$-signature indexed adjunction of tables
$\mathrmbf{List}(X)^{\mathrm{op}}\xrightarrow{\;\mathrmbfit{tbl}_{\mathcal{A}}}\mathrmbf{Adj}
:{\langle{I,s}\rangle}\mapsto{\mathrmbf{Tbl}_{\mathcal{A}}(I,s)}$.
%
\footnote{Here,
the span 
$K'\xleftarrow{\,e\,}{\widehat{K}}\xrightarrow{\,\hat{t}\,}\mathrmbfit{tup}_{\mathcal{A}}(I,s)$
is the pullback in the context $\mathrmbf{Set}$ of the opspan
$K'\xrightarrow{\,t'\,}\mathrmbfit{tup}_{\mathcal{A}}(I',s')\xleftarrow{\,\mathrmbfit{tup}_{\mathcal{A}}(h)\,}\mathrmbfit{tup}_{\mathcal{A}}(I,s)$ and
$\widehat{K}\xleftarrow{\,k'\,}K$
is the mediating morphism for the span
$K'\xleftarrow{\,k\,}K\xrightarrow{\,t\,}\mathrmbfit{tup}_{\mathcal{A}}(I,s)$.}
%
\begin{table}
\begin{center}
\begin{tabular}{@{\hspace{5pt}}c@{\hspace{30pt}}c}
\\ 
{{\begin{tabular}[b]{c}
\setlength{\unitlength}{0.48pt}
\begin{picture}(320,220)(-40,-30)
\put(0,160){\makebox(0,0){\footnotesize{$
\underset{=(\mathrmbf{Set}{\,\downarrow\,}\mathrmbfit{tup}_{\mathcal{A}}(I',s'))}
{\mathrmbf{Tbl}_{\mathcal{A}}(I',s')}$}}}
\put(0,0){\makebox(0,0){\footnotesize{$
\underset{={\wp}\mathrmbfit{tup}_{\mathcal{A}}(I',s')}
{\mathrmbf{Rel}_{\mathcal{A}}(I',s')}$}}}
\put(240,160){\makebox(0,0){\footnotesize{$
\underset{=(\mathrmbf{Set}{\,\downarrow\,}\mathrmbfit{tup}_{\mathcal{A}}(I,s))}
{\mathrmbf{Tbl}_{\mathcal{A}}(I,s)}$}}}
\put(240,0){\makebox(0,0){\footnotesize{$
\underset{={\wp}\mathrmbfit{tup}_{\mathcal{A}}(I,s)}
{\mathrmbf{Rel}_{\mathcal{A}}(I,s)}$}}}
\put(120,190){\makebox(0,0){\scriptsize{${\scriptstyle\sum}_{{h}}$}}}
\put(122,162){\makebox(0,0){\scriptsize{${h}^{\ast}$}}}
\put(120,130){\makebox(0,0){\scriptsize{${\scriptstyle\prod}_{{h}}$}}}
\put(120,30){\makebox(0,0){\scriptsize{$\exists_{h}$}}}
\put(122,2){\makebox(0,0){\scriptsize{${h}^{{\scriptscriptstyle-}1}$}}}
\put(120,-30){\makebox(0,0){\scriptsize{$\forall_{h}$}}}
\put(-15,80){\makebox(0,0)[r]{\scriptsize{$\mathrmbfit{im}^{\mathcal{A}}_{(I',s')}$}}}
\put(16,80){\makebox(0,0)[l]{\scriptsize{$\mathrmbfit{inc}^{\mathcal{A}}_{(I',s')}$}}}
\put(225,80){\makebox(0,0)[r]{\scriptsize{$\mathrmbfit{im}^{\mathcal{A}}_{(I,s)}$}}}
\put(256,80){\makebox(0,0)[l]{\scriptsize{$\mathrmbfit{inc}^{\mathcal{A}}_{(I,s)}$}}}
\put(170,180){\vector(-1,0){100}}
\qbezier(70,160)(85,160)(100,160)\qbezier(140,160)(155,160)(170,160)\put(170,160){\vector(1,0){0}}
\put(170,140){\vector(-1,0){100}}
\put(170,20){\vector(-1,0){100}}
\qbezier(70,0)(85,0)(100,0)\qbezier(140,0)(155,0)(170,0)\put(170,0){\vector(1,0){0}}
\put(170,-20){\vector(-1,0){100}}
\put(-10,125){\vector(0,-1){90}}
\put(10,35){\vector(0,1){90}}
\put(230,125){\vector(0,-1){90}}
\put(250,35){\vector(0,1){90}}
\end{picture}
\end{tabular}}}
&
{{\begin{tabular}[b]{l}
{\scriptsize\setlength{\extrarowheight}{2pt}$\begin{array}[b]{l}
\mathrmbfit{tup}_{\mathcal{A}}(I',s')\xleftarrow{\mathrmbfit{tup}_{\mathcal{A}}(h)}\mathrmbfit{tup}_{\mathcal{A}}(I,s)
\rule[-10pt]{0pt}{10pt}
\end{array}$}
\\ 
{\scriptsize\setlength{\extrarowheight}{2pt}$\begin{array}[b]{|l|}
\hline
{\scriptstyle\sum}_{{h}} \dashv {h}^{\ast} \dashv {\scriptstyle\prod}_{{h}} 
\\
\exists_{{h}} \dashv {h}^{{\scriptscriptstyle-}1} \dashv \forall_{{h}}
\\
\mathrmbfit{im}_{\mathcal{A}}(I',s') \dashv \mathrmbfit{inc}_{\mathcal{A}}(I',s')
\\
\mathrmbfit{im}_{\mathcal{A}}(I,s) \dashv \mathrmbfit{inc}_{\mathcal{A}}(I,s) 
\\ \hline
\mathrmbfit{inc}^{\mathcal{A}}_{(I',s')}{\;\circ\;}{h}^{\ast}	
 = {h}^{{\scriptscriptstyle-}1}{\;\circ\;}\mathrmbfit{inc}^{\mathcal{A}}_{(I,s)} 
\\
\mathrmbfit{inc}^{\mathcal{A}}_{(I,s)}{\;\circ\;}{\scriptstyle\prod}_{{h}}	
 = \forall_{h}{\;\circ\;}\mathrmbfit{inc}^{\mathcal{A}}_{(I',s')} 
\\ \hline
{\scriptstyle\sum}_{{h}}{\;\circ\;}\mathrmbfit{im}^{\mathcal{A}}_{(I',s')}
 \cong \mathrmbfit{im}^{\mathcal{A}}_{(I,s)}{\;\circ\;}\exists_{{h}}
\\
{h}^{\ast}{\;\circ\;}\mathrmbfit{im}^{\mathcal{A}}_{(I,s)}
 \cong \mathrmbfit{im}^{\mathcal{A}}_{(I',s')}{\;\circ\;}{h}^{{\scriptscriptstyle-}1}
\\ \hline
\end{array}$}
\\ \\
\end{tabular}}}
\\ \\
\multicolumn{2}{c}{\textit{small fibers -- short distance}}
\end{tabular}
\end{center}
\caption{Reflection: Type Domain}
\label{tbl:tbl-rel:refl:sml:shrt}
\end{table}
\begin{theorem}\label{thm:fib:cxt:tbl:A:list:X}
The fibered context of $\mathcal{A}$-tables
$\mathrmbf{Tbl}(\mathcal{A})\xrightarrow{\mathrmbfit{sign}_{\mathcal{A}}}\mathrmbf{List}(X)^{\mathrm{op}}$
is the Grothendieck construction $\int_{\mathrmbf{List}(X)}$ 
--- visualized in the lower-right quadrant of Fig.~\ref{fig:tbl:gro:constr} ---
of the $X$-signature indexed adjunction
$\mathrmbf{List}(X)^{\mathrm{op}}\xrightarrow{\;\mathrmbfit{tbl}_{\mathcal{A}}}\mathrmbf{Adj}$.$^{\ref{grothendieck}}$
\end{theorem}
%

\subsubsection{Upper Aspect.}\label{sub:sub:sec:tbl:typ:dom:up}



Here, we show that the context of tables $\mathrmbf{Tbl}$ is fibered over type domains
via the projection passage $\mathrmbf{Tbl}\xrightarrow{\mathrmbfit{data}}\mathrmbf{Cls}^{\mathrm{op}}$.
We use the Grothendieck construction $\int_{\mathrmbf{Cls}}$
on the indexed adjunction
$\mathrmbf{Cls}^{\mathrm{op}}\xrightarrow{\!\mathrmbfit{tbl}\;}\mathrmbf{Adj}
:\mathcal{A}\mapsto\mathrmbf{Tbl}(\mathcal{A})$.
We use the same definitions as in \S~\ref{sub:sec:tbl}.
A $\mathrmbf{Tbl}$-object $\mathcal{T} = {\langle{I,s,\mathcal{A},K,t}\rangle}$,
called an table (database relation),
consists of 
a type domain $\mathcal{A} = {\langle{X,Y,\models_{\mathcal{A}}}\rangle}$ and
an $\mathcal{A}$-table ${\langle{I,s,K,t}\rangle}\in\mathrmbf{Tbl}(\mathcal{A})$.
%
A $\mathrmbf{Tbl}$-morphism
$\mathcal{T}_{2} = {\langle{I_{2},s_{2},\mathcal{A}_{2},K_{2},t_{2}}\rangle} 
\xleftarrow{{\langle{h,f,g,k}\rangle}} 
{\langle{I_{1},s_{1},\mathcal{A}_{1},K_{1},t_{1}}\rangle} = \mathcal{T}_{1}$,
called a table morphism
(see Fig.~\ref{fig:tbl:mor}),
consists of 
a signed domain morphism
${\langle{I_{2},s_{2},\mathcal{A}_{2}}\rangle}\xrightarrow{{\langle{h,f,g}\rangle}}{\langle{I_{1},s_{1},\mathcal{A}_{1}}\rangle}$
%
\footnote{A signed domain morphism factors into
a signature morphism
$\mathcal{S}_{2} = {\langle{I_{2},s_{2},X_{2}}\rangle} 
\xrightarrow{{\langle{h,f}\rangle}} 
{\langle{I_{1},s_{1},X_{1}}\rangle} = \mathcal{S}_{1}$ and
a type domain morphism 
$\mathcal{A}_{2} = {\langle{X_{2},Y_{2},\models_{\mathcal{A}_{2}}}\rangle} 
\xrightleftharpoons{{\langle{f,g}\rangle}} 
{\langle{X_{1},Y_{1},\models_{\mathcal{A}_{1}}}\rangle} = \mathcal{A}_{1}$
with common sort function $X_{2} \xrightarrow{f} X_{1}$.}
%
and a key function $K_{2}\xleftarrow{k}K_{1}$,
which satisfy the condition 
(using Lem.~\ref{lem:tup:fn:fact} in \S\ref{sub:sub:sec:tup:fn:fact}):
{\footnotesize{$k{\,\cdot\,}t_{2} = t_{1}{\,\cdot\,}\mathrmbfit{tup}(h,f,g)
=(t_{1}{\,\cdot\,}\acute{\tau}_{{\langle{f,g}\rangle}}(I_{1},s_{1})){\,\cdot\,}\mathrmbfit{tup}_{\mathcal{A}_{2}}(\hat{h})
=t_{1}{\,\cdot\,}\mathrmbfit{tup}_{\mathcal{A}_{1}}(h){\,\cdot\,}\grave{\tau}_{{\langle{f,g}\rangle}}(I_{2},s_{2})$}}.
This gives two alternate, but equivalent, definitions in terms of fibers.
%


\begin{lemma}\label{lem:tbl:fbr:fact:typ:dom}
For any signed domain morphism
${\langle{I_{2},s_{2},\mathcal{A}_{2}}\rangle}\xrightarrow{{\langle{h,f,g}\rangle}}{\langle{I_{1},s_{1},\mathcal{A}_{1}}\rangle}$,
the tuple resolution 
{\footnotesize{$
\mathrmbfit{tup}(h,f,g)
=\acute{\tau}_{{\langle{f,g}\rangle}}(I_{1},s_{1}){\,\cdot\,}\mathrmbfit{tup}_{\mathcal{A}_{2}}(\hat{h})
=\mathrmbfit{tup}_{\mathcal{A}_{1}}(h){\,\cdot\,}\grave{\tau}_{{\langle{f,g}\rangle}}(I_{2},s_{2})$
}}
(Lem.~\ref{lem:tup:fn:fact} in \S~\ref{sub:sub:sec:tup:fn:fact})
resolves the table fiber adjunction
\[\mbox{\footnotesize{
$
\mathrmbf{Tbl}(I_{2},s_{2},\mathcal{A}_{2})
{\;\xleftrightharpoons[\;\;\;\;{\langle{h,f,g}\rangle}^{\ast}]{\;{\scriptscriptstyle\sum}_{{\langle{h,f,g}\rangle}}\;\;\;}\;}
\mathrmbf{Tbl}(I_{1},s_{1},\mathcal{A}_{1})$.
}\normalsize}\]
into the table fiber adjunction factorization in Fig.~\ref{fig:tbl:fbr:fact:typ:dom}.
\end{lemma}
\begin{figure}
\begin{center}
{{\begin{tabular}{c}
\setlength{\unitlength}{0.54pt}
\begin{picture}(240,210)(0,-50)
\put(0,160){\makebox(0,0){\footnotesize{$
\underset{=\,(\mathrmbf{Set}{\,\downarrow\,}\mathrmbfit{tup}_{\mathcal{A}_{2}}(I_{2},s_{2}))}
{\mathrmbf{Tbl}_{\mathcal{A}_{2}}(I_{2},s_{2})}$}}}
\put(240,160){\makebox(0,0){\footnotesize{$
\underset{=\,(\mathrmbf{Set}{\,\downarrow\,}\mathrmbfit{tup}_{\mathcal{A}_{1}}({\scriptstyle\sum}_{f}(I_{2},s_{2})))}
{\mathrmbf{Tbl}_{\mathcal{A}_{1}}({\scriptscriptstyle\sum}_{f}(I_{2},s_{2}))}$}}}
\put(0,0){\makebox(0,0){\footnotesize{$
\underset{=\,(\mathrmbf{Set}{\,\downarrow\,}\mathrmbfit{tup}_{\mathcal{A}_{2}}({f}^{\ast}(I_{1},s_{1})))}
{\mathrmbf{Tbl}_{\mathcal{A}_{2}}({f}^{\ast}(I_{1},s_{1}))}$}}}
\put(240,0){\makebox(0,0){\footnotesize{$
\underset{=\,(\mathrmbf{Set}{\,\downarrow\,}\mathrmbfit{tup}_{\mathcal{A}_{1}}(I_{1},s_{1}))}
{\mathrmbf{Tbl}_{\mathcal{A}_{1}}(I_{1},s_{1})}$}}}
\put(120,170){\makebox(0,0){\scriptsize{${\grave{\tau}_{{\langle{f,g}\rangle}}(I_{2},s_{2})}^{\ast}$}}}
\put(130,-15){\makebox(0,0){\scriptsize{${\scriptstyle\sum}_{\acute{\tau}_{{\langle{f,g}\rangle}}(I_{1},s_{1})}$}}}
\put(-6,80){\makebox(0,0)[r]{\scriptsize{${\scriptstyle\sum}_{\mathrmbfit{tup}_{\mathcal{A}_{2}}(\hat{h})}$}}}
\put(246,80){\makebox(0,0)[l]{\scriptsize{${\mathrmbfit{tup}_{\mathcal{A}_{1}}(h)}^{\ast}$}}}
\put(120,60){\makebox(0,0)[r]{\scriptsize{${\scriptstyle\sum}_{{\langle{h,f,g}\rangle}}$}}}
\put(124,104){\makebox(0,0)[l]{\scriptsize{${{\langle{h,f,g}\rangle}}^{\ast}$}}}
\put(70,160){\vector(1,0){100}}
\put(180,0){\vector(-1,0){100}}
\put(0,35){\vector(0,1){90}}
\put(240,125){\vector(0,-1){90}}
\put(60,135){\vector(3,-2){160}}
\put(180,25){\vector(-3,2){160}}
\put(360,100){\makebox(0,0)[l]{\footnotesize{$\underset{\textit{{sort function}}}
{X_{2}\xrightarrow{f}X_{1}}$}}}
\put(360,60){\makebox(0,0)[l]{\scriptsize{$\underset{X_{1}-\textit{{signature morphism}}}
{{\scriptstyle\sum}_{f}(I_{2},s_{2})\xrightarrow{h}{\langle{I_{1},s_{1}}\rangle}}$}}}
\put(360,20){\makebox(0,0)[l]{\scriptsize{$\underset{X_{2}-\textit{{signature morphism}}}
{{\langle{I_{2},s_{2}}\rangle}\xrightarrow{\hat{h}}{f}^{\ast}(I_{1},s_{1})}$}}}
\put(120,-60){\makebox(0,0){\footnotesize{$\underset{\textit{{type domain morphism}}}
{\mathcal{A}_{2}\xrightarrow{{\langle{f,g}\rangle}}\mathcal{A}_{1}}$}}}
\end{picture}
\end{tabular}}}
\end{center}
\caption{Table Fiber Adjunction Factorization}
\label{fig:tbl:fbr:fact:typ:dom}
\end{figure}
\mbox{}\newline
For any type domain $\mathcal{A}$,
the fibered context of $\mathcal{A}$-tables $\mathrmbf{Tbl}(\mathcal{A})$
separates into the partition
$\mathrmbf{Tbl}(\mathcal{A})
= \coprod_{\!\!\!\underset{\in\mathrmbf{List}(X)}{\langle{I,s}\rangle}\!\!\!}\mathrmbf{Tbl}_{\mathcal{A}}(I,s)$.
For any type domain morphism
${\mathcal{A}_{2}\xrightleftharpoons{{\langle{f,g}\rangle}}\mathcal{A}_{1}}$,
we can sum the partitions of fibered passages as follows:
\begin{center}
{{\begin{tabular}{c}
\setlength{\unitlength}{0.54pt}
\begin{picture}(300,120)(0,-10)
\put(0,40){\makebox(0,0){\footnotesize{$
\underset{\coprod_{\!\!\!\underset{\in\mathrmbf{List}(X_{2})}{\langle{I_{2},s_{2}}\rangle}
\!\!\!\!\!}\mathrmbf{Tbl}_{\mathcal{A}_{2}}(I_{2},s_{2})}{\underbrace{\mathrmbf{Tbl}(\mathcal{A}_{2})}}$}}}
\put(305,40){\makebox(0,0){\footnotesize{$
\underset{\coprod_{\!\!\!\underset{\in\mathrmbf{List}(X_{1})}
{\langle{I_{1},s_{1}}\rangle}\!\!\!\!\!}\mathrmbf{Tbl}_{\mathcal{A}_{1}}(I_{1},s_{1})}
{\underbrace{\mathrmbf{Tbl}(\mathcal{A}_{1})}}$}}}
\put(150,80){\makebox(0,0){\scriptsize{$
\xrightarrow
[\coprod_{\!\!\!\underset{\in\mathrmbf{List}(X_{2})}{\langle{I_{2},s_{2}}\rangle}\!\!\!}{\grave{\tau}_{{\langle{f,g}\rangle}}(I_{2},s_{2})}^{\ast}]
{\grave{\mathrmbfit{tbl}}_{{\langle{f,g}\rangle}}}$}}}
\put(150,15){\makebox(0,0){\scriptsize{$
\xleftarrow
[\coprod_{\!\!\!\underset{\in\mathrmbf{List}(X_{1})}{\langle{I_{1},s_{1}}\rangle}\!\!\!}{\scriptscriptstyle\sum}_{\acute{\tau}_{{\langle{f,g}\rangle}}(I_{1},s_{1})}]
{\acute{\mathrmbfit{tbl}}_{{\langle{f,g}\rangle}}}$}}}
\end{picture}
\end{tabular}}}
\end{center}
\mbox{}\newline
The factorization in Fig.~\ref{fig:tbl:fbr:fact:typ:dom},
suggests the following definitions of table fiber passages,
where
the fiber passage
$\mathrmbf{Tbl}(\mathcal{A}_{2})\xleftarrow{\acute{\mathrmbfit{tbl}}_{{\langle{f,g}\rangle}}}\mathrmbf{Tbl}(\mathcal{A}_{1})$
is define in terms of the tuple function
{\footnotesize{$
\mathrmbfit{tup}_{\mathcal{A}_{2}}(f^{\ast}(I_{1},s_{1}))
\xleftarrow[\hat{f} \cdot {(\mbox{-})} \cdot g]{\acute{\tau}_{{\langle{f,g}\rangle}}(I_{1},s_{1})} 
\mathrmbfit{tup}_{\mathcal{A}_{1}}(I_{1},s_{1})
$}}
%
and the substitution (inverse image, pullback) function
$\mathrmbf{List}(X_{2})\xleftarrow{f^{\ast}}\mathrmbf{List}(X_{1})$,
and
the adjoint fiber passage
$\mathrmbf{Tbl}(\mathcal{A}_{2})\xrightarrow{\grave{\mathrmbfit{tbl}}_{{\langle{f,g}\rangle}}}\mathrmbf{Tbl}(\mathcal{A}_{1})$
is define in terms the adjoints,
the tuple function
{\footnotesize{$
\mathrmbfit{tup}_{\mathcal{A}_{2}}(I_{2},s_{2})
\xleftarrow[{(\mbox{-})} \cdot g]{\grave{\tau}_{{\langle{f,g}\rangle}}(I_{2},s_{2})} 
\mathrmbfit{tup}_{\mathcal{A}_{1}}({\scriptstyle\sum}_{f}(I_{2},s_{2}))
$}}
%
and the existential quantifier (direct image) function
$\mathrmbf{List}(X_{2})\xrightarrow{{\scriptscriptstyle\sum}_{f}}\mathrmbf{List}(X_{1})$.


%
\begin{definition}(adjoint table fiber passages)\label{def:fbr:pass:typ:dom}
\begin{description}
\item[$\acute{\mathrmbfit{tbl}}_{{\langle{f,g}\rangle}}$:] 
An $\mathcal{A}_{1}$-table is mapped to an $\mathcal{A}_{2}$-table as follows:
\[\mbox{\footnotesize{
${\langle{f^{\ast}(I_{1},s_{1}),{\scriptscriptstyle\sum}_{\acute{\tau}_{{\langle{f,g}\rangle}}(I_{1},s_{1})}(K_{1},t_{1})}\rangle}
\stackrel{\acute{\mathrmbfit{tbl}}_{{\langle{f,g}\rangle}}}{\mapsfrom}
{\langle{I_{1},s_{1},K_{1},t_{1}}\rangle}$,
}}\]
where 
${\langle{f^{\ast}(I_{1},s_{1}),\mathcal{A}_{2}}\rangle}$-tuple
${\scriptstyle\sum}_{\acute{\tau}_{{\langle{f,g}\rangle}}(I_{1},s_{1})}(K_{1},t_{1})
= {\langle{K_{1},t_{1}{\cdot\,}\acute{\tau}_{{\langle{f,g}\rangle}}(I_{1},s_{1})}\rangle}
 \in \mathrmbfit{tup}_{\mathcal{A}_{2}}(f^{\ast}(I_{1},s_{1}))$
is the existential (direct) image of 
${\langle{I_{1},s_{1},\mathcal{A}_{1}}\rangle}$-tuple 
${\langle{K_{1},t_{1}}\rangle} \in \mathrmbfit{tup}_{\mathcal{A}_{1}}(I_{1},s_{1})$
along $\acute{\tau}_{{\langle{f,g}\rangle}}(I_{1},s_{1})$.
A morphism of $\mathcal{A}_{1}$-tables
${\langle{h_{1},k_{1}}\rangle} : {\langle{I_{1},s_{1},K_{1},t_{1}}\rangle} \rightarrow {\langle{I'_{1},s'_{1},K'_{1},t'_{1}}\rangle}$
is mapped to the morphism of $\mathcal{A}_{2}$-tables
\newline
${\langle{f^{\ast}(h_{1}),k_{1}}\rangle} :
{\langle{f^{\ast}(I_{1},s_{1}),K_{1},t_{1}{\cdot\,}\acute{\tau}_{{\langle{f,g}\rangle}}(I_{1},s_{1})}\rangle}
\rightarrow
{\langle{f^{\ast}(I'_{1},s'_{1}),K'_{1},t'_{1}{\cdot\,}\acute{\tau}_{{\langle{f,g}\rangle}}(I'_{1},s'_{1})}\rangle}$.
\newline
\item[$\grave{\mathrmbfit{tbl}}_{{\langle{f,g}\rangle}}$:] 
An $\mathcal{A}_{2}$-table is mapped to an $\mathcal{A}_{1}$-table as follows: 
\[\mbox{\footnotesize{
${\langle{I_{2},s_{2},K_{2},t_{2}}\rangle}
\stackrel{\grave{\mathrmbfit{tbl}}_{{\langle{f,g}\rangle}}}{\mapsto}
{\langle{{\scriptstyle\sum}_{f}(I_{2},s_{2}),(\grave{\tau}_{{\langle{f,g}\rangle}}(I_{2},s_{2}))^{\ast}(K_{2},t_{2})}\rangle}$,
}}\]
where
${\langle{{\scriptstyle\sum}_{f}(I_{2},s_{2}),\mathcal{A}_{1}}\rangle}$-tuple
$(\grave{\tau}_{{\langle{f,g}\rangle}}(I_{2},s_{2}))^{\ast}(K_{2},t_{2}) = {\langle{\widehat{K}_{1},\hat{t}_{1}}\rangle}$
is the substitution (inverse image) of 
${\langle{I_{2},s_{2},\mathcal{A}_{2}}\rangle}$-tuple 
${\langle{K_{2},t_{2}}\rangle}$
along $\grave{\tau}_{{\langle{f,g}\rangle}}(I_{2},s_{2})$
(Fig.~\ref{fig:tbl:mor:large:typ:dom}).
A morphism of $\mathcal{A}_{2}$-tables
${\langle{h_{2},k_{2}}\rangle} : {\langle{I_{2},s_{2},K_{2},t_{2}}\rangle} \rightarrow {\langle{I'_{1},s'_{1},K'_{1},t'_{1}}\rangle}$
is mapped to the morphism of $\mathcal{A}_{1}$-tables
${\langle{{\scriptstyle\sum}_{f}(h_{2}),k_{1}}\rangle} : 
{\langle{{\scriptstyle\sum}_{f}(I_{2},s_{2}),\widehat{K}_{1},\hat{t}_{1}}\rangle} \rightarrow
{\langle{{\scriptstyle\sum}_{f}(I_{2}',s_{2}'),\widehat{K}_{1}',\hat{t}_{1}'}\rangle}$, 
where
$k_{1} : \widehat{K}_{1}\rightarrow\widehat{K}_{1}'$ is the unique mediating function for the span
$K_{2}'\xleftarrow{\hat{k}{\,\cdot\,}k_{2}}K_{1} 
\xrightarrow{\hat{t}_{1}{\,\cdot\,}\mathrmbfit{tup}_{\mathcal{A}_{1}}({\scriptstyle\sum}_{f}(h_{2}))} 
\mathrmbfit{tup}_{\mathcal{A}_{1}}({\scriptstyle\sum}_{f}(I_{2}',s_{2}'))$,
since
$(\hat{k}{\,\cdot\,}k_{2}) \cdot t_{2}'
= \hat{k}{\,\cdot\,}t_{2}{\,\cdot\,}\mathrmbfit{tup}_{\mathcal{A}_{2}}(h_{2})
= \hat{t}_{1}{\,\cdot\,}\grave{\tau}_{{\langle{f,g}\rangle}}(I_{2},s_{2}){\,\cdot\,}\mathrmbfit{tup}_{\mathcal{A}_{2}}(h_{2})
= (\hat{t}_{1}{\,\cdot\,}\mathrmbfit{tup}_{\mathcal{A}_{1}}({\scriptstyle\sum}_{f}(h_{2}))){\,\cdot\,}\grave{\tau}_{{\langle{f,g}\rangle}}(I_{2}',s_{2}')$.
\end{description}
\end{definition}
%


\begin{figure}
\begin{center}
{{\begin{tabular}{c@{\hspace{25pt}}c@{\hspace{25pt}}c}
%
{{\begin{tabular}{c}
{\setlength{\unitlength}{0.4pt}\begin{picture}(360,180)(0,-65)
\put(0,80){\makebox(0,0){\footnotesize{$K_{2}$}}}
\put(180,80){\makebox(0,0){\footnotesize{$K_{1}$}}}
\put(360,80){\makebox(0,0){\footnotesize{$K_{1}$}}}
\put(0,-25){\makebox(0,0){\footnotesize{$\mathrmbfit{tup}_{\mathcal{A}_{2}}(\mathcal{S}_{2})$}}}
\put(180,-25){\makebox(0,0){\footnotesize{$\mathrmbfit{tup}_{\mathcal{A}_{2}}(f^{\ast}(\mathcal{S}_{1}))$}}}
\put(360,-25){\makebox(0,0){\footnotesize{$\mathrmbfit{tup}_{\mathcal{A}_{1}}(\mathcal{S}_{1})$}}}
\put(0,-56){\makebox(0,0){\footnotesize{$\underset{\textstyle{\mathcal{T}_{2}}}{\underbrace{\rule{40pt}{0pt}}}$}}}
\put(180,-56){\makebox(0,0){\footnotesize{$\underset{\textstyle{\acute{\mathrmbfit{tbl}}_{{\langle{f,g}\rangle}}(\mathcal{T}_{1})}}{\underbrace{\rule{40pt}{0pt}}}$}}}
\put(360,-56){\makebox(0,0){\footnotesize{$\underset{\textstyle{\mathcal{T}_{1}}}{\underbrace{\rule{40pt}{0pt}}}$}}}
\put(-5,26){\makebox(0,0)[r]{\scriptsize{$t_{2}$}}}
\put(188,26){\makebox(0,0)[l]{\scriptsize{$t_{1}{\,\cdot\,}\acute{\tau}_{{\langle{f,g}\rangle}}(I_{1},s_{1})$}}}
\put(370,26){\makebox(0,0)[l]{\scriptsize{$t_{1}$}}}
\put(85,-8){\makebox(0,0){\scriptsize{$\mathrmbfit{tup}_{\mathcal{A}_{2}}\!(\hat{h})$}}}
\put(85,-38){\makebox(0,0){\scriptsize{${(\mbox{-})}{\,\cdot\,}\hat{h}$}}}
\put(290,-8){\makebox(0,0){\scriptsize{$\acute{\tau}_{{\langle{f,g}\rangle}}\!(I_{1},\!s_{1})$}}}
\put(275,-38){\makebox(0,0){\scriptsize{$\hat{f}{\,\cdot\,}{(\mbox{-})}{\,\cdot\,}g$}}}
\put(180,125){\makebox(0,0){\scriptsize{$k$}}}
\put(95,92){\makebox(0,0){\scriptsize{$k$}}}
\put(155,80){\vector(-1,0){130}}
\put(270,80){\makebox(0,0){\normalsize{$=$}}}
\put(0,60){\vector(0,-1){68}}
\put(180,60){\vector(0,-1){68}}
\put(100,-25){\vector(-1,0){40}}
\put(300,-25){\vector(-1,0){40}}
\put(360,60){\vector(0,-1){68}}
\put(10,100){\oval(20,20)[tl]}
\put(10,110){\line(1,0){160}}
\put(0,94){\vector(0,-1){0}}
\put(170,110){\line(1,0){180}}
\put(350,100){\oval(20,20)[tr]}
\put(90,-98){\makebox(0,0){\footnotesize{$\underset{\textstyle{\mathrmbf{Tbl}(\mathcal{A}_{2})}}{\underbrace{\rule{100pt}{0pt}}}$}}}
\put(360,-98){\makebox(0,0){\footnotesize{$\underset{\textstyle{\mathrmbf{Tbl}(\mathcal{A}_{1})}}{\underbrace{\rule{40pt}{0pt}}}$}}}
\end{picture}}
\end{tabular}}}
&
{{\begin{tabular}{c}
{\setlength{\unitlength}{0.4pt}\begin{picture}(0,180)(0,5)
\put(0,100){\makebox(0,0){\normalsize{$\cong$}}}
\end{picture}}
\end{tabular}}}
&
{{\begin{tabular}{c}
{\setlength{\unitlength}{0.4pt}\begin{picture}(360,180)(0,-65)
\put(0,80){\makebox(0,0){\footnotesize{$K_{2}$}}}
\put(180,80){\makebox(0,0){\footnotesize{$\widehat{K}_{1}$}}}
\put(360,80){\makebox(0,0){\footnotesize{$K_{1}$}}}
\put(0,-25){\makebox(0,0){\footnotesize{$\mathrmbfit{tup}_{\mathcal{A}_{2}}(\mathcal{S}_{2})$}}}
\put(180,-25){\makebox(0,0){\footnotesize{$\mathrmbfit{tup}_{\mathcal{A}_{1}}({\scriptstyle\sum}_{f}(\mathcal{S}_{2}))$}}}
\put(360,-25){\makebox(0,0){\footnotesize{$\mathrmbfit{tup}_{\mathcal{A}_{1}}(\mathcal{S}_{1})$}}}
\put(0,-56){\makebox(0,0){\footnotesize{$\underset{\textstyle{\mathcal{T}_{2}}}{\underbrace{\rule{40pt}{0pt}}}$}}}
\put(180,-56){\makebox(0,0){\footnotesize{$\underset{\textstyle{\grave{\mathrmbfit{tbl}}_{{\langle{f,g}\rangle}}(\mathcal{T}_{2})}}{\underbrace{\rule{40pt}{0pt}}}$}}}
\put(360,-56){\makebox(0,0){\footnotesize{$\underset{\textstyle{\mathcal{T}_{1}}}{\underbrace{\rule{40pt}{0pt}}}$}}}
\put(-5,26){\makebox(0,0)[r]{\scriptsize{$t_{2}$}}}
\put(188,26){\makebox(0,0)[l]{\scriptsize{$\hat{t}_{1}$}}}
\put(370,26){\makebox(0,0)[l]{\scriptsize{$t_{1}$}}}
\put(85,-10){\makebox(0,0){\scriptsize{$\grave{\tau}_{{\langle{f,g}\rangle}}(I_{2},s_{2})$}}}
\put(85,-38){\makebox(0,0){\scriptsize{${(\mbox{-})}{\,\cdot\,}g$}}}
\put(290,-10){\makebox(0,0){\scriptsize{$\mathrmbfit{tup}_{\!\mathcal{A}_{1}}\!(h)$}}}
\put(280,-38){\makebox(0,0){\scriptsize{$h{\,\cdot\,}{(\mbox{-})}$}}}
\put(95,92){\makebox(0,0){\scriptsize{$\hat{k}$}}}
\put(180,125){\makebox(0,0){\scriptsize{$k$}}}
\put(270,92){\makebox(0,0){\scriptsize{$\tilde{k}$}}}
\put(335,80){\vector(-1,0){130}}
\put(155,80){\vector(-1,0){130}}
\put(0,60){\vector(0,-1){68}}
\put(180,60){\vector(0,-1){68}}
\put(100,-25){\vector(-1,0){40}}
\put(300,-25){\vector(-1,0){40}}
\put(360,60){\vector(0,-1){68}}
\put(10,100){\oval(20,20)[tl]}
\put(10,110){\line(1,0){160}}
\put(0,94){\vector(0,-1){0}}
\put(170,110){\line(1,0){180}}
\put(350,100){\oval(20,20)[tr]}
\qbezier(20,20)(20,20)(40,20)
\qbezier(40,4)(40,12)(40,20)
\put(90,30){\makebox(0,0){\scriptsize{\textit{pullback}}}}
\put(270,70){\makebox(0,0){\scriptsize{\textit{mediator}}}}
\put(0,-98){\makebox(0,0){\footnotesize{$\underset{\textstyle{\mathrmbf{Tbl}(\mathcal{A}_{2})}}{\underbrace{\rule{40pt}{0pt}}}$}}}
\put(270,-98){\makebox(0,0){\footnotesize{$\underset{\textstyle{\mathrmbf{Tbl}(\mathcal{A}_{1})}}{\underbrace{\rule{100pt}{0pt}}}$}}}
\end{picture}}
\end{tabular}}}
\\&&\\&&\\&&\\
\multicolumn{3}{c}{\footnotesize{$k{\,\cdot\,}t_{2} = t_{1}{\,\cdot\,}\mathrmbfit{tup}(h,f,g)
\newline
=t_{1}{\,\cdot\,}\acute{\tau}_{{\langle{f,g}\rangle}}(I_{1},s_{1}){\,\cdot\,}\mathrmbfit{tup}_{\mathcal{A}_{2}}(\hat{h})
=t_{1}{\,\cdot\,}\mathrmbfit{tup}_{\mathcal{A}_{1}}(h){\,\cdot\,}\grave{\tau}_{{\langle{f,g}\rangle}}(I_{2},s_{2})$}}
\end{tabular}}}
\end{center}
\caption{Table Morphism: Type Domain}
\label{fig:tbl:mor:large:typ:dom}
\end{figure}
\begin{description}
\item[levo:] 
A table morphism (left side Fig.~\ref{fig:tbl:mor:large:typ:dom})
consists of
type domain morphism 
$\mathcal{A}_{2}\xrightleftharpoons{{\langle{f,g}\rangle}}\mathcal{A}_{1}$
\underline{and} 
a morphism 
$\mathcal{T}_{2}\xleftarrow{{\langle{\hat{h},k}\rangle}}\acute{\mathrmbfit{tbl}}_{{\langle{f,g}\rangle}}(\mathcal{T}_{1})$
in the fiber context $\mathrmbf{Tbl}(\mathcal{A}_{2})$,
where the fiber passage
$\mathrmbf{Tbl}(\mathcal{A}_{2})\xleftarrow{\acute{\mathrmbfit{tbl}}_{{\langle{f,g}\rangle}}}\mathrmbf{Tbl}(\mathcal{A}_{1})$
along the
type domain morphism $\mathcal{A}_{2}\xrightleftharpoons{{\langle{f,g}\rangle}}\mathcal{A}_{1}$
is defined in 
Def.~\ref{def:fbr:pass:typ:dom}.
This fiber passage is a component of
the type domain indexed context of tables 
$\mathrmbf{Cls}^{\mathrm{op}}\!\xrightarrow{\;\acute{\mathrmbfit{tbl}}\;}\mathrmbf{Cxt}$.
\newline
\item[dextro:] 
A table morphism (right side Fig.~\ref{fig:tbl:mor:large:typ:dom})
consists of
type domain morphism 
$\mathcal{A}_{2}\xrightleftharpoons{{\langle{f,g}\rangle}}\mathcal{A}_{1}$
\underline{and} 
a morphism 
$\grave{\mathrmbfit{tbl}}_{{\langle{f,g}\rangle}}(\mathcal{T}_{2})\xleftarrow{{\langle{h,\tilde{k}}\rangle}}\mathcal{T}_{1}$
in the fiber context $\mathrmbf{Tbl}(\mathcal{A}_{1})$,
where the fiber passage
$\mathrmbf{Tbl}(\mathcal{A}_{2})\xrightarrow{\grave{\mathrmbfit{tbl}}_{{\langle{f,g}\rangle}}}\mathrmbf{Tbl}(\mathcal{A}_{1})$
along the
type domain morphism $\mathcal{A}_{2}\xrightleftharpoons{{\langle{f,g}\rangle}}\mathcal{A}_{1}$
is defined in 
Def.~\ref{def:fbr:pass:typ:dom}.
This fiber passage is a component of
the type domain indexed context of tables 
$\mathrmbf{Cls}\xrightarrow{\;\grave{\mathrmbfit{tbl}}\;}\mathrmbf{Cxt}$.
\end{description}
\begin{proposition}\label{prop:typ:dom:tbl:fbr:adj}
For any type domain morphism $\mathcal{A}_{2}\xrightleftharpoons{{\langle{f,g}\rangle}}\mathcal{A}_{1}$,
there is a table fiber adjunction
$\mathrmbf{Tbl}(\mathcal{A}_{2})
\xleftarrow{{\langle{
\acute{\mathrmbfit{tbl}}_{{\langle{f,g}\rangle}}{\!\dashv\,}\grave{\mathrmbfit{tbl}}_{{\langle{f,g}\rangle}}
}\rangle}}
\mathrmbf{Tbl}(\mathcal{A}_{1})$.
%
This fiber adjunction is a component of
a type domain indexed adjunction of tables 
$\mathrmbf{Cls}^{\mathrm{op}}\xrightarrow{\;\mathrmbfit{tbl}\;}\mathrmbf{Adj}$.
\end{proposition}



%
\begin{theorem}\label{thm:fib:cxt:tbl:cls}
The fibered context of tables
$\mathrmbf{Tbl}\xrightarrow{\mathrmbfit{data}}\mathrmbf{Cls}^{\mathrm{op}}$
is the Grothendieck construction $\int_{\mathrmbf{Cls}}$ 
--- visualized in the upper-right quadrant of Fig.~\ref{fig:tbl:gro:constr} ---
of the type domain indexed adjunction
$\mathrmbf{Cls}^{\mathrm{op}}\xrightarrow{\!\mathrmbfit{tbl}\;}\mathrmbf{Adj}$.
$^{\ref{grothendieck}}$
\end{theorem}
%
%

%
\begin{figure}
\begin{center}{\footnotesize{
\begin{tabular}{@{\hspace{-30pt}}c@{\hspace{0pt}}c@{\hspace{0pt}}c}
\begin{tabular}{c}
\setlength{\unitlength}{0.4pt}
\begin{picture}(0,70)(120,-35)
\put(0,0){\makebox(0,0){\footnotesize{$\mathrmbf{Cls}$}}}
\end{picture}
\end{tabular}
&
\multicolumn{2}{c}{$\mathcal{A}_{2} = {\langle{X_{2},Y_{2},\models_{\mathcal{A}_{2}}}\rangle}
\xrightleftharpoons{{\langle{f,g}\rangle}} 
{\langle{X_{1},Y_{1},\models_{\mathcal{A}_{1}}}\rangle} = \mathcal{A}_{1}$}
\\ && \\ \hline
&& \\
\begin{tabular}{c}
\setlength{\unitlength}{0.4pt}
\begin{picture}(0,70)(120,-35)
\put(8,35){\makebox(0,0){\footnotesize{$\mathrmbf{Cls}^{\mathrm{op}}$}}}
\put(10,0){\makebox(0,0)[l]{\scriptsize{$\mathrmbfit{tup}$}}}
\put(0,-35){\makebox(0,0){\footnotesize{$\bigl(\mathrmbf{Adj}{\,\Uparrow\,}\mathrmbf{Set}\bigr)$}}}
\put(0,20){\vector(0,-1){40}}
\end{picture}
\end{tabular}
&
{\footnotesize{$\begin{array}{c@{\hspace{10pt}}}
{f^{\ast}}^{\mathrm{op}}{\;\circ\;}\mathrmbfit{tup}_{\mathcal{A}_{2}}\xLeftarrow{\acute{\tau}_{{\langle{f,g}\rangle}}}\mathrmbfit{tup}_{\mathcal{A}_{1}}
\\
\acute{\tau}_{{\langle{f,g}\rangle}}
= (\varepsilon_{f}^{\mathrm{op}}{\;\circ\;}\mathrmbfit{tup}_{\mathcal{A}_{1}})\bullet({f^{\ast}}^{\mathrm{op}}{\;\circ\;}\grave{\tau}_{{\langle{f,g}\rangle}})
\\
\text{(Eqn.~\ref{typ:dom:tup:bridge} in \S~\ref{sub:sub:sec:tup:bridge:typ:dom})}
\end{array}$}}
&
{\footnotesize{$\begin{array}{@{\hspace{30pt}}c}
\mathrmbfit{tup}_{\mathcal{A}_{2}}\xLeftarrow{\grave{\tau}_{{\langle{f,g}\rangle}}}{\scriptstyle\sum}_{f}^{\mathrm{op}}{\;\circ\;}\mathrmbfit{tup}_{\mathcal{A}_{1}} 
\\
\grave{\tau}_{{\langle{f,g}\rangle}}
= ({\scriptstyle\sum}_{f}^{\mathrm{op}}{\;\circ\;}\acute{\tau}_{{\langle{f,g}\rangle}})\bullet(\eta_{f}^{\mathrm{op}}{\;\circ\;}\mathrmbfit{tup}_{\mathcal{A}_{2}})
\\
\text{(Eqn.~\ref{sign:tup:bridge} in \S~\ref{sub:sub:sec:tup:bridge:sign})}
\end{array}$}}
\\
&& \\ 
&& \\
\hline
&& \\&& \\
\begin{tabular}{c}
\setlength{\unitlength}{0.4pt}
\begin{picture}(0,360)(-80,-120)
\put(-80,180){\makebox(0,0)[r]{{$\mathrmbf{fibered}\;\;\left\{\rule{0pt}{24pt}\right.$}}}
\put(-80,60){\makebox(0,0)[r]{{$\mathrmbf{indexed}\;\;\left\{\rule{0pt}{24pt}\right.$}}}
\end{picture}
\end{tabular}
&
\begin{tabular}{c}
\setlength{\unitlength}{0.4pt}
\begin{picture}(160,340)(0,-100)
\put(0,240){\makebox(0,0){\footnotesize{$\mathrmbf{Tbl}(\mathcal{A}_{2})$}}}
\put(160,240){\makebox(0,0){\footnotesize{$\mathrmbf{Tbl}(\mathcal{A}_{1})$}}}
\put(5,120){\makebox(0,0){\footnotesize{${\mathrmbf{List}(X_{2})}^{\mathrm{op}}$}}}
\put(165,120){\makebox(0,0){\footnotesize{${\mathrmbf{List}(X_{1})}^{\mathrm{op}}$}}}
\put(0,0){\makebox(0,0){\footnotesize{$\mathrmbf{Cxt}$}}}
\put(160,0){\makebox(0,0){\footnotesize{$\mathrmbf{Cxt}$}}}
\put(-6,180){\makebox(0,0)[r]{\scriptsize{$\mathrmbfit{sign}_{\mathcal{A}_{2}}^{\mathrm{op}}$}}}
\put(170,180){\makebox(0,0)[l]{\scriptsize{$\mathrmbfit{sign}_{\mathcal{A}_{1}}^{\mathrm{op}}$}}}
\put(-6,60){\makebox(0,0)[r]{\scriptsize{$\mathrmbfit{tbl}_{\mathcal{A}_{2}}$}}}
\put(170,60){\makebox(0,0)[l]{\scriptsize{$\mathrmbfit{tbl}_{\mathcal{A}_{1}}$}}}
\put(80,258){\makebox(0,0){\scriptsize{$\acute{\mathrmbfit{tbl}}_{{\langle{f,g}\rangle}}$}}}
\put(90,138){\makebox(0,0){\scriptsize{${{f^{\ast}}^{\mathrm{op}}}$}}}
\put(80,-12){\makebox(0,0){\scriptsize{$\mathrmbfit{id}$}}}
\put(80,60){\makebox(0,0){\large{$\overset{\acute{\tau}^{{\Sigma}}_{{\langle{f,g}\rangle}}}{\Longleftarrow}$}}}
\put(105,240){\vector(-1,0){50}}
\put(100,120){\vector(-1,0){40}}
\put(130,0){\vector(-1,0){100}}
\put(0,210){\vector(0,-1){60}}
\put(160,210){\vector(0,-1){60}}
\put(0,90){\vector(0,-1){60}}
\put(160,90){\vector(0,-1){60}}
\put(80,-70){\makebox(0,0){{$\overset{\underbrace{\rule{60pt}{0pt}}}{\rule{0pt}{10pt}\mathrmbf{left}\;\;\mathrmbf{adjoint}}$}}}
\end{picture}
\end{tabular}
&
\begin{tabular}{c}
\setlength{\unitlength}{0.4pt}
\begin{picture}(160,340)(0,-100)
\put(0,240){\makebox(0,0){\footnotesize{$\mathrmbf{Tbl}(\mathcal{A}_{2})$}}}
\put(160,240){\makebox(0,0){\footnotesize{$\mathrmbf{Tbl}(\mathcal{A}_{1})$}}}
\put(5,120){\makebox(0,0){\footnotesize{${\mathrmbf{List}(X_{2})}^{\mathrm{op}}$}}}
\put(165,120){\makebox(0,0){\footnotesize{${\mathrmbf{List}(X_{1})}^{\mathrm{op}}$}}}
\put(0,0){\makebox(0,0){\footnotesize{$\mathrmbf{Cxt}$}}}
\put(160,0){\makebox(0,0){\footnotesize{$\mathrmbf{Cxt}$}}}
\put(-6,180){\makebox(0,0)[r]{\scriptsize{$\mathrmbfit{sign}_{\mathcal{A}_{2}}^{\mathrm{op}}$}}}
\put(170,180){\makebox(0,0)[l]{\scriptsize{$\mathrmbfit{sign}_{\mathcal{A}_{1}}^{\mathrm{op}}$}}}
\put(-6,60){\makebox(0,0)[r]{\scriptsize{$\mathrmbfit{tbl}_{\mathcal{A}_{2}}$}}}
\put(170,60){\makebox(0,0)[l]{\scriptsize{$\mathrmbfit{tbl}_{\mathcal{A}_{1}}$}}}
\put(80,258){\makebox(0,0){\scriptsize{$\grave{\mathrmbfit{tbl}}_{{\langle{f,g}\rangle}}$}}}
\put(90,138){\makebox(0,0){\scriptsize{${{{\Sigma}_{f}}^{\mathrm{op}}}$}}}
\put(80,-12){\makebox(0,0){\scriptsize{$\mathrmbfit{id}$}}}
\put(80,60){\makebox(0,0){\large{$\overset{\grave{\tau}^{\ast}_{{\langle{f,g}\rangle}}}{\Longleftarrow}$}}}
\put(55,240){\vector(1,0){50}}
\put(60,120){\vector(1,0){40}}
\put(30,0){\vector(1,0){100}}
\put(0,210){\vector(0,-1){60}}
\put(160,210){\vector(0,-1){60}}
\put(0,90){\vector(0,-1){60}}
\put(160,90){\vector(0,-1){60}}
\put(80,-70){\makebox(0,0){{$\overset{\underbrace{\rule{60pt}{0pt}}}{\rule{0pt}{10pt}\mathrmbf{right}\;\;\mathrmbf{adjoint}}$}}}
\end{picture}
\end{tabular}
\\ && \\ 
\begin{tabular}{c}
\setlength{\unitlength}{0.4pt}
\begin{picture}(0,70)(120,-35)
\put(8,35){\makebox(0,0){\footnotesize{$\mathrmbf{Cls}^{\mathrm{op}}$}}}
\put(10,0){\makebox(0,0)[l]{\scriptsize{$\mathrmbfit{tbl}$}}}
\put(0,-35){\makebox(0,0){\footnotesize{$\mathrmbf{Adj}$}}}
\put(0,20){\vector(0,-1){40}}
\end{picture}
\end{tabular}
&
{\footnotesize{$\begin{array}{c@{\hspace{10pt}}}
\mathrmbfit{tup}_{\mathcal{A}_{2}}(f^{\ast}(\mathcal{S}_{1}))
\xleftarrow{\acute{\tau}_{{\langle{f,g}\rangle}}(\mathcal{S}_{1})}
\mathrmbfit{tup}_{\mathcal{A}_{1}}(\mathcal{S}_{1})
\\
\mathrmbfit{tbl}_{\mathcal{A}_{2}}(f^{\ast}(\mathcal{S}_{1}))
\xleftarrow{{\Sigma}_{\acute{\tau}_{{\langle{f,g}\rangle}}(\mathcal{S}_{1})}}
\mathrmbfit{tbl}_{\mathcal{A}_{1}}(\mathcal{S}_{1})
\end{array}$}}
&
{\footnotesize{$\begin{array}{@{\hspace{30pt}}c}
\mathrmbfit{tup}_{\mathcal{A}_{2}}(\mathcal{S}_{2})
\xleftarrow{\grave{\tau}_{{\langle{f,g}\rangle}}(\mathcal{S}_{2})}
\mathrmbfit{tup}_{\mathcal{A}_{1}}({\Sigma}_{f}(\mathcal{S}_{2}))
\\
\mathrmbfit{tbl}_{\mathcal{A}_{2}}(\mathcal{S}_{2})
\xrightarrow{{\grave{\tau}_{{\langle{f,g}\rangle}}(\mathcal{S}_{2})}^{\ast}}
\mathrmbfit{tbl}_{\mathcal{A}_{1}}({\Sigma}_{f}(\mathcal{S}_{2}))
\end{array}$}}
\\
&& \\
&& \\
\end{tabular}
}\normalsize}
\end{center}
\caption{Indexed Adjunction of Tables}
\label{fig:ind:fbr:tbl:cls}
\end{figure}
%

\subsection{Fibered Contexts of \texttt{FOLE} Tables}\label{sub:sec:fbr:cxt:tbl}

The Grothendieck constructions for \texttt{FOLE} tables are listed in Tbl.~\ref{tbl:grothen:construct}.
Here we indicated whether the construction is a fibration, an opfibration or a bifibration.
We also list the proposition or theorem proving the construction and its location.
The Grothendieck constructions for \texttt{FOLE} tables are displayed in Fig.~\ref{fig:tbl:gro:constr}.

\begin{table}
\begin{center}
{{\begin{tabular}{c}
{\scriptsize\setlength{\extrarowheight}{3pt}
$\begin{array}{|@{\hspace{3pt}}r@{\hspace{4pt}}|@{\hspace{5pt}}l@{\hspace{3pt}}l@{\hspace{3pt}}|@{\hspace{5pt}}l@{\hspace{15pt}=\int\hspace{15pt}}l@{\hspace{5pt}}|}
\multicolumn{1}{r}{}
&
\multicolumn{2}{l}{}
&
\multicolumn{1}{l}{\textsf{fibered construct}}
&
\multicolumn{1}{l}{\textsf{indexed construct}}
\\\hline
\bowtie
&
\S\ref{sub:sec:fole:comps:sign}
&
\text{(Thm.\ref{thm:fib:cxt:sign:set})}
&
\rule[8pt]{0pt}{5pt}
\mathrmbf{List}\xrightarrow{\;\mathrmbfit{sort}\;}\mathrmbf{Set}
&
\mathrmbf{Set}\xrightarrow{\;\mathrmbfit{list}\;}\mathrmbf{Adj}
\\
\triangleleft
&
\S\ref{sub:sec:fole:comps:typ:dom}
&
\text{(Thm.\ref{thm:fib:cxt:cls:set})}
&
\mathrmbf{Cls}\xrightarrow{\;\mathrmbfit{sort}\;}\mathrmbf{Set}
&
\mathrmbf{Set}^{\mathrm{op}}\!\xrightarrow{\;\mathrmbfit{cls}\;}\mathrmbf{Cxt}
\\\hline\hline
\bowtie
&
\S\ref{sub:sec:tbl:sign:dom}
&
\text{(Thm.\ref{thm:fib:cxt:tbl:dom})}
&
\rule[8pt]{0pt}{5pt}
\mathrmbf{Tbl}\xrightarrow{\;\mathrmbfit{dom}}\mathrmbf{Dom}^{\mathrm{op}}
&
\mathrmbf{Dom}^{\mathrm{op}}\!\xrightarrow{\;\mathrmbfit{tbl}\;}\mathrmbf{Adj}
\\\hline
\bowtie
&
\S\ref{sub:sub:sec:tbl:sign:low}
&
\text{(Thm.\ref{thm:fib:cxt:tbl:S:cls:X})}
&
\rule[8pt]{0pt}{5pt}
\mathrmbf{Tbl}(\mathcal{S})\xrightarrow{\;\mathrmbfit{data}_{\mathcal{S}}}\mathrmbf{Cls}(X)^{\mathrm{op}}
& 
\mathrmbf{Cls}(X)^{\mathrm{op}}\!\xrightarrow{\;\mathrmbfit{tbl}_{\mathcal{S}}\;}\mathrmbf{Adj}
\\
\triangleright
&
\S\ref{sub:sub:sec:tbl:sign:up}
&
\text{(Thm.\ref{thm:fib:cxt:tbl:sign})}
&
\mathrmbf{Tbl}\xrightarrow{\;\mathrmbfit{sign}\;}\mathrmbf{List}^{\mathrm{op}}
&
\mathrmbf{List}^{\mathrm{op}}\!\xrightarrow{\;\mathrmbfit{tbl}\;}\mathrmbf{Cxt}
\\\hline
\bowtie
&
\S\ref{sub:sub:sec:tbl:typ:dom:low}
&
\text{(Thm.\ref{thm:fib:cxt:tbl:A:list:X})}
&
\rule[10pt]{0pt}{5pt}
\mathrmbf{Tbl}(\mathcal{A})\xrightarrow{\;\mathrmbfit{sign}_{\mathcal{A}}\;}\mathrmbf{List}(X)^{\mathrm{op}}
&
\mathrmbf{List}(X)^{\mathrm{op}}\!\xrightarrow{\;\mathrmbfit{tbl}_{\mathcal{A}}\;}\mathrmbf{Adj}
\\
\bowtie
&
\S\ref{sub:sub:sec:tbl:typ:dom:up}
&
\text{(Thm.\ref{thm:fib:cxt:tbl:cls})}
&
\mathrmbf{Tbl}\xrightarrow{\;\mathrmbfit{data}\;}\mathrmbf{Cls}^{\mathrm{op}}
&
\mathrmbf{Cls}^{\mathrm{op}}\!\xrightarrow{\;\mathrmbfit{tbl}\;}\mathrmbf{Adj}
\\\hline\hline
\multicolumn{1}{|r}{}
&&
&
\multicolumn{2}{l|}{\bowtie\;\,=\textit{bifibration},\;\;\triangleleft\,=\textit{fibration},\;\;\triangleright\,=\textit{opfibration}}
\\\hline
\end{array}$}
\end{tabular}}}
\end{center}
\caption{Grothendieck Constructions}
\label{tbl:grothen:construct}
\end{table}
\begin{figure}
\begin{center}
{{\begin{tabular}{c}
\setlength{\unitlength}{0.66pt}
\begin{picture}(320,380)(0,0)
\put(160,320){\makebox(0,0){\footnotesize{$\mathrmbf{Tbl}$}}}
\put(165,370){\makebox(0,0){\footnotesize{$\mathrmbf{Dom}^{\mathrm{op}}$}}}
\put(95,321){\makebox(0,0){\footnotesize{$\mathrmbf{List}^{\mathrm{op}}$}}}
\put(230,321){\makebox(0,0){\footnotesize{$\mathrmbf{Cls}^{\mathrm{op}}$}}}
\put(164,343){\makebox(0,0)[l]{\scriptsize{$\mathrmbfit{dom}$}}}
\put(131.8,327){\makebox(0,0){\scriptsize{$\mathrmbfit{sign}$}}}
\put(190,327){\makebox(0,0){\scriptsize{$\mathrmbfit{data}$}}}
\put(160,330){\vector(0,1){30}}
\put(143,320){\vector(-1,0){30}}
\put(177,320){\vector(1,0){30}}
%
\put(320,160){\makebox(0,0){\footnotesize{$
\underset{\underset{\mathcal{A}={\langle{X,Y,\models_{\mathcal{A}}}\rangle}}
{\mathcal{A}{\;\mapsto\;}\mathrmbf{Tbl}(\mathcal{A})}}
{\mathrmbf{Cls}^{\mathrm{op}}\xrightarrow{\;\mathrmbfit{tbl}\;}\mathrmbf{Adj}}$}}}
\put(0,160){\makebox(0,0){\footnotesize{$\underset{
\underset{\mathcal{S}={\langle{I,s,X}\rangle}}
{\mathcal{S}{\;\mapsto\;}\mathrmbf{Tbl}(\mathcal{S})}}
{\mathrmbf{List}^{\mathrm{op}}{\!\xrightarrow{\:\mathrmbfit{tbl}\;}\;}\mathrmbf{Cxt}}$}}}
\put(160,0){\makebox(0,0){\footnotesize{$
\underset{{\langle{I,s,\mathcal{A}}\rangle}\mapsto\mathrmbf{Tbl}_{{\langle{I,s,\mathcal{A}}\rangle}}}
{\mathrmbf{Dom}\xrightarrow{\;\mathrmbfit{tbl}\;}\mathrmbf{Adj}}$}}}
\put(160,160){\makebox(0,0){\footnotesize{$
\underset{\stackrel{\text{signed domain}}{\text{Sec.~\ref{sub:sec:tbl:sign:dom}}}}{\int_{\mathrmbf{Dom}}}$}}}
\put(240,230){\makebox(0,0){\footnotesize{$
\underset{\stackrel{\text{type domain}}{\text{Sec.~\ref{sub:sub:sec:tbl:typ:dom:up}}}}{\int_{\mathrmbf{Cls}}}$}}}
\put(80,230){\makebox(0,0){\footnotesize{$
\underset{\stackrel{\text{signature}}{\text{Sec.~\ref{sub:sub:sec:tbl:sign:up}}}}{\int_{\mathrmbf{List}}}$}}}
\put(15,175){\vector(1,1){130}}
\put(305,175){\vector(-1,1){130}}
\put(160,20){\vector(0,1){280}}
\dottedline[$\cdot$]{5}(145,15)(20,140)\put(17,143){\vector(-1,1){0}}
\dottedline[$\cdot$]{5}(175,15)(300,140)\put(303,143){\vector(1,1){0}}
%
%
\put(425,120){\makebox(0,0){\scriptsize{$
\Bigl\{\mathrmbf{Tbl}(\mathcal{A})\xrightarrow{\mathrmbfit{sign}_{\mathcal{A}}}\mathrmbf{List}(X)^{\mathrm{op}}
\mid\mathcal{A}\in\mathrmbf{Cls}\Bigr\}$}}}
\put(340,20){\makebox(0,0){\scriptsize{$
\Bigl\{\underset{{\langle{I,s}\rangle}\mapsto{\mathrmbf{Tbl}_{\mathcal{A}}(I,s)}}
{\mathrmbf{List}(X)\xrightarrow{\;\mathrmbfit{tbl}_{\mathcal{A}}}\mathrmbf{Adj}}
\mid\mathcal{A}\in\mathrmbf{Cls}\Bigr\}$}}}
\put(320,60){\makebox(0,0){\scriptsize{$
\underset{\stackrel{\text{$X$-signature}}{\text{Sec.~\ref{sub:sub:sec:tbl:typ:dom:low}}}}{\int_{\mathrmbf{List}(X)}}
$}}}
\put(320,35){\vector(1,1){70}}
%
%
\put(-100,120){\makebox(0,0){\scriptsize{$
\Bigl\{\mathrmbf{Tbl}(\mathcal{S})\xrightarrow{\mathrmbfit{data}_{\mathcal{S}}}\mathrmbf{Cls}(X)^{\mathrm{op}}
\mid\mathcal{S}\in\mathrmbf{List}\Bigr\}$}}}
\put(-20,20){\makebox(0,0){\scriptsize{$
\Bigl\{\underset{{\langle{X,Y,\models_{\mathcal{A}}}\rangle}\mapsto{\mathrmbf{Tbl}_{\mathcal{A}}(I,s)}}
{\mathrmbf{Cls}(X)\xrightarrow{\;\mathrmbfit{tbl}_{\mathcal{S}}}\mathrmbf{Cxt}}
\mid\mathcal{S}\in\mathrmbf{List}\Bigr\}$}}}
\put(0,60){\makebox(0,0){\scriptsize{$
\underset{\stackrel{\text{$X$-type domain}}{\text{Sec.~\ref{sub:sub:sec:tbl:sign:low}}}}{\int_{\mathrmbf{Cls}(X)}}$}}}
\put(0,35){\vector(-1,1){70}}
%
%
\put(385,140){\setlength{\unitlength}{0.35pt}
\begin{picture}(160,340)(0,-100)
\put(0,240){\makebox(0,0){\footnotesize{$\mathrmbf{Tbl}(\mathcal{A}_{2})$}}}
\put(160,240){\makebox(0,0){\footnotesize{$\mathrmbf{Tbl}(\mathcal{A}_{1})$}}}
\put(5,120){\makebox(0,0){\footnotesize{${\mathrmbf{List}(X_{2})}^{\mathrm{op}}$}}}
\put(165,120){\makebox(0,0){\footnotesize{${\mathrmbf{List}(X_{1})}^{\mathrm{op}}$}}}
\put(0,0){\makebox(0,0){\footnotesize{$\mathrmbf{Cxt}$}}}
\put(160,0){\makebox(0,0){\footnotesize{$\mathrmbf{Cxt}$}}}
\put(-6,180){\makebox(0,0)[r]{\scriptsize{$\mathrmbfit{sign}_{\mathcal{A}_{2}}$}}}
\put(170,180){\makebox(0,0)[l]{\scriptsize{$\mathrmbfit{sign}_{\mathcal{A}_{1}}$}}}
\put(-6,60){\makebox(0,0)[r]{\scriptsize{$\mathrmbfit{tbl}_{\mathcal{A}_{2}}$}}}
\put(170,60){\makebox(0,0)[l]{\scriptsize{$\mathrmbfit{tbl}_{\mathcal{A}_{1}}$}}}
\put(80,258){\makebox(0,0){\scriptsize{$\acute{\mathrmbfit{tbl}}_{{\langle{f,g}\rangle}}$}}}
\put(95,138){\makebox(0,0){\scriptsize{${{f^{\ast}}^{\mathrm{op}}}$}}}
\put(80,-12){\makebox(0,0){\scriptsize{$\mathrmbfit{id}$}}}
\put(80,60){\makebox(0,0){\large{$\overset{\acute{\tau}^{{\Sigma}}_{{\langle{f,g}\rangle}}}{\Longleftarrow}$}}}
\put(105,240){\vector(-1,0){50}}
\put(100,120){\vector(-1,0){40}}
\put(130,0){\vector(-1,0){100}}
\put(0,210){\vector(0,-1){60}}
\put(160,210){\vector(0,-1){60}}
\put(0,90){\vector(0,-1){60}}
\put(160,90){\vector(0,-1){60}}
\thicklines
\put(-94,190){\vector(1,0){0}}
\put(-121,190){\line(1,0){14}}
\put(-120,50){\line(0,1){140}}
\put(-121,50){\line(1,0){14}}
\thinlines
\put(-130,120){\makebox(0,0)[r]{\scriptsize{${\int_{\mathrmbf{Cls}}}$}}}
%
\put(80,-65){\makebox(0,0){{$\overset{\underbrace{\rule{60pt}{0pt}}}{\rule{0pt}{10pt}\mathrmbf{left}\;\;\mathrmbf{adjoint}}$}}}
\end{picture}}
%
%
\put(-165,140){\setlength{\unitlength}{0.35pt}
\begin{picture}(160,340)(0,-100)
\put(0,240){\makebox(0,0){\footnotesize{$\mathrmbf{Tbl}(\mathcal{S}_{2})$}}}
\put(160,240){\makebox(0,0){\footnotesize{$\mathrmbf{Tbl}(\mathcal{S}_{1})$}}}
\put(5,120){\makebox(0,0){\footnotesize{${\mathrmbf{Cls}(X_{2})}^{\mathrm{op}}$}}}
\put(165,120){\makebox(0,0){\footnotesize{${\mathrmbf{Cls}(X_{1})}^{\mathrm{op}}$}}}
\put(0,0){\makebox(0,0){\footnotesize{$\mathrmbf{Cxt}$}}}
\put(160,0){\makebox(0,0){\footnotesize{$\mathrmbf{Cxt}$}}}
\put(-6,180){\makebox(0,0)[r]{\scriptsize{$\mathrmbfit{data}_{\mathcal{S}_{2}}$}}}
\put(170,180){\makebox(0,0)[l]{\scriptsize{$\mathrmbfit{data}_{\mathcal{S}_{1}}$}}}
\put(-6,60){\makebox(0,0)[r]{\scriptsize{$\mathrmbfit{tbl}_{\mathcal{S}_{2}}$}}}
\put(170,60){\makebox(0,0)[l]{\scriptsize{$\mathrmbfit{tbl}_{\mathcal{S}_{1}}$}}}
\put(80,258){\makebox(0,0){\scriptsize{$\acute{\mathrmbfit{tbl}}_{{\langle{h,f}\rangle}}$}}}
\put(95,138){\makebox(0,0){\scriptsize{${{f^{-1}}^{\mathrm{op}}}$}}}
\put(80,-12){\makebox(0,0){\scriptsize{$\mathrmbfit{id}$}}}
\put(80,60){\makebox(0,0){\large{$\overset{\tau^{{\Sigma}}_{{\langle{h,f}\rangle}}}{\Longleftarrow}$}}}
\put(105,240){\vector(-1,0){50}}
\put(100,120){\vector(-1,0){40}}
\put(130,0){\vector(-1,0){100}}
\put(0,210){\vector(0,-1){60}}
\put(160,210){\vector(0,-1){60}}
\put(0,90){\vector(0,-1){60}}
\put(160,90){\vector(0,-1){60}}
\thicklines
\put(254,190){\vector(-1,0){0}}
\put(281,190){\line(-1,0){14}}
\put(280,50){\line(0,1){140}}
\put(281,50){\line(-1,0){14}}
\thinlines
\put(290,120){\makebox(0,0)[l]{\scriptsize{${\int_{\mathrmbf{List}}}$}}}
%
%
\end{picture}}
\end{picture}
\\\\
\end{tabular}}}
%
\end{center}
\caption{The Fibered Hierarchy of \texttt{FOLE} Tables}
\label{fig:tbl:gro:constr}
\end{figure}
%

\newpage
\section{Table Constructions}\label{sec:tbl:lim:colim}

In this section we use properties of comma contexts and the Grothendieck construction
to prove that the various (sub)contexts of \texttt{FOLE} tables are complete (joins exist) and cocomplete (unions exist).

\subsection{Preliminaries}\label{sub:sec:prelims}

%
\begin{proposition}\label{prop:cls:co:compl}
The mathematical context of classifications (type domains) $\mathrmbf{Cls}$ is (co)complete, and
its type (sort) and instance (data) projections
$\mathrmbf{Set}^{\mathrm{op}}\!\!\xleftarrow{\;\mathrmbfit{data}^{\mathrm{op}}\!\!}\mathrmbf{Cls}\xrightarrow{\mathrmbfit{sort}}\mathrmbf{Set}$
are (co)continuous.
\end{proposition}
\begin{proof}
Barwise and Seligman~\cite{barwise:seligman:97}.
\mbox{}\hfill\rule{5pt}{5pt}
\end{proof}
%


%
\begin{proposition}\label{prop:clsX:compl}
For any sort set $X$,
the context of $X$-sorted type domains $\mathrmbf{Cls}(X)$ is complete, and
its instance (data) projection
$\mathrmbf{Set}\xleftarrow{\;\mathrmbfit{data}_{X}}\mathrmbf{Cls}(X)^{\mathrm{op}}$
is cocontinuous.
\end{proposition}
\begin{proof}
To prove the proposition in general, use the three special cases: 
any collection of $X$-sorted type domains has a \emph{product}, 
whose instance set is the coproduct (disjoint union) of the collection of instance (data) sets;
there is a \emph{terminal} $X$-sorted type domain,
whose instance (data) set is the empty set $\emptyset$; and
any opspan of $X$-sorted type domain morphisms has a \emph{pullback}, 
whose instance set is the pushout of the instance (data) projection span.
%
\mbox{}\hfill\rule{5pt}{5pt}
\end{proof}

\mbox{}\newline

\begin{proposition}\label{prop:fbr:pass:cts}
\mbox{}\newline
{{\begin{tabular}{c@{\hspace{30pt}}c}
{{\begin{tabular}[t]{p{200pt}}
For any sort function $X_{2}\xrightarrow{f}X_{1}$, 
the type domain fiber passage
$\mathrmbf{Cls}(X_{2})\xleftarrow{\mathrmbfit{cls}(f)}\mathrmbf{Cls}(X_{1})$
is continuous (preserves limits):
$\prod(\mathrmbfit{A}{\,\circ}{f}^{-1}) = {f}^{-1}(\prod\mathrmbfit{A})$
for any diagram
$\mathrmbf{I}\xrightarrow{\mathrmbfit{A}}\mathrmbf{Cls}(X_{1})$.
\end{tabular}}}
&
{{\begin{tabular}[c]{c}
\setlength{\unitlength}{0.63pt}
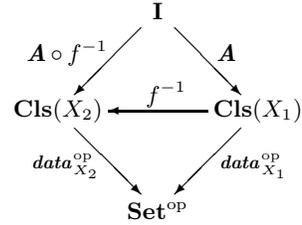
\begin{picture}(120,40)(0,57)
\put(60,120){\makebox(0,0){\footnotesize{$\mathrmbf{I}$}}}
\put(0,60){\makebox(0,0){\footnotesize{$\mathrmbf{Cls}(X_{2})$}}}
\put(120,60){\makebox(0,0){\footnotesize{$\mathrmbf{Cls}(X_{1})$}}}
\put(60,0){\makebox(0,0){\footnotesize{$\mathrmbf{Set}^{\mathrm{op}}$}}}
\put(30,95){\makebox(0,0)[r]{\footnotesize{$\mathrmbfit{A}{\,\circ\,}{f}^{-1}$}}}
\put(95,95){\makebox(0,0)[l]{\footnotesize{$\mathrmbfit{A}$}}}
\put(65,70){\makebox(0,0){\footnotesize{${f}^{-1}$}}}
\put(25,27){\makebox(0,0)[r]{\scriptsize{$\mathrmbfit{data}_{X_{2}}^{\mathrm{op}}$}}}
\put(97,27){\makebox(0,0)[l]{\scriptsize{$\mathrmbfit{data}_{X_{1}}^{\mathrm{op}}$}}}
\put(70,110){\vector(1,-1){40}}
\put(50,110){\vector(-1,-1){40}}
\put(90,60){\vector(-1,0){60}}
\put(10,50){\vector(1,-1){40}}
\put(110,50){\vector(-1,-1){40}}
\end{picture}
\end{tabular}}}
\end{tabular}}}
\end{proposition}
\begin{proof}
To prove this, show that the inverse image of the limit is the limit of the inverse image of any diagram in $\mathrmbf{Cls}(X_{1})$.
We need only show this for products and pullbacks.
We note that inverse image preserves data projection:  
${f}^{-1}{\;\circ\;}\mathrmbfit{data}_{X_{2}}^{\mathrm{op}} = \mathrmbfit{data}_{X_{1}}^{\mathrm{op}}$.%
\mbox{}\hfill\rule{5pt}{5pt}
\end{proof}
%


%
\begin{proposition}\label{prop:tup:pass:cts}
The tuple passage 
$\mathrmbf{List}(X)^{\mathrm{op}}\xrightarrow{\mathrmbfit{tup}_{\mathcal{A}}}\mathrmbf{Set}$
is continuous.
\end{proposition}
\begin{proof}
We need only show that the tuple passage 
maps the initial object in ${\mathrmbf{List}(X)}$ 
to the terminal object in $\mathrmbf{Set}$
and maps the pushout of a span in ${\mathrmbf{List}(X)}$ 
to the pullback of the image opspan in $\mathrmbf{Set}$.
\mbox{}\hfill\rule{5pt}{5pt}
\end{proof}



\newpage
\subsection{Propositions}\label{sub:sub:sec:props}

%
\begin{table}
\begin{center}
{{\begin{tabular}{c}
{\footnotesize\setlength{\extrarowheight}{3.5pt}
$\begin{array}{|@{\hspace{3pt}}r@{\hspace{10pt}}cccccccc@{\hspace{3pt}}|}
\hline
{\textsf{complete:}}\rule[-5pt]{0pt}{5pt}
&
\overset{\downarrow}{\mathrmbf{List}(X)},
&
\overset{\downarrow,\int}{\mathrmbf{List}},
&
\overset{\ast}{\mathrmbf{Cls}},
&
\overset{\downarrow}{\mathrmbf{Dom}},
&
\overset{\downarrow}{\mathrmbf{Tbl}_{\mathcal{A}}(I,s)},
&
&
\overset{\downarrow,\int}{\mathrmbf{Tbl}(\mathcal{A})},
&
\overset{\int}{\mathrmbf{Tbl}}
\\\hline\hline
{\textsf{cocomplete:}}\rule[-5pt]{0pt}{5pt}
&
\overset{\downarrow}{\mathrmbf{List}(X)},
&
\overset{\downarrow,\int}{\mathrmbf{List}},
&
\overset{\ast}{\mathrmbf{Cls}},
&
\overset{\downarrow}{\mathrmbf{Dom}},
&
\overset{\downarrow}{\mathrmbf{Tbl}_{\mathcal{A}}(I,s)},
&
\overset{\downarrow,\int}{\mathrmbf{Tbl}(\mathcal{S})},
&
\overset{\downarrow,\int}{\mathrmbf{Tbl}(\mathcal{A})},
&
\overset{\downarrow,\int}{\mathrmbf{Tbl}}
\\\hline\hline
{\textit{proven by:}}\rule[-5pt]{0pt}{5pt}
&
\multicolumn{8}{l|}{{\ast=\text{Info. Flow.}
\text{\cite{barwise:seligman:97}},
\;\;\downarrow\,=\text{comma context},
\;\;\int=\text{Grothendieck}
}}
\\\hline
\end{array}$}
\end{tabular}}}
\end{center}
\caption{Complete/Cocomplete Contexts}
\label{tbl:co:complete:cxts}
\end{table}
%

\subsubsection{Using Comma Contexts.}\label{sub:sub:sec:props:comma}

These propositions use Facts~\ref{fact:comma:lim},~\ref{fact:comma:colim} in \S~\ref{sub:sec:comma:cxt}.

\begin{proposition}\label{prop:com:cxt:lim:colim}
The comma contexts of $X$-signatures, signatures, signed domains, ${\langle{\mathcal{A},I,s}\rangle}$-tables, and $\mathcal{A}$-tables 
are associated with the following passage opspans:
\begin{center}
{{\footnotesize\setlength{\extrarowheight}{4pt}$\begin{array}{r@{\;=\;}l@{\hspace{6pt}:\hspace{25pt}}l}
\multicolumn{2}{c}{\textit{comma context}} &
\textit{passage opspan}
\\
\mathrmbf{List}(X)
& \bigl(\mathrmbf{Set}{\,\downarrow\,}X\bigr)
& \mathrmbf{Set}\xrightarrow{\;\mathrmbfit{1}\;}\mathrmbf{Set}\xleftarrow{\;X\;}\mathrmbf{1}
\\
\mathrmbf{List}
& \bigl(\mathrmbf{Set}{\,\downarrow\,}\mathrmbf{Set}\bigr)
& \mathrmbf{Set}\xrightarrow{\;\mathrmbfit{1}\;}\mathrmbf{Set}\xleftarrow{\;\mathrmbfit{1}\;}\mathrmbf{Set}
\\
\mathrmbf{Dom}
& \bigl(\mathrmbf{Set}{\,\downarrow\,}\mathrmbfit{sort}\bigr)
& \mathrmbf{Set}\xrightarrow{\;\mathrmbfit{1}\;}\mathrmbf{Set}\xleftarrow{\;\mathrmbfit{sort}\;}\mathrmbf{Cls}
\\
\mathrmbf{Tbl}_{\mathcal{A}}(I,s)
& \bigl(\mathrmbf{Set}{\,\downarrow\,}\mathrmbfit{tup}_{\mathcal{A}}(I,s)\bigr)
& \mathrmbf{Set}\xrightarrow{\;\mathrmbfit{1}\;}\mathrmbf{Set}\xleftarrow{\;\mathrmbfit{tup}_{\mathcal{A}}(I,s)\;}\mathrmbf{1}
\\
\mathrmbf{Tbl}(\mathcal{A})
& \bigl(\mathrmbf{Set}{\,\downarrow\,}\mathrmbfit{tup}_{\mathcal{A}}\bigr)
& \mathrmbf{Set}\xrightarrow{\;\mathrmbfit{1}\;}\mathrmbf{Set}\xleftarrow{\;\mathrmbfit{tup}_{\mathcal{A}}\;}\mathrmbf{List}(X)^{\mathrm{op}}
\text{.}
\end{array}$}}
\end{center}
respectively.
Hence, they are (co)complete and their projections
\begin{center}
{\begin{tabular}{@{\hspace{20pt}}c}
{{\footnotesize\setlength{\extrarowheight}{4pt}$\begin{array}[t]{@{\hspace{5pt}}r@{\hspace{5pt}}c@{\hspace{5pt}}l@{\hspace{5pt}}}
\mathrmbf{Set}\xleftarrow{\mathrmbfit{arity}_{X}}
\!\!\!\!\!\!
&
{\mathrmbf{List}(X)}
&
\!\!\!\!\!
\xrightarrow{\;\;\;}\mathrmbf{1}
\\
\mathrmbf{Set}\xleftarrow{\;\mathrmbfit{arity}\;}
\!\!\!\!\!\!\!\!\!\!
&
{\mathrmbf{List}}
&
\!\!\!\!\!\!\!\!\!
\xrightarrow{\;\mathrmbfit{sort}\;}\mathrmbf{Set}
\\
\mathrmbf{Set}\xleftarrow{\;\mathrmbfit{arity}\;}
\!\!\!\!\!\!\!\!\!
&
{\mathrmbf{Dom}}
&
\!\!\!\!\!\!\!
\xrightarrow{\;\mathrmbfit{data}\;}\mathrmbf{Cls}
\\
\mathrmbf{Set}\xleftarrow{\;\mathrmbfit{key}_{\mathcal{A}}(I,s)\;}
\!\!\!
&
{\mathrmbf{Tbl}_{\mathcal{A}}(I,s)}
&
\!
\xrightarrow{\;\;\;}\mathrmbf{1}
\\
\mathrmbf{Set}\xleftarrow{\;\mathrmbfit{key}_{\mathcal{A}\;}}
\!\!\!\!\!\!
&
{\mathrmbf{Tbl}(\mathcal{A})}
& 
\!\!\!\!\!
\xrightarrow{\mathrmbfit{sign}_{\mathcal{A}}}\mathrmbf{List}(X)^{\mathrm{op}}
\end{array}$}}
\end{tabular}}
\end{center}
are (co)continuous.
%
\comment{\footnote{The coproduct of two signatures ${\langle{I_{1},s_{1}}\rangle}$ and ${\langle{I_{2},s_{2}}\rangle}$
is the signature 
${\langle{I_{1},s_{1}}\rangle} + {\langle{I_{2},s_{2}}\rangle} = {\langle{I_{1} + I_{2},[s_{1},s_{2}]}\rangle}$
with injections
${\langle{I_{1},s_{1}}\rangle}\xrightarrow{\iota_{1}}
{\langle{I_{1}+I_{2},[s_{1},s_{2}]}\rangle}
\xleftarrow{\iota_{2}}{\langle{I_{2},s_{2}}\rangle}$.
Given any opspan of signature morphisms
${\langle{I_{1},s_{1}}\rangle} \xrightarrow{h_{1}} {\langle{I,s}\rangle} \xleftarrow{h_{2}} {\langle{I_{2},s_{2}}\rangle}$,
the pairing function ${\langle{I_{1} + I_{2},[s_{1},s_{2}]}\rangle}\xrightarrow{[h_{1},h_{2}]} {\langle{I,s}\rangle}$
is the mediating morphism to the coproduct.
Any $X$-signature ${\langle{I,s}\rangle} \in \mathrmbf{List}(X)$
resolves into the coproduct ${\langle{I,s}\rangle} = \coprod_{i\in{I}} {\langle{1,s_{i}}\rangle}$;
entity types (sorts) form the subcontext
$X \hookrightarrow \mathrmbf{List}(X) : x \mapsto {\langle{1,x}\rangle}$.}}
%
\end{proposition}
\begin{proof}
The contexts $\mathrmbf{Set}$, $\mathrmbf{Cls}$, $\mathrmbf{1}$ and $\mathrmbf{List}(X)^{\mathrm{op}}$ are (co)complete; 
the passage $\mathrmbf{Set}\xrightarrow{\mathrmbfit{1}}\mathrmbf{Set}$ is (co)cocontinuous; and
the passages 
$\mathrmbf{Set}\xrightarrow{\mathrmbfit{1}}\mathrmbf{Set}$,
$\mathrmbf{Cls}\xrightarrow{\mathrmbfit{sort}}\mathrmbf{Set}$,
$\mathrmbf{1}\xrightarrow{X,\;\mathrmbfit{tup}_{\mathcal{A}}(I,s)}\mathrmbf{Set}$, and 
$\mathrmbf{List}(X)^{\mathrm{op}}\xrightarrow{\mathrmbfit{tup}_{\mathcal{A}}}\mathrmbf{Set}$
are continuous.
\mbox{}\hfill\rule{5pt}{5pt}
\end{proof}

\newpage

\begin{proposition}\label{prop:com:cxt:colim}
The comma contexts of 
tables and $\mathcal{S}$-tables 
are associated with the following passage opspans:
\begin{center}
{{\footnotesize\setlength{\extrarowheight}{4pt}$\begin{array}{r@{\;=\;}l@{\hspace{6pt}:\hspace{25pt}}l}
\multicolumn{2}{c}{\textit{comma context}} &
\textit{passage opspan}
\\
\mathrmbf{Tbl}
& \bigl(\mathrmbf{Set}{\,\downarrow\,}\mathrmbfit{tup}\bigr)
& \mathrmbf{Set}\xrightarrow{\;\mathrmbfit{1}\;}\mathrmbf{Set}\xleftarrow{\;\mathrmbfit{tup}\;}\mathrmbf{Dom}^{\mathrm{op}}
\\
\mathrmbf{Tbl}(\mathcal{S})
& \bigl(\mathrmbf{Set}{\,\downarrow\,}\mathrmbfit{tup}_{\mathcal{S}}\bigr)
& \mathrmbf{Set}\xrightarrow{\;\mathrmbfit{1}\;}\mathrmbf{Set}\xleftarrow{\;\mathrmbfit{tup}_{\mathcal{S}}\;}\mathrmbf{Cls}(X)^{\mathrm{op}}
\text{.}
\end{array}$}}
\end{center}
respectively.
Hence, they are cocomplete and their projections
\begin{center}
{\begin{tabular}{@{\hspace{20pt}}c}
{{\footnotesize\setlength{\extrarowheight}{4pt}$\begin{array}[t]{@{\hspace{5pt}}r@{\hspace{5pt}}c@{\hspace{5pt}}l@{\hspace{5pt}}}
\mathrmbf{Set}\xleftarrow{\;\mathrmbfit{key}\;}
\!\!\!\!\!\!
&
{\mathrmbf{Tbl}}
& 
\!\!\!\!\!
\xrightarrow{\mathrmbfit{dom}}\mathrmbf{Dom}^{\mathrm{op}}
\\
\mathrmbf{Set}\xleftarrow{\;\mathrmbfit{key}_{\mathcal{S}\;}}
\!\!\!
&
{\mathrmbf{Tbl}(\mathcal{S})}
& 
\!\!
\xrightarrow{\mathrmbfit{data}_{\mathcal{S}}}\mathrmbf{Cls}(X)^{\mathrm{op}}
\end{array}$}}
\end{tabular}}
\end{center}
are cocontinuous.
\end{proposition}
\begin{proof}
The contexts 
$\mathrmbf{Set}$, $\mathrmbf{Dom}^{\mathrm{op}}$ and $\mathrmbf{Cls}(X)^{\mathrm{op}}$ are cocomplete
(Prop.~\ref{prop:clsX:compl}); 
and the passage $\mathrmbf{Set}\xrightarrow{\mathrmbfit{1}}\mathrmbf{Set}$ is cococontinuous.
\mbox{}\hfill\rule{5pt}{5pt}
\end{proof}
%

\subsubsection{Using the Grothendieck Construction.}\label{sub:sub:sec:props:groth}


%
\begin{proposition}\label{prop:grothen:lim:colim}
The fibered contexts (Tbl.~\ref{tbl:grothen:construct}) of 
signatures $\mathrmbf{List}$, 
tables $\mathrmbf{Tbl}$ and 
$\mathcal{A}$-table $\mathrmbf{Tbl}(\mathcal{A})$
are (co)complete and their projections are (co)continuous.
\begin{center}
{{\begin{tabular}{c}
{\scriptsize\setlength{\extrarowheight}{3pt}
$\begin{array}{|@{\hspace{5pt}}l@{\hspace{3pt}}l@{\hspace{3pt}}|@{\hspace{5pt}}l@{\hspace{15pt}=\int\hspace{15pt}}l@{\hspace{5pt}}|}
\multicolumn{2}{l}{}
&
\multicolumn{1}{l}{\textsf{fibered construct}}
&
\multicolumn{1}{l}{\textsf{indexed construct}}
\\\hline
\S\ref{sub:sec:fole:comps:sign}
&
\text{(Thm.\ref{thm:fib:cxt:sign:set})}
&
\rule[8pt]{0pt}{5pt}
\mathrmbf{List}\xrightarrow{\;\mathrmbfit{sort}\;}\mathrmbf{Set}
&
\mathrmbf{Set}\xrightarrow{\;\mathrmbfit{list}\;}\mathrmbf{Adj}
\\
\S\ref{sub:sec:tbl:sign:dom}
&
\text{(Thm.\ref{thm:fib:cxt:tbl:dom})}
&
\rule[8pt]{0pt}{5pt}
\mathrmbf{Tbl}\xrightarrow{\;\mathrmbfit{dom}}\mathrmbf{Dom}^{\mathrm{op}}
&
\mathrmbf{Dom}^{\mathrm{op}}\!\xrightarrow{\;\mathrmbfit{tbl}\;}\mathrmbf{Adj}
\\
\S\ref{sub:sub:sec:tbl:typ:dom:low}
&
\text{(Thm.\ref{thm:fib:cxt:tbl:A:list:X})}
&
\rule[10pt]{0pt}{5pt}
\mathrmbf{Tbl}(\mathcal{A})\xrightarrow{\;\mathrmbfit{sign}_{\mathcal{A}}\;}\mathrmbf{List}(X)^{\mathrm{op}}
&
\mathrmbf{List}(X)^{\mathrm{op}}\!\xrightarrow{\;\mathrmbfit{tbl}_{\mathcal{A}}\;}\mathrmbf{Adj}
\\
\S\ref{sub:sub:sec:tbl:typ:dom:up}
&
\text{(Thm.\ref{thm:fib:cxt:tbl:cls})}
&
\mathrmbf{Tbl}\xrightarrow{\;\mathrmbfit{data}\;}\mathrmbf{Cls}^{\mathrm{op}}
&
\mathrmbf{Cls}^{\mathrm{op}}\!\xrightarrow{\;\mathrmbfit{tbl}\;}\mathrmbf{Adj}
\\\hline
\end{array}$}
\end{tabular}}}
\end{center}
\end{proposition}
\begin{proof}
This proposition uses Fact.~\ref{fact:groth:adj:lim:colim} of \S\ref{sub:sec:ind:fbr}.
The indexing contexts $\mathrmbf{Set}$, $\mathrmbf{Dom}^{\mathrm{op}}$, $\mathrmbf{List}(X)^{\mathrm{op}}$, $\mathrmbf{Cls}^{\mathrm{op}}$
are (co)complete (Prop.\ref{prop:com:cxt:lim:colim}),
and
the fiber contexts $\mathrmbf{List}(X)$, $\mathrmbf{Tbl}(I,s,\mathcal{A})$ and $\mathrmbf{Tbl}(\mathcal{A})$
are (co)complete (Prop.\ref{prop:com:cxt:lim:colim}).
%
\mbox{}\hfill\rule{5pt}{5pt}
\end{proof}

%

\comment{
\begin{proposition}
The fibered context $\mathrmbf{Cls}$ is complete and 
the projection $\mathrmbf{Cls}\xrightarrow{\mathrmbfit{sort}}\mathrmbf{Set}$ is continuous.
\end{proposition}
\begin{proof}
Uses
Thm.~\ref{thm:groth:lim} in \S\ref{sub:sec:ind:fbr}
and
Prop.\ref{prop:fib:cxt:cls:set}.
%
\begin{enumerate}
\item
The indexing context $\mathrmbf{Set}$ is complete,
\item
the fiber context $\mathrmbf{Cls}(X)$
is complete for each sort set $X$ (Prop.\ref{prop:clsX:compl}), and
\item
the type domain fiber passage
$\mathrmbf{Cls}(X_{2})
\xleftarrow[{f}^{-1}]{\mathrmbfit{cls}(f)}
\mathrmbf{Cls}(X_{1})$
is continuous
for each sort function $X_{2}\xrightarrow{\,f\,}X_{1}$
(Prop.\ref{prop:fbr:pass:cts}).
\end{enumerate}
\end{proof}
}

\begin{proposition}\label{prop:groth:sign:lwr:colim}
[signature-lower]
The fibered context of $\mathcal{S}$-tables $\mathrmbf{Tbl}(\mathcal{S})$ is cocomplete and 
the projection $\mathrmbf{Tbl}(\mathcal{S})\xrightarrow{\mathrmbfit{data}_{\mathcal{S}}}\mathrmbf{Cls}(X)^{\mathrm{op}}$ is cocontinuous.
\end{proposition}
\begin{proof}
Uses
Fact.~\ref{fact:groth:colim} in \S\ref{sub:sec:ind:fbr}
and
the discussion in \S~\ref{sub:sub:sec:tbl:sign:low}.
For the contravariant pseudo-passage
$\mathrmbf{Cls}(X)^{\mathrm{op}}\!\xrightarrow{\acute{\mathrmbfit{tbl}}_{\mathcal{S}}}\mathrmbf{Cxt}$
that uses the existential quantification, 
\begin{enumerate}
\item
the indexing context $\mathrmbf{Cls}(X)^{\mathrm{op}}$ is cocomplete,
\item
the fiber context $\mathrmbf{Tbl}_{\mathcal{S}}(\mathcal{A})$ is cocomplete
for each $X$-sorted type domain $\mathcal{A}$,
and
\item
the fiber passage
$\mathrmbf{Tbl}_{\mathcal{S}}(\mathcal{A})
\xleftarrow[{\scriptscriptstyle\sum}_{g}]{\acute{\mathrmbfit{tbl}}_{\mathcal{S}}(g)}
\mathrmbf{Tbl}_{\mathcal{S}}(\widetilde{\mathcal{A}})$ is cocontinuous
(being left adjoint)
for each $X$-sorted type domain morphism 
$\mathcal{A}\xrightleftharpoons{{\langle{\mathrmit{1}_{X},g}\rangle}}\widetilde{\mathcal{A}}$.
\mbox{}\hfill\rule{5pt}{5pt}
\end{enumerate}
\end{proof}
%

\subsection{Constructive Proof}\label{sub:sec:cons:pf}

%
\begin{proposition}\label{prop:construct}
The context of $\mathcal{A}$-tables $\mathrmbf{Tbl}(\mathcal{A})$ is complete.
\end{proposition}
\begin{proof}
We have already proved this using comma contexts and the Grothendieck construction.
Now we give a constructive proof of this fact,
which illustrates that 
\newline
\mbox{}\hfill{{``limits (natural joins) are resolvable into substitutions followed by meets.''}}\hfill\mbox{}
\newline
Suppose that $\mathrmbfit{T} : G \rightarrow \mathrmbf{Tbl}(\mathcal{A})$ 
is a diagram of $\mathcal{A}$-tables and $\mathcal{A}$-table morphisms
{\footnotesize\[
\mathrmbfit{T} = \{ 
\mathrmbfit{T}_{n} = {\langle{I_{n},s_{n},K_{n},t_{n}}\rangle} 
\xrightarrow{\langle{h_{e},k_{e}}\rangle} 
{\langle{I_{m},s_{m},K_{m},t_{m}}\rangle} = \mathrmbfit{T}_{m} 
\mid (n \xrightarrow{e} m) \in G 
\}. 
\]\normalsize}
\begin{center}
{\footnotesize{$\begin{array}{c@{\hspace{5pt}\stackrel{\cong}{\leftrightarrow}\hspace{5pt}}c}
\multicolumn{2}{c}{{\langle{I_{n},s_{n},K_{n},t_{n}}\rangle} \xrightarrow{\langle{h_{e},k_{e}}\rangle} 
{\langle{I_{m},s_{m},K_{m},t_{m}}\rangle}}
\\
{\scriptstyle\sum}_{h_{e}}(K_{n},t_{n}) \xrightarrow{k_{e}} {\langle{K_{m},t_{m}}\rangle}
& {\langle{K_{n},t_{n}}\rangle} \xrightarrow{\hat{k}_{e}} h_{e}^{\ast}(K_{m},t_{m})
\\
\multicolumn{2}{c}{{\scriptstyle\sum}_{h_{e}}(\hat{k}_{e}) \cdot \varepsilon^{h_{e}}_{\mathcal{T}_{m}}
= \hat{k}_{e} \cdot \varepsilon^{h_{e}}_{\mathcal{T}_{m}} = k_{e}}
\end{array}$}}
\end{center}
Let 
$\mathrmbfit{S} = \mathrmbfit{T}^{\mathrm{op}} \circ \mathrmbfit{sign}_{\mathcal{A}} 
: G^{\mathrm{op}} \rightarrow {\mathrmbf{List}(X)}$
be the underlying diagram of signatures and signature morphisms
{\footnotesize\[
\mathrmbfit{S} = \{ 
\mathrmbfit{S}_{n} = {\langle{I_{n},s_{n}}\rangle} 
\xleftarrow{h_{e}} 
{\langle{I_{m},s_{m}}\rangle} = \mathrmbfit{S}_{m} 
\mid (n \xrightarrow{e} m) \in G 
\}. 
\]\normalsize}
\hspace{-9pt}
Assume that $\gamma : \mathrmbfit{S} \Rightarrow \Delta{\langle{I,s}\rangle}$ is a colimiting cocone 
$\gamma = \{ \gamma_{n} : {\langle{I_{n},s_{n}}\rangle} \rightarrow {\langle{I,s}\rangle} \mid n \in G \}$
with base diagram $\mathrmbfit{S}$ and colimit signature ${\langle{I,s}\rangle}$,
so that $h_{e} \cdot \gamma_{n} =  \gamma_{m}$
for all edges $n \xrightarrow{e} m$ in $G$. 
For each $G$-node $n$,
use substitution to move
fiber tables and fiber table morphisms 
from the peripheral fiber categories $\{ \mathrmbf{Tbl}_{\mathcal{A}}(I_{n},s_{n}) \}$
to the central fiber context $\mathrmbf{Tbl}_{\mathcal{A}}(I,s)$:
\begin{center}
{\scriptsize\setlength{\extrarowheight}{2pt}
{$\begin{array}{|c@{\hspace{5pt}{\mapsto}\hspace{5pt}}c|}
\multicolumn{1}{c}{\rule[-6pt]{0pt}{10pt}\mbox{\normalsize{\sf{peripheral}}}} & 
\multicolumn{1}{c}{\mbox{\normalsize{\sf{central}}}}
\\ \hline
{\langle{K_{n},t_{n}}\rangle}
& \gamma_{n}^{\ast}(K_{n},t_{n}) = {\langle{K_{n}^{\ast},t_{n}^{\ast}}\rangle}
\\
{\langle{K_{n},t_{n}}\rangle} \xrightarrow{\hat{k}_{e}} h_{e}^{\ast}(K_{m},t_{m})
& \gamma_{n}^{\ast}(K_{n},t_{n}) \xrightarrow{\gamma_{n}^{\ast}(\hat{k}_{e})} \gamma_{n}^{\ast}(h_{e}^{\ast}(K_{m},t_{m})) 
= \gamma_{m}^{\ast}(K_{m},t_{m}) 
\\
{\scriptstyle\sum}_{h_{e}}(K_{n},t_{n}) \xrightarrow{k_{e}} {\langle{K_{m},t_{m}}\rangle}
& \gamma_{m}^{\ast}({\scriptstyle\sum}_{h_{e}}(K_{n},t_{n})) \xrightarrow{\gamma_{m}^{\ast}(k_{e})} \gamma_{m}^{\ast}(K_{m},t_{m})
\\
&
\stackrel{\cong}{\leftrightarrow}\;
{\scriptstyle\sum}_{\gamma_{m}}(\gamma_{m}^{\ast}({\scriptstyle\sum}_{h_{e}}(K_{n},t_{n}))) 
\xrightarrow{\widehat{\gamma_{m}^{\ast}(k_{e})}} 
{\langle{K_{m},t_{m}}\rangle}
\\ \hline
\end{array}$}}
\end{center}
Hence,
there is diagram $\mathrmbfit{T}^{\ast} : G \rightarrow \mathrmbf{Tbl}_{\mathcal{A}}(I,s)$
in the central fiber
{\footnotesize\[
\mathrmbfit{T}^{\ast} = \{
{\langle{K_{n}^{\ast},t_{n}^{\ast}}\rangle} = \gamma_{n}^{\ast}(K_{n},t_{n}) 
\xrightarrow{\gamma_{n}^{\ast}(\hat{k}_{e})} 
\gamma_{n}^{\ast}(h_{e}^{\ast}(K_{m},t_{m})) \cong \gamma_{m}^{\ast}(K_{m},t_{m}) = {\langle{K_{m}^{\ast},t_{m}^{\ast}}\rangle}
\mid (n \xrightarrow{e} m) \in G \}. 
\]\normalsize}
\hspace{-9pt}
%
Assume that $\pi : \mathrmbfit{T}^{\ast} \Rightarrow \Delta{\langle{K,t}\rangle}$ is a limiting cone 
$\pi = \{ {\langle{K,t}\rangle} \xrightarrow{\pi_{n}} {\langle{K_{n}^{\ast},t_{n}^{\ast}}\rangle} \mid n \in G \}$
with base diagram $\mathrmbfit{T}^{\ast}$ and 
join table ${\langle{K,t}\rangle} = \prod_{n \in G} {\langle{K_{n}^{\ast},t_{n}^{\ast}}\rangle}$
with fiber projections,
so that $\pi_{n} \cdot \gamma_{n}^{\ast}(\hat{k}_{e}) =  \pi_{m}$
for all edges $n \xrightarrow{e} m$ in $G$. 
We claim that the composite $\mathcal{A}$-table morphism
$\mathcal{T} = {\langle{I,s,K,t}\rangle} \xrightarrow{\langle{\gamma_{n},\pi_{n}{\cdot}\varepsilon^{\gamma_{n}}_{\mathrmbfit{T}_{n}}}\rangle} 
{\langle{I_{n},s_{n},K_{n},t_{n}}\rangle}$
is the $n^{\mathrm{th}}$ component of a limiting cone 
$\widehat{\gamma} : {\langle{I,s,K,t}\rangle} \Rightarrow \mathrmbfit{T}$ 
for $\mathrmbfit{T}$ in $\mathrmbf{Tbl}(\mathcal{A})$,
where each component 
has signature morphism 
${\langle{I_{n},s_{n}}\rangle} \xrightarrow{\gamma_{n}} {\langle{I,s}\rangle}$ 
and key function 
$\pi_{n}{\cdot}\varepsilon^{\gamma_{n}}_{\mathrmbfit{T}_{n}} : K \rightarrow K_{n}$. 
It is natural with respect to the diagram $\mathrmbfit{T}$.
Now suppose that 
$\alpha : {\langle{I',s',K',t'}\rangle} \Rightarrow \mathrmbfit{T}$ 
is another cone
$\alpha = \{ 
{\langle{I',s',K',t'}\rangle} \xrightarrow{\langle{h_{n},k_{n}}\rangle} {\langle{I_{n},s_{n},K_{n},t_{n}}\rangle} 
\mid n \in G \}$ over $\mathrmbfit{T}$,
each component
with signature morphism ${\langle{I_{n},s_{n}}\rangle} \xrightarrow{h_{n}} {\langle{I',s'}\rangle}$ 
and key function $k_{n} : K' \rightarrow K_{n}$
satisfying the condition $t' \cdot \mathrmbfit{tup}_{\mathcal{A}}(h_{n}) = k_{n} \cdot t_{n}$.
Since $\gamma$ is a colimiting cocone,
there is a unique signature morphism 
${\langle{I,s}\rangle} \xrightarrow{h} {\langle{I',s'}\rangle}$ 
such that $\alpha = \gamma \bullet \Delta{h}$,
or $\alpha_{n} = \gamma_{n} \cdot h$,
and hence,
$\mathrmbfit{tup}_{\mathcal{A}}(\alpha_{n}) =
\mathrmbfit{tup}_{\mathcal{A}}(h) \cdot \mathrmbfit{tup}_{\mathcal{A}}(\gamma_{n})$,
for each node $n \in G$.
Since
$k_{n} \cdot t_{n}
= t' \cdot \mathrmbfit{tup}_{\mathcal{A}}(h_{n}) 
= t' \cdot \mathrmbfit{tup}_{\mathcal{A}}(h) \cdot \mathrmbfit{tup}_{\mathcal{A}}(\gamma_{n})$,
there is a unique mediating key function $K' \xrightarrow{k^{\ast}_{n}} K_{n}^{\ast}$
satisfying
$k^{\ast}_{n} \cdot t_{n}^{\ast} = t' \cdot \mathrmbfit{tup}_{\mathcal{A}}(h)$ and
$k^{\ast}_{n} \cdot \varepsilon^{\gamma_{n}}_{\mathrmbfit{T}_{n}} = k_{n}$.
Hence,
we have the $\mathcal{A}$-table morphism
$\mathcal{T}' = {\langle{I',s',K',t'}\rangle} \xrightarrow{\langle{h,k^{\ast}_{n}}\rangle} {\langle{I,s,K^{\ast}_{n},t^{\ast}_{n}}\rangle} = \gamma^{\ast}_{n}(\mathrmbfit{T}_{n})$,
which satisfies
${\langle{h,k^{\ast}_{n}}\rangle} \cdot {\langle{\gamma_{n},\varepsilon^{\gamma_{n}}_{\mathrmbfit{T}_{n}}}\rangle}
= {\langle{h_{n},k_{n}}\rangle}$
for each $n \in G$.
The central fiber table morphism
${\scriptstyle\sum}_{h}(K',t')
 \xrightarrow{k^{\ast}_{n}} {\langle{K^{\ast}_{n},t^{\ast}_{n}}\rangle} = \gamma^{\ast}_{n}(\mathrmbfit{T}_{n})$,
is the $n^{\mathrm{th}}$ component of a central fiber cone 
$\alpha^{\ast} : {\langle{K',t'}\rangle} \Rightarrow \Delta\mathrmbfit{T}^{\ast}$.
Hence,
there is a unique mediating function 
$K' \xrightarrow{k} K$
such that
$\Delta{k} \bullet \pi = \alpha^{\ast}$,
or
$k \bullet \pi_{n} = k^{\ast}_{n}$
for each $n \in G$.
Hence,
we have the commuting diagram of $\mathcal{A}$-table morphisms
$\mathcal{T}' \xrightarrow{\langle{h,k^{\ast}_{n}}\rangle} \mathcal{T}
\xrightarrow{\langle{\gamma_{n},\pi_{n}{\cdot}\varepsilon^{\gamma_{n}}_{\mathrmbfit{T}_{n}}}\rangle} \mathrmbfit{T}_{n}
= \mathcal{T}' \xrightarrow{\langle{h_{n},k_{n}}\rangle} \mathrmbfit{T}_{n}$.
Uniqueness is straightforward.
\mbox{}\hfill\rule{5pt}{5pt}
\end{proof}
%

\subsection{Example}\label{sub:sec:eg}

We illustrate the use of these semantic operations
by using the observation made in Prop.~\ref{prop:construct}
that limits are resolvable into substitutions followed by meets.
Here we discuss the special case of pullback --- the join of two $\mathcal{A}$-tables.
Consider the $\mathrmbf{Tbl}(\mathcal{A})$-opspan 
%
\footnotesize
\begin{align}\label{tbl:opspan}
\mathcal{T}_{1}={\langle{I_{1},s_{1},K_{1},t_{1}}\rangle} \xrightarrow{\langle{h_{1},k_{1}}\rangle} 
\overset{\textstyle\mathcal{T}}{\overbrace{\langle{I,s,K,t}\rangle}}
\xleftarrow{\langle{h_{2},k_{2}}\rangle} {\langle{I_{2},s_{2},K_{2},t_{2}}\rangle}=\mathcal{T}_{2}
\end{align}
\normalsize
%
%
illustrated in the bottom part of Figure~\ref{binary:join} 
with key opspan
$K_{1} \xrightarrow{k_{1}} K \xleftarrow{k_{2}} K_{2}$
and signature span 
${\langle{I_{1},s_{1}}\rangle} \xleftarrow{h_{1}} {\langle{I,s}\rangle} \xrightarrow{h_{2}} {\langle{I_{2},s_{2}}\rangle}$.
Since $\mathrmbf{List}(X)$ is cocomplete,
we can form the colimiting cocone (opspan) of this signature span, 
with pushout signature
${\langle{I_{1}{+}_{I}I_{2},[s_{1},s_{2}]}\rangle}$
and injection signature morphisms
\footnotesize\[
{\langle{I_{1},s_{1}}\rangle} \xrightarrow{\iota_{1}} {\langle{I_{1}{+}_{I}I_{2},[s_{1},s_{2}]}\rangle} \xleftarrow{\iota_{2}} {\langle{I_{2},s_{2}}\rangle}
\]\normalsize
that satisfies the commutative diagram
$h_{1} \cdot \iota_{1} = h_{2} \cdot \iota_{2}$.
Apply the continuous tuple passage 
$\mathrmbfit{tup}_{\mathcal{A}} : {\mathrmbf{List}(X)}^{\mathrm{op}} \longrightarrow \mathrmbf{Set}$
to this signature opspan
to get the limiting cone 
over the $\mathrmbf{Set}$-opspan
$\mathrmbfit{tup}_{\mathcal{A}}(I_{1},s_{1})
\xrightarrow{\mathrmbfit{tup}_{\mathcal{A}}(h_{1})}
\mathrmbfit{tup}_{\mathcal{A}}(I,s)
\xleftarrow{\mathrmbfit{tup}_{\mathcal{A}}(h_{2})}
\mathrmbfit{tup}_{\mathcal{A}}(I_{2},s_{2})$
with pullback set
$\mathrmbfit{tup}_{\mathcal{A}}(I_{1},s_{1}){\times}_{\mathrmbfit{tup}_{\mathcal{A}}(I,s)}\mathrmbfit{tup}_{\mathcal{A}}(I_{2},s_{2})$
and projection functions
\footnotesize\[
\mathrmbfit{tup}_{\mathcal{A}}(I_{1},s_{1})
\xleftarrow{\mathrmbfit{tup}_{\mathcal{A}}(\iota_{1})}
\mathrmbfit{tup}_{\mathcal{A}}(I_{1}{+}_{I}I_{2},[s_{1},s_{2}])
\xrightarrow{\mathrmbfit{tup}_{\mathcal{A}}(\iota_{2})}
\mathrmbfit{tup}_{\mathcal{A}}(I_{2},s_{2})
\]\normalsize
This is illustrated in the top part of Figure~\ref{binary:join}.
%

In general,
the join (limit) of an arbitrary diagram in $\mathrmbf{Tbl}(\mathcal{A})$
is obtained by 
(1) inverse image (substitution) of the component tables along the colimit signature injections over the underlying signature diagram,
followed by
(2) meet (conjunction) at the colimit signature.
In particular,
the pullback of $\mathrmbf{Tbl}(\mathcal{A})$-opspan (\ref{tbl:opspan}) is the table
$\mathcal{T}_{1}{\times}_{\mathcal{T}}\mathcal{T}_{2}$
whose signature is the pushout signature ${\langle{I_{1}{+}_{I}I_{2},[s_{1},s_{2}]}\rangle}$,
whose key set is the pullback set $K_{1}{\times}_{K}K_{2}$, and
whose tuple function
\footnotesize\[
t_{1}{\times}_{t}t_{2} : 
K_{1}{\times}_{K}K_{2} \rightarrow \mathrmbfit{tup}_{\mathcal{A}}(I_{1}{+}_{I}I_{2},[s_{1},s_{2}]) = \mathrmbfit{tup}_{\mathcal{A}}(I_{1},s_{1}){\times}_{\mathrmbfit{tup}_{\mathcal{A}}(I,s)}\mathrmbfit{tup}_{\mathcal{A}}(I_{2},s_{2})
\]\normalsize
is the mediating function
obtained by taking the pullback of sources and targets in (\ref{tbl:opspan}).
For proof,
use a continuity proposition for comma categories,
and show that the key set and projection functions,
obtained by inverse image (substitution) and meet,
forms the pullback.
%
\footnote{Since we identify database joins with limits in $\mathrmbf{Tbl}(\mathcal{A})$,
this allows us to compute joins as inverse images followed by meets,
both of which are elementary logical operations.
The dual approach will identify database unions with colimits in $\mathrmbf{Tbl}(\mathcal{A})$.
This is the key insight for a structured/logical approach to database formalism
using fiber Boolean operations (conjunction and disjunction), substitution and the quantifiers.}

%
\begin{figure}
\begin{center}
\begin{tabular}{c}
\\ \\
\setlength{\unitlength}{0.5pt}
\begin{picture}(320,320)(0,0)
\put(340,60){{\scriptsize$\begin{array}{l}
\pi_{1} = \mathrmbfit{tup}_{\mathcal{A}}(\iota_{1}) \\
\pi_{2} = \mathrmbfit{tup}_{\mathcal{A}}(\iota_{2}) \\
\widehat{K}_{1} = \pi_{1}^{-1}(K_{1}) \\
\widehat{K}_{2} = \pi_{2}^{-1}(K_{2}) \\
K_{1}{\times}_{K}K_{2} = \pi_{1}^{-1}(K_{1}){\wedge}\pi_{2}^{-1}(K_{2}) \\
\widehat{t}_{1} = \pi_{1}^{-1}(t_{1}) \\
\widehat{t}_{2} = \pi_{2}^{-1}(t_{2}) \\
t_{1}{\times}_{t}t_{2} = \pi_{1}^{-1}(t_{1}){\wedge}\pi_{2}^{-1}(t_{2})
\end{array}$}}
\put(0,160){\begin{picture}(0,0)(0,0)
\put(5,0){\makebox(0,0)[r]{\footnotesize{$\mathcal{T}_{1}\;\left\{\rule{0pt}{10pt}\;\;\mathrmbfit{tup}_{\mathcal{A}}(I_{1},s_{1})\right.$}}}
\put(315,0){\makebox(0,0)[l]{\footnotesize{$\left.\mathrmbfit{tup}_{\mathcal{A}}(I_{2},s_{2})\;\;\rule{0pt}{10pt}\right\}\;\mathcal{T}_{2}$}}}
\put(80,0){\makebox(0,0){\footnotesize{$K_{1}$}}}
\put(160,-180){\makebox(0,0){\footnotesize{$
\overset{\underbrace{\;\;\rule[-6pt]{0pt}{10pt}\textstyle\mathrmbfit{tup}_{\mathcal{A}}(I,s)\;\;}}
{\mathcal{T}}$}}}
\put(160,-80){\makebox(0,0){\footnotesize{$K$}}}
\put(240,0){\makebox(0,0){\footnotesize{$K_{2}$}}}
\put(160,182){\makebox(0,0){\footnotesize{$
\overset{\textstyle\;\;\mathcal{T}_{1}{\times}_{\mathcal{T}}\mathcal{T}_{2}}
{\overbrace{\;\;\rule[0pt]{0pt}{10pt}\mathrmbfit{tup}_{\mathcal{A}}(I_{1}{+}_{I}I_{2},[s_{1},s_{2}])\;\;}}$}}}
\put(212,202){\makebox(0,0)[l]{\footnotesize{$=\pi_{1}^{\ast}(\mathcal{T}_{1})
\;{\wedge}_{{\langle{I_{1}{+}_{I}I_{2},[s_{1},s_{2}]}\rangle}}\;
\pi_{2}^{\ast}(\mathcal{T}_{2})$}}}
\put(254,160){\makebox(0,0)[l]{\footnotesize{$=\mathrmbfit{tup}_{\mathcal{A}}(I_{1},s_{1}){\times}_{\mathrmbfit{tup}_{\mathcal{A}}(I,s)}\mathrmbfit{tup}_{\mathcal{A}}(I_{2},s_{2})$}}}
\put(136,86){\makebox(0,0){\scriptsize{$\widehat{K}_{1}$}}}
\put(184,86){\makebox(0,0){\scriptsize{$\widehat{K}_{2}$}}}
\put(160,56){\makebox(0,0){\scriptsize{$K_{1}{\times}_{K}K_{2}$}}}
\put(160,62){\vector(0,1){82}}
\put(76,-84){\makebox(0,0)[r]{\scriptsize{$\mathrmbfit{tup}_{\mathcal{A}}(h_{1})$}}}
\put(244,-84){\makebox(0,0)[l]{\scriptsize{$\mathrmbfit{tup}_{\mathcal{A}}(h_{2})$}}}
\put(116,-44){\makebox(0,0)[r]{\scriptsize{$k_{1}$}}}
\put(204,-44){\makebox(0,0)[l]{\scriptsize{$k_{2}$}}}
\put(76,84){\makebox(0,0)[r]{\scriptsize{$\pi_{1}$}}}
\put(244,84){\makebox(0,0)[l]{\scriptsize{$\pi_{2}$}}}
\put(112,52){\makebox(0,0)[r]{\scriptsize{$\widehat{\pi}_{1}$}}}
\put(208,52){\makebox(0,0)[l]{\scriptsize{$\widehat{\pi}_{2}$}}}
\put(168,-116){\makebox(0,0)[l]{\scriptsize{$t$}}}
\put(50,-10){\makebox(0,0){\scriptsize{$t_{1}$}}}
\put(280,-10){\makebox(0,0){\scriptsize{$t_{2}$}}}
\put(144,120){\makebox(0,0)[r]{\scriptsize{$\widehat{t}_{1}$}}}
\put(176,120){\makebox(0,0)[l]{\scriptsize{$\widehat{t}_{2}$}}}
\put(160,100){\makebox(0,0){\scriptsize{$t_{1}{\times}_{t}t_{2}$}}}
\put(138,98){\vector(1,3){15}}
\put(130,74){\vector(-2,-3){40}}
\put(182,98){\vector(-1,3){15}}
\put(190,74){\vector(2,-3){40}}
\put(145,48){\vector(-4,-3){53}}
\put(175,48){\vector(4,-3){53}}
\qbezier(156,60)(148,70)(140,80)\put(140,80){\vector(-3,4){0}}
\qbezier(164,60)(172,70)(180,80)\put(180,80){\vector(3,4){0}}
\put(12,-12){\vector(1,-1){136}}
\put(308,-12){\vector(-1,-1){136}}
\put(148,148){\vector(-1,-1){136}}
\put(172,148){\vector(1,-1){136}}
\put(92,-12){\vector(1,-1){56}}
\put(228,-12){\vector(-1,-1){56}}
\put(68,0){\vector(-1,0){56}}
\put(160,-92){\vector(0,-1){56}}
\put(252,0){\vector(1,0){56}}
\end{picture}}
\end{picture}
\\ \\
\end{tabular}
\end{center}
\caption{Binary Join}
\label{binary:join}
\end{figure}
%

\newpage
\section{Conclusion and Future Work}\label{sec:conclu}

\subsection{This Paper in Review}\label{sub:sec:paper:review}

%
\begin{table}
\begin{center}
{\fbox{\scriptsize\setlength{\extrarowheight}{3pt}{\begin{tabular}{|@{\hspace{3pt}}r@{\hspace{3pt}:\hspace{6pt}}l|}
\multicolumn{2}{l}{\S\ref{sec:tbl:basics}
: \textbf{Table Basics}}
\\\hline
Thm.~\ref{thm:fib:cxt:sign:set}
&
$\mathrmbf{List}\xrightarrow{\mathrmbfit{sort}}\mathrmbf{Set}
=\check{\int}\bigl(\mathrmbf{Set}\xrightarrow{\;\mathrmbfit{list}\;}\mathrmbf{Adj}\bigr)$
\\
Thm.~\ref{thm:fib:cxt:cls:set}
&
$\mathrmbf{Cls}\xrightarrow{\mathrmbfit{sort}}\mathrmbf{Set}
=\grave{\int}\bigl(
\mathrmbf{Set}^{\mathrm{op}}\xrightarrow{\;\mathrmbfit{cls}\;}\mathrmbf{Cxt}\bigr)$
\\
Lem.~\ref{lem:nat:iso}
&
natural isomorphisms (\textbf{levo} $\cong$ \textbf{dextro}) inclusion \& tuple
\\
Prop.~\ref{prop:typ:dom:inc:tup}
&
inclusion/tuple passages
$\left\{\text{
{\scriptsize\setlength{\extrarowheight}{2pt}$\begin{array}{l}
\mathrmbf{Cls}\xrightarrow{\mathrmbfit{inc}}\bigl(\mathrmbf{Adj}{\,\Uparrow\,}\mathrmbf{Dom}\bigr)
\\
\mathrmbf{Cls}^{\mathrm{op}}\!\xrightarrow{\mathrmbfit{tup}}\bigl(\mathrmbf{Adj}{\,\Uparrow\,}\mathrmbf{Set}\bigr)
\end{array}$}
}\right.$
\\
Lem.~\ref{lem:tup:fn:fact}
&
tuple function factorizations: type domain \& signature
\\\hline
\multicolumn{2}{l}{\S\ref{sec:tbl:hier}
: \textbf{Hierarchy}}
\\\hline
Thm.~\ref{thm:fib:cxt:tbl:dom}
&
$\mathrmbf{Tbl}\xrightarrow{\;\mathrmbfit{dom}\;}\mathrmbf{Dom}^{\mathrm{op}}
=\check{\int}\bigl(\mathrmbf{Dom}^{\mathrm{op}}\!\xrightarrow{\;\mathrmbfit{tbl}\;}\mathrmbf{Adj}\bigr)$

\\\cline{1-1}
Thm.~\ref{thm:fib:cxt:tbl:S:cls:X}
&
$\mathrmbf{Tbl}(\mathcal{S})\xrightarrow{\mathrmbfit{data}_{\mathcal{S}}}\mathrmbf{Cls}(X)^{\mathrm{op}}
=\check{\int}\bigl(\mathrmbf{Cls}(X)^{\mathrm{op}}\xrightarrow{\;\mathrmbfit{tbl}_{\mathcal{S}}}\mathrmbf{Adj}\bigr)$
\\
Lem.~\ref{lem:tbl:fbr:fact:sign}
&
table fiber adjunction factorization (signature) 
\\
Thm.~\ref{thm:fib:cxt:tbl:sign}
&
$\mathrmbf{Tbl}\xrightarrow{\;\mathrmbfit{sign}\;}\mathrmbf{List}^{\mathrm{op}}
=\acute{\int}\bigl(\mathrmbf{List}^{\mathrm{op}}\!\xrightarrow{\;\mathrmbfit{tbl}\;}\mathrmbf{Cxt}\bigr)$
\\\cline{1-1}
Thm.~\ref{thm:fib:cxt:tbl:A:list:X}
&
$\mathrmbf{Tbl}(\mathcal{A})\xrightarrow{\mathrmbfit{sign}_{\mathcal{A}}}\mathrmbf{List}(X)^{\mathrm{op}}
=\check{\int}\bigl(\mathrmbf{List}(X)^{\mathrm{op}}\xrightarrow{\;\mathrmbfit{tbl}_{\mathcal{A}}}\mathrmbf{Adj}\bigr)$
\\
Lem.~\ref{lem:tbl:fbr:fact:typ:dom}
&
table fiber adjunction factorization (type domain)
\\
Prop.~\ref{prop:typ:dom:tbl:fbr:adj}
&
adjunction 
{\scriptsize{$\mathrmbf{Tbl}(\mathcal{A}_{2})
\xleftarrow{{\langle{
\acute{\mathrmbfit{tbl}}_{{\langle{f,g}\rangle}}{\!\dashv\,}\grave{\mathrmbfit{tbl}}_{{\langle{f,g}\rangle}}
}\rangle}}
\mathrmbf{Tbl}(\mathcal{A}_{1})$}}
\\
Thm.~\ref{thm:fib:cxt:tbl:cls}
&
$\mathrmbf{Tbl}\xrightarrow{\mathrmbfit{data}}\mathrmbf{Cls}^{\mathrm{op}}
=\check{\int}\bigl(\mathrmbf{Cls}^{\mathrm{op}}\xrightarrow{\!\mathrmbfit{tbl}\;}\mathrmbf{Adj}\bigr)$
\\\hline
\multicolumn{2}{l}{\S\ref{sec:tbl:lim:colim}: \textbf{Table Constructions}}
\\\hline
\multicolumn{2}{|l|}{\underline{preliminaries}}
\\ Prop.~\ref{prop:cls:co:compl} & (co)completeness of $\mathrmbf{Cls}$
\\ Prop.~\ref{prop:clsX:compl} & completeness of $\mathrmbf{Cls}(X)$
\\ Prop.~\ref{prop:fbr:pass:cts} & continuity of $\mathrmbf{Cls}(X_{2})\xleftarrow{\;\mathrmbfit{cls}(f)\;}\mathrmbf{Cls}(X_{1})$
\\ Prop.~\ref{prop:tup:pass:cts} & continuity of $\mathrmbf{List}(X)^{\mathrm{op}}\!\!\xrightarrow{\;\mathrmbfit{tup}_{\mathcal{A}}\;}\mathrmbf{Set}$
\\
\multicolumn{2}{|l|}{\underline{using comma contexts}}
\\
Prop.~\ref{prop:com:cxt:lim:colim}
&
(co)completeness of $\mathrmbf{List}(X),\mathrmbf{List},\mathrmbf{Dom},\mathrmbf{Tbl}_{\mathcal{A}}(I,s)$ \& $\mathrmbf{Tbl}_{\mathcal{A}}$
\\
Prop.~\ref{prop:com:cxt:colim}
&
cocompleteness of $\mathrmbf{Tbl}$ \& $\mathrmbf{Tbl}(\mathcal{S})$
\\
\multicolumn{2}{|l|}{\underline{using Grothendieck construction}}
\\ 
Prop.~\ref{prop:grothen:lim:colim} & (co)completeness of $\mathrmbf{List}$, $\mathrmbf{Tbl}$ and $\mathrmbf{Tbl}(\mathcal{A})$
\\ 
Prop.~\ref{prop:groth:sign:lwr:colim} & cocompleteness of $\mathrmbf{Tbl}(\mathcal{S})$
\\
\multicolumn{2}{|l|}{\underline{by construction}}
\\
Prop.~\ref{prop:construct}
&
completeness of $\mathrmbf{Tbl}(\mathcal{A})$
\\\hline
\multicolumn{2}{l}{\S~\ref{sec:append}: \textbf{Appendix}}
\\\hline
\multicolumn{2}{|l|}{\underline{$\mathcal{A}$-relations}}
\\
Prop.~\ref{prop:tbl:rel:refl}
&
reflection
$\mathrmbf{Tbl}_{\mathcal{A}}
\xrightarrow{{\langle{\mathrmbfit{im}_{\mathcal{A}}{\,\dashv\,}\mathrmbfit{inc}_{\mathcal{A}}}\rangle}}
\mathrmbf{Rel}_{\mathcal{A}}$
\\
\multicolumn{2}{|l|}{\underline{comma contexts}}
\\
Fact.~\ref{fact:comma:lim}
&
comma context completeness
\\
Fact.~\ref{fact:comma:colim}
&
comma context cocompleteness
\\
\multicolumn{2}{|l|}{\underline{Grothendieck construction}}
\\
Fact.~\ref{fact:groth:lim}
&
fibration completeness
\\
Fact.~\ref{fact:groth:colim}
&
opfibration cocompleteness
\\
Fact.~\ref{fact:groth:adj:lim:colim}
&
bifibration (co)completeness
\\\hline
\end{tabular}}}}
\end{center}
\caption{Lemmas, Propositions and Theorems}
\label{tbl:thms:props:lems:cors}
\end{table}
%

A precise mathematical basis for \texttt{FOLE} interpretation consists of two notions:
relational tables and relation databases.
This paper has developed the notion of \emph{relational table} in terms of 
comma contexts and the Grothendieck construction.
Table~\ref{tbl:thms:props:lems:cors} lists the lemmas, propositions and theorems in this paper.

The table concept is built upon 
the three more elementary concepts of
signature, type domain, and signed domain.
In \S~\ref{sec:tbl:basics}, 
we have discussed the mathematical contexts
for these three elementary concepts:
Thm{.}\,\ref{thm:fib:cxt:sign:set} 
describes the fibered context of signatures $\mathrmbf{List}$ 
as a Grothendieck construction
indexed by sort sets; and
Thm.\,\ref{thm:fib:cxt:cls:set}
describes the fibered context of type domains $\mathrmbf{Cls}$
as a Grothendieck construction
also indexed by sort sets.

In \S~\ref{sec:tbl:hier},
we have described 
how each elementary concept provides a distinct, but related, approach to the fibered nature of the table concept 
via the Grothendieck construction
(illustrated in Tbl.~\ref{fig:tbl:gro:constr} of \S~\ref{sub:sec:fbr:cxt:tbl})
---
each fixed elementary concept providing a fiber subcontext of tables:
Thm.\,\ref{thm:fib:cxt:tbl:dom}
describes the fibered context of tables $\mathrmbf{Tbl}$
as a Grothendieck construction
indexed by signed domains;
Thm.\,\ref{thm:fib:cxt:tbl:sign}
describes the fibered context of tables $\mathrmbf{Tbl}$
as a Grothendieck construction
indexed by signatures,
with the indexing defined by means of
Thm.\,\ref{thm:fib:cxt:tbl:S:cls:X}; and
Thm.\,\ref{thm:fib:cxt:tbl:cls}
describes the fibered context of tables $\mathrmbf{Tbl}$
as a Grothendieck construction
indexed by type domains,
with the indexing defined by means of
Thm.\,\ref{thm:fib:cxt:tbl:A:list:X}.

In \S~\ref{sec:tbl:lim:colim},
we proved the existence of sum and product constructions (database unions and joins) 
on various fiber contexts of tables by using both comma contexts and the Grothendieck construction:
Prop.~\ref{prop:cls:co:compl}--\ref{prop:groth:sign:lwr:colim}
prove that the contexts of signatures, type domains, signed domains, and tables have limits and colimits (joins and unions); and
Prop.~\ref{prop:construct} gives a detailed description of the limit construction (join) for tables with fixed type domain,
arguing that limits are resolvable into substitutions followed by meets.

In the appendix \S\ref{sec:append}, 
we discuss relations, comma contexts and fibrations:
Prop.~\ref{prop:tbl:rel:refl} describes the reflection between tables and relations,
thus linking traditional logic interpretation with relational database interpretation; and
Facts.~\ref{fact:comma:lim}--\ref{fact:groth:adj:lim:colim}\; 
state facts about comma contexts and the Grothendieck construction.

%

\newpage
\subsection{The Presentation of \texttt{FOLE}}\label{sub:sec:present:FOLE}

The first-order logical environment \texttt{FOLE} 
(Fig.~\ref{fig:tbl:papers}:~{\scriptsize{\textbf{0}}})
was first described in Kent~\cite{kent:iccs2013}.  
A series of three papers
(Fig.~\ref{fig:tbl:papers}:~{\scriptsize{\textbf{1},\textbf{2},\textbf{5}}})
describe in detail a mathematical representation
for ontologies within {\ttfamily FOLE}.
The {\ttfamily FOLE} representation can be expressed in two forms: 
a classification form and interpretative form.
The foundation paper (Kent~\cite{kent:fole:era:found}) and 
the superstructure paper (Kent~\cite{kent:fole:era:supstruc}) 
developed the classification form of {\ttfamily FOLE}.
A third paper (Kent~\cite{kent:fole:era:interp}) will develop the interpretative form of {\ttfamily FOLE} 
as a transformational passage from sound logics (Kent~\cite{kent:iccs2013}),
\footnote{Following the relational model,
we assume a semantic structure and use a logical theory consistent with that structure 
in terms of first-order logic (E.F. Codd~\cite{codd:90}).}
thereby defining the formalism and semantics of first-order logical/relational database systems (Kent~\cite{kent:db:sem}).
A series of two papers
(Fig.~\ref{fig:tbl:papers}:~{\scriptsize{\textbf{3},\textbf{4}}})
provide a rigorous mathematical foundation for the interpretation of {\ttfamily FOLE}:
the first [this paper] describes the notion of a \texttt{FOLE} table and
the second describes the notion of a \texttt{FOLE} database.
System interoperability, 
in the general setting of institutions and logical environments,
was defined in the paper ``System Consequence'' (Kent~\cite{kent:iccs2009}). 
This was inspired by the channel theory of information flow 
(Barwise and Seligman \cite{barwise:seligman:97}).
Since {\ttfamily FOLE} is a logical environment (Kent~\cite{kent:fole:era:supstruc}),
in two further papers 
(Fig.~\ref{fig:tbl:papers}:~{\scriptsize{\textbf{6},\textbf{7}}})
we apply this approach to interoperability
for information systems based on first-order logic and relational databases:
one paper discusses integration over a fixed type domain and
the other paper discusses integration over a fixed universe.
%

%
\begin{figure}
\begin{center}
{\scriptsize{\setlength{\unitlength}{0.5pt}{
\begin{tabular}{c}
{\fbox{{\tiny{\textbf{0.}}}~\texttt{FOLE}
\cite{kent:iccs2013}}}
\\\\
{{\begin{tabular}{c}
$\underset{\textstyle{
\rule[5pt]{0pt}{10pt}
\underbrace{
\rule[-10pt]{0pt}{10pt}
\text{\fbox{\begin{tabular}{l@{\hspace{4pt}}l}
{\fbox{{\tiny{\textbf{3.}}}~\textbf{Table}}}
&
{\fbox{{\tiny{\textbf{4.}}}~\textbf{Database}}}
\end{tabular}}}}
}}
{\underbrace{
\rule[-11pt]{0pt}{10pt}
\text{\fbox{{
\begin{tabular}{l@{\hspace{4pt}}l@{\hspace{4pt}}l}
{\fbox{{\tiny{\textbf{1.}}}~\textbf{Foundation}
\cite{kent:fole:era:found}}}
&
{\fbox{{\tiny{\textbf{2.}}}~\textbf{Superstructure}
\cite{kent:fole:era:supstruc}}}
&
{\fbox{{\tiny{\textbf{5.}}}~\textbf{Interpretation}}}
\end{tabular}
}}}}}$
\end{tabular}}}
\\
\rule[15pt]{0pt}{10pt}
{\fbox{{$
\begin{array}{l}
{\textbf{System Interoperability:}}
\left\{\rule{0pt}{18pt}\right.
{\text{{\setlength{\unitlength}{0.5pt}\begin{tabular}{l}
{\fbox{{\tiny{\textbf{6.}}}~\textbf{Fixed Type Domain}}}
\\
{\fbox{{\tiny{\textbf{7.}}}~\textbf{Fixed Universe}}}
\rule[7pt]{0pt}{5pt}
\end{tabular}}}}
\end{array}$}}}
\end{tabular}}}}
\end{center}
\caption{\texttt{FOLE} Papers: Sequence and Dependency}
\label{fig:tbl:papers}
\end{figure}

\appendix

\section{Appendix}\label{sec:append}

\subsection{$\mathcal{A}$-Relations.}\label{append:sub:sec:rel:tbl}

Let $\mathcal{A} = {\langle{X,Y,\models_{\mathcal{A}}}\rangle}$ be a fixed type domain.
The mathematical contexts of $\mathcal{A}$-relations and $\mathcal{A}$-tables 
\footnote{For fixed type domain $\mathcal{A}$,
the context of $\mathcal{A}$-Tables
is discussed in \S~\ref{sub:sub:sec:tbl:typ:dom:low}.}
are used for satisfaction and interpretation
(\cite{kent:fole:era:found}),
$\mathcal{A}$-relations for traditional interpretation and 
$\mathcal{A}$-tables for database interpretation.


\paragraph{Fiber Contexts.}

Let ${\langle{I,s}\rangle}$ be any signature. 
The ${\langle{I,s}\rangle}^{\text{th}}$-fiber context of relations is the subset order
\[\mbox{\footnotesize{$
\mathrmbf{Rel}_{\mathcal{A}}(I,s)={\langle{{\wp}\mathrmbfit{tup}_{\mathcal{A}}(I,s),\subseteq}\rangle}
$.}\normalsize}\]
An object 
$R{\;\in\;}\mathrmbf{Rel}_{\mathcal{A}}(I,s)$
consists of 
a subset of tuples $R{\;\subseteq\;}\mathrmbfit{tup}_{\mathcal{A}}(I,s)$.
\footnote{More abstractly,
we could define a relation to be a \emph{subobject}
of ${\langle{I,s,\mathcal{A}}\rangle}$-tuples;
that is,
an isomorphism class of monomorphisms
$R\xhookrightarrow{i}\mathrmbfit{tup}_{\mathcal{A}}(I,s)$.
These correspond to the proper or uncorrupted relational tables of Codd~\cite{codd:90}.}
%
A morphism 
$R' \leftarrow R$
in $\mathrmbf{Rel}_{\mathcal{A}}(I,s)$
consists of subset order
$R' \supseteq R$.

\paragraph{Fibered Context.}

The fibered context $\mathrmbf{Rel}(\mathcal{A})$
has indexed $\mathcal{A}$-relations ${\langle{I,s,R}\rangle}$ as objects
with $R\subseteq\mathrmbfit{tup}_{\mathcal{A}}(I,s)$ and 
morphisms ${\langle{I',s',R'}\rangle}\xleftarrow{\;h\;}{\langle{I,s,R}\rangle}$
%
\footnote{We use this orientation to accord with both relational fibers and table morphisms.}
%
consisting of a signature morphism ${\langle{I',s'}\rangle}\xrightarrow{h}{\langle{I,s}\rangle}$
satisfying either of the adjoint fiber orderings 
\begin{equation}\label{def:rel:cxt}
{{\begin{picture}(120,0)(0,-4)
\put(60,0){\makebox(0,0){\footnotesize{$
\underset{\textstyle{\text{in}\;\mathrmbf{Rel}_{\mathcal{A}}(I',s')}}
{R'{\;\supseteq\;}\exists_{{h}}(R)}
\;\;\;\;\rightleftarrows\;\;\;\;
\underset{\textstyle{\text{in}\;\mathrmbf{Rel}_{\mathcal{A}}(I,s)}}
{{h}^{{\scriptscriptstyle-}1}(R'){\;\supseteq\;}{R}}
$}}}
\end{picture}}}
\end{equation}
defined in terms of 
the fiber adjunction 
${\langle{\exists_{{h}}{\;\dashv\;}{h}^{{\scriptscriptstyle-}1}}\rangle}
:\mathrmbf{Rel}_{\mathcal{A}}(I',s'){\;\leftrightarrows\;}\mathrmbf{Rel}_{\mathcal{A}}(I,s)$
(Tbl.~\ref{tbl:tbl-rel:refl:sml:shrt} in \S~\ref{sub:sub:sec:tbl:typ:dom:low}).
As we show below,
the context $\mathrmbf{Rel}(\mathcal{A})$ of $\mathcal{A}$-relations
can be viewed as a mathematical subcontext of
the context $\mathrmbf{Tbl}(\mathcal{A})$ of $\mathcal{A}$-tables.


\paragraph{Inclusion.}

Let ${\langle{I,s}\rangle}$ be any signature. 
The ${\langle{I,s}\rangle}^{\text{th}}$-fiber inclusion passage 
\footnote{For fixed signed domain ${\langle{I,s,\mathcal{A}}\rangle}$,
the fiber mathematical context of ${\langle{I,s,\mathcal{A}}\rangle}$-tables is is discussed in \S~\ref{sub:sec:tbl:sign:dom}.}
%
\[\mbox{\footnotesize{$
\mathrmbf{Rel}_{\mathcal{A}}(I,s)\xrightarrow{\;\mathrmbfit{inc}^{\mathcal{A}}_{{\langle{I,s}\rangle}}\;\;}\mathrmbf{Tbl}_{\mathcal{A}}(I,s)
$}\normalsize}\]
is defined as follows.
An fiber relation 
$R{\;\in\;}\mathrmbf{Rel}_{\mathcal{A}}(I,s)$
is mapped to the fiber table
${\langle{R,inc}\rangle}{\;\in\;}\mathrmbf{Tbl}_{\mathcal{A}}(I,s)$.
A fiber morphism $R'{\,\supseteq\,}R$ in $\mathrmbf{Rel}_{\mathcal{A}}(I,s)$
is mapped to the fiber morphism
${\langle{R',inc}\rangle}\xleftarrow{inc}{\langle{R,inc}\rangle}$ in $\mathrmbf{Tbl}_{\mathcal{A}}(I,s)$.
The fibered inclusion passage 
\[\mbox{\footnotesize{$
\mathrmbf{Rel}(\mathcal{A})\xrightarrow{\;\mathrmbfit{inc}_{\mathcal{A}}\;\;}\mathrmbf{Tbl}(\mathcal{A})
$}\normalsize}\]
can be defined in terms of the fiber passages 
$\{ \mathrmbfit{inc}^{\mathcal{A}}_{{\langle{I,s}\rangle}} \mid {\langle{I,s}\rangle}\in\mathrmbf{List}(X)\}$.
An $\mathcal{A}$-relation 
${\langle{I,s,R}\rangle}\in\mathrmbf{Rel}(\mathcal{A})$
\underline{is mapped to} 
the
$\mathcal{A}$-table 
${\langle{I,s,R,inc}\rangle}
={\langle{I,s,\mathrmbfit{inc}^{\mathcal{A}}_{(I,s)}(R)}\rangle}
{\;\in\;}\mathrmbf{Tbl}(\mathcal{A})$.
An $\mathcal{A}$-relation morphism 
${\langle{I',s',R'}\rangle}\xleftarrow{\;h\;}{\langle{I,s,R}\rangle}$
consisting of a signature morphism ${\langle{I',s'}\rangle}\xrightarrow{h}{\langle{I,s}\rangle}$
satisfying either of the adjoint fiber orderings in Eqn.~\ref{def:rel:cxt}
\underline{is mapped to} 
the
$\mathcal{A}$-table morphism 
${\langle{I',s',R',inc}\rangle}\xleftarrow{\langle{h,r}\rangle}{\langle{I,s,R,inc}\rangle}$,
where the key function $R'\xleftarrow{\,r\,}R$
satisfying the condition $r{\;\cdot\;}{inc}={inc}{\;\cdot\;}\mathrmbfit{tup}_{\mathcal{A}}(h)$
is a restriction of the tuple function
$\mathrmbfit{tup}_{\mathcal{A}}(I',s')\xleftarrow{\mathrmbfit{tup}_{\mathcal{A}}(h)}\mathrmbfit{tup}_{\mathcal{A}}(I,s)$.
Hence,
we have the adjointly-related fiber context morphisms (see Eqn.~\ref{def:tbl:cxt}).
\begin{equation}\label{inc:mor}
{{\begin{picture}(100,40)(0,-8)
\put(35,10){\makebox(0,0){\footnotesize{$
\overset{\textstyle{\mathrmbfit{inc}^{\mathcal{A}}_{{\langle{I',s'}\rangle}}\Bigl(
R'{\;\supseteq\;}\exists_{{h}}(R)
\Bigr)}}
{\overbrace{\mathrmbfit{inc}^{\mathcal{A}}_{(I',s')}(R')
\underset{\textstyle{\;\xleftarrow{\;\acute{r}\;}\exists_{{h}}(\mathrmbfit{inc}^{\mathcal{A}}_{(I,s)}(R))}}
{\xhookleftarrow{inc}\mathrmbfit{inc}^{\mathcal{A}}_{(I,s)}(\exists_{{h}}(R))}
}}
{\;\;\;\;\rightleftarrows\;\;\;\;}
\overset{\textstyle{\mathrmbfit{inc}^{\mathcal{A}}_{{\langle{I,s}\rangle}}\Bigl(
{h}^{{\scriptscriptstyle-}1}(R'){\;\supseteq\;}{R}
\Bigr)}}
{\overbrace{
{h}^{{\scriptscriptstyle-}1}(\mathrmbfit{inc}^{\mathcal{A}}_{(I',s')}(R'))
\xhookleftarrow{inc}
\mathrmbfit{inc}^{\mathcal{A}}_{(I,s)}(R)}}
$}}}
\end{picture}}}
\end{equation}
Either pullback or image factorization can be used
(Tbl.~\ref{fig:inc:tbl:mor})
to define the key function $R'\xleftarrow{\,r\,}R$.
Using pullback,
the $\mathcal{A}$-table morphism is the composition of
the fiber morphism
$\exists_{{h}}\Bigl({h}^{{\scriptscriptstyle-}1}(\mathrmbfit{inc}^{\mathcal{A}}_{(I',s')}(R'))
\xhookleftarrow{inc}\mathrmbfit{inc}^{\mathcal{A}}_{(I,s)}(R)\Bigr)$
with 
the $\mathrmbfit{inc}^{\mathcal{A}}_{(I',s')}(R')^{\text{th}}$ counit component 
$\mathrmbfit{inc}^{\mathcal{A}}_{(I',s')}(R')
\xleftarrow{\grave{r}}\exists_{{h}}\Bigl({h}^{{\scriptscriptstyle-}1}(\mathrmbfit{inc}^{\mathcal{A}}_{(I',s')}(R'))\Bigr)$
for the fiber adjunction
${\langle{\exists_{{h}}{\;\dashv\;}{h}^{{\scriptscriptstyle-}1}}\rangle}
:\mathrmbf{Tbl}_{\mathcal{A}}(I',s'){\;\leftrightarrows\;}\mathrmbf{Tbl}_{\mathcal{A}}(I,s)$
(Tbl.~\ref{tbl:tbl-rel:refl:sml:shrt} in \S~\ref{sub:sub:sec:tbl:typ:dom:low}).
\begin{figure}
\begin{center}
{{\begin{tabular}{c}
{\setlength{\unitlength}{0.65pt}\begin{picture}(240,160)(0,-35)
\put(-40,110){\makebox(0,0){\footnotesize{$R'$}}}
\put(-20,95){\makebox(0,0){\footnotesize{$\supseteq$}}}
\put(0,80){\makebox(0,0){\footnotesize{$\exists_{h}({h}^{{\scriptscriptstyle-}1}(R'))$}}}
\put(20,65){\makebox(0,0){\footnotesize{$\supseteq$}}}
\put(40,50){\makebox(0,0){\footnotesize{$\exists_{h}(R)$}}}
\put(280,110){\makebox(0,0){\footnotesize{$R$}}}
\put(260,95){\makebox(0,0){\footnotesize{$\supseteq$}}}
\put(240,80){\makebox(0,0){\footnotesize{${h}^{{\scriptscriptstyle-}1}(\exists_{h}(R))$}}}
\put(220,65){\makebox(0,0){\footnotesize{$\supseteq$}}}
\put(200,50){\makebox(0,0){\footnotesize{${h}^{{\scriptscriptstyle-}1}(R')$}}}
\put(0,0){\makebox(0,0){\footnotesize{$\mathrmbfit{tup}_{\mathcal{A}}(I',s')$}}}
\put(240,0){\makebox(0,0){\footnotesize{$\mathrmbfit{tup}_{\mathcal{A}}(I,s)$}}}
\put(120,120){\makebox(0,0){\scriptsize{$r$}}}
\put(120,-10){\makebox(0,0){\scriptsize{$\mathrmbfit{tup}_{\mathcal{A}}(h)$}}}
\put(255,110){\vector(-1,0){270}}
\put(190,0){\vector(-1,0){140}}
\put(0,65){\vector(0,-1){50}}
\put(240,65){\vector(0,-1){50}}
\put(278,96){\vector(-1,-4){20}}
\put(215,40){\vector(3,-4){20}}
\put(-38,96){\vector(1,-4){20}}
\put(35,40){\vector(-3,-4){20}}
\qbezier(72,57)(155,90)(240,100)\put(72,57){\vector(-3,-1){0}}
\put(70,77){\makebox(0,0){\tiny{\itshape{pullback}}}}
\put(70,95){\makebox(0,0){\scriptsize{$\grave{r}$}}}
\put(170,77){\makebox(0,0){\tiny{\itshape{factor}}}}
\put(170,95){\makebox(0,0){\scriptsize{$\acute{r}$}}}
\qbezier(0,100)(85,90)(168,57)\put(0,100){\vector(-4,1){0}}
\put(-39,-26){\makebox(0,0){\footnotesize{$\underset{\mathrmbfit{inc}^{\mathcal{A}}_{(I',s')}(R')}{\underbrace{\rule{44pt}{0pt}}}$}}}
\put(144,-52){\makebox(0,0){\footnotesize{$\underset{\exists_{{h}}(\mathrmbfit{inc}^{\mathcal{A}}_{(I,s)}(R))}{\underbrace{\rule{182pt}{0pt}}}$}}}
\put(201,-26){\makebox(0,0){\footnotesize{$\underset{{h}^{{\scriptscriptstyle-}1}(\mathrmbfit{inc}^{\mathcal{A}}_{(I',s')}(R'))}{\underbrace{\rule{44pt}{0pt}}}$}}}
\put(279,-26){\makebox(0,0){\footnotesize{$\underset{\mathrmbfit{inc}^{\mathcal{A}}_{(I,s)}(R)}{\underbrace{\rule{44pt}{0pt}}}$}}}
\end{picture}}
\end{tabular}}}
\end{center}
\caption{Inclusion Table Morphism}
\label{fig:inc:tbl:mor}
\end{figure}
%


\paragraph{Image.}

Let ${\langle{I,s}\rangle}$ be any signature. 
The ${\langle{I,s}\rangle}^{\text{th}}$-fiber image passage 
\[\mbox{\footnotesize{$
\mathrmbf{Tbl}_{\mathcal{A}}(I,s)\xrightarrow{\;\mathrmbfit{im}^{\mathcal{A}}_{{\langle{I,s}\rangle}}\;\;}\mathrmbf{Rel}_{\mathcal{A}}(I,s)
$}\normalsize}\]
is defined as follows.
A fiber table
${\langle{K,t}\rangle}{\;\in\;}\mathrmbf{Tbl}_{\mathcal{A}}(I,s)$
is mapped to the fiber relation 
${\wp{t}}(K){\;\in\;}\mathrmbf{Rel}_{\mathcal{A}}(I,s)$.
A fiber morphism
${\langle{K',t'}\rangle}\xleftarrow{k}{\langle{K,t}\rangle}$ in $\mathrmbf{Tbl}_{\mathcal{A}}(I,s)$
is mapped to the fiber morphism ${\wp{t'}}(K'){\,\supseteq\,}{\wp{t}}(K)$ in $\mathrmbf{Rel}_{\mathcal{A}}(I,s)$
guaranteed by the table morphism condition $k{\;\cdot\;}t'=t$.
The fibered image passage 
\[\mbox{\footnotesize{$
\mathrmbf{Tbl}(\mathcal{A})\xrightarrow{\;\mathrmbfit{im}_{\mathcal{A}}\;\;}\mathrmbf{Rel}(\mathcal{A})
$}\normalsize}\]
can be defined in terms of the fiber image passages 
$\{ \mathrmbfit{im}^{\mathcal{A}}_{{\langle{I,s}\rangle}} \mid {\langle{I,s}\rangle}\in\mathrmbf{List}(X)\}$.
An $\mathcal{A}$-table 
${\langle{I,s,K,t}\rangle}\in\mathrmbf{Tbl}(\mathcal{A})$
with signature ${\langle{I,s}\rangle}$
and table ${\langle{K,t}\rangle}\in\mathrmbf{Tbl}_{\mathcal{A}}(I,s)$
\underline{is mapped to} 
the $\mathcal{A}$-relation 
${\langle{I,s,{\wp{t}}(K)}\rangle}\in\mathrmbf{Rel}(\mathcal{A})$
with the same signature 
and the relation
${\wp{t}}(K)=\mathrmbfit{im}_{{\langle{I,s}\rangle}}^{\mathcal{A}}(K,t)\in\mathrmbf{Rel}_{\mathcal{A}}(I,s)$.
An $\mathcal{A}$-table morphism 
$\mathcal{T}'={\langle{I',s',K',t'}\rangle}\xleftarrow{\langle{h,k}\rangle}{\langle{I,s,K,t}\rangle}=\mathcal{T}$
consisting of signature morphism ${\langle{I',s'}\rangle}\xrightarrow{\,h\,}{\langle{I,s}\rangle}$
satisfying either of the adjoint fiber orderings in Eqn.~\ref{def:tbl:cxt}
\underline{is mapped to} 
the $\mathcal{A}$-relation morphism 
$\mathrmbfit{im}_{\mathcal{A}}(\mathcal{T}')={\langle{I',s',R'}\rangle}
\xleftarrow{h}
{\langle{I,s,R}\rangle}=\mathrmbfit{im}_{\mathcal{A}}(\mathcal{T})$
with the same signature morphism
and satisfying either of the adjoint fiber orderings 
\[\mbox{\footnotesize{$
\underset{\textstyle{\text{in}\;\mathrmbf{Rel}_{\mathcal{A}}(I',s')}}
{\mathrmbfit{im}_{{\langle{I',s'}\rangle}}^{\mathcal{A}}(K',t')
{\;\supseteq}
\overset{\textstyle{\;\;\mathrmbfit{im}_{{\langle{I',s'}\rangle}}^{\mathcal{A}}(\exists_{h}(K,t))}}
{\overbrace{\exists_{h}(\mathrmbfit{im}_{{\langle{I,s}\rangle}}^{\mathcal{A}}(K,t))}}}
{\;\;\;\rightleftarrows\;\;\;}
\underset{\textstyle{\text{in}\;\mathrmbf{Rel}_{\mathcal{A}}(I,s)}}
{\overset{\textstyle{\mathrmbfit{im}_{{\langle{I,s}\rangle}}^{\mathcal{A}}({h}^{{\scriptscriptstyle-}1}(K',t'))}}
{\overbrace{{h}^{{\scriptscriptstyle-}1}(\mathrmbfit{im}_{{\langle{I',s'}\rangle}}^{\mathcal{A}}(K',t'))}}
{\;\supseteq\;}
\mathrmbfit{im}_{{\langle{I,s}\rangle}}^{\mathcal{A}}(K,t)}
$.}\normalsize}\]
%

\paragraph{Reflection.}

The inclusion passage 
$\mathrmbf{Rel}(\mathcal{A})\xrightarrow{\;\mathrmbfit{inc}_{\mathcal{A}}\;}\mathrmbf{Tbl}(\mathcal{A})$ is full.
The composite passage 
$\mathrmbf{Rel}(\mathcal{A})\xrightarrow{\;\mathrmbfit{inc}_{\mathcal{A}}{\;\circ\;}\mathrmbfit{im}_{\mathcal{A}}\;}\mathrmbf{Rel}(\mathcal{A})$
is the identity passage.
\begin{definition}
There is an image-factorization bridge
$\mathrm{1}_{\mathrmbf{Tbl}(\mathcal{A})}\xRightarrow{\;\eta\;}\mathrmbfit{im}_{\mathcal{A}}{\;\circ\;}\mathrmbfit{inc}_{\mathcal{A}}$.
\end{definition}
%
The $\mathcal{T}^{\text{th}}$-component
for $\mathcal{A}$-table $\mathcal{T}={\langle{I,s,K,t}\rangle}$
is the $\mathcal{A}$-table morphism
$\mathcal{T}\xrightarrow[{\langle{1,e}\rangle}]{\eta_{\mathcal{T}}}\mathrmbfit{inc}_{\mathcal{A}}(\mathrmbfit{im}_{\mathcal{A}}(\mathcal{T}))$,
where $K\xrightarrow{e}{\wp{t}}(K)\xrightarrow{inc}\mathrmbfit{tup}_{\mathcal{A}}(I,s)$
is the image factorization of the tuple function $K \xrightarrow{t} \mathrmbfit{tup}_{\mathcal{A}}(I,s)$.
The naturality diagram
\begin{center}
{{\begin{tabular}{c}
{\setlength{\unitlength}{0.5pt}\begin{picture}(160,150)(0,-15)
\put(0,120){\makebox(0,0){\footnotesize{$K'$}}}
\put(0,60){\makebox(0,0){\footnotesize{${\wp{t'}}(K)'$}}}
\put(160,120){\makebox(0,0){\footnotesize{$K$}}}
\put(160,60){\makebox(0,0){\footnotesize{${\wp{t}}(K)$}}}
\put(0,0){\makebox(0,0){\footnotesize{$\mathrmbfit{tup}_{\mathcal{A}}(I',s')$}}}
\put(160,0){\makebox(0,0){\footnotesize{$\mathrmbfit{tup}_{\mathcal{A}}(I,s)$}}}
\put(80,130){\makebox(0,0){\scriptsize{$k$}}}
\put(80,70){\makebox(0,0){\scriptsize{$r$}}}
\put(80,-10){\makebox(0,0){\scriptsize{$\mathrmbfit{tup}_{\mathcal{A}}(h)$}}}
\put(145,120){\vector(-1,0){130}}
\put(130,60){\vector(-1,0){100}}
\put(110,0){\vector(-1,0){60}}
\put(0,105){\vector(0,-1){30}}
\put(160,105){\vector(0,-1){30}}
\put(0,42){\vector(0,-1){27}}\put(4,43){\makebox(0,0){\scriptsize{$\cap$}}}
\put(160,42){\vector(0,-1){27}}\put(164,43){\makebox(0,0){\scriptsize{$\cap$}}}
\put(-40,30){\makebox(0,0)[r]{\footnotesize{$\mathrmbfit{inc}_{\mathcal{A}}(\mathrmbfit{im}_{\mathcal{A}}(\mathcal{T}'))
\left\{\rule{0pt}{24pt}\right.$}}}
\put(200,30){\makebox(0,0)[l]{\footnotesize{$\left.\rule{0pt}{24pt}\right\}
\mathrmbfit{inc}_{\mathcal{A}}(\mathrmbfit{im}_{\mathcal{A}}(\mathcal{T}))$}}}
%
%
\end{picture}}
\end{tabular}}}
\end{center}
factors the condition $k{\;\cdot\;}t'=t{\;\cdot\;}\mathrmbfit{tup}_{\mathcal{A}}(h)$ by diagonal fill-in.
This gives the $\mathcal{A}$-table morphism 
$\mathrmbfit{inc}_{\mathcal{A}}(\mathrmbfit{im}_{\mathcal{A}}(\mathcal{T}'))
\xleftarrow{\langle{h,r}\rangle}
\mathrmbfit{inc}_{\mathcal{A}}(\mathrmbfit{im}_{\mathcal{A}}(\mathcal{T}))$,
which is the image-inclusion composite passage 
applied to 
the $\mathcal{A}$-table morphism 
$\mathcal{T}'
\xleftarrow{\langle{h,k}\rangle}
\mathcal{T}$.
%
\begin{proposition}\label{prop:tbl:rel:refl}
There is a reflection
$\mathrmbf{Tbl}_{\mathcal{A}}
\;\xrightarrow{{\langle{\mathrmbfit{im}_{\mathcal{A}}{\;\dashv\;}\mathrmbfit{inc}_{\mathcal{A}}}\rangle}}\;
\mathrmbf{Rel}_{\mathcal{A}}$.
\end{proposition}
This reflection embodies the notion of informational equivalence.

\subsection{Comma Contexts}\label{sub:sec:comma:cxt}

\begin{fact}\label{fact:comma:lim}
Let $\mathrmbf{A}\xrightarrow{\mathrmbfit{L}}\mathrmbf{C}\xleftarrow{\mathrmbfit{R}}\mathrmbf{B}$ be a passage opspan
with both $\mathrmbf{A}\xrightarrow{\mathrmbfit{L}}\mathrmbf{C}$ and $\mathrmbf{B}\xrightarrow{\mathrmbfit{R}}\mathrmbf{C}$ continuous passages.
%
\footnote{A passage $\mathrmbf{C}\xrightarrow{\mathrmbfit{F}}\mathrmbf{D}$ is \emph{continuous} 
when it preserves all small limits that exist in $\mathrmbf{C}$.}
%
If $\mathrmbf{A}$ and $\mathrmbf{B}$ are complete contexts,
then the comma context $\bigl(\mathrmbfit{L}{\,\downarrow\,}\mathrmbfit{R}\bigr)$ is complete and 
the projection passages 
$\mathrmbf{A}\leftarrow\bigl(\mathrmbfit{L}{\,\downarrow\,}\mathrmbfit{R}\bigr)\rightarrow\mathrmbf{B}$ 
are continuous.
\end{fact}

\begin{fact}\label{fact:comma:colim}
Let $\mathrmbf{A}\xrightarrow{\mathrmbfit{L}}\mathrmbf{C}\xleftarrow{\mathrmbfit{R}}\mathrmbf{B}$ be a passage opspan
with $\mathrmbf{A}\xrightarrow{\mathrmbfit{L}}\mathrmbf{C}$ cocontinuous.
%
\footnote{A passage $\mathrmbf{C}\xrightarrow{\mathrmbfit{F}}\mathrmbf{D}$ is \emph{cocontinuous} 
when it preserves all small colimits that exist in $\mathrmbf{C}$. 
A passage $\mathrmbf{C}\xrightarrow{\mathrmbfit{F}}\mathrmbf{D}$ is cocontinuous \underline{iff} 
the opposite passage $\mathrmbf{C}^{\mathrm{op}}\xrightarrow{\mathrmbfit{F}^{\mathrm{op}}}\mathrmbf{D}^{\mathrm{op}}$
between opposite contexts is a continuous passage.}
%
If $\mathrmbf{A}$ and $\mathrmbf{B}$ are cocomplete contexts,
then the comma context $\bigl(\mathrmbfit{L}{\,\downarrow\,}\mathrmbfit{R}\bigr)$ is cocomplete and 
the projection passages 
$\mathrmbf{A}\leftarrow\bigl(\mathrmbfit{L}{\,\downarrow\,}\mathrmbfit{R}\bigr)\rightarrow\mathrmbf{B}$ 
are cocontinuous.
\end{fact}


\subsection{The Grothendieck Construction}
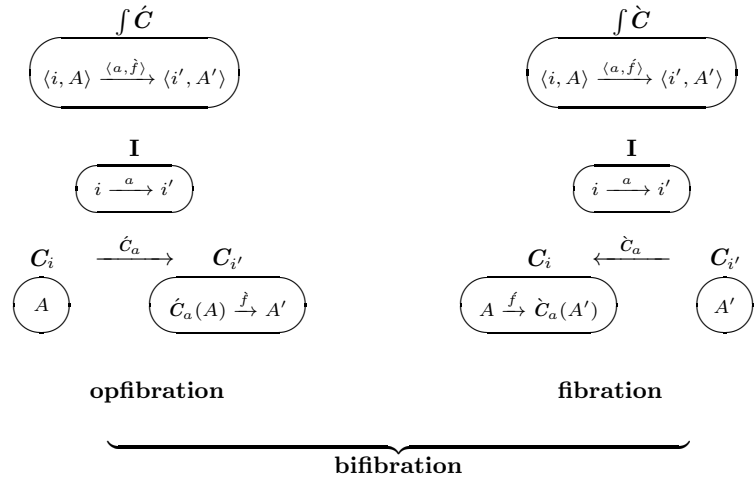
\label{sub:sec:ind:fbr}

%
\begin{figure}
\begin{center}
{{\begin{tabular}{c@{\hspace{40pt}}c}
{{\begin{tabular}{c}
\setlength{\unitlength}{0.44pt}
\begin{picture}(260,300)(33,0)
\put(140,200){\begin{picture}(0,0)(0,0)
\put(0,30){\oval(180,60)}
\put(0,75){\makebox(0,0){\footnotesize{$\int{\!\acute{\mathrmbfit{C}}}$}}}
\put(0,30){\makebox(0,0){\scriptsize{${\langle{i,A}\rangle}\xrightarrow{{\langle{a,\grave{f}}\rangle}}{\langle{i',A'}\rangle}$}}}
\end{picture}}
\put(140,130){\begin{picture}(0,0)(0,0)
\put(0,35){\makebox(0,0){\footnotesize{$\mathrmbf{I}$}}}
\put(0,3){\makebox(0,0){\scriptsize{$i\xrightarrow{\;\;a\;\;}i'$}}}
\put(0,0){\oval(100,40)}
\end{picture}}
\put(60,30){\begin{picture}(0,0)(0,0)
\put(160,0){\oval(134,48)}
\put(160,40){\makebox(0,0){\footnotesize{$\mathrmbfit{C}_{i'}$}}}
\put(160,0){\makebox(0,0){\scriptsize{$\acute{\mathrmbfit{C}}_{a}(A)\xrightarrow{\grave{f}}A'$}}}
\put(80,48){\makebox(0,0){\footnotesize{$\xrightarrow{\;\;\;\acute{\mathrmbfit{C}}_{a}\;\;\;\;}$}}}
\put(0,0){\oval(48,48)}
\put(0,40){\makebox(0,0){\footnotesize{$\mathrmbfit{C}_{i}$}}}
\put(0,0){\makebox(0,0){\scriptsize{$A$}}}
\end{picture}}
\end{picture}
\end{tabular}}}
&
{{\begin{tabular}{c}
\setlength{\unitlength}{0.44pt}
\begin{picture}(260,300)(-5,0)
\put(140,200){\begin{picture}(0,0)(0,0)
\put(0,30){\oval(180,60)}
\put(0,75){\makebox(0,0){\footnotesize{$\int{\!\grave{\mathrmbfit{C}}}$}}}
\put(0,30){\makebox(0,0){\scriptsize{${\langle{i,A}\rangle}\xrightarrow{{\langle{a,\acute{f}}\rangle}}{\langle{i',A'}\rangle}$}}}
\end{picture}}
\put(140,130){\begin{picture}(0,0)(0,0)
\put(0,35){\makebox(0,0){\footnotesize{$\mathrmbf{I}$}}}
\put(0,3){\makebox(0,0){\scriptsize{$i\xrightarrow{\;\;a\;\;}i'$}}}
\put(0,0){\oval(100,40)}
\end{picture}}
\put(60,30){\begin{picture}(0,0)(0,0)
\put(0,0){\oval(134,48)}
\put(0,40){\makebox(0,0){\footnotesize{$\mathrmbfit{C}_{i}$}}}
\put(0,0){\makebox(0,0){\scriptsize{$A\xrightarrow{\acute{f}}\grave{\mathrmbfit{C}}_{a}(A')$}}}
\put(80,48){\makebox(0,0){\footnotesize{$\xleftarrow{\;\;\;\grave{\mathrmbfit{C}}_{a}\;\;\;\;}$}}}
\put(160,0){\oval(48,48)}
\put(160,40){\makebox(0,0){\footnotesize{$\mathrmbfit{C}_{i'}$}}}
\put(160,0){\makebox(0,0){\scriptsize{$A'$}}}
\end{picture}}
\end{picture}
\end{tabular}}}
\\&\\
{\small\bfseries{opfibration}} & {\small\bfseries{fibration}}
\\&\\
\multicolumn{2}{c}{\small{$\overset{\underbrace{\rule{220pt}{0pt}}}{\text{\small\bfseries{bifibration}}}$}}
\\&\\
\multicolumn{2}{c}{\scriptsize{\begin{tabular}{p{300pt}}
The adjunction
$\mathrmbfit{C}_{i}\xrightarrow{{\langle{\acute{\mathrmbfit{C}}_{a}{\;\dashv\;}\grave{\mathrmbfit{C}}_{a}}\rangle}}\mathrmbfit{C}_{i'}$
has unit $\mathrmbfit{1}_{\mathrmbfit{C}_{i}}\xRightarrow{\eta_{a}}\acute{\mathrmbfit{C}}_{a}{\;\circ\;}\grave{\mathrmbfit{C}}_{a}$
with the $\mathrmbfit{C}_{i}$-morphism $A\xrightarrow{\eta_{a}(A)}\grave{\mathrmbfit{C}}_{a}(\acute{\mathrmbfit{C}}_{a}(A))$
as its $A^{\mathrm{th}}$ component, and
has counit $\grave{\mathrmbfit{C}}_{a}{\;\circ\;}\acute{\mathrmbfit{C}}_{a}\xRightarrow{\varepsilon_{a}}\mathrmbfit{1}_{\mathrmbfit{C}_{i'}}$
with the $\mathrmbfit{C}_{i'}$-morphism $\acute{\mathrmbfit{C}}_{a}(\grave{\mathrmbfit{C}}_{a}(A'))\xrightarrow{\varepsilon_{a}(A')}A'$
as its $A'^{\mathrm{th}}$ component.
\begin{center}
{\scriptsize\setlength{\extrarowheight}{3pt}$\begin{array}{r@{\hspace{5pt}=\hspace{5pt}}l}
A\xrightarrow{\acute{f}}\grave{\mathrmbfit{C}}_{a}(A')
&
A\xrightarrow{\eta_{a}(A)}\grave{\mathrmbfit{C}}_{a}(\acute{\mathrmbfit{C}}_{a}(A))
\xrightarrow{\grave{\mathrmbfit{C}}_{a}(\grave{f})}\grave{\mathrmbfit{C}}_{a}(A')
\\
\acute{\mathrmbfit{C}}_{a}(A)\xrightarrow{\grave{f}}A'
&
\acute{\mathrmbfit{C}}_{a}(A)\xrightarrow{\acute{\mathrmbfit{C}}_{a}(\acute{f})}
\acute{\mathrmbfit{C}}_{a}(\grave{\mathrmbfit{C}}_{a}(A'))\xrightarrow{\varepsilon_{a}(A')}A'
\end{array}$}
\end{center}
\end{tabular}}}
\end{tabular}}}
\end{center}
\caption{Bifibration}
\label{fig:bi:fbr}
\end{figure}
\begin{description}
\item[fibration:] 
A fibration (fibered context) $\int{\!\grave{\mathrmbfit{C}}}$
is the Grothendieck construction of a contravariant pseudo-passage (indexed context)
$\mathrmbf{I}^{\mathrm{op}}\xrightarrow{\grave{\mathrmbfit{C}}}\mathrmbf{Cxt}$,
where the action on any indexing object $i$ in $\mathrmbf{I}$ is the fiber context
$\grave{\mathrmbfit{C}}_{i}$
and
the action on any indexing morphism $i\xrightarrow{a}i'$ is the fiber passage
$\mathrmbfit{C}_{i}\xleftarrow{\grave{\mathrmbfit{C}}_{a}}\mathrmbfit{C}_{i'}$. 
An object in $\int{\!\grave{\mathrmbfit{C}}}$ is a pair ${\langle{i,A}\rangle}$,
where $i$ is an indexing object in $\mathrmbf{I}$ and $A$ is an object in the fiber context $\grave{\mathrmbfit{C}}_{i}$.
A morphism in $\int{\!\grave{\mathrmbfit{C}}}$ is a pair 
${\langle{i,A}\rangle}\xrightarrow{{\langle{a,\acute{f}}\rangle}}{\langle{i',A'}\rangle}$,
where $i\xrightarrow{a}i'$ is an indexing morphism in $\mathrmbf{I}$ 
and $A\xrightarrow{\acute{f}}\grave{\mathrmbfit{C}}_{a}(A')$ is a fiber morphism in $\grave{\mathrmbfit{C}}_{i}$. 
There is a projection passage
$\int{\!\grave{\mathrmbfit{C}}}\rightarrow\mathrmbf{I}$.
\newline
\item[opfibration:] 
An opfibration $\int{\!\acute{\mathrmbfit{C}}}$
is the Grothendieck construction of a covariant pseudo-passage (indexed context)
$\mathrmbf{I}\xrightarrow{\acute{\mathrmbfit{C}}}\mathrmbf{Cxt}$,
where the action on any indexing object $i$ in $\mathrmbf{I}$ is the fiber context
$\acute{\mathrmbfit{C}}_{i}$
and
the action on any indexing morphism $i\xrightarrow{a}i'$ is the fiber passage
$\mathrmbfit{C}_{i}\xrightarrow{\acute{\mathrmbfit{C}}_{a}}\mathrmbfit{C}_{i'}$. 
An object in $\int{\!\acute{\mathrmbfit{C}}}$ is a pair ${\langle{i,A}\rangle}$,
where $i$ is an indexing object in $\mathrmbf{I}$ and $A$ is an object in the fiber context $\acute{\mathrmbfit{C}}_{i}$.
A morphism in $\int{\!\acute{\mathrmbfit{C}}}$ is a pair 
${\langle{i,A}\rangle}\xrightarrow{{\langle{a,\grave{f}}\rangle}}{\langle{i',A'}\rangle}$,
where $i\xrightarrow{a}i'$ is an indexing morphism in $\mathrmbf{I}$ 
and $\acute{\mathrmbfit{C}}_{a}(A)\xrightarrow{\grave{f}}A'$ is a fiber morphism in $\acute{\mathrmbfit{C}}_{i'}$. 
There is a projection passage
$\int{\!\acute{\mathrmbfit{C}}}\rightarrow\mathrmbf{I}$.
\newline
\item[bifibration:] 
A bifibration $\int{\!{\mathrmbfit{C}}}$
(Fig.\ref{fig:bi:fbr})
is the Grothendieck construction of an indexed adjunction
$\mathrmbf{I}\xrightarrow{{\mathrmbfit{C}}}\mathrmbf{Adj}$
consisting of 
a left adjoint covariant pseudo-passage
$\mathrmbf{I}\xrightarrow{\acute{\mathrmbfit{C}}}\mathrmbf{Cxt}$
and a right adjoint contravariant pseudo-passage
$\mathrmbf{I}^{\mathrm{op}}\xrightarrow{\grave{\mathrmbfit{C}}}\mathrmbf{Cxt}$. 
The action on any indexing object $i$ in $\mathrmbf{I}$ is the fiber context
$\mathrmbfit{C}_{i} = \grave{\mathrmbfit{C}}_{i} = \acute{\mathrmbfit{C}}_{i}$
and
the action on any indexing morphism $i\xrightarrow{a}i'$ is the fiber adjunction
$\bigl(\mathrmbfit{C}_{i}\xrightarrow{\mathrmbfit{C}_{a}}\mathrmbfit{C}_{i'}\bigr) = 
\bigl(\mathrmbfit{C}_{i}\xrightarrow{{\langle{\acute{\mathrmbfit{C}}_{a}{\;\dashv\;}\grave{\mathrmbfit{C}}_{a}}\rangle}}\mathrmbfit{C}_{i'}\bigr)$. 
The Grothendieck constructions of component fibration and component opfibration are isomorphic
$\int{\!\grave{\mathrmbfit{C}}}{\;\cong\;}\int{\!\acute{\mathrmbfit{C}}}$
\[\mbox{\footnotesize{$
\bigl({\langle{i,A}\rangle}\xrightarrow{{\langle{a,\acute{f}}\rangle}}{\langle{i',A'}\rangle}\bigr)
{\;\;\;\overset{\cong}{\rightleftarrows}\;\;\;}
\bigl({\langle{i,A}\rangle}\xrightarrow{{\langle{a,\grave{f}}\rangle}}{\langle{i',A'}\rangle}\bigr)
$}\normalsize}\]
via (Fig.~\ref{fig:bi:fbr}) the adjoint pair
$A\xrightarrow{\acute{f}}\grave{\mathrmbfit{C}}_{a}(A'){\;\cong\;}
\acute{\mathrmbfit{C}}_{a}(A)\xrightarrow{\grave{f}}A'$.
Define the Grothendieck construction of the bifibration to be the Grothendieck construction of component fibration
$\int{\!\mathrmbfit{C}}{\;\doteq}\int{\!\grave{\mathrmbfit{C}}}$
with projection 
$\int{\!\mathrmbfit{C}}\rightarrow\mathrmbf{I}$.
\end{description}
%


%
\begin{fact}\label{fact:groth:lim}
If\, $\mathrmbf{I}^{\mathrm{op}}\!\xrightarrow{\mathrmbfit{C}}\mathrmbf{Cxt}$ 
is a contravariant pseudo-passage (indexed context)
s.t.
\begin{enumerate}
\item 
the indexing context $\mathrmbf{I}$ is complete,
\item 
the fiber context $\mathrmbf{C}_{i}$ is complete for each $i\in\mathrmbf{I}$, and
\item 
the fiber passage $\mathrmbf{C}_{i}\xleftarrow{\mathrmbfit{C}_{a}}\mathrmbf{C}_{j}$ is continuous for each $i\xrightarrow{a}j$ in $\mathrmbf{I}$,
\end{enumerate}
then the fibered context (Grothendieck construction) $\int\mathrmbf{C}$ is complete 
and the projection $\int\mathrmbf{C}\xrightarrow{\mathrmbfit{P}}\mathrmbf{I}$ is continuous.
\end{fact}
\begin{proof}
Tarlecki, Burstall and Goguen~\cite{tarlecki:burstall:goguen:91}.
\end{proof}


%
\begin{fact}\label{fact:groth:colim}
If\, $\mathrmbf{I}\!\xrightarrow{\mathrmbfit{C}}\mathrmbf{Cxt}$ 
is a covariant pseudo-passage (indexed context)
s.t.
\begin{enumerate}
\item 
the indexing context $\mathrmbf{I}$ is cocomplete,
\item 
the fiber context $\mathrmbf{C}_{i}$ is cocomplete for each $i\in\mathrmbf{I}$, and
\item 
the fiber passage $\mathrmbf{C}_{i}\xrightarrow{\mathrmbfit{C}_{a}}\mathrmbf{C}_{j}$ is cocontinuous for each $i\xrightarrow{a}j$ in $\mathrmbf{I}$,
\end{enumerate}
then the fibered context (Grothendieck construction) $\int\mathrmbf{C}$ is cocomplete 
and the projection $\int\mathrmbf{C}\xrightarrow{\mathrmbfit{P}}\mathrmbf{I}$ is cocontinuous.
\end{fact}
\begin{proof}
Dual to the above.
\mbox{}\hfill\rule{5pt}{5pt}
\end{proof}
%


%
\begin{fact}\label{fact:groth:adj:lim:colim}
If\, $\mathrmbf{I}\!\xrightarrow{\mathrmbfit{C}}\mathrmbf{Adj}$ 
is an indexed adjunction consisting of 
a contravariant pseudo-passage 
$\mathrmbf{I}^{\mathrm{op}}\xrightarrow{\grave{\mathrmbfit{C}}}\mathrmbf{Cxt}$
and
a covariant pseudo-passage
$\mathrmbf{I}\!\xrightarrow{\acute{\mathrmbfit{C}}}\mathrmbf{Cxt}$
that are locally adjunctive
$\bigl(
\mathrmbfit{C}_{i}\xrightarrow{{\langle{\acute{\mathrmbfit{C}}_{a}{\;\dashv\;}\grave{\mathrmbfit{C}}_{a}}\rangle}}\mathrmbfit{C}_{i'}
\bigr)$
for each $i\xrightarrow{a}j$ in $\mathrmbf{I}$,
s.t.
\begin{enumerate}
\item 
the indexing context $\mathrmbf{I}$ is complete and cocomplete,
\item 
the fiber context $\mathrmbf{C}_{i}$ is complete and cocomplete for each $i\in\mathrmbf{I}$,
\end{enumerate}
then the fibered context (Grothendieck construction) $\int\mathrmbf{C}\rightarrow\mathrmbf{I}$ is complete and cocomplete
and the projection
$\int\mathrmbf{C}\rightarrow\mathrmbf{I}$
is continuous and cocontinuous.
\end{fact}
\begin{proof}
Use Facts.~\ref{fact:groth:lim}~\&~\ref{fact:groth:colim},
since the fiber passage $\mathrmbf{C}_{i}\xleftarrow{\grave{\mathrmbfit{C}}_{a}}\mathrmbf{C}_{i'}$ is continuous (being right adjoint)
and the fiber passage $\mathrmbf{C}_{i}\xrightarrow{\acute{\mathrmbfit{C}}_{a}}\mathrmbf{C}_{i'}$ is cocontinuous (being left adjoint)
for each $i\xrightarrow{a}j$ in $\mathrmbf{I}$.
\mbox{}\hfill\rule{5pt}{5pt}
\end{proof}
%



\end{document}